%% file: thesis.tex
\documentclass[defaultstyle,11pt]{thesis}

\usepackage{amssymb}		
\usepackage{graphicx}  
\usepackage{bmpsize}
\usepackage{hyperref}		
\usepackage{paralist}
\usepackage{amsmath,amsopn}
\usepackage{xcolor}
\usepackage{tikz,pgf}
\usepackage{caption,subcaption}
\usepackage{wrapfig}
\usepackage{ntheorem}
\usepackage{pifont}
\usepackage{multirow}
\usepackage{morefloats}

\newtheorem{theorem}{Theorem}

\newtheorem{defi}{Theorem}
\newtheorem{definition}[defi]{Definition}

\newtheorem{exam}{Theorem}
\newtheorem{example}[exam]{Example}

\newtheorem{rema}{Theorem}
\newtheorem{remark}[rema]{Remark}

\newtheorem*{proo}{Theorem}
\newtheorem*{proof}[proo]{Proof}

\newtheorem{lemm}{Theorem}
\newtheorem{lemma}[lemm]{Lemma}

\newtheorem{syst}{Theorem}
\newtheorem{system}[syst]{System}

\title{Inductive Certificate Synthesis for Control Design}

\author{Hadi}{Ravanbakhsh}

\otherdegrees{B.E., University of Tehran, 2011 \\
	      M.S., University of Colorado, Boulder, 2014}

\degree{Doctor of Philosophy}		
	{Ph.D., Computer Science}		

\dept{Department of}			
	{Computer Science}		

\advisor{Prof.}				
	{Sriram Sankaranarayanan}			

\reader{Prof. Christoffer Heckman}		
\readerThree{Prof. Ashutosh Trivedi}		

\abstract{  \OnePageChapter
The focus of this thesis is developing a framework for designing correct-by-construction controllers using control certificates. We use nonlinear dynamical systems to model the physical environment (plants). The goal is to synthesize controllers for these plants while guaranteeing formal correctness w.r.t. given specifications. We consider different fundamental specifications including stability, safety, and reach-while-stay. Stability specification states that the execution traces of the system remain close to an equilibrium state and approach it asymptotically. Safety specification requires the execution traces to stay in a safe region. Finally, for reach-while-stay specification, safety is needed until a target set is reached.

The design task consists of two phases. In the first phase, the control design problem is reduced to the question of finding a control certificate. More precisely, the goal of the first phase is to define a class of control certificates with a specific structure. This definition should guarantee the following:  ``Having a control certificate, one can systematically design a controller and prove its correctness at the same time."
The goal in the second phase is to find such a control certificate. We define a potential control certificate space (hypothesis space) using parameterized functions. Next, we provide an inductive search framework to find proper parameters, which yield a control certificate. 

Finally, we evaluate our framework. We show that discovering control certificates is practically feasible and demonstrate the effectiveness of the automatically designed controllers through simulations and real physical systems experiments.
	}

\dedication[Dedication]{	
	 \centerline{To My Parents}
	}

\acknowledgements{	\OnePageChapter	

	First and foremost, I would like to thank my supervisor Sriram Sankaranarayanan, whose invaluable support made my graduate studies a delightful experience. I am grateful for the freedom he gave me to pursue my passion and his guidance through the process. I am further thankful to the committee/advisory members Pavol \v Cern\' y, Evan Chang, John Hauser, Christoffer Heckman, Fabio Somenzi, Behrouz Touri, and Ashutosh Trivedi for their help and feedback throughout my studies.
    
    I would also like to acknowledge NSF support under award numbers CNS-0953941 and SHF-1527075 for funding my studies on different projects.
    
    I am thankful to my collaborator, Sina Aghli, who helped me elevate my research to the application level, bringing more excitement to my projects. I also would like to thank former/current members of CUPLV group for helpful discussions and feedbacks, especially Aditya, Aleks, Amin, Souradeep, Vris, and Xin. I extend my gratitude to the CS graduate advisors Jacqueline DeBoard and Rajshree Shrestha for their efforts.
    
    I am grateful to my friends Al, Amir, Arash, Azadeh, Farhad, Ghazaleh, Hamid, Homa, Hooman, Liam, Mahdi, Mahnaz, Mahshab, Mohammad, Paria, Reza, Reihaneh, Romik, Saman, Sanaz, Sepideh, Sina, Sorayya, and others for the good times in Boulder.
    
    Last but not least, I am thankful to my family, especially Forough for constant support and encouragement.
	}

\ToCisShort	

\LoFisShort	

\LoTisShort	

\begin{document}

\input macros.tex

\input{intro/intro}

\input{related/related}

\input{certificate/certificate}

\input{synt/synt}

\input{eval/eval}

\input{conclusion}

\bibliographystyle{plain}	
\nocite{*}		
\bibliography{refs}		

\appendix
\input{benchmarks/benchmarks}

\end{document}

%% file: macros.tex
\usetikzlibrary{arrows,shapes,shapes.arrows,shadows,snakes,backgrounds,decorations,decorations.markings,decorations.pathmorphing,positioning,fit,automata,calc}

\newcommand{\todo}[1]{\textcolor{red}{\texttt{#1}}}

\newcommand{\diff}[2]{\frac{\partial #1}{\partial #2}}
\newcommand{\diffr}[1]{\diff{#1}{r}}
\newcommand{\diffth}[1]{\diff{#1}{\theta}}
\newcommand{\diffz}[1]{\diff{#1}{z}}

\newcommand{\vth}{V_{\theta}}

\newcommand{\twochoices}[2]{\left\{ \begin{array}{lcc}
        \displaystyle #1 \\ \vspace{-10pt} \\
        \displaystyle #2 \end{array} \right. } 

\newcommand{\threechoices}[3]{\left\{ \begin{array}{lcc}
        #1 \\ #2 \\ #3 \end{array} \right. }    

\newcommand{\fourchoices}[4]{\left\{ \begin{array}{lcc}
        #1 \\ #2 \\ #3 \\ #4 \end{array} \right. }      

\newcommand{\twovec}[2]{\left(\begin{array}{c} #1 \\ #2 \end{array}\right)}
\newcommand{\threevec}[3]{\left(\begin{array}{c} #1 \\ #2 \\ #3 \end{array}\right)}
\newcommand{\twomatrix}[4]{\left(\begin{array}{cc} #1 & #2 \\ #3 & #4 \end{array}\right)}

\newcommand\scr[1]{\mathcal{#1}}
\newcommand \reals {\ensuremath \mathbb{R}}
\newcommand \bools {\ensuremath \mathbb{B}}
\renewcommand{\vec}[1]{\mathbf{#1}}
\newcommand \vx {\vec{x}}
\newcommand \vy {\vec{y}}
\newcommand \vz {\vec{z}}
\newcommand \vu {\vec{u}}
\newcommand \vc {\vec{c}}
\newcommand \vd {\vec{d}}
\newcommand \ve {\vec{e}}
\newcommand \vf {\vec{f}}
\newcommand \vg {\vec{g}}
\newcommand \vh {\vec{h}}
\newcommand \va {\vec{a}}
\newcommand \vb {\vec{b}}
\newcommand \vr {\vec{r}}
\newcommand \vs {\vec{s}}
\newcommand \vq {\vec{q}}
\newcommand \vv {\vec{v}}
\newcommand \vw {\vec{w}}
\newcommand \vm {\vec{m}}
\newcommand \vlam {\boldsymbol{\lambda}}
\newcommand \K {\scr{K}}
\renewcommand \P {\scr{P}}
\newcommand \Tau {\scr{T}}
\newcommand \U {\scr{U}}
\newcommand \trX {\sigma_X}
\newcommand \trU {\sigma_U}
\newcommand \tr {\sigma}
\newcommand \dtrX {\dot{\sigma}_X}
\newcommand \trUp {\sigma^+_U}
\newcommand \trXs {\sigma^*_X}
\newcommand \vzero {\vec{0}}
\newcommand \veps {\boldsymbol{\epsilon}}

\newcommand \tupleof[1] { \left\langle #1 \right \rangle}

\newcommand \B {\scr{B}}
\newcommand \C {\scr{C}}
\newcommand \D {\scr{D}}
\newcommand \E {\scr{E}}
\newcommand \Hy {\scr{H}}
\newcommand \I {\scr{I}}
\newcommand \V {\scr{V}}
\newcommand \T {\scr{T}}
\newcommand \X {\scr{X}}
\newcommand \degr {D}

\newcommand \Vol {\mbox{Vol}}
\newcommand{\arginf}[1]{\underset{#1}{\mbox{arginf}}\ }
\newcommand{\inter}[1]{\overset{\circ}{#1}}

\newcommand \false {\mathit{false}}
\newcommand \true  {\mathit{true}}

\newcommand \ab[1]{\scr{A}(#1)}
\newcommand \ac[1]{\scr{F}(#1)}

\newcommand\cond{\mathsf{cond}}
\newcommand\RWS{reach-while-stay}
\newcommand\RS{region-stability}

\newcommand{\argmin}[1]{\underset{#1}{\mathrm{argmin}}}
\newcommand\SK{\mathcal{\kappa}}

\newcommand\tick{\ding{51}}
\newcommand\crossMark{\ding{54}}

\definecolor{myBlue}{rgb}{0.041536,0.306316,0.700569}
\definecolor{myGreen}{rgb}{0.061778,0.471387,0.126909}
\definecolor{myRed}{rgb}{0.729085,0.069717,0.037669}

\newcommand*{\QED}{\hfill\ensuremath{\blacksquare}}%

\newcommand{\hadi}[1]{\emph{\color[rgb]{0.041536,0.306316,0.700569}#1}}

%% file: intro/intro.tex
\chapter{Introduction}\label{ch:intro}
A control system consists of a controller that interacts with its physical environment (plant) to perform specific tasks. For example, an artificial insulin delivery system includes a  plant, which is the patient's body, and a controller, which is the insulin infusion device~\cite{cobelli2011artificial}. An autonomous car is another example, where the vehicle moves in an environment and different controllers are responsible for various subsystems including powertrain, cutoff fuel injection, idle speed, cruise, and autonomous driving~\cite{Jin2014,balluchi2000,STURSBERG2004}. Humanoid robots also depend on controllers to accomplish complex tasks using different subtasks such as walking and object manipulation~\cite{Nguyen2015RobustCLF,Ames2013}. As these systems are becoming ubiquitous, the increasing need for analyzing these systems is undeniable. In fact, in specific safety-critical domains, a control system failure could have catastrophic consequences. 

Formal methods provide tools and techniques to study correctness of control systems through mathematical models. The goal of formal verification is to check whether the system works correctly, where correctness is expressed through some specifications. For example, for an artificial insulin delivery system, one should make sure the device keeps the glucose level of the patient in a proper range. The specification in this example can be expressed as $(\forall t) \ \ 70 < G(t) < 200 $, where $G(t)$ is the glucose level in the patient's blood as a function of time. 

The verification process can have two outcomes. Either the verification is successful, and correctness of the closed-loop system is proven, or the analysis declares the possibility of failure. In the latter case, usually, the verifier reveals a counterexample. This counterexample describes a possible behavior of the system, where the specification is violated. Failure in verification is not the last step. The system designer should update the controller by considering the faulty scenario. Then, the revised version of the system can be verified again to see if the update fixes the issue. For example, assume that verification of the artificial insulin delivery system fails. The verifier yields a scenario wherein the patient does not get a meal as expected, and after one hour, the glucose level increases to $230$ units ($ G(60) = 230 $). Then, the designer of the insulin infusion device tries to update the controller and remove such behaviors. Such update is achieved, perhaps, by reducing the maximum amount of insulin that can be given to the patient by the controller. However, such fixes can be quite hard to realize in practice since updating the controller may introduce other faulty scenarios. 

A more appealing approach involves the automatic synthesis of the controller from specifications rather than a manual-design/verification loop. Instead of verifying whether the closed-loop system behaves correctly w.r.t the specification, the goal is to design a controller for which the correctness of the closed-loop system is mathematically guaranteed. These controllers are referred to as ``correct-by-construction" controllers. Unfortunately, such synthesis procedures have higher complexity compared to verification procedures, and their development remains challenging.

In this thesis, the problem of synthesizing correct-by-construction controllers is investigated. We develop learning-based tools and methods for automated synthesis of such controllers. In particular, we are interested in nonlinear dynamical systems. Furthermore, providing a scalable solution is the primary objective, and as such, we only consider basic specifications including safety and stability. This thesis incorporates previously published papers~\cite{ravanbakhsh2015counterexample,ravanbakhsh-others/2015/Counter,ravanbakhsh2015counter,ravanbakhsh2016robust,ravanbakhsh2017class,ravanbakhsh2017demonstration} and papers under review~\cite{ravanbakhsh2018demonstration,ravanbakhsh2018funnel}.

The goal of this chapter is to define control systems and specifications formalism, along with the problem statement and motivating examples.

\input{intro/control}
\input{intro/spec}

%% file: intro/control.tex
\section{Control Systems}\label{sec:system}
We are particularly interested in state feedback control systems. There are two types of such systems: (i) smooth feedback systems, and (ii) switched feedback systems. For each of these systems, a description of the system model is provided in two steps. In the first step, the syntax of the control system is explained, and subsequently, the behavior of the system (semantics) is discussed by defining execution traces of the system.

\subsection{Smooth Feedback}
The system of interest consists of a \emph{plant} and a smooth state feedback \emph{controller}. The plant has $n$ continuous variables defining its state. A state $\vx$ belong to set $X :\ \reals^n$. The controller provides feedback (input) $\vu \in U \subseteq \reals^m$ for a measured state $\vx$. The dynamics are defined using an ordinary differential equation (ODE). Figure~\ref{fig:model}(a) shows a schematic view of the closed-loop system. 

\begin{figure}[t]
\begin{center}
    \includegraphics[width=0.8\textwidth]{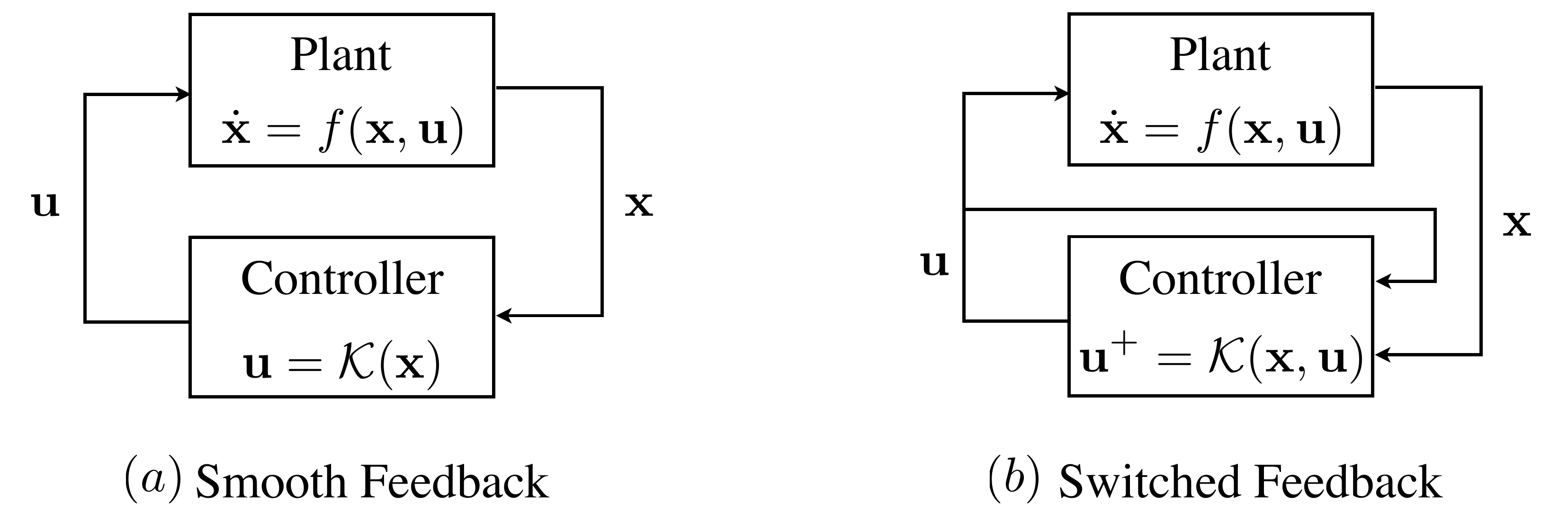}
\end{center}
\caption{Models of state feedback systems.}\label{fig:model} 
\end{figure}

We now provide a formal definition of the plant, the controller, and the closed-loop system is given in the following paragraphs.

\begin{definition}[Smooth Plant]\label{def:smooth-plant}
\label{def:plant}
 A smooth plant is a tuple $\P(X, U, f)$ describing the physical environment, where
 \begin{compactenum}
        \item $X :\ \reals^n$ is domain of plant state,
        \item $U$ ($\subseteq \reals^m$) is the range of control feedback,
        \item $f :\ X \times U \mapsto X$ is a smooth function, which defines the vector field for an ODE.
    \end{compactenum}
\end{definition}

\begin{definition}[Smooth Controller]\label{def:smooth-controller}
\label{def:control}
Given a smooth plant $\P(X, U, f)$, a smooth controller is a smooth function $\K :\ X \mapsto U$ that maps current state $\vx \in X$ to the feedback $\vu \in U$.
\end{definition}

\begin{definition}[Smooth Feedback System]\label{def:smooth-system}
\label{def:control}
The combination of the smooth plant $\P(X, U, f)$ and the smooth controller $\K$ yields a smooth closed-loop system $\Psi(\P, \K)$.
\end{definition}

\begin{definition}[Smooth Trace]
\label{def:trace}
    Given a smooth closed-loop system $\Psi$, a trace $\tr$ of $\Psi$ is a tuple $\tr:(\vx(\cdot),\vu(\cdot))$, where $\vx(\cdot)$ ($\vu(\cdot)$) maps the time to the state (input) of the system.
\end{definition}

Each possible execution of the system can be modeled using a valid trace.

\begin{definition}[Smooth Valid Trace]
\label{def:trace-valid-smooth}
 Given a smooth closed-loop system $\Psi(\P, \K)$, a
trace $\tr$ of the system is said to be valid iff for all times $t$:
        \begin{align*}
            \vu(t) & = \K(\vx(t)) \\
            \dot{\vx}(t) & = f(\vx(t), \vu(t)) \,.
        \end{align*}
\end{definition}

Notice that given an initial value $\vx(0) = \vx_0$, by continuity of $f$ and $\K$, Peano existence theorem guarantees that at least one valid trace exists (though, there may be more than one). Also, if $f$ and $\K$ are uniformly Lipschitz continuous, by Picard-Lindel\"{o}f theorem, a unique valid trace exists. For example, if the trace does not diverge to infinity and $f$ and $\K$ are finite over a compact set, then valid trace is unique for a given $\vx_0$. Also, each trace $\tr$ has a ``escape time'', $\Tau(\tr) \in \reals^+ \cup \{\infty\}$, and for all time $t > \Tau(\tr)$, the trace is not defined. We use $\bot$ to denote absence of a solution ($(\forall\ t \geq \Tau(\tr)) \ \vx(t) = \bot \land \vu(t) = \bot$). In the rest of this thesis, we use word \emph{trace} instead of \emph{valid trace}.

\begin{example}
Consider a plant with single state $x$, and a single input $u$ ($X :\ \reals$, $U:\reals$), where $f(x, u) = 1 + u$. For feedback law $\K(x) = x^2$, a unique valid trace $x(t)$ exists s.t. $\dot{x}(t) = 1 + x(t)^2$. Assuming $x_0 = 1$, $x(t) = tan(t + \frac{\pi}{4})$ is a valid trace. In other words, the solution blows up to infinity within a finite time. Here $\Tau(\tr) = \frac{\pi}{4}$ and $x(t) = \bot$ for $t \geq \frac{\pi}{4}$.
\end{example}

\subsection{Switched Feedback}
For a switched system, the model is slightly different. First, $U$ is a finite set with size $|U| = m'$, yielding $m'$ different modes for the plant. Also, for technical reasons, we need $\K$ to be a function of the current plant state $\vx$ and current mode $\vu$ (Figure~\ref{fig:model}(b)).

\begin{definition}[Switched Plant]\label{def:switched-plant}
\label{def:plant}
 A plant is a tuple $\P(X, U, f)$ describing the physical environment:
 \begin{compactenum}
        \item $X :\ \reals^n$ is domain of plant state,
        \item $U$ is a finite set of control feedback,
        \item $f :\ (X \times U) \mapsto X$ defines the smooth vector field for each mode $\vu \in U$.
    \end{compactenum}
\end{definition}

For readability, we use $f_\vu(\vx)$ instead of $f(\vx, \vu)$ for switched systems.

\begin{definition}[Switched Controller]\label{def:switched-plant}
\label{def:control}
Given a plant $\P(X, U, f)$, a controller is a function $\K :\ X \times U \mapsto U$ that maps current state $\vx \in X$ and current mode $\vu$ to the next mode $\vu^* \in U$.
\end{definition}

\begin{definition}[Switched Feedback System]\label{def:switched-system}
\label{def:control}
The combination of the switched plant $\P(X, U, f)$ and the switched controller $\K$ yields a closed-loop switched system $\Psi(\P, \K)$.
\end{definition}

\begin{definition}[Switched Valid Trace]
\label{def:trace-valid}
 Given a switched system $\Psi(\P, \K)$, a
trace $\tr$ of the system is said to be valid iff \begin{compactenum}
        \item For all times $t$
        \begin{align*}
            \vu^+(t) & = \K(\vu(t), \vx(t)) \\
            \dot{\vx}(t) & = f_{\vu^+(t)}(\vx(t)) \,,
        \end{align*}
        where $\dot{\vx}(t)$ is the right derivative of $\vx$ w.r.t. time,
        \item $\vu$ is continuous from left ($\vu^-(t) = \vu(t)$) for all $t > 0$,
        \item The set of switched times 
            \[\mathsf{SwitchTimes}(\vu(\cdot)) : \{t \in \reals^+ | \vu(t) \neq
\vu^+(t)\} \,,\]
            is a countable set.
    \end{compactenum}
\end{definition}

Similar to smooth traces, given $\vx(0) = \vx_0$ and $\vu(0) = \vu_0$, there exists at least one valid trace. In addition, a valid trace can have Zeno behavior.
\begin{definition}[Zeno Behavior]
    A trace $\tr$ has Zeno behavior if $(\exists \Delta > 0)$ s.t.  $\mathsf{SwitchTimes}(\vu(\cdot)) $ $\cap [0, \Delta]$ is infinite. Moreover, $\vx(t) = \vu(t) = \bot$ for $t \geq \Delta$.
\end{definition}
This phenomenon is named after Zeno of Elea, greek philosopher who noticed the paradoxes one may get when the number of switches (discrete events) are infinite in a finite (continuous) time interval. Thus, avoiding Zeno behavior is not merely a practical concern, as it leads to theoretical flaws as well.

\begin{example} Consider a plant with two states $x_1$ and $x_2$, and two modes $u_1$, and $u_2$ with the following dynamics:
\[
f_{u_1}(x_1, x_2) = \left[\begin{array}{c} -1 \\ -2\end{array} \right] \, , \, f_{u_2}(x_1, x_2) = \left[\begin{array}{c} -1 \\ 2\end{array} \right]\,.
\]
The feedback function $\K$ is the following:
\[
\K(x_1, x_2, u) = \begin{cases} u_1 & u = u_2 \land x_2 \geq x_1 \\ u_2 & u = u_1 \land x_2 \leq -x_1 \\ u & \mbox{otherwise}\,. \end{cases} 
\]
Starting from $\vx(0) = [1, 0]$ and $u(0) = u_1$, the trace is shown in Figure~\ref{fig:zeno}. It is easy to show that $x_1(t) = 1 - t$. On the other hand, the trace does not leave $\{\vx \ | \ x_2 \leq x_1 \land x_2 \geq -x_1\}$, which is a contradiction. However, as illustrated in Figure~\ref{fig:zeno}, the number of switches is infinite in interval $[0, 1]$. Zeno behavior occurs and the time stops at $\Delta = 1$. Therefore, $\vx(t) = \bot$ for $t \geq 1$.

\begin{figure}[t]
\begin{center}
    \includegraphics[width=0.2\textwidth]{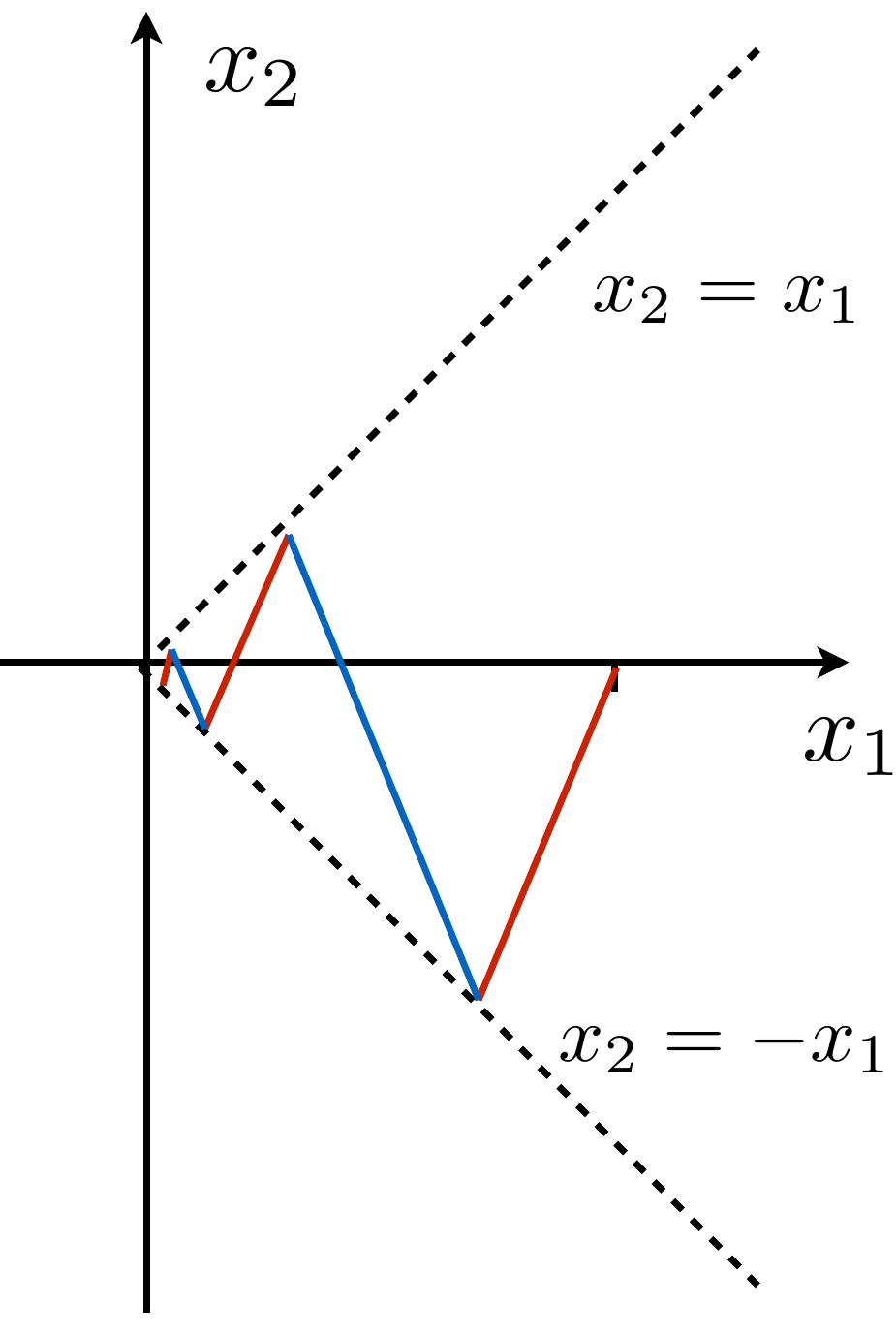}
\\ The color of execution trace is red (blue) when mode $u_1$ ($u_2$) is selected.
\end{center}
\caption{An example of a Zeno behavior.}\label{fig:zeno} 
\end{figure}
\end{example}

%% file: intro/spec.tex
\section{Specifications}\label{sec:spec}
A specification describes the desired behavior of all possible system traces $\tr :\ (\vx(\cdot), \vu(\cdot))$. In general, a specification is a logical formula over a trace $\tr$. However, we restrict the specification to be a logical formula over $\vx(\cdot)$ ($\varphi:\ \vx(\cdot) \mapsto \mathbb{B}$), and given a state trace $\vx(\cdot)$, its value is decided. A trace $\vx(\cdot)$ \emph{respects} a specification $\varphi$, if $\varphi(\vx(\cdot))$ holds. For example, to ensure the state eventually reaches a set $G$ ($\Diamond G$), we use the following specification:
\[
\varphi(\vx(\cdot)) :\ (\exists t \geq 0) \ \vx(t) \in G \,.
\]

 Recall that $\vx(t) = \bot$ when the trace is not defined (either the trace escapes in finite time or Zeno behavior occurs). Therefore, we implicitly require $\vx(t) \neq \bot$ if $\vx(t)$ to be used in $\varphi(\vx(\cdot))$.

The ultimate goal is to find a $\K$ function that guarantees all traces admit the specifications.
\begin{definition}[Control Synthesis Problem]\label{def:problem}
    Given a plant $\P(X, Q, f)$ and a specification $\varphi$, the control synthesis problem is to find a function $\K$ s.t. for all traces $\tr$ of the closed-loop system $\Psi(\P, \K)$, $\varphi(\vx(\cdot))$ holds.
\end{definition}

One may use temporal logics~\cite{Koymans1990,maler2004monitoring} to describe $\varphi$. However, in this thesis, we focus on basic specifications, including stability, safety, and reach-while-stay. Other specifications are left for future work.

\subsection{Safety}
Safety property requires the traces of the system to remain inside a safe set $S$. Formally, $\varphi(\vx(\cdot)) :\ (\forall t \geq 0) \ \vx(t) \in S$ ($\Box S$).
However, if $\vx(0) \not\in S$, the safety cannot be guaranteed. Therefore, we also need to enforce initial condition $\vx(0) \in I$ for some initial set $I \subseteq S$:
\[
\varphi(\vx(\cdot)) \ :\ (\vx(0) \in I) \implies (\forall t \geq 0) \ \vx(t) \in S\,.
\]

\begin{example}[Inverted Pendulum Problem]\label{ex:inverted-pendulum}
    Consider the problem of keeping an inverted pendulum in a vertical position. This model has applications in balancing two-wheeled robots. The system has two degrees of freedom: the location of the cart $x$, and the degree of the inverted pendulum $\alpha$. The goal is to keep the pendulum in a vertical position by moving the cart with input $u$ (Figure~\ref{fig:inverted-pendulum}).

The system has four state variables $[x, \dot{x}, \alpha, \dot{\alpha}]$ with the following dynamics~\cite{landry2005dynamics}:
\begin{align*}\label{eq:inverted-pendulum-dyn}
        \ddot{x} & = \frac{4u - 4\epsilon\dot{x} + 4ml\dot{\alpha}^2 \sin(\alpha) - 3mg\sin(\alpha)\cos(\alpha)}{4(M+m)-3m\cos^2(\alpha)} \\ 
        \ddot{\alpha} & = \frac{ (M+m)g\sin(\alpha) - (u - \epsilon\dot{x}) \cos(\alpha) - ml\dot{\alpha}^2\sin(\alpha)\cos(\alpha)}{l(\frac{4}{3}(M+m)-m\cos(\alpha)^2)}\,,
\end{align*}
where $m = 0.21$ and $M=0.815$ are masses of the pendulum and the cart respectively, $g=9.8$ is the gravitational acceleration, and $l=0.305$ is distance of center of mass of the pendulum from the cart. Also, the input is saturated $U:[-20, 20]$.
We are interested in safety property where the state is initially in set $I :\ \{\vx \ | \ ||\vx||_2 \leq 0.1\}$ (almost in vertical position) and the safe region is $S:\ [-1, 1]^4$. The goal is to design a controller to satisfy this safety property.
\begin{figure}[t]
\begin{center}
    \includegraphics[width=0.3\textwidth]{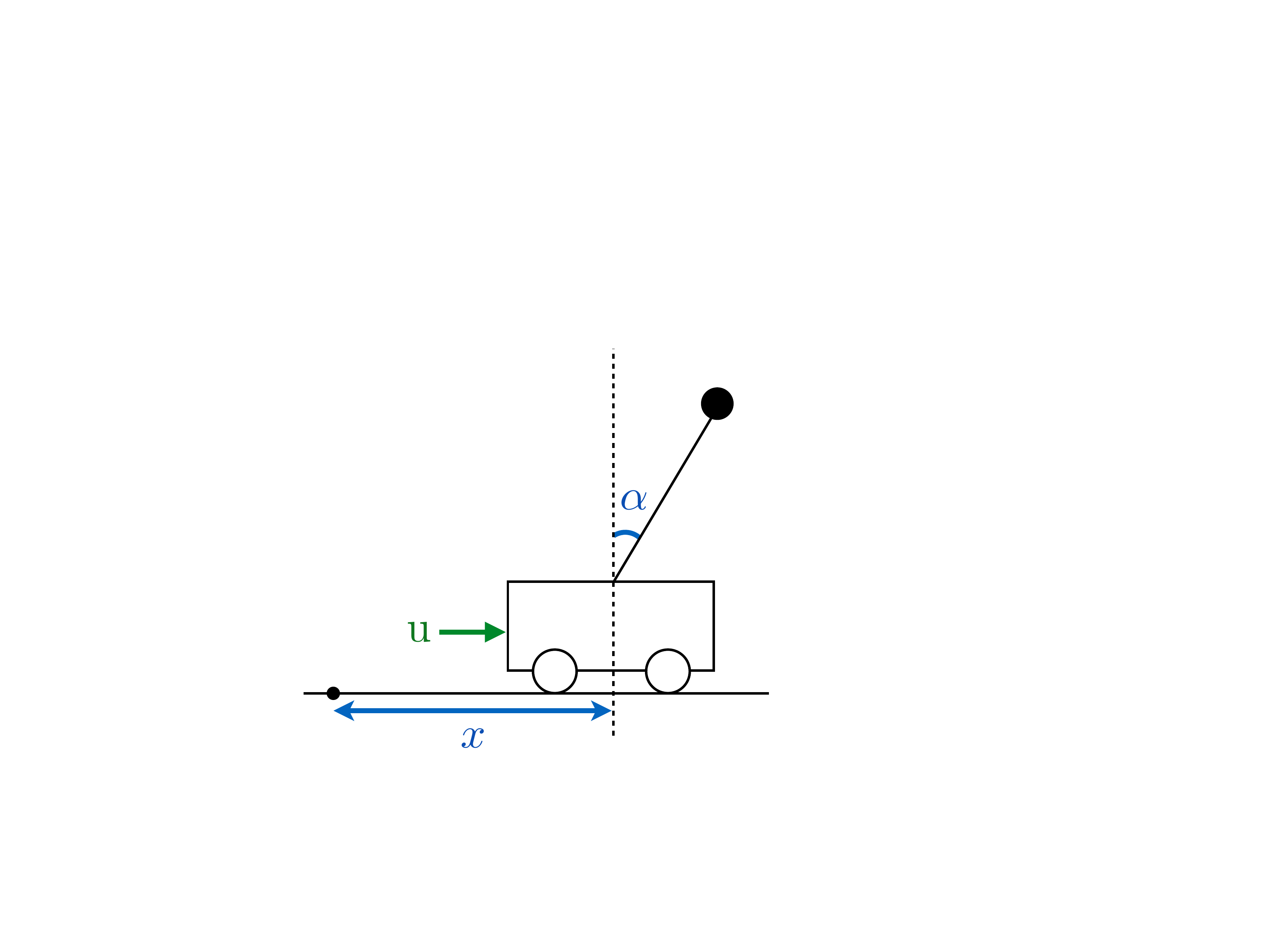}
\end{center}
\caption{A schematic view of the ``inverted pendulum on a cart."}\label{fig:inverted-pendulum} 
\end{figure}
\end{example}

\subsection{Stability}
The goal in stability is getting the state $\vx$ close to a desired point $\vx_r \in X$ (equilibrium state). Formally, global \emph{asymptotic} stability consists of two parts:

\begin{align*}
    (a) \ & (\forall \epsilon > 0) \ (\exists \delta > 0) \ \vx(0) \in \B_\delta(\vx_r) \implies (\forall t) \ \vx(t) \in \B_\epsilon(\vx_r) \\
    (b) \ & (\forall \epsilon > 0) \ (\exists T) \ (\forall t \geq T) \ \vx(t) \in \B_\epsilon(\vx_r)\,,
\end{align*}
where $\B_\delta(\vx)$ is a ball of appropriate dimension centered at $\vx$ with radius $\delta$. The first part makes sure that the trace does not diverge and stays close to the equilibrium (Lyapunov stability). Meanwhile, the second part guarantees that the state gets arbitrarily close (and stays close) to the equilibrium. Putting these two properties together, one can show that the state converges to the equilibrium ($\lim_{t\rightarrow\infty} ||\vx(t) - \vx_r|| = 0$).

\subsubsection{Reference Tracking}
While stability is used to keep the state close to the equilibrium, many specifications involve keeping the state close to a \emph{moving} state. For example, assume a car is chasing another car. The goal is not to park the car somewhere, but to follow the other car. The moving state defines a reference trajectory $\vx_r(\cdot):\ \reals^+ \mapsto X$. 
Formally, we wish $\lim_{t \rightarrow \infty} \vx(t) = \vx_r(t)$ (global asymptotic trajectory tracking). Global asymptotic trajectory tracking holds iff the following conditions hold:
\begin{align*}
    (a) \ & (\forall \epsilon > 0) \ (\exists \delta > 0) \ \vx(0) \in \B_\delta(\vx_r(0)) \implies (\forall t) \ \vx(t) \in \B_\epsilon(\vx_r(t)) \\
    (b) \ & (\forall \epsilon > 0) \ (\exists T) \ (\forall t \geq T) \ \vx(t) \in \B_\epsilon(\vx_r(t))\,.
\end{align*}
Similar to stability property, condition (a) guarantees that if the state is initially near $\vx_r(0)$, then the state will not diverge from $\vx_r(t)$ and stays close to the reference $\vx_r(t)$ at time $t$ (for all times). Also, condition (b) guarantees that the state gets arbitrarily close (and stays close) to the reference $\vx_r(t)$ (at time $t$) as $t \rightarrow \infty$.

Stability properties are interesting merely from a theoretical perspective as globally asymptotic stable system never exists. We consider practical stability in the next section.

\subsection{Reach-While-Stay}
The \emph{Reach-While-Stay} (RWS) combines the safety and reachability properties to serve as the $\U$ operator in temporal logics. More specifically, given an initial set $I$, a goal set $G$, and a safe set $S$, the specification requires the trace to reach from $I$ to $G$ while staying inside $S$ ($I \implies S \ \U \ G$ in temporal logic). Formally:

\[
(\vx(0) \in I) \implies (\exists T \geq 0) \left(\begin{array}{l}
    \vx(T) \in G \ \land \\
    (\forall t, 0 \leq t \leq T) \ \vx(t) \in S
\end{array} \right) \,.
\]

\begin{example}[Forward Flight Problem]\label{ex:ducted-fan-forward}
Caltech ducted-fan has been used to study the aerodynamics of a single wing of a thrust vectored, fixed-wing aircraft~\cite{jadbabaie2002control}.  
In this example, we wish to design forward flight control in which the angle of attack needs to be fixed for a stable forward flight. The model of the system is carefully calibrated through wind tunnel experiments. The system has four states: $v$ is the velocity, $\gamma$ defines the moving direction of the ducted-fan, $\mu$ is the rotational position, and $q$ is the angular velocity. The control inputs are the thrust $\tau$ and the angle $\delta$ at which the thrust is applied (Figure~\ref{fig:ducted-fan}).     
Also, the inputs are saturated: $\tau :\ [0, 13.5]$ and $\delta :\ [-0.45, 0.45]$.
The dynamics are:
\begin{align*}\label{eq:ducted-fan-forward-dyn}
m \dot{v} & = -D(v, \alpha) - W \sin(\gamma) + u \cos(\alpha + \delta_u) \, \ & \dot{\mu} & = q \\
m v \dot{\gamma} & = L(v, \alpha) - W \cos(\gamma) + u \sin(\alpha + \delta_u) \, \ & J \dot{q} & = M(v, \alpha) - u l_T \sin(\delta_u) \,,
\end{align*}
where the angle of attack $\alpha = \mu - \gamma$, and drag ($D$), lift ($L$), and moment ($M$) terms are polynomials in $v$ and $\alpha$. 
For full list of parameters, see~\cite{jadbabaie2002control}. According to the dynamics, $\vx_r:\ [6, 0, 0.1771, 0]$ is a stable equilibrium (for $\vu_r:\ [3.2, -0.138]$) where the ducted-fan can move forward with velocity $6$.
Thus, the goal is to reach near $\vx_r$. 
\begin{figure}[t]
\begin{center}
    \includegraphics[width=0.3\textwidth]{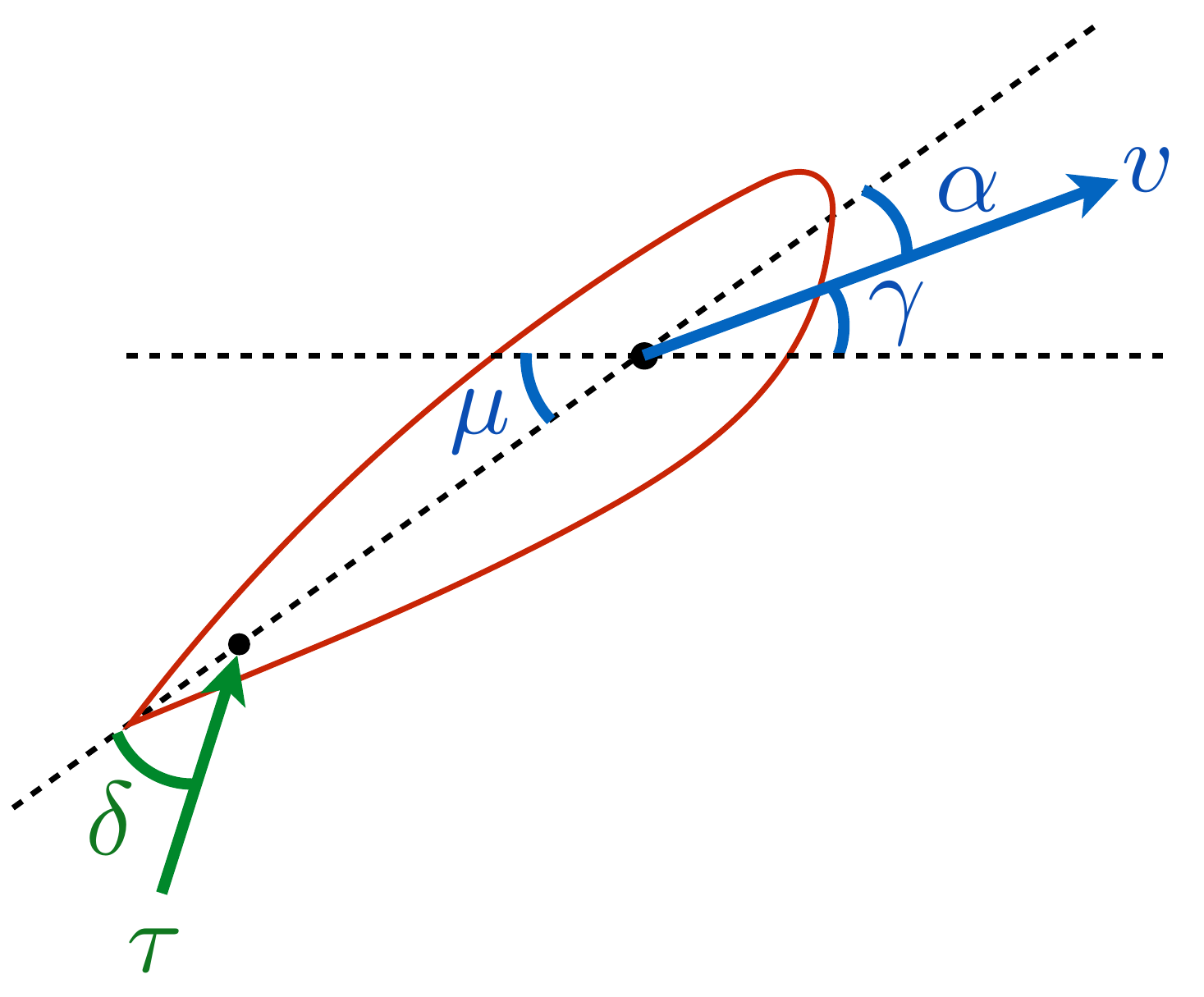}
\end{center}
\caption{A schematic view of the Caltech ducted-fan.}\label{fig:ducted-fan} 
\end{figure}

To guarantee practical stability, we use a RWS specification with the following sets:
\begin{align*}
S&:\{[v, \gamma, \mu, q]^t + \vx_r | v \in [3, 9], \gamma \in [-0.75, 0.75], \mu \in [-0.75, 0.75], q \in [-2, 2]\} \\
I&:\{[v, \gamma, \mu, q]^t + \vx_r |(0.4v)^2 + \gamma^2 + \mu^2 + q^2 < 0.4^2 \} \, \\
G&:\{[v, \gamma, \mu, q]^t + \vx_r |(0.4v)^2 + \gamma^2 + \mu^2 + q^2 < 0.05^2 \} \,.
\end{align*}
\end{example}

\begin{example}[Hover Mode Problem]\label{ex:ducted-fan-hover}
In this example, we consider a problem for the \emph{planar} Caltech ducted-fan~\cite{jadbabaie2002control}.
The goal is to keep the ducted-fan in a hover mode. The system (planar model) has three degrees of freedom, $x$, $y$, and $\mu$, which define the position and orientation of the ducted-fan. There are six state variables $x$, $y$, $\mu$, $\dot{x}$, $\dot{y}$, $\dot{\mu}$ and two control inputs $u_1$, $u_2$ ($U \in [-10, 10]\times[0, 10]$).
The dynamics are
\begin{align*}
m \ddot{x} & = -d_c\dot{x} + u_1 \cos(\mu) - u_2 \sin(\mu) \\
m \ddot{y} & = -d_c\dot{y} + u_2 \cos(\mu) + u_1 \sin(\mu) - mg \\
J \ddot{\mu} & = r u_1 \,,
\end{align*}
where $m = 11.2$, $g = 0.28$, $J = 0.0462$, $r = 0.156$ and $d_c = 0.1$. The system is stable at the origin for $\vu_r:\ [0, mg]$. Therefore, we set $\vu_r$ as the origin for the input space. For practical stability, we consider RWS and the sets are defined as
\begin{align*}
S :\ \vx_r \oplus ([-1,1]\times[-1,1]\times[-0.7,0.7]\times[-1, 1]^3) \,,\, I :\B_{0.25}(\vx_r) \,,\, G:\B_{0.1}(\vx_r) \,,
\end{align*}
where $\oplus$ is the Minkowski sum.
\end{example}

\begin{example}[Bicycle Problem]\label{ex:bicycle}
    This system is a two-wheeled mobile robot modeled with five states $[x, y, v, \alpha, \gamma]$ and two control inputs~\cite{francis2016models}, where $x$ and $y$ define the position of the robot, $v$ is its velocity, $\alpha$ is the rotational position and $\gamma$ is the angle between the front and rear axles.  The goal is to stabilize the robot to a reference velocity $v_r=5$, and $\alpha_r = \gamma_r = y_r = 0$ as shown in Figure~\ref{fig:bicycle}. The dynamics of the model is as follows:
\begin{equation*}\label{eq:bicycle-dyn}
\dot{x} = v\cos(\alpha) \,,\, \dot{y} = v\sin(\alpha) \,,\, \dot{v} = \tau \,,\, \dot{\alpha} = \frac{v}{l} \sigma \,,\, \dot{\sigma} = \delta,
\end{equation*}
where $l = 1$ is the distance between the wheels and $\sigma = tan(\gamma)$ (see Figure~\ref{fig:bicycle}). 
Variable $x$ is immaterial in the stabilization problem and is dropped to obtain a model with four state variables $[y, v, \alpha, \sigma]$. The equilibrium is $\vx_r:[y_r, v_r, \alpha_r, \sigma_r]$. The inputs are saturated: $\tau \in [-10, 10]$, $\delta \in [-10, 10]$. For practical stability, we consider RWS property with the following sets
\[
        S:\ \vx_r \oplus ([-2, 2]\times[3, 7]\times[-1, 1]\times[-1, 1]) \,,\, I:\ \B_{0.4}(\vx_r) \,,\, G:\ \B_{0.1}(\vx_r) \,.
\]
\begin{figure}[t]
\begin{center}
    \includegraphics[width=0.4\textwidth]{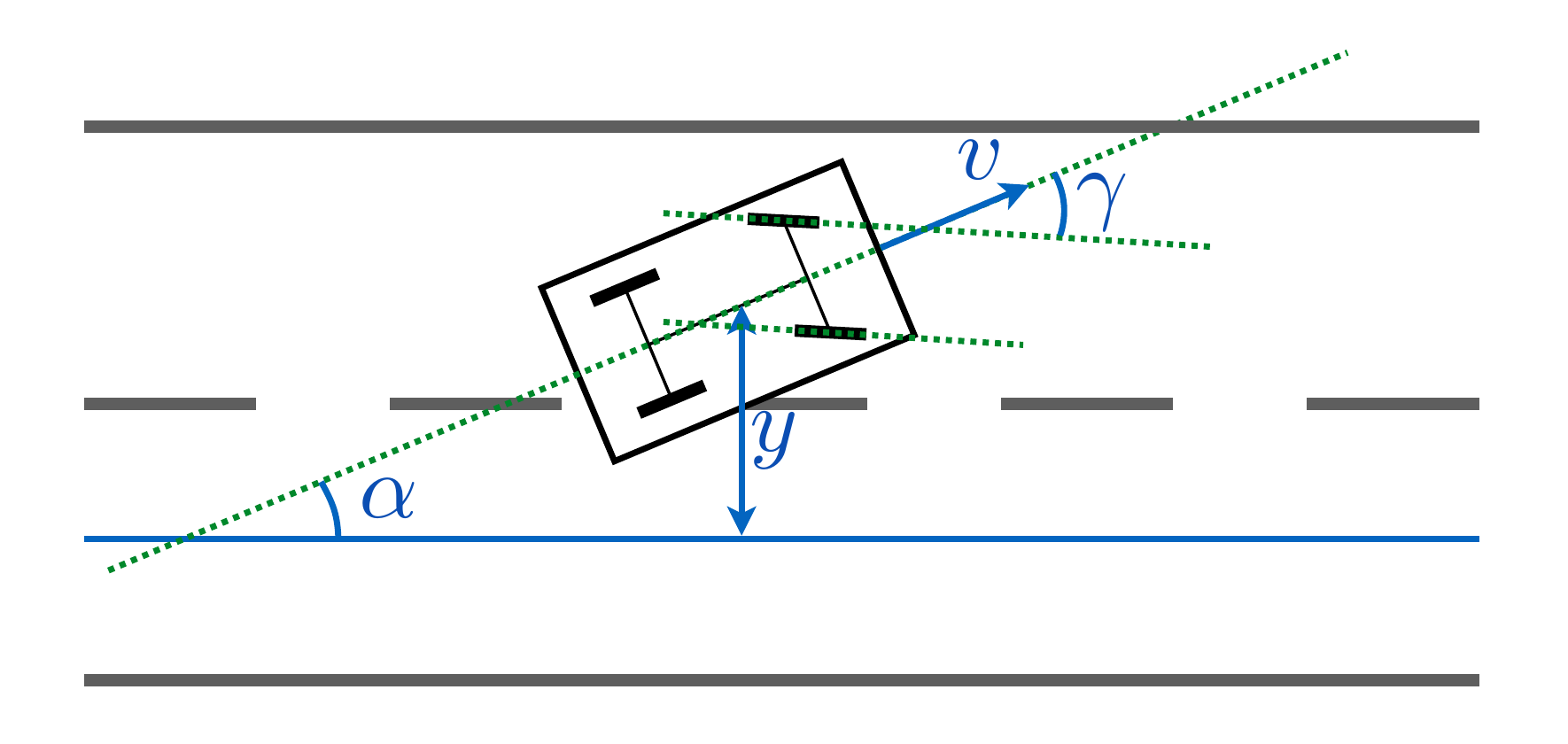}
      \end{center}
\caption{A schematic view of the bicycle model.}\label{fig:bicycle} 
\end{figure}
\end{example}
\subsubsection{Uninitialized RWS}
A RWS problem is \emph{uninitialized} if the initial set is the whole safe set $I = S$ ($S \implies S \ \U \ G$. To avoid technical difficulties for uninitialized RWS problems, we assume the safe set $S$ is a \emph{nondegenerate} basic semi-algebraic set.

\begin{definition}[Nondegenerate Basic Semialgebraic Set]\label{Def:basic-semialgebraic-set}
A nondegenerate basic semialgebraic set $K$ is a nonempty set defined by a conjunction polynomial inequalities:
\[ K :\ \{\vx\ |\ p_{K,1}(\vx) \leq 0\ \land\ \cdots\ \land p_{K,i}(\vx) \leq 0\} \,,\]  
where $\vx \in \reals^n$. For each $j \in [1,i]$, we define \[ H_{K,j} =
\{\vx\ |\ \vx \in K\ \land\ p_{K,j}(\vx) = 0\} \neq \emptyset \,.\]
\end{definition}

\subsubsection{RWS with Reference Trajectory}
For RWS with reference trajectory, in addition to sets $I$, $G$, and $S$, a \emph{valid} reference trace segment is also provided. The \emph{valid} trace segment, is defined over interval $[0, T]$: $\sigma_r :\ (\vx_r(\cdot), \vu_r(\cdot))$, where $\vx_r(\cdot) :\ [0, T] \mapsto X$ and $\vu_r(\cdot) :\ [0, T] \mapsto U$.
Then, given initial set $I \ni \vx_r(0)$, a goal set $G \ni \vx_r(T)$, and a safe set $S \ni \vx_r(t)$ ($\forall t \in [0, T]$), the goal is to reach $G$ from $I$ while remaining inside $S$:
\[
\varphi(\vx(\cdot)) \ :\ 
(\vx(0) \in I) \implies \left(\begin{array}{l}
    \vx(T) \in G \ \land \\
    (\forall t, 0 \leq t \leq T) \ \vx(t) \in S
\end{array} \right)\,.
\]
The difference between this property and the original RWS property is the fact that for this case, the reference trajectory $\vx_r(\cdot)$ (which is a feasible trajectory) is encoded into the specification (as a hint). Also, a time window $[0, T]$ is given a priori. This extra piece of information leads to a different solution as discussed later in this thesis.

We also consider another variation of this specification in which the time constraints are removed:

\[
\varphi(\vx(\cdot)) \ :\ 
(\vx(0) \in I) \implies (\exists \ T' \geq 0) \left(\begin{array}{l}
    \vx(T') \in G \ \land \\
    (\forall t, 0 \leq t \leq T') \ \vx(t) \in S
\end{array} \right)\,.
\]

\begin{example}[Obstacle Avoidance Problem]\label{ex:obstacle-car}
    Consider the bicycle model in Example~\ref{ex:bicycle}, where the state is $\vx :\ [\alpha, x, y, v]$ and inputs are $\gamma$ and $\tau$. Also, $l = 0.34$, $\tau \in [-4, 4]$, and $\gamma \in [-\frac{\pi}{4}, \frac{\pi}{4}]$. Consider a scenario where the vehicle moving with speed $2$ needs to circumnavigate an obstacle as shown in Figure~\ref{fig:obs-avoid}. First, a planner generates a reference trajectory $\vx_r$, which performs the task in $2$ time units (shown with the solid red line). Note that, by design, the reference trajectory keeps some distance from the obstacle. For RWS with reference trajectory, we define $S$ as
\[
S:\ \{ \vx \ | \ ([x, y] \oplus \B_{0.25}(\vzero)) \cap O = \emptyset \}\,,
\]
where $O$ is the obstacle and $0.25$ is the distance between the center of the vehicle and its corners (the body of the vehicle fits in a ball with radius $0.25$). This trick allows us to reason only about the center of the vehicle, and safety is guaranteed as long as the center of the vehicle is in $S$. Next, we set $I :\ \B_{0.25}(\vx_r(0))$ (a ball around the start point of the reference trajectory) and $G :\ \B_{0.5}(\vx_r(T))$ (a ball around the end point of the reference trajectory). Notice that the RWS property is defined over sets $I$, $G$, $S$, and horizon $T$. However, the reference trajectory $\sigma_r(\vx_r(\cdot), \vu_r(\cdot))$ is also provided as a hint.
\begin{figure}[t]
\begin{center}
    \includegraphics[width=0.45\textwidth]{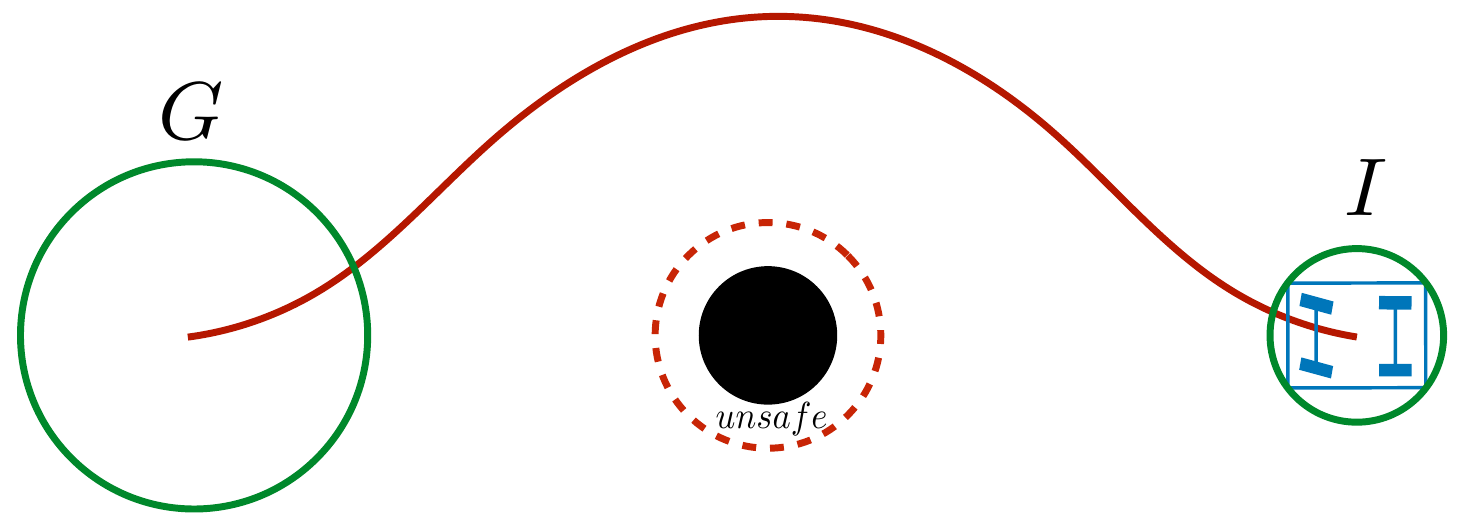}
\\ The reference trajectory (shown in solid red line) and sets are projected on $xy$ plane. Initial set $I$ and goal set $G$ are shown with green circles. The obstacle is the black circle, and the boundary of the safe region for the center of the vehicle is shown with the red dashed circle.
\end{center}
\caption{Obstacle avoidance for the bicycle model.}\label{fig:obs-avoid}
\end{figure}
\end{example}
\paragraph{Summary:} In this chapter, we discussed control system models, including smooth feedback systems and switched feedback systems. We formally defined the specifications we wish to solve for. Furthermore, we defined the control synthesis problem, along with several motivating examples.

%% file: related/related.tex
\chapter{Overview}\label{ch:overview}
 In this chapter, we provide background and discuss where our proposed method stands in the related work. We start with control verification problems and classify existing methods. Afterward, we consider extensions of these methods for control synthesis problems. Finally, we show how our contributions relate to existing work.

\section{Finding Certificates}
The correctness of a control system w.r.t. a specification is equivalent to the existence of a certificate, and a verification problem is equivalent to the problem of finding a certificate. However, such problems are undecidable even for linear hybrid systems~\cite{HENZINGER1998}. Nevertheless, all hope is not lost and incomplete methods have been investigated.
In these methods, the search for a certificate is restricted. The search is conducted over a search space $\Hy$, namely hypothesis space. If a certificate is found in $\Hy$, the correctness of the system is proven. However, failure to find a certificate does not translate to ``non-existence of certificates."

A search method has two components: (i) a search space, and (ii) a search tool. We discuss each component in the followings. Throughout this section, we use the following running example. Consider a discrete-time dynamical system, where the transition relation is defined as $\vx_{t+1} = F(\vx_{t})$. To guarantee safety w.r.t. compact sets $I$ and $S$ ($I \implies \Box S$) one needs to find an invariant set $Inv$ s.t. 
\begin{equation}\label{eq:running-example-invariant}
    \begin{array}{rl}
        (a):\ & I \subseteq Inv \\
        (b):\ & Inv \subseteq S \\
        (c):\ & \vx \in Inv \implies F(\vx) \in Inv\,.
    \end{array}
\end{equation}
We will discuss different methods for finding an invariant $Inv$.

\subsection{Search Space}
To find a control certificate, one needs to define a hypothesis space $\Hy$ and search in $\Hy$ for a member that is a certificate. $\Hy$ is defined in two ways. Either $\Hy$ is taken to be a finite set, or is defined using a set of parameters.

\paragraph{Finite $\Hy$:}
In this method, $\Hy$ is designed to be finite. Considering the running example, in the first step, the state space is discretized into a finite set of cells $\Gamma$ and each cell $\gamma \in \Gamma$ is a subset of $X$ ($\gamma \subset X$). The hypothesis space is the set of all subsets of $\Gamma$ ($\Hy:\ 2^\Gamma$). Now, the goal is to find $h^* \in \Hy$ s.t. $Inv : \bigcup_{\gamma \in h^*} \gamma$ is an invariant.

\paragraph{Parameterization:}
A more general approach is parameterization. The goal is to restrict the search space and search for a certificate with a specific structure (template). In this method, a template $T:\ \C \mapsto \Hy$ is defined over a set of parameters $\vc \in \C$. Suppose we wish to find an invariant for the running example. One can restrict the invariant to be a box. Therefore, we define the following template:
\begin{align}\label{eq:invariant-template}
T(\vc) :\ \{ \vx \ | \ \bigwedge_{i=1}^n \left( \ve_i^t \vx \geq c_{2i} \land \ve_i^t \vx \leq c_{2i+1} \right) \} \,,
\end{align}
where $\ve_i$ is the unit vector in the direction of the $i^{th}$ dimension. Then, $\Hy$ is defined to be $\Hy:\ \{T(\vc) \ | \ \vc \in \C\}$.

We note that the parameterization technique is more general as the template could include some finite discretization. For example, Huang et al.~\cite{huang2015controller} use piecewise constant template functions, Ravanbakhsh et al.~\cite{ravanbakhsh2014abstract} use union of template polyhedra and Wu et al.~\cite{wu1996induced} uses a set of quadratic functions.

\subsection{Search Tool}
Aside from defining $\Hy$, one needs a search tool to find a certificate. Fundamentally, there are two classes of search tools: fixed-point computation (FPC) and constraint solving. 

\paragraph{Fixed-point Computation:}
In this method, members of $ \Hy $ form a (possibly infinite) complete lattice. An order-preserving function $f:\Hy \mapsto \Hy$ is defined over the lattice in a way that a fixed-point of $f$ satisfies sufficient conditions for being a certificate. Therefore, the goal is to find a fixed-point in this lattice.

Consider the running example, where the hypothesis space is finite and defined through discretization $\Hy :\ 2^\Gamma$. Members of $\Hy$ form a complete lattice with $\emptyset$ at the bottom and $\Gamma$ at the top. Then, $f$ is defined as
\[
f(h) :\ \{\gamma \in \Gamma \ | \ \gamma \subseteq S \land (\forall \vx \in \gamma) \ F(\vx) \in \gamma' \land \gamma' \in h\} \cap h\,.
\]
In other words, $f(h)$ contains all the cells in $h$ which are safe, and if the state is in $f(h)$, the state remains in $h$ in the next time step. Now, for a fixed point $h^*$ ($h^* = f(h^*)$), if $\left(\bigcup_{\gamma \in h^*} \gamma\right) \supseteq I$, for all states in $\bigcup_{\gamma \in h^*} \gamma$, in the next step the state remains in $\left(\bigcup_{\gamma \in h^*} \gamma\right) \subseteq S$, and thus, $\bigcup_{\gamma \in h^*} \gamma$ is an invariant.
Using lattice theory (see ~\cite{tabuada2009} for details), to find the greatest fixed-point, simply $h$ is initialized to be $\Gamma$ (greatest member of the lattice). Then, iteratively $h$ is set to be $f(h)$, until a fixed-point $h = f(h)$ is reached. If $\left(\bigcup_{\gamma \in h} \gamma\right) \supseteq I$, then $\bigcup_{\gamma \in h} \gamma$ is an invariant. 
The fixed-point computation method has been investigated for finite systems~\cite{tabuada2009}. This method has connections to viability theory and Hamilton-Jacobi equations used for reachability analysis over continuous domains~\cite{aubin2011viability}. Solving Hamilton-Jacobi equations is hard and in practice, approximate solutions over bounded time intervals are considered~\cite{mitchell2005,spaceex2011,flowstar2013}.

\paragraph{Constraint Solving:}
Given a hypothesis space $\Hy$, one can define a set of constraints $\eta:\Hy \mapsto \bools$ s.t. if $\eta(h)$ holds, then $h$ is a certificate. Next, $(\exists h) \ \eta(h)$ is solved using constraint solvers. There are two possible outcomes: (i) no $h$ exists, which means no certificate with the specific structure exists, (ii) one $h$ is returned and $h$ is a certificate. 
Consider the running example, where the hypothesis space is parameterized using the template $T$ (Eq.~\eqref{eq:invariant-template}). In order to find an invariant, we simply use Eq.~\eqref{eq:running-example-invariant}:
\[
\eta(\vc) :\ T(\vc) \supseteq I \, \land \, T(\vc) \subseteq S \, \land \, \left(\vx \in T(\vc) \implies F(\vx) \in T(\vc)\right)\,.
\]

For dynamical systems, the problem of finding certificates for the stability of linear~\cite{boyd1994linear} and polynomial~\cite{papachristodoulou2002} systems have been addressed using constraint solving. Also, safety is discussed through barrier function~\cite{prajna2004safety}. Recently, Dimitrova et al. ~\cite{dimitrova2014deductive} have shown how such certificates can be extended to address more complicated specifications such as parity games~\cite{thomas2002automata}.

\section{Control Synthesis}
 Recall that for correct-by-construction controller design, we wish to find a function $\K$, which guarantees the closed-loop system correctness. However, the reach-avoid controller synthesis problem is shown to be undecidable even for simplest switched systems~\cite{krishna2017}. Nevertheless, there are two classes of incomplete solutions based on \emph{certificates}. 
 
 In the first set of solutions, in addition to $\Hy$, a search space $\mathfrak{K}$ is also defined to search for a proper feedback law $\K$. In other words, the search is conducted over $\Hy \times \mathfrak{K}$ to find a certificate $h$ and a feedback law $\K$ at the same time. Therefore, the control design problem is reduced to finding a ``certificate + feedback law." The second set of solutions is based on ``control certificates." A control certificate (i) provides a strategy to control the system, and (ii) the closed-loop system correctness is guaranteed if the strategy is respected by the feedback law. The existence of a control certificate is sufficient to address the control design problem. Once a control certificate is found, $\K$ is mechanically generated from the control certificate. Methods based on control certificates are less conservative as they do not enforce $\K$ to have a specific structure. Now, we discuss the details.
 
\subsection{Certificate + Feedback Law} Searching for a certificate and a feedback law at the same time is known as static feedback design, where feedback function $\K$ has a specific structure. All these methods use templates along with constraint solvers to find a solution. Static (fixed structure) feedback has been widely studied in control theory. For example, one of the fundamental approaches to address linear system design is LQR~\cite{boyd1994linear}, in which a linear feedback function $\K$ and a quadratic certificate (known as Lyapunov function) are searched simultaneously. An alternative is to use a semi-definite programming (SDP) formulation~\cite{lofberg2004yalmip}. However, the problem is harder for switched linear systems~\cite{lin2009survey} or nonlinear control systems~\cite{prieur1999uniting}. In fact, Prieur et al.~\cite{prieur1999uniting} show that the set of feasible solutions for such a problem may not only be non-convex but also disconnected. There have been attempts to solve such problems by using bilinear matrix inequality (BMI) solvers~\cite{henrion2005solving} or satisfiability modulo theories (SMT) solvers~\cite{huang2015controller}. An alternative method is to search for a local solution using alternating maximizations (a.k.a. policy iteration)~\cite{el1994synthesis,majumdar2013control}. Without going into the technical details, in these methods, starting from a feedback function and a potential certificate, alternatively (i) the feedback law is fixed to search for a better potential certificate, and (ii) the potential certificate is fixed to search for a better feedback law. This procedure continues until convergence to a local solution, and the method declares success only if the optimal local solution is feasible (the potential certificate is in fact a certificate).
 A similar approach is used in formal methods community, in which a feedback function is fixed, and the closed-loop system is verified (to search for a certificate). If the verification is successful, the problem is solved. Otherwise, the analysis returned by the verifier (usually a counterexample) is used to choose a better feedback function~\cite{raman2015reactive,Abate2017}. Others also provide similar methods to find a set of parameters for which the closed-loop system is correct~\cite{yordanov2008parameter,donze2009parameter}.
 Finally, Dimitrova et al. ~\cite{dimitrova2014deductive} have shown how to address parity games by finding a ``certificate + feedback law." While their work shows that constraint solving based methods can be applied toward more complicated specification, no way of finding such certificates is provided.

\subsection{Control Certificate}
Unlike the problem of finding a ``certificate + feedback law," the problem of finding a ``control certificate" has been investigated using both fixed-point computation (FPC) and constraint solving.

Returning to the running example, suppose we wish to solve a control synthesis problem. The transition relation is now defined as $\vx_{t+1} = F(\vx_{t}, {\color{red}\vu_{t}})$. To guarantee safety w.r.t. compact sets $I$ and $S$ ($I \implies \Box S$) one needs to find a \emph{control} invariant set $CInv$ s.t. 
\begin{equation}\label{eq:running-example-control-invariant}
    \begin{array}{rl}
        (a):\ & I \subseteq CInv \\
        (b):\ & CInv \subseteq S \\
        (c):\ & \vx \in CInv \implies {\color{red}(\exists \vu)} \ F(\vx, {\color{red}\vu}) \in CInv\,.
    \end{array}
\end{equation}

For FPC methods, a (possibly infinite) lattice is defined over the members of the hypothesis space and a function $f :\ \Hy \mapsto \Hy$ is defined s.t. a fixed-point of $f$ (namely $h^*$) yields a control certificate $CInv : \bigcup_{\gamma \in h^*} \gamma$. 

Consider the running example, where the hypothesis space is finite and defined through discretization $\Hy :\ 2^\Gamma$. Members of $\Hy$ form a complete lattice with $\emptyset$ at a bottom and $\Gamma$ at the top. Then, $f$ is defined as:
\[
f(h) :\ \{\gamma \in \Gamma \ | \ \gamma \subseteq S \land (\forall \vx \in \gamma) \ {\color{red}(\exists \vu \in U)} \ F(\vx, {\color{red}\vu}) \in \gamma' \land \gamma' \in h\} \cap h\,.
\]
In other words, $f(h)$ contains all the cells in $h$, which are safe, and if the state is in $f(h)$, there exists a control input to keep the state in $h$ in the next time step.
Now, for a fixed-point $h^*$ ($h^* = f(h^*)$), if $\left(\bigcup_{\gamma \in h^*} \gamma\right) \supseteq I$, for all states in $\bigcup_{\gamma \in h^*} \gamma$, there exists a control input to keep the state inside $\left(\bigcup_{\gamma \in h^*} \gamma\right) \subseteq S$, and thus $\bigcup_{\gamma \in h^*} \gamma$ is a control invariant.
FPC for finding control certificates is investigated for finite systems~\cite{tabuada2009} and dynamical systems.
For dynamical systems, Hamilton-Jacobi equation (used for reachability analysis) is extended to Hamilton-Jacobi-Bellman (HJB) equation~\cite{bellman1957dynamic,tomlin1998synthesizing}, wherein the ultimate goal is to find the maximum controllable region. However, solving the HJB equation is hard. 

The FPC technique gives a complete solution for finite systems~\cite{tabuada2009}. As such, the FPC is mainly used along with a finite hypothesis space. In these methods, the feedback space is usually taken to be finite (switched systems), and the transition relation is approximated for performance benefits. The time, on the other hand, can be discrete or continuous. Zamani et al.~\cite{zamani2012symbolic} develop a technique for time discretization, leading to state-of-the-art toolboxes~\cite{roy2011pessoa,rungger2016scots}. To address continuous-time systems, usually Zeno behavior~\cite{liberzon2012switching} needs to be considered~\cite{Asarin2000,liu2013synthesis,ozay2013computing,wongpiromsarn2011tulip}. 
These methods have been shown to be more scalable when compared to solving HJB equations. However, the complexity remains exponential in the number of state variables (because of discretization). 
To combat state space explosion, Mouelhi et al.~\cite{mouelhi2013cosyma} provide a tool which uses multi-scale abstractions meaning that the cells used for defining the search space are smaller for specific regions. While this method seems to scale better, it is not applicable to all kinds of systems. Another interesting technique~\cite{nilssonincremental} is to use counterexample-guided abstraction refinement (CEGAR)~ framework\cite{clarke2000counterexample}. In this technique, first, the problem is solved with some initial hypothesis space. If the method fails to find a control certificate, either there is no solution or the hypothesis space is not expressive enough. In the latter case, the hypothesis space is refined, and the problem is resolved. This trick helps to use smaller cells only when it is needed. 

While fixed-point computation (FPC) is mostly used with finite hypothesis spaces, it is possible to use FPC for finding control certificates in a template form~\cite{slanina2007controller,ravanbakhsh2014abstract} through abstract interpretation framework~\cite{cousot1977abstract}.

Constraint solving for finding control certificates has received less attention. 
Consider the running example, where the hypothesis space is parameterized using the template $T$ defined in Eq.~\eqref{eq:invariant-template}. In order to find a \emph{control} invariant, we simply use Eq.~\eqref{eq:running-example-control-invariant}:
\[
\eta(\vc) :\ T(\vc) \supseteq I \, \land \, T(\vc) \subseteq S \, \land \, \left(\vx \in T(\vc) \implies {\color{red}(\exists \vu \in U)} \ F(\vx, {\color{red} \vu}) \in T(\vc)\right)\,.
\]

Tan et al.~\cite{tan2004searching} investigate the problem of finding control Lyapunov functions (control certificates for designing stable systems). 
In another line of work, Taly et al.~\cite{taly2010switching,taly2009synthesizing} use a combination of simulations and quantifier elimination to find control certificates. The control certificates in the mentioned line of work are quite complicated, and the provided search method is not scalable.
 
 \section{Thesis Overview}
 In this thesis, the proposed solution is based on control certificates. We discuss the proposed framework in the next three chapters. In Chapter~\ref{ch:certificates}, we investigate control certificates to address the specifications introduced in the previous chapter. For each specification, we reduce the control synthesis problem to problem of finding a class of control certificates, and we show how a feedback law $\K$ is designed using a control certificate. In Chapter~\ref{ch:search}, the problem of finding control certificates is discussed. We use parameterization for defining the hypothesis space. Then, we employ constraint solvers to find a control certificate in the hypothesis space. We show that finding a control certificate is a computationally hard problem in general. We propose a combination of techniques from formal methods, control theory, and machine learning to tackle this problem. Finally, in Chapter~\ref{ch:eval}, we evaluate the effectiveness of the proposed method. We also show how automatically extracted controllers behave in simulations and practice.

%% file: certificate/certificate.tex
\chapter{Certificates}\label{ch:certificates}

In this chapter, we investigate different classes of control certificates for basic specifications. For each class of problems, we introduce a class of control certificates and show how a control synthesis problem is solved if a corresponding control certificate is available.

A \emph{certificate}, given a plant $\Psi(\P, \K)$ (see Definitions.~\ref{def:smooth-system} and~\ref{def:switched-system}), and a specification $\varphi$, guarantees that $\varphi$ holds for all traces. Each specification/control system has a corresponding class of certificates. For example, to prove the stability of a smooth closed-loop system, Lyapunov functions are used as certificates~\cite{boyd1994linear}. The function is named after Aleksandr Lyapunov, who introduced the concept of Lyapunov function for establishing the stability of ODEs. Recall that asymptotic stability consists of remaining close to the equilibrium state $ \vx_r $ (Lyapunov stability), and converging to the equilibrium $\vx_r$ (see Section~\ref{sec:spec}). To guarantee both conditions, a function $V :\ X \mapsto \reals$ with a unique (local) minimum $0$ for $\vx_r$ ($V(\vx_r) = 0$) is defined. For a given state trace $\vx(\cdot)$, we refer to $V(\vx(t))$ as the value of $V$ at time $t$. Intuitively speaking, if we could show that value of $V$ always decreases as time goes to infinity, and value of $V$ converges to its minimum $0$, the state converges to $\vx_r$.

\begin{definition}[Lyapunov Function]
    Given a plant $\P$ and a smooth feedback function $\K$, a radially unbounded smooth function $V:X \mapsto \reals$ is a Lyapunov function iff the following conditions hold:
    \[V(\vx_r) = 0\,, \ \ \ \ (\forall \vx \neq \vx_r) \ V(\vx) > 0 \,, \ \ \ \ (\forall \vx \neq \vx_r) \ \dot{V}(\vx) < 0\,,\]
    where $\dot{V}(\vx)$ is the derivative of $V$ w.r.t. time: $\dot{V}(\vx) = \nabla V(\vx) \cdot \dot{x} = \nabla V \cdot f(\vx, \K(\vx))$.
\end{definition}

\begin{theorem}[Lyapunov~\cite{lyapunov1992general}]
    Given a plant $\P$ and a smooth feedback function $\K$, the existence of a Lyapunov function guarantees global asymptotic stability of the closed-loop system.
\end{theorem}

In addition to proving a specification, certificates such as Lyapunov functions can be extended for correct-by-construction control design. More specifically, we address the control synthesis problem (Definition~\ref{def:problem}) through control certificates. The existence of a \emph{control} certificate guarantees the existence of a \emph{strategy} that satisfies the specification. Furthermore, a control certificate is ``constructive" if there is an automated method for designing a feedback law $\K$ that admits the strategy, and thus, respects the specifications. As the ultimate goal is designing the controller, in this thesis, we will focus on constructive control certificates (certificates for control design).

\input{certificate/stability}
\input{certificate/tracking}
\input{certificate/safety}

\section{Reach-While-Stay}\label{sec:rws}
For the \emph{reach-while-stay} (RWS) specification, the goal is to reach a target set $G$ from an initial set $I$, while staying in a safe set $S$, wherein $S$ and $I$ are compact and $I \subseteq \inter{S}$. For a verification problem, it is sufficient to first find a \emph{barrier-like} function $B$
\begin{align*}
\begin{array}{rrl}
\noindent\textbf{(a)}:\ & (\forall \vx \in I) & \ B(\vx) < 0 \\ 
\noindent\textbf{(b)}:\ & (\forall \vx \in \partial S) & \ B(\vx) > 0 \\ 
\noindent\textbf{(c)}:\ & (\forall \vx \in S \setminus (\inter{I} \cup \inter{G}) & \ B(\vx) = 0 \Rightarrow \nabla B.f(\vx) < 0 \,,
\end{array}
\end{align*}
where $f(\vx)$ is simply the vector field. Notice that for condition (c), we do not need to consider states in $G$. The existence of $B$ proves that $\Box(S \setminus G) \lor S \ \U \ G$. Then, by finding a \emph{Lyapunov-like} function, one can prove $\Diamond \lnot (S \setminus G)$ (prove $S \ \U \ G$):
\[
(\forall \vx \in S \setminus \inter{G}) \ \nabla V.f(\vx) < 0 \,.
\] 

Xu et al.~\cite{xu2015robustness} consider combination of control Lyapunov functions and control barrier functions. However, the control Lyapunov function in their method is merely used for performance improvement and does not provide formal correctness. Nguyen et al.~\cite{Nguyen2016CBF} provide a QP method to find an input which respects CLF and CBF conditions at the same time. However, unlike verification, searching for a control barrier-like function and a control Lyapunov-like function must be a joint search. For example, the feasibility of the QP method used in~\cite{Nguyen2016CBF} (input selection) is not formally guaranteed, unless the CLF and CBF are generated jointly. In other words, we should make sure that selecting an input $\vu$ which respects both CLF and CBF conditions is always feasible. A straightforward solution is to perform a joint search, and find a certificate which consists of two functions $(V, B)$, with the following conditions:
\begin{align} \label{eq:eq:general-clf-cbf}
\begin{array}{rrl}
\noindent\textbf{(a)}:\ & (\forall \vx \in I) & \ B(\vx) < 0 \\ 
\noindent\textbf{(b)}:\ & (\forall \vx \in \partial S) & \ B(\vx) > 0 \\ 
\noindent\textbf{(c)}:\ & (\forall \vx \in S \setminus \inter{G}) &\ (\exists \vu \in U) \ \nabla V.f(\vx, \vu) < 0 \\ 
\noindent\textbf{(d)}:\ & (\forall \vx \in S \setminus \inter{G}) & \ B(\vx) = 0 \Rightarrow (\exists \vu \in U) \left( \nabla B.f(\vx, \vu) < 0 \land \nabla V.f(\vx, \vu) < 0 \right) \,.
\end{array}
\end{align}
The first three conditions were discussed previously. The fourth condition guarantees that for the states on the barrier, there exists a control input for which both $V$ and $B$ decrease.

In the rest of this section, first we provide a classes of certificates with a simpler structure, when compared to Eq.~\eqref{eq:eq:general-clf-cbf}. Next, we extend this class to address RWS with reference tracking. Finally, we consider a more general class of control certificates, with multiple barrier functions.
\input{certificate/lyapunov-barrier}

\input{certificate/funnel}

\input{certificate/multiple-barrier}
\input{certificate/disturbance}

%% file: certificate/stability.tex
\section{Stability}
Lyapunov functions were extended to \emph{Control} Lyapunov Functions (CLF) by Artstein~\cite{artstein1983stabilization} for the stabilization of dynamical systems. Similar to stability analysis, we define a function $V :\ X \mapsto \reals$ with a unique (local) minimum $0$ only for $\vx_r$ ($V(\vx_r) = 0$). The controller, always, decreases the value of $V$ by choosing a proper feedback (Figure~\ref{fig:clf}). Then, as $V$ reaches its minimum $0$, the state converges to $\vx_r$.

\begin{definition}[Control Lyapunov Function]\label{def:clf}
    Given a smooth plant $\P$, a radially unbounded smooth function $V:X \mapsto \reals$ is a control Lyapunov function iff the following conditions hold:
    \begin{equation}\label{eq:clf}
V(\vx_r) = 0 \,, \ \ \ \ (\forall \vx \neq \vx_r) \ V(\vx) > 0 \,, \ \ \ \ (\forall \vx \neq \vx_r) \ (\exists \vu \in U) \ \dot{V}(\vx, \vu) < 0\,,
\end{equation}
    where $\dot{V}(\vx, \vu) = \nabla V \cdot f(\vx, \vu)$.
\end{definition}
Figure~\ref{fig:clf} shows the cenceptual view of a CLF, wherein at each non-equilibrium state, there is a value of the control input that can instantaneously decrease the value of the CLF.
Given a CLF, suppose that, the controller always selects a control input $ \vu $ (given current state $\vx$) s.t. the value of $ V $ decreases ($ \nabla V. f(\vx(t), \vu) < 0 $). This strategy leads to decrease of the value of $V$, which in turn guarantees the closed-loop system asymptotic stability. Artstein proved such $\K$ exists~\cite{artstein1983stabilization}. 

\begin{figure}[t]
\begin{center}
\includegraphics[width=0.4\textwidth]%
    {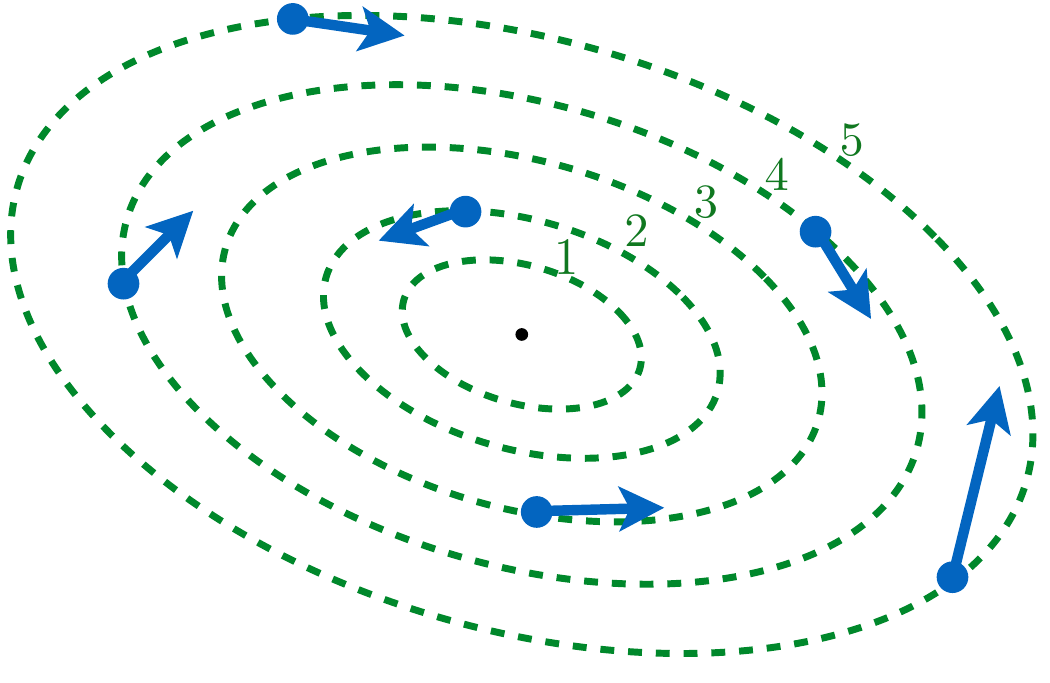}
\\ The black dot is the equilibrium $\vx_r$ for the stability property. Green dashed lines show different level sets of a CLF. The values for the level sets are shown in green. For some states (blue dots), the vector field of the closed-loop system is shown with blue arrows.
    \caption{A schematic view of a CLF.}
\label{fig:clf}
\end{center}
\end{figure}

\begin{theorem}[Artstein~\cite{artstein1983stabilization}]
Given a smooth plant $\P$, the existence of a control Lyapunov function guarantees the existence of an almost-everywhere smooth feedback law $\K$, which guarantees asymptotic stability of the closed-loop system.
\end{theorem}

This theorem proves there exists an almost everywhere continuous feedback function $\K$, where the continuity holds over $X \setminus \vx_r$, which is practically acceptable. 
However, Artstein theorem is not constructive and not applicable to control design problems. Sontag~\cite{sontag1989universal} provides a constructive method for designing not only continuous but also smooth feedbacks laws (over $X \setminus \vx_r$) using CLFs. Furthermore, Sontag shows $\K$ is smooth (even in $\vx_r$) under the ``small control property" assumption: for any $\epsilon > 0$, there exists a $\delta > 0$ s.t. $\K(\vx) \in \B_\epsilon(\vu_r)$ for all $\vx \in \B_\delta(\vx_r)$, where $f(\vx_r, \vu_r) = \vzero$.
This method is based on two assumptions. The first assumption indicates that inputs are not saturated ($U:\reals^m$). The second assumption states that the system is \emph{control affine}. I.e.,
\begin{equation}\label{eq:control-affine-plant}
    \dot{x} = f(\vx, \vu) = f_0(\vx) + \sum_{i=1}^m f_i(\vx) u_i \,.
\end{equation}
While this assumption seems restrictive, many control problems have this property. Additionally, if the dynamics of a system is not affine in control, one could approximately model the system using a control affine system. Such approximation is possible by adding additional states and integrators. In the rest of this thesis, we only address control affine dynamics for \emph{smooth} plants. Following Sontag, others have provided methods for cases where the inputs are saturated, and $U$ is a ball~\cite{LIN1991UNIVERSAL} or a polytope~\cite{suarez2001global}. Moreover, Nguyen et al.~\cite{Nguyen2015RobustCLF} show how to implement effective controllers from CLFs.

\begin{example}\label{ex:clf}
    Consider a one-input, two-states system with dynamics $\dot{x_1} = -x_2 \,, \dot{x_2} = x_1 + u$. To stabilize to the origin, we use CLF: $V(x_1, x_2) = x_1^2 + x_2^2 + x_1x_2$. Note that $\dot{V}(x_1, x_2, u) = (2x_2-x_2) u - x_1^2 + x_2^2$.
     Using Sontag formula the following feedback law is extracted from $V$:
    \[
    \K(x_1, x_2) = \begin{cases}
        0 & 2x_2-x_1 = 0 \\
        -\frac{-x_1^2 + x_2^2 + \sqrt{(-x_1^2+x_2^2)^2 + (2x_2-x_1)^2}}{2x_2-x_1} & 2x_2-x_1 \neq 0\,,
    \end{cases}
    \]
    which yields the vector field shown in Figure~\ref{fig:clf-example}.
\begin{figure}[t]
\begin{center}
\includegraphics[width=0.4\textwidth]%
    {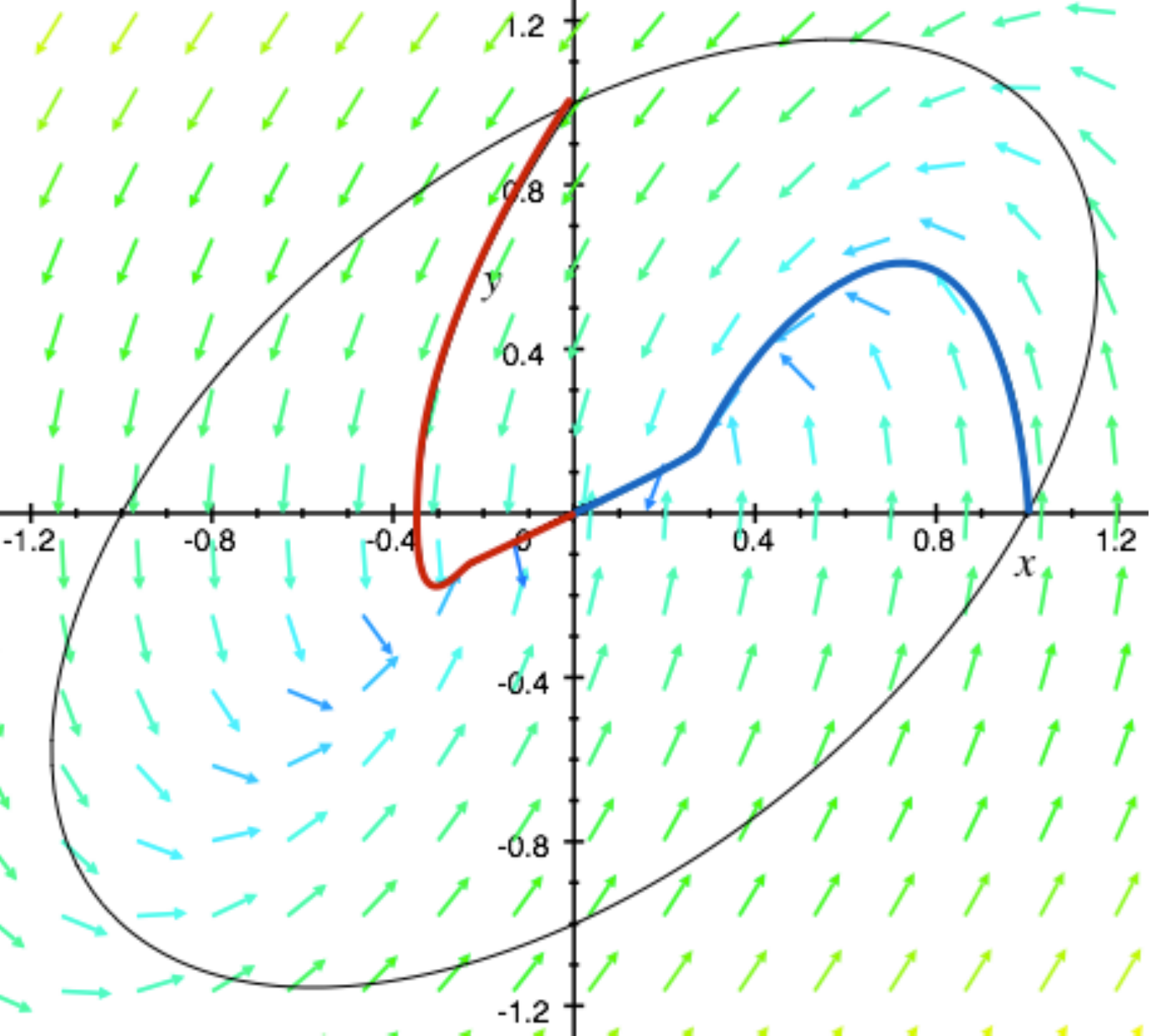} \\
The ellipse is a level set of $V$ ($\{\vx \ | \ V(\vx) = 1\}$). Initial state for the red trace is $[0, 1]$ and for the blue trace is $[1, 0]$.
    \caption{Vector field and two execution traces for the closed-loop system described in Example~\ref{ex:clf}.}
\label{fig:clf-example}
\end{center}
\end{figure}

\end{example}

\subsection{Non-Zeno CLF}
So far, we discussed the CLF for \emph{smooth} feedback systems. However, designing switched feedback functions using the strategy provided by a CLF is not straightforward as Zeno behavior may occur. We introduce non-Zeno CLFs which, are constructive control certificates for \emph{switched} feedback systems.
In this section, we define a large class of CLFs that can be used to synthesize controllers with guaranteed ``minimum dwell-time," i.e., the time between two switches is bounded from below. For this class of CLFs, the min dwell-time, and thus asymptotic stability, is only guaranteed on a compact set rather than globally. In particular, the bound on the minimum dwell-time is inversely proportional to the diameter of this set.

Before going further, we need to define some terminologies. Without loss of generality, we assume that the desired equilibrium is the origin $\vx_r = \vzero$. Let $\partial K$ and $\inter{K}$ denote the boundary and interior of a closed set $K$. Let $V^{\bowtie \beta} :\ \{ \vx \ | \ V(\vx) \bowtie \beta \}$ where $\bowtie \in \{ < , \leq, =, \geq, > \}$.
\begin{definition}[Associated Region] 
Given a smooth CLF $V$ and a compact region $P$, $\vzero \in \inter{P}$, we associate region $P^*$ to $V$ as $P^* :\ V^{\leq \beta}$, where $\beta:\ \min_{\vx \not\in \inter{P}} V(\vx)$.
\end{definition}
Notice that $V$ is radially unbounded and thus $\beta$ exists. Moreover, since $V$ has only one local minimum (at the origin), $V(\vx) = \beta \implies \vx \in \partial P$. Therefore, $P^* \subseteq P$.
\begin{definition}[$\phi$-boundedness] Given functions $p, \phi :\ X \mapsto \reals$, $p$ is said to be $\phi$-bounded iff for every bounded region $S \subset X$ there exists a constant $\Lambda_S$ s.t.  $(\forall \vx \in S) \ p(\vx) \leq \Lambda_S \phi(\vx)$.
\end{definition}
\begin{example}
Consider $\phi(x,y):\ x^2 +y^2$.  Any multivariate polynomial $p(x,y)$ whose lowest degree terms have degree at least two is $\phi$-bounded. Examples include $x^2 + 2x^3 + 3 xy$, $xy$, and $x^6- 3 y^3$. On the other hand, the function $p(x,y)= x+y$ is not $\phi$-bounded since no bound of the form $p(x,y) \leq \Lambda_S (x^2 + y^2)$ exists when $S$ is taken to be a region containing $(0,0)$.  Similarly, the function $3 + x$ is not $\phi$-bounded.
\end{example}

\begin{definition}[Non-Zeno CLF] \label{def:clf-non-zeno} Given a switched plant $\P$,  a CLF is said to be non-Zeno iff there exist positive definite functions $\phi_\vu :\ X \mapsto \reals$ s.t.
\begin{align} 
    &\ddot{V}_\vu(\vx) \mbox{ is } \phi_\vu\mbox{-bounded} \label{eq:phi-q-1} \\
    &\dot{\phi}_\vu(\vx) \mbox{ is } \phi_\vu\mbox{-bounded} \label{eq:phi-q-2}\\
    &    (\forall \vx \neq \vzero) \ (\exists \ \vu \in U) \
    \dot{V}_\vu(\vx) < - \phi_\vu(\vx), \label{eq:RCLF}
\end{align}
where $\dot{V}_\vu(\vx) = \nabla V \cdot f_\vu(\vx)$, $\ddot{V}_\vu(\vx) = \nabla \dot{V}_\vu \cdot f_\vu(\vx)$ and $\dot{\phi}_\vu(\vx) = \nabla \phi_\vu \cdot f_\vu(\vx)$.
\end{definition}
Informally, the goal is to make sure not only $\dot{V}_\vu$ is negative definite, but also is smaller than a class of negative  (definite) functions. Now we explain how such property helps to guarantee min-dwell time property. As depicted in Figure~\ref{fig:min-dwell-time}, we force the system to switch to a mode $\vu$ for which $\dot{V}_{\vu}(\vx) \leq - \phi_{\vu}(\vx) < 0$. Also, we keep the system in mode $\vu$ until $\dot{V}_\vu(\vx) \geq -\frac{\phi_\vu(\vx)}{\lambda}$, wherein $\lambda > 1$.

Given a non-Zeno CLF $V$, let a class of suitable feedback laws $\SK$ associated to $V$ be defined as 
\begin{equation}
    \label{eq:controller}
    \K \in \SK \iff \K(\vu, \vx) \mbox{ matches} \begin{cases} 
    \vu^* \hspace{0.5cm}
        \dot{V}_{\vu}(\vx) \geq -\frac{\phi_\vu(\vx)}{\lambda} \wedge \dot{V}_{\vu^*}(\vx) \leq -\phi_{\vu^*}(\vx)  \land \vx \in P\\
    \vu \hspace{0.5cm} \dot{V}_{\vu}(\vx) < -\frac{\phi_\vu(\vx)}{\lambda} \land \vx \in P \\
    \overline{\vu} \hspace{0.5cm}  \vx \not\in P \,,
    \end{cases}
\end{equation} 
wherein $\lambda > 1$ is a chosen scale constant. In other words, as long as $\vx \in P$, after switching to a mode $\vu$, rather than switching when the CLF $\dot{V}_\vu(\vx) = 0$, we force the system to switch when $\dot{V}_\vu(\vx) \geq -\frac{\phi_\vu(\vx)}{\lambda}$. We also force the system to switch to a mode $\vu^*$ for which $\dot{V}_{\vu^*}(\vx) \leq - \phi_{\vu^*}(\vx)$. The definition of a non-Zeno CLF guarantees that such a mode $\vu^*$ will exist.

\begin{figure}[t]
\begin{center}
\includegraphics[width=0.4\textwidth]%
    {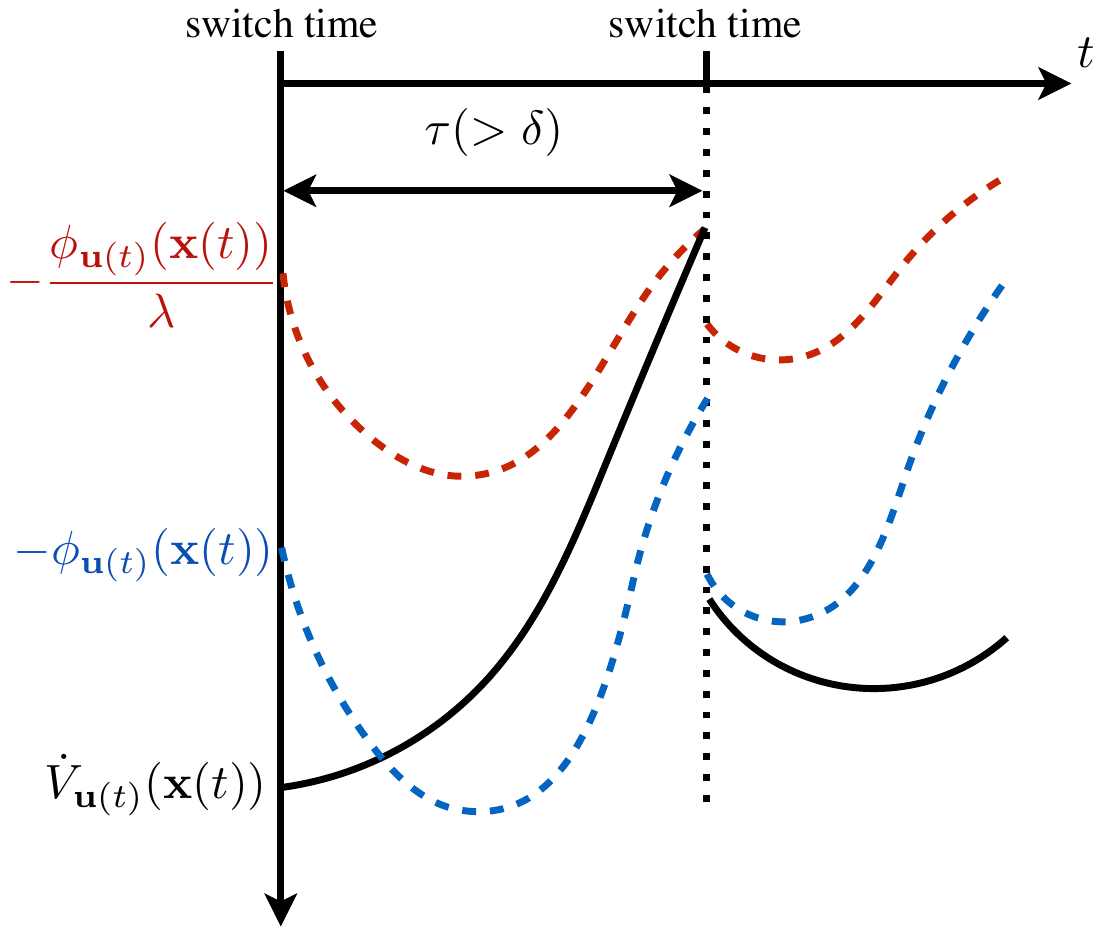} \\
Right after switching, $\dot{V}_{\vu(t)}(\vx(t)) < -\phi_{\vu(t)}(\vx(t))$ and the next switch occur only if $\dot{V}_{\vu(t)}(\vx(t)) \geq -\frac{\phi_{\vu(t)}(\vx(t))}{\lambda}$ ($\lambda = 2$).
    \caption{A conceptual view of $\dot{V}$ over time.}
\label{fig:min-dwell-time}
\end{center}
\end{figure}

The key observation here is that the constraints on $\ddot{V_\vu}$, $\dot{V_\vu}$, $\dot{\phi_\vu}$ altogether guarantee that when the controller switches at time $T$, the controller need not switch again in interval $[T, T + \delta)$ for some fixed  $\delta > 0$ (i.e.,  $\dot{V}_{\vu}(\vx(t)) \leq -\frac{\phi_\vu(\vx)}{\lambda}$ for all $t \in [T, T+\delta]$ ). A bound for $\delta$ is given directly in the proof of the following theorem.

\begin{theorem} \label{thm:asymp-stable} Given a plant $\P$, a compact regions $P$, and a non-Zeno CLF $V(\vx)$ (and associated region $P^*$ for $V$ w.r.t. $P$), assuming $\vx(0) \in P^*$, a switching function that admit the description of Eq.~\eqref{eq:controller} results in a system which satisfies the following properties. \begin{compactenum}
   \item all the traces of the system are time-divergent,
   \item $P^*$ is a positive invariant,
   \item system is asymptotically stable.
   \end{compactenum}
\end{theorem}

We note that stability is an infinite horizon property and we need to make sure $\vx(t) \neq \bot$ ($\vx(\cdot)$ is defined) for all time $t$. The first property guarantees Zeno behavior is avoided and the second property guarantees that the trace does not escape to infinity in finite time. Having these properties, it is guaranteed that the trace is defined at all times and $\vx(t) \neq \bot$ (see Chapter~\ref{ch:intro} for details).

\begin{proof}
Consider the class of suitable feedback laws $\SK$ defined in Eq.~\eqref{eq:controller}. We note that $\SK$ is non-empty as $\vu^*$ exists by construction of the CLF $V$. Also, we assume if $\vx(t) = \vzero$ at a switch time, then the controller switches to a mode $\vu_0$ where $f_{\vu_0}(\vx) = \vzero$ and the state remains in $\vx$ forever, without switching.

In the first step, we show that there exists a min dwell-time between two switching times as long as the state remains in set $P$. Assume there is a switch time $t_1$ s.t. $\vx(t_1) \in P^*$ (and as mentioned $\vx(t_1) \neq \vzero$) and mode switches to $\vu$. Thus,
\begin{equation}
    \label{eq:dot-v-t1-eq}
    \dot{V}_{\vu}(\vx(t_1)) = \nabla V . f(\vx(t_1), \K(\vx(t_1), \vu(t_1))) < - \phi_\vu(\vx(t_1))\,.
\end{equation}
Let $t_2 > t_1$ be the next time instance when the controller switches to mode $\vu^*$ and for all $t \in [t_1, t_2]$, $\vx(t) \in P$.
By definition of the controller, we can conclude
\begin{equation}
    \label{eq:dot-v-t2-eq}
    \dot{V}_{\vu}(\vx(t_2)) \geq -\frac{\phi_\vu(\vx(t_2))}{\lambda}\,.
\end{equation}
It is sufficient to show $\delta = t_2 - t_1$ has a lower bound and it can not be arbitrarily small.

From Eqs.~\eqref{eq:phi-q-1} and~\eqref{eq:phi-q-2} and boundedness of $P$ there are constants $\Lambda_1$ and $\Lambda_2$  s.t. for all $\vx \in P$
\begin{align}
\ddot{V}_{\vu}(\vx) &\leq \Lambda_1 \phi_\vu(\vx) \label{eq:ddot-v-upper-bound}\\
\dot{\phi}_{\vu}(\vx) &\leq \Lambda_2 \phi_\vu(\vx)\,.
\label{eq:dot-phi-upper-bound}
\end{align}

From the fact that $\vx(t) \in P$ and Eq.~\eqref{eq:dot-phi-upper-bound}, we get
\begin{align*}
(\forall t \in [t_1, t_2]) \ \ \
    \phi_\vu(\vx(t)) = \phi_\vu(\vx(t_1)) + \int_{t_1}^{t} 
                                \dot{\phi}_{\vu}(\vx(\tau)) d\tau 
                            \leq \phi_\vu(\vx(t_1)) + \int_{t_1}^{t} 
                                \Lambda_2 \phi_\vu(\vx(\tau)) d\tau\,,
\end{align*}
and therefore $(\forall t \in [t_1, t_2])$
\begin{equation}
    \label{eq:phi-t-ineq}
    \phi_\vu(\vx(t)) \leq e^{\Lambda_2 \delta} \phi_\vu(\vx(t_1))\,.
\end{equation}

A lower bound on $\dot{V}_{\vu}(\vx(t_2))$ by Eq.~\eqref{eq:dot-v-t2-eq}
\begin{equation}
    \label{eq:dot-v-t2-lower-bound}
    \dot{V}_{\vu}(\vx(t_2)) \geq -\frac{\phi_\vu(\vx(t_2))}{\lambda} 
    \overset{Eq.~\eqref{eq:phi-t-ineq}}{\implies}                         \dot{V}_{\vu}(\vx(t_2)) \geq 
                             -\frac{e^{\Lambda_2 \delta} 
                             \phi_\vu(\vx(t_1))}{\lambda} \,.
\end{equation}

Also
\begin{align*}
(\forall t \in [t_1, t_2])& \\
    \dot{V}_{\vu}(\vx(t)) &= \dot{V}_{\vu}(\vx(t_1)) + \int_{t_1}^{t} 
                                \ddot{V}_{\vu}(\vx(\tau)) d\tau \\
    \overset{Eq.~\eqref{eq:ddot-v-upper-bound}}{\implies}
                        &\leq \dot{V}_{\vu}(\vx(t_1)) + \Lambda_1 \int_{t_1}^{t} 
                                \phi_\vu(\vx(\tau)) d\tau \\
    \overset{Eq.~\eqref{eq:phi-t-ineq}}{\implies}
                        &\leq \dot{V}_{\vu}(\vx(t_1)) + \Lambda_1 \int_{t_1}^{t} 
                                e^{\Lambda_2 \delta} \phi_\vu(\vx(t_1)) d\tau\,,
\end{align*}
and therefore an upper bound on $\dot{V}_{\vu}(\vx(t_2))$ is
\begin{align}
    \dot{V}_{\vu}(\vx(t_2)) &\leq \dot{V}_{\vu}(\vx(t_1)) + \Lambda_1  
                                e^{\Lambda_2 \delta} \phi_\vu(\vx(t_1)) \delta 
                                \nonumber \\
    \overset{Eq.~\eqref{eq:dot-v-t1-eq}}{\implies}
    &< -\phi_\vu(\vx(t_1)) + \Lambda_1  
                                e^{\Lambda_2 \delta} \phi_\vu(\vx(t_1)) \delta\,.
    \label{eq:dot-v-t2-upper-bound}
\end{align}

From Eqs.~\eqref{eq:dot-v-t2-lower-bound} and~\eqref{eq:dot-v-t2-upper-bound}

\begin{align*}
     -\frac{e^{\Lambda_2 \delta} \phi_\vu(\vx(t_1))}{\lambda}
     \leq 
     \dot{V}_{\vu}(\vx(t_2))
     <
     -\phi_\vu(\vx(t_1)) + \Lambda_1  
                                e^{\Lambda_2 \delta} \phi_\vu(\vx(t_1)) \delta\,,
\end{align*}
and finally assuming $\vx(t_1) \not= \vzero$ ($\phi(\vx(t_1) \neq 0$), we have $\phi(\vx(t_1)) > 0$:
\begin{align*}
    -\frac{e^{\Lambda_2 \delta}}{\lambda} < 
    - 1 + \Lambda_1  
                                e^{\Lambda_2 \delta} \delta
    \implies 1 &< \frac{\lambda \Lambda_1 e^{\Lambda_2 \delta} \delta}
    {(\lambda - e^{\Lambda_2 \delta}) } = h(\delta)\,.
\end{align*}

Notice that if $\delta \geq \frac{\log(\lambda)}{\Lambda_2}$, then $\delta$ has a lower-bound. Otherwise $h(\delta) \geq0$. Furthermore
\begin{compactenum}
\item $0 \leq e^{\Lambda_2 \delta} < \lambda \iff h(\delta) > 0$
\item $h$ is a monotone function of $\delta$ in domain
  $0 \leq e^{\Lambda_2 \delta} < \lambda$ by showing that
  $\frac{dh}{d\delta}$ is positive
    \item $h(0) = 0$ and 
    $\lim_{\delta \rightarrow \frac{\log(\lambda)}{\Lambda_2}} h(\delta) = +\infty$ ($h:\reals^+ \mapsto \reals^+$).
\end{compactenum}
Therefore, $h^{-1}:\reals^+ \mapsto \reals^+$ is defined and
$h^{-1}(1) < \delta$. Therefore, $h^{-1}(1)$ is a
lower bound on $\delta$, and min dwell-time exists as long as $\vx \in P$. 

 Next, we prove that $P^*$ is a positive invariant. Recall that $P^* :\ V^{\leq \beta} \subseteq P$ is a compact set containing $\vzero$. 
 Assume $\vx(0) \in P^*$. We obtain $V(\vx(0)) \leq \beta$.  Also, the $\K$ function ensures that as long as $\vx(t) \in P$,
$\dot{V}_{\vu(t)}(\vx(t)) \leq -\frac{\phi_{\vu(t)}(\vx(t))}{\lambda} < 0$ by definition of $\K$. Therefore
\begin{align*}
    V(\vx(t_b)) = V(\vx(0)) + \int_0^{t_b} \dot{V}_{\vu(t)}(\vx(t)) \ dt \leq V(\vx(0))\,.
\end{align*}
Since $V(\vx(0)) \leq \beta$, we have $V(\vx(t_b)) \leq \beta$ for all $t_b \geq 0$. Therefore, by definition $\vx(t) \in P^*$.

In the next step of the proof, we want to show the system is asymptotically stable. Since $P^*$ is a compact set, $(\forall t > 0) \ \vx(t) \in P^*$ and time diverges, by Bolzano-Weierstrass Theorem~\cite{bartle2011introduction}, $\vx(t)$ converges to some $\vx^* \in P^*$.
Assume $\vx^* \neq \vzero$ and therefore $\min_\vu (\phi_\vu(\vx^*)) = R > 0$. By continuity of $\phi_\vu$ and divergence of time, one can find $\epsilon > 0$ s.t.
\begin{align*}
(a): & \ (\exists T > 0) \ (\forall t \geq T) \ \vx(t) \in \scr{B}_{\epsilon}(\vx^*) \subseteq P^* \\
(b): & \ (\forall \vu \in U) \ (\forall \vx \in \scr{B}_{\epsilon}(\vx^*)) \ 
\phi_\vu(\vx) \geq \frac{R}{2}\,.
\end{align*}
Also, $V$ is bounded in $\scr{B}_{\epsilon}(\vx^*)$ and decreases through time. Formally,
\begin{align*}
(\forall t \geq T) \ \dot{V}_{\vu(t)}(\vx(t)) &\leq 
-\frac{\phi_{\vu(t)}(\vx(t))}{\lambda}
                                 \leq     -\frac{R}{2 \lambda}\,.
\end{align*}
As a result 
\begin{align*}
V(\vx(T+t)) &= V(\vx(T)) + \int_{T}^{T+t} \dot{V}_{\vu(\tau)}(\vx(\tau)) d\tau \\ 
             &\leq V(\vx(T)) -\frac{R}{2 \lambda} t\,.
\end{align*}
which means eventually $V$ becomes negative as time goes to infinity and that is a contradiction. Therefore, $\vx^* = \vzero$ and the system is asymptotically stable.
\QED
\end{proof}

\paragraph{Implementation:}
Once a non-Zeno CLF is found, the controller can be implemented in many ways. We can implement an operational amplifier circuit that selects the appropriate mode by computing $\phi(\vx)$ and $\dot{V_\vu}(\vx)$ from the state feedback $\vx$. Such a circuit will not need to know the minimum dwell time: however, the minimum dwell time provides us with a guideline on the maximum permissible delay.

Another approach is to find an under-approximation of min-dwell time $\underline{\delta}$ and use a discrete time controller that changes the modes every $\underline{\delta}$ time units. Yet another software-based solution is to use a model predictive control scheme, where the controller switches to a mode $\vu$ at time $t_s$ given $\vx(t_s)$ ($\dot{V}_\vu(\vx(t_s)) < - \phi_\vu(\vx(t_s))$). Also, the controller predicts the first time instance $t_f > t_s$ s.t. $\dot{V}_{\vu}(\vx(t_f)) \geq -\frac{\phi_\vu(\vx(t_f))}{\lambda}$.
Then the controller sets a wake-up timer for time $t=t_f$ and re-evaluates at that point. The minimum dwell time provides a design guideline on the shortest possible wake-up time $t_f$.

%% file: certificate/tracking.tex
\section{Reference Tracking}
Stability for non-holonomic systems is a challenging problem. Brockett~\cite{brockett1983} showed that even a simple unicycle model is not stabilizable using continuous feedback laws. However, continuous feedback laws exist for stabilization to non-stationary \emph{reference trajectories}. In this section, we consider stability w.r.t. a reference trajectory using control Lyapunov functions. We discuss smooth feedback laws derived from the CLFs, noting that the case of switched feedback law extraction may be derived using the ideas from the previous section.

\subsection{Trajectory Tracking}
Stability property appears in trajectory tracking, wherein the goal is to stabilize to a reference (feasible) trajectory $\vx_r(t)$, assuming $\vx_r(t)$ (and $\vu_r(t)$) is provided as input.
Formally, let $\vx_d(t):\ \vx(t) - \vx_r(t)$ describe the deviation from the reference trajectory state at time $t$. As illustrated in Figure~\ref{fig:new-system}(a), the goal is to stabilize $\vx_d$ to the equilibrium $\vzero$ under the time-varying reference frame that places $\vx_r(t)$ as the origin at time $t$. The dynamics for $\vx_d$ is defined as: $\dot{\vx}_d = f(\vx, \vu) - \vr(t)$, wherein $\dot{\vx}_r = \frac{d\vx_r}{dt} = \vr(t)$. Furthermore, $\K$ is a function of both $\vx$ and $t$.

In order to stabilize $\vx_d$ ($\vx(t) \rightarrow \vx_r(t)$ as $t \rightarrow \infty$), we define a CLF as a function $V$ over $\vx_d$, that respects the following constraints:
\begin{equation}\label{eq:tt-clf}
\begin{array}{l}
    \noindent\textbf{(a)}:\ V (\vzero) = 0\ \mbox{and}\ (\forall\vx_d \neq \vzero)\; V(\vx_d) > 0 \\
    \noindent\textbf{(b)}:\ (\forall\ t \geq 0,\ \vx_d \neq \vzero)\; (\exists \vu \in U) \; \dot{V}(t, \vx, \vu) < 0\,.\\
\end{array}
\end{equation}
Note that $V$ is a function of $\vx_d$ and depends on $t$ only through $\vx_d$. Also its derivative depends on $\vx_d$, $t$, and $\vu$:
\begin{align*}
\dot{V}(t, \vx_d, \vu) & = \nabla V(\vx_d) \cdot \dot{\vx}_d = \nabla V(\vx_d) \cdot (f(\vx_d + \vx_r(t), \vu) - \vr(t))\,.
\end{align*}

Conditions (a) and (b) are identical to those in Definition~\ref{def:clf}, requiring the function $V$ to be positive definite, and the ability to choose controls to decrease $V$. The only difference is the fact that the dynamics are time-varying and $f$ depends on $t$.

\begin{theorem}\label{thm:trajectory-tracking}
  Given a plant $\P$, a reference trajectory $\vx_r(\cdot)$, and a radially unbounded smooth function $V$ which respects Eq.~\eqref{eq:tt-clf}, there exists an almost-everywhere smooth feedback law $\K$, for which $\vx_d$ stabilizes to the origin.
\end{theorem}

\begin{proof}
    We appeal directly to Sontag's result to obtain an almost-everywhere smooth feedback function $\K(\vx_d, t)$ that guarantees that $\dot{V}(\vx_d(t), t) < 0$ with $\vx_d(t) \not= 0$~\cite{sontag1989universal}. The feedback law is guaranteed to be smooth everywhere except at the origin.

 Since $\dot{V}(\vx_d(t), t) \leq 0$, and $V(\vx_d(t))$ is radially unbounded, La Salle's theorem~\cite{lasalle1960} guarantees that the dynamical system will stabilize to the 0 level set $V^{=0}:\ \{\vx_d |\ V(\vx_d) = 0 \} = \{\vzero\}$.
  \QED
\end{proof}

One of the key drawbacks of trajectory tracking is that it specifies the reference trajectory $\vx_r(t)$ along with the \emph{reference timing}, wherein the state $\vx_r(t)$ must ideally be achieved at time $t$. This poses a challenge for control design unless the timing is designed very carefully. Imagine, a reference trajectory that traverses a winding and hilly road at constant speeds. This compels the control to continually accelerate the vehicle on upslopes only to ``slam the brakes'' on downhill sections~\cite{Hauser+Saccon/2006/Motorcycle}.

\begin{figure}[t]
\begin{center}
    \includegraphics[width=0.8\textwidth]{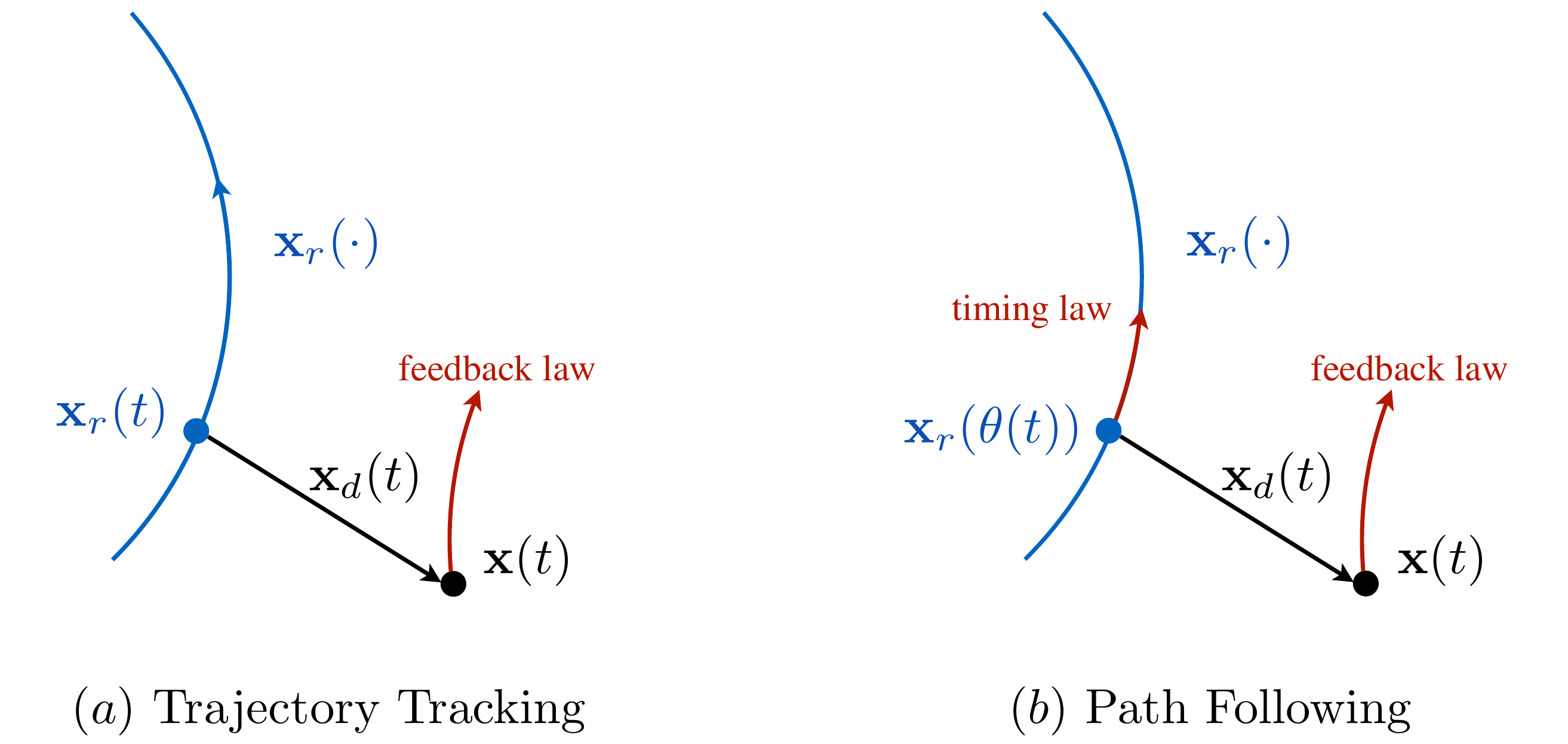}
\end{center}
\caption{A schematic view of a system along with a reference trajectory.}\label{fig:new-system} 
\end{figure}

\subsection{Path Following}
As an alternative to trajectory tracking, path following allows the user to specify a reference path parameterized with a scalar, $\theta$ (instead of time), $\vx_r(\theta)$ yields a state for each $\theta$ and $\frac{d\vx_r}{d\theta} = \vr(\theta)$.

As proposed by Hauser et al., one could define $\Pi$ as a function that maps a state $\vx$ to the closest state on the reference trajectory $\vx_r(\cdot)$, using an auxiliary map $\pi$ ~\cite{HAUSER1995,Saccon+Hauser+Beghi/2013/Virtual}:
\[ \pi(\vx):\ \underset{\theta}{\mbox{argmin}} \ ||\vx -
  \vx_r(\theta)||_P^2 \ \,, \ \Pi(\vx) :\ \vx_r(\pi(\vx)) \,, \] where $||\vx||_P^2 :\ \vx^t P \vx$ is a Lyapunov function for the linearized dynamics around the reference trajectory.  In order to stabilize the system to the reference path, Hauser et al.\ propose to decrease the value of $||\vx - \Pi(\vx)||_P^2$. However, as the projection function $\pi$ can get complicated, they use local approximations of $\pi$.

Following this, others have proposed to design a control law for a virtual input $u_0$ that controls $\theta$ as a function of time (called the \emph{timing feedback law}). In other words $u_0$ controls the progress (or sometimes regress) along the reference~\cite{SKJETNE2004,AGUIAR2004}. Therefore, the deviation is now defined as $\vx_d(t):\ \vx(t) - \vx_r(\theta(t))$\,, wherein $\diff{\theta}{t} = u_0$. As depicted in Figure~\ref{fig:new-system}(b), $\theta$ is mapped to a state on the path $\vx_r(\theta)$. For example, Faulwasser et al.~\cite{Faulwasser2014} design the timing law as a function of $\theta$, the deviation ($||\vx_d||$ for state feedback systems), and their higher derivatives:
\[
g(\theta^{(k)},||\vx_d||^{(k)},\ldots,\theta,||\vx_d||,u_0) = 0\,,
\]
wherein $\theta^{(k)}$ is the $k^{th}$ derivative of $\theta$. However, defining the function $g$ is a nontrivial problem.

We now present the design of a path following scheme by specifying a timing law as well as control for deviation from the reference trajectory based on a control Lyapunov function.

Let input path $(\vx_r(\cdot), \vu_r(\cdot))$ be a valid trace when  $\dot{\theta} = 1$. We define a new coordinate system, in which the state of the system is $\vz^t :\ [\theta, \vx_d^t]$, wherein $\vx(t) = \vx_d(t) + \vx_r(\theta(t))$ is the original state of the system. Also, the control inputs are $\vv^t :\ [u_0, \vu^t]$.  We assume $\theta$ is directly controllable using $u_0$ ($\dot{\theta} = u_0$). Therefore, $\dot{\vx}_r = \frac{d \vx_r}{d\theta} \dot{\theta} = \vr(\theta) u_0$, and thus
\[
\dot{\vx}_d = f(\vx_d + \vx_r(\theta), \vu) - \vr(\theta) u_0 \,.
\]

Our goal is to design a control that seeks to stabilize $\vx_d = \vzero$, which in turn guarantees $\vx(t) \rightarrow \vx_r(\theta(t))$  as $t \rightarrow \infty$ meaning $\vx(\cdot)$ converges to path $\{\vx_r(\theta) \ | \ \theta \in \reals\}$.  We define a CLF as a function $V$ over $\vx_d$, that respects the following constraints:
\begin{equation}\label{eq:pf-clf}
\begin{array}{l}
    \noindent\textbf{(a)}:\ V (\vzero) = 0\ \mbox{and}\ (\forall\vx_d \neq \vzero)\; V(\vx_d) > 0 \\
    \noindent\textbf{(b)}:\ (\forall\ \theta,\ \vx_d \neq \vzero)\; (\exists \vu \in U, u_0 \in U_0) \; \dot{V}(\theta, \vx, u_0, \vu) < 0\,,\\
\end{array}
\end{equation}
where $U_0:\ [\underline{u_0}, \overline{u_0}]$. Note that $V$ depends on $\theta$ only through $\vx_d$. Also, its derivative depends of $\vx_d, \theta, u_0, \vu$.
\begin{align*}
\dot{V}(\theta, \vx_d, u_0, \vu) & = \nabla V(\vx_d) \cdot \dot{\vx}_d  = \nabla V(\vx_d) \cdot (f(\vx_d + \vx_r(\theta), \vu) - \vr(\theta) u_0)\,.
\end{align*}

The $(\forall\ \theta)$ quantifier in condition (b) guarantees that this decrease is achieved no matter where the current reference state $\vx(\theta)$ lies relative to the current state $\vx_d$. Also, the value of $u_0$ (rate of change for $\theta$) is obtained through the feedback law derived from the CLF $V$.

\begin{theorem}\label{thm:path-following}
  Given a plant $\P$, a reference trajectory $\vx_r(\cdot)$, and a radially unbounded smooth function $V$ which respects Eq.~\eqref{eq:pf-clf}, there exists an almost-everywhere smooth feedback law $\K$ for which $\vx_d(\cdot)$ stabilizes to the origin.
\end{theorem}
\begin{proof}
    We appeal directly to Sontag's result to obtain an almost-everywhere smooth feedback function $\K(\vx_d, \theta)$ which guarantees that $\dot{V}(\vx_d(t), \theta(t)) < 0$ for all $(\vx_d(t), \theta(t))$ with $\vx_d(t) \not= 0$~\cite{sontag1989universal}. The feedback law is guaranteed to be smooth everywhere except at the origin.

  Since $\dot{V}(\vx_d(t), \theta(t)) \leq 0$, and $V(\vx_d(t))$ is radially unbounded, La Salle's theorem~\cite{lasalle1960} guarantees that the dynamical system will stabilize to the 0 level set $V^{=0}:\ \{(\vx_d, \theta) |\ V(\vx_d) = 0 \} = \{ (\vzero, \theta)\ |\ \theta \in \reals \}$, which is the reference trajectory since $\vx_d(\cdot) = \vzero$. \QED
\end{proof}

Thus far, we have just demonstrated how CLFs can be used to decide on a timing law as well as a control to nullify the deviation to $0$. However, this does not address a key requirement of progress: we need to ensure that $\dot{\theta} > 0$ so that we make progress along the reference from one end to another. This can be enforced by saturating $u_0 = \dot{\theta}$:$u_0 \in [\underline{u_0}, \overline{u_0}]$, where $0 < \underline{u_0} \leq 1 \leq \overline{u_0} < \infty$.
Notice that we force $1 \in [\underline{u_0}, \overline{u_0}]$ to ensure that the reference trajectory remains feasible.

%% file: certificate/safety.tex
\section{Safety}\label{sec:certificate-safety}
We now consider control certificates for proving safety properties. A Control Barrier Function (CBF) is a control certificate for safety property. Recall that safety is defined over sets $I$ and $S$ ($I \implies \Box S$). To avoid finite escape time, we assume $S$ and $I$ are compact, and $I \subseteq \inter{S}$ for simplicity. A control barrier function is a smooth function $ B $, which is positive on the boundary of the safe set and negative on the boundary of the initial set. Also, value of $ B $ decreases when $ B(\vx) = 0 $. 

\begin{definition}[CBF\cite{taly2010switching,xu2015robustness,blanchini2008set}]
    Given a plant $\P$, a smooth function $B$ is a CBF iff the following conditions hold 
\begin{equation} \label{eq:cbf-original}
\begin{array}{rrl}
\noindent\textbf{(a)}:\ & (\vx \in \partial S) &\ B(\vx) > 0 \\
\noindent\textbf{(b)}:\ & (\vx \in I) & \ B(\vx) < 0    \\
\noindent\textbf{(c)}:\ & (\vx \in S \setminus \inter{I}) & \ \left( B(\vx) = 0  \implies (\exists \vu \in U)\  \dot{B}(\vx, \vu) \ < \ 0 \right) \,.
\end{array}
\end{equation}
\end{definition}

Intuitively $B^{=0} \cap S$ forms a barrier and $\partial S$ is unreachable (Figure~\ref{fig:barrier}).
Eq.~\eqref{eq:cbf-original}, combined with the smoothness of $B$ and $f$ ensures that as soon as the state is adequately ``close" to the barrier, it is possible to choose a control mode that ensures the local decrease of the $B$.
In other words, set $ P^* :\ (B^{\leq 0} \cup I) \cap S$ is a control invariant, and it is guaranteed that there exists a switching strategy for which the trace never leaves $ P^* $ (safety is guaranteed).

We note that especial care is necessary if one wishes to use the following condition instead of  condition (c) in Eq.~\eqref{eq:cbf-original}:
\[
(\vx \in S \setminus \inter{I}) \ \left( B(\vx) = 0  \implies (\exists \vu \in U) \ \dot{B}(\vx, \vu) \ {\color{red} \leq} \ 0 \right)\,,
\]
as discussed in~\cite{blanchini2008set,taly2009synthesizing,ravanbakhsh2017class}.

\begin{figure}[t]
\begin{center}
\includegraphics[width=0.4\textwidth]%
    {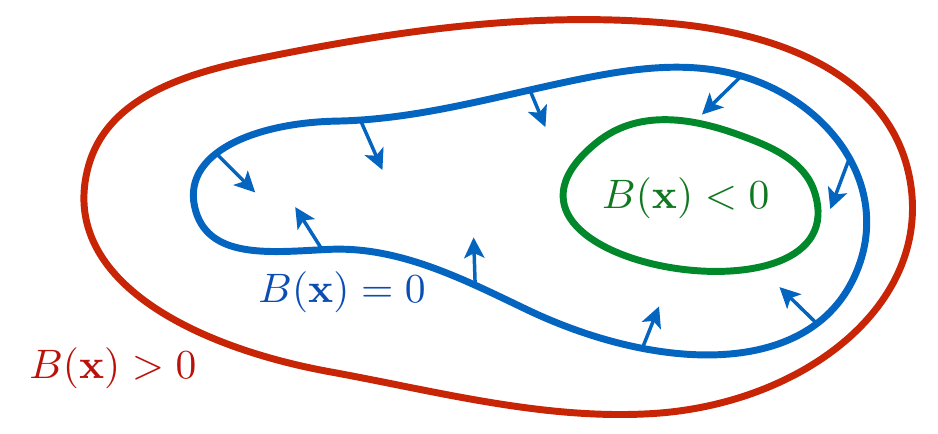}
\\ Boundary of safe and initial regions are shown in red and green, respectively. The barrier is shown in blue.
    \caption{A schematic view of a control barrier function.}
\label{fig:barrier}
\end{center}
\end{figure}

\begin{theorem}[Wieland et al.~\cite{WIELAND2007462}]\label{thm:cbf} Given a plant $\P$, sets $I$ and $S$, and a control barrier function $ B $, there is a \emph{smooth} feedback law which guarantees $ I \Rightarrow \Box S $ for the closed-loop system.
\end{theorem}

As we will discuss later, it is hard to search for such a function in our framework because of the equality constraint. Inspired by Kong et al.~\cite{kong2013exponential}, we find that the following relaxation is particularly effective in our experiments:
\begin{equation}
\label{eq:cbf-safety}
    \begin{array}{rrl}
\noindent\textbf{(a)}:\ & (\vx \in \partial S) & \ B(\vx) > 0 \\
\noindent\textbf{(b)}:\ &  (\vx \in  I) & \  B(\vx) < 0 \\
\noindent\textbf{(c)}:\ &  (\vx \in S \setminus \inter{I}) & \ (\exists \vu \in U) \ \left( \begin{array}{c}
        \dot{B}(\vx,\vu) - \lambda^* B(\vx) < 0 
        \lor
        \dot{B}(\vx,\vu) + \lambda^* B(\vx) < 0
    \end{array} \right) \,,
    \end{array}
\end{equation}
for some constant $\lambda^*$.

Intuitively, by choosing $\lambda^* = 0$, the conditions are similar to that of Barrier functions, introduced by Prajna et al.~\cite{prajna2004safety}. Also, as $|\lambda^*|$ gets larger, the conditions get less conservative. In fact, for a large enough $\lambda^*$, new conditions are equivalent to original ones, considering smoothness of $B$ and $f$, and compactness of $S$.
This would yield a trade-off, the degree of which is decided by $\lambda^*$. We also note that this formulation is less conservative than the one introduced by Kong et al.~\cite{kong2013exponential} as our formulation uses two exponential conditions which only forces decrease in the value of $B$ around $\partial P^*$.
Xu et al.~\cite{xu2015robustness} discuss conditions for the so-called ``control zeroing'' barrier functions for safety. The difference lies in the fact that our method provides a logical formulation (using disjunction), while their approach uses a bilinear formulation.

We now show how a \emph{switched} feedback law is extracted from a CBF. Given a CBF $B$ satisfying Eq.~\eqref{eq:cbf-safety}, the choice of a switching mode is dictated by a function $\cond^*_\vu(\vx) :\ \min (\dot{B}_\vu(\vx) +\lambda^* B(\vx), \dot{B}_\vu(\vx) -\lambda^* B(\vx))$ defined for each state $\vx \in S \setminus \inter{I}$ and mode $\vu \in U$.
First, we note that by compactness of $I$, there exists a $\lambda > \lambda^*$ s.t.
\[
(\forall \vx \in {\color{red} S}) \ (\exists \vu \in U) \ \cond_\vu(\vx) < 0\,,
\]
where $\cond_\vu(\vx) :\ \min (\dot{B}_\vu(\vx) +\lambda B(\vx), \dot{B}_\vu(\vx) -\lambda B(\vx))$. Then, by compactness of $S$
\[
(\forall \vx \in S) \ (\exists \vu \in U) \ \cond_\vu(\vx) < {\color{red}-\epsilon}\,,
\]
for some constant $\epsilon > 0$. Ultimately, we wish $\cond_{\vu(t)}(\vx(t)) < 0$ for all $t \geq 0$.
The idea is that whenever (at time $T$) the controller switches to a mode $\vu$, we make sure $\cond_\vu(\vx(T)) < -\epsilon$. Moreover, one can guarantee $\cond_\vu(\vx(t)) < 0$ for all time $t \in [T, T+\delta]$, for some minimum time $\delta > 0$ (there is no need for switching).
We define a class of suitable feedback laws on compact sets to discuss the details.

\begin{definition}[Class of Suitable Feedback Laws] Given a plant $\P$, a compact set $S$, and condition functions $\cond_\vu :\ X \mapsto \reals$ (for all $\vu \in U$), a class of suitable feedback laws $\SK$ is defined as
\begin{equation}\label{eq:suitable-feedbacks}
    \K \in \SK(\cond, S, \veps) \iff \K(\vu, \vx) \mbox{ matches} \begin{cases} 
    \vu^* & 
        \cond_{\vu}(\vx) \geq - \epsilon_s \land \cond_{\vu^*}(\vx) < - \epsilon \land \vx \in S\\
    \vu &  \cond_{\vu}(\vx) < -\epsilon_s \land \vx \in S \\
    \overline{u} & \vx \not\in S \,,
    \end{cases}
\end{equation}
where $\veps :\ (\epsilon, \epsilon_s)$ and $0 < \epsilon_s < \epsilon$.
\end{definition}
In other words, as long as $\vx(\cdot) \in S$, for a feedback function $\K \in \SK(\cond, S, \epsilon)$, the controller mode $\vu$ persists when $\cond_\vu(\vx(t)) < -\epsilon_s$ and the mode changes only when $-\epsilon_s \leq \cond_\vu(\vx(T+\delta))$ and for all $t \in [T, T+\delta] \ \cond_\vu(\vx(t)) \leq -\epsilon_s$. We wish to show under some mild conditions, there is a lower bound on $\delta$ (min dwell-time exists).

\begin{lemma}[Min Dwell-time for $\SK$]~\label{lem:min-dwell-exists} Given a switched plant $\P$, and a compact set $S$, assuming (i)
    $(\forall \vx \in S) \ (\exists \vu \in U ) \ \cond_\vu(\vx) < 0$, (ii)
    $\cond_\vu$ is continuous and piecewise differentiable, and (iii) $\dot{x}$ is bounded on $S$,
    there exists $\veps$, s.t. $\SK(\cond, S, \veps)$ is non-empty, and for any $\K \in \SK(\cond, S, \veps)$, (a) min dwell-time exists, and (b) $\cond_{\vu(t)}(\vx(t)) \leq -\epsilon_s$, as long as $\vx(\cdot) \in S$.
\end{lemma}
\begin{proof}
     $\cond_\vu$ is continuous and piecewise differentiable. Then, by compactness of $S$, it is guaranteed that 
\begin{equation}\label{eq:mode-exists}
(\forall \vx \in S) \ (\exists \vu \in U) \ \cond_\vu(\vx) < -\epsilon\,,
\end{equation}
 for some $\epsilon > 0$. Let $\epsilon_s > 0$ be a design parameter s.t. $\epsilon_s < \epsilon$. $\SK(\cond, S, \veps)$ is non-empty as $\vu^*$ exists for the first case by Eq.~\eqref{eq:mode-exists}.

We now show min dwell-time exists.
     Suppose the controller switches to mode $\vu$ at time $T$. Let $T + \delta$ be the earliest time instant, where $\cond_{\vu}(\vx(T+\delta)) \geq - \epsilon_s$ while at the same time
    \[ 
        (\forall t \in (T , T + \delta ])\  \vu(t) = \vu,\ \vx(t) \in S \,.
    \]
    At time $T$, $\cond_{\vu}(\vx(T)) < -\epsilon$ and at time $T+\delta$, $\cond_\vu(\vx(T+\delta)) \geq -\epsilon_s$. Note that $\cond_\vu$ is a continuous and piecewise differentiable function of $\vx$. As a result, there is a Lipschitz constant $A_\vu$ such that
        \[ | \cond_\vu(\vx(T+ \tau)) - \cond_\vu(\vx(T)) | \leq A_\vu ||\vx(T+\tau) - \vx(T) || \,. \]
        Now, $\vx(t)$ in the interval $t \in [T,T+\tau]$ is the solution of an ODE where $\dot{x}$ is bounded. As a result, there exists a constant $B_\vu$ such that 
        \[ \|\vx(T+\tau) - \vx(T) \| \leq B_\vu \tau \,. \]
        Combining, we have $ | \cond_\vu(\vx(T+ \tau)) - \cond_\vu(\vx(T)) |  \leq  \Lambda_\vu \tau $ wherein $\Lambda_\vu = A_\vu B_\vu$. 

    Let $\Lambda = \max_{\vu \in U} \Lambda_\vu$. Let us choose a $\delta$ such that 
    \begin{align} \label{eq:delta-lower-bound}
       \delta :\ \frac{\epsilon - \epsilon_s}{\Lambda} \,.
    \end{align}
    The above arguments show that for all $t \in [T,T+\delta]$,
    \[
    | \cond_{\vu}(\vx(T+t)) - \cond_\vu(\vx(T)) | \leq \Lambda_\vu t \leq \Lambda \delta \leq \epsilon - \epsilon_s \,.
      \]
          Therefore, using that $\cond_\vu(\vx(T)) < -\epsilon$, we obtain for all $t \in [T,T+\delta]$, $\cond_\vu(\vx(T+t)) < -\epsilon_s$, and the controller would not switch in interval $[T, T+\delta)$. Therefore, min dwell-time exists as long as $\vx(\cdot) \in S$. Also, by the definition (Eq.~\eqref{eq:suitable-feedbacks}), $\cond_{\vu(t)}(\vx(t)) \leq -\epsilon_s$ as long as $\vx(\cdot) \in S$. \QED
 \end{proof}
 
 Now, we can show any feedback function $\K \in \SK(\cond, S, \veps)$ is a solution to the safety property.

\begin{theorem}
\label{thm:safety-switch}
    Given a switched plant $\P$, sets $I$ and $S$ and a control barrier function $ B $ (satisfying Eq.~\eqref{eq:cbf-safety}), there exist $\lambda$ and $\veps$ s.t. $\SK(\cond, S, \veps)$ (wherein $\cond_\vu :\ \dot{B}_\vu +\lambda B$) is non-empty, and any feedback function $\K \in \SK(\cond, S, \veps)$  guarantees the safety property defined by $S,I$:\ $I \implies \Box S$.
\end{theorem}
\begin{proof}
    Recall that $P^* :\ B^{\leq 0} \cap S$.
    By condition (a) of Eq.~\eqref{eq:cbf-safety}, $I \subset B^{< 0}$ ($I \subset \inter{B^{\leq 0}}$). Therefore, $I \subset \inter{P^*}$ ($\vx(0) \in \inter{P^*}$). Also, by condition (b) of Eq.~\eqref{eq:cbf-safety} $B^{\leq 0} \subset \inter{S}$ and therefore $P^* \subset \inter{S}$.
    
    By Eq.~\eqref{eq:cbf-safety} and compactness of $I$, there exists a $\lambda > \lambda^*$ s.t.
    \[
    (\forall \vx \in {\color{red} S}) \ (\exists \vu \in U) \cond_\vu(\vx) < 0\,, 
    \]
    where $\cond_\vu(\vx):\ \min(\dot{B}_\vu(\vx) +\lambda B(\vx),\dot{B}_\vu(\vx) -\lambda B(\vx))$. By Lemma~\ref{lem:min-dwell-exists} there exists $\veps$ s.t. $\K \in \SK(\cond, S, \veps)$ is non-empty. Moreover, for any $\K \in \SK(\cond, S, \veps)$ (i) there exists a min dwell-time and (ii) $\cond_{\vu(t)}(\vx(t)) \leq -\epsilon_s$ as long as $\vx(t) \in S$.
    Initially $\vx(0) \in P^*$. As long as the trace remains in $P^* \subseteq S$, time diverges and safety holds. If the safety property is violated, $\vx(\cdot)$ must reach boundary of $P^*$ and leave $P^*$ at some time $t$.
    Let $T$ be the first time this happens. This means that $B(\vx(T)) = 0$ and $\dot{B}_{\vu(T)}(\vx(T)) \geq 0$. However, for all time $t \leq T$,  
    \[
    \cond_{\vu(t)}(\vx(t)) :\ \min (\dot{B}_{\vu(t)}(\vx(t)) + \lambda B(\vx(t)) \lor \dot{B}_{\vu(t)}(\vx(t)) - \lambda B(\vx(t))) \leq -\epsilon_s \,,
    \]
which is a contradiction. As a result, the trace never leaves $P^*$ and time diverges. \QED
\end{proof}

\begin{example}[Inverted Pendulum]\label{ex:control-implementation}
Consider an inverted pendulum on a cart. The problem is to keep the pendulum in a vertical position $S :\ \{[\theta \ \omega]^t \ | \ \theta \in [-1, 1] \land \omega \in [-3, 3] \}$, having $I:\ \B_{0.5}(\vzero)$. The dynamics of the system is described by the ODEs
\[
\dot{\theta} =\omega \, \ \ , \ \ \, \dot{\omega}  =\frac{g}{l}sin(\theta)-\frac{h}{ml^2}\omega+\frac{1}{ml}cos(\theta)u \,,
\]
where $g = 9.8$, $h = 2$, $l = 2$, and $m = 0.5$. We use Taylor expansion to approximate the trigonometric function. We assume the control feedback is discrete, $\vu \in U :\{-30, 0, 30\}$. A CBF $B$ is provided: $B([\theta \ \omega]^t) :\ 10 \ \theta^2 + 1.5312 \ \theta\omega + 2.5859 \ \omega^2$.
In order to implement a discrete-time controller, we need a minimal dwell-time $\delta$ for switching strategy. Such $\delta$ exists by Theorem~\ref{thm:safety-switch}, and we can find a lower-bound for $\delta = 0.2ms$ from the expressions derived in its proof (using $\lambda^* = 0$, $\epsilon = 0.05$, $\epsilon_s = 0.01$).
We implement the plant and the controller in a Simulink diagram in MATLAB and simulate the system for $2.5s$ for initial state $[\theta \ \omega] = [1 \ -2]$. Figure~\ref{fig:simulation-inverted}  shows the simulation trace for $\theta$, $\omega$, and $B$. In fact, initially $\vx(0) \not\in I$. Nevertheless, the state is in the controllable region and once value of $B$ reaches $0$, the state will remain safe forever.

\begin{figure*}[t!]
\begin{center}
\includegraphics[width=0.5\textwidth]%
    {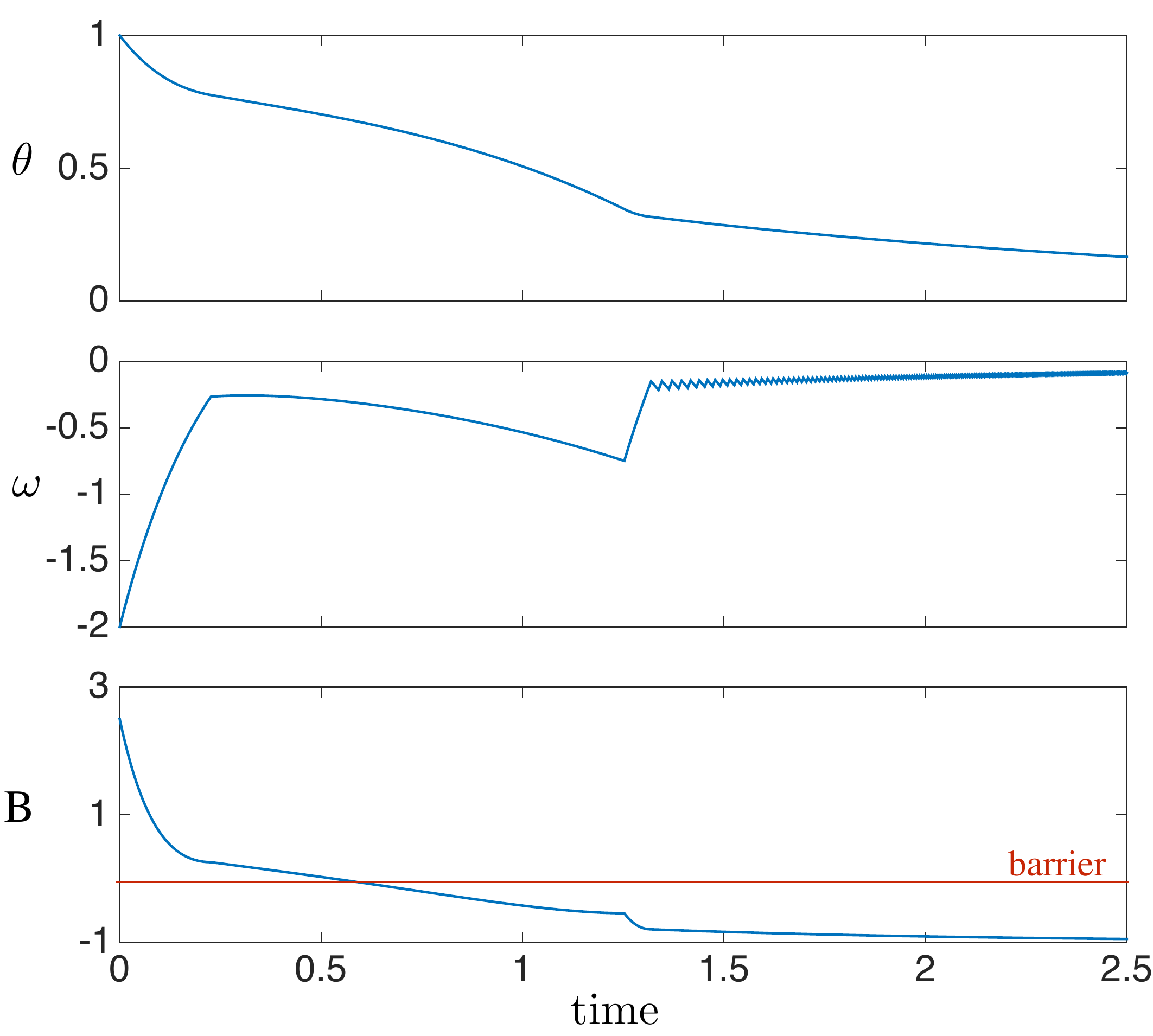}
\caption{Execution traces for Example~\ref{ex:control-implementation}.}
\label{fig:simulation-inverted}
\end{center}
\end{figure*}
\end{example}

%% file: certificate/lyapunov-barrier.tex
\subsection{Basic RWS}

The original control certificate $(V, B)$ for RWS can be simplified when we assume $ V = B $. In fact, function $ V $ can play both roles.
\begin{definition}[Control Lyapunov-Barrier Function] A smooth function $V$ is a control Lya-punov-barrier function (CLBF) iff
\begin{equation} \label{eq:cllf-rws}
    \begin{array}{rrl} 
    \noindent\textbf{(a)}:\ & (\forall \vx \in I) & \ V(\vx) < 0 \\ 
    \noindent\textbf{(b)}:\ & (\forall \vx \in \partial S) & \ V(\vx) > 0 \\ 
    \noindent\textbf{(c)}:\ & (\forall \vx \in S \setminus \inter{G}) & \ (\exists \vu \in U) \ \nabla V.f(\vx, \vu) < 0 \,.\end{array}
\end{equation}
\end{definition}

Using these conditions, the controller always provides a control input for decreasing the value of $V$. Then, having $\vx(0) \in P^* :\ V^{\leq 0} \cap S$, the trace never leaves $P^*$ without entering $G$ (Figure~\ref{fig:rws}). Moreover, as the value of $V$ decreases, it is guaranteed that the trace cannot stay in $P^* \setminus G$ and thus $I \implies S \ \U \ G$. 

\begin{lemma}\label{lem:rws}
Given a plant $\P$, compact sets $ G, I, S $, a smooth function $ V $ satisfying Eq.~\eqref{eq:cllf-rws}, and a feedback function $\K$ s.t. for all traces (wherein $\vx(0) \in I$) (i) $\dot{V}(\vx(t)) \leq -\epsilon_s$ , and (ii) time progresses as long as $\vx(t) \in S \setminus \inter{G}$, it is guaranteed that $(\exists T \geq 0)$ s.t.
\begin{compactenum}
    \item $(\forall t \in [0, T]) \ \vx(t) \in S$,
    \item $\vx(T) \in G$\,.
\end{compactenum}    
\end{lemma}

\begin{proof}
Recall that $P^* :\ V^{\leq 0} \cap S$.
By condition (b) of Eq.~\eqref{eq:cllf-rws} $V^{\leq 0} \subset \inter{S}$ and therefore $P^* \subset \inter{S}$. Also, by condition (a) of Eq.~\eqref{eq:cllf-rws}, $I \subset V^{< 0}$ ($I \subset \inter{V^{\leq 0}}$). Therefore, $I \subset \inter{P^*}$ ($\vx(0) \in \inter{P}$).

If $\vx(0) \in G$, the conditions trivially hold. Therefore, we assume $\vx(0) \in \inter{P^*} \setminus G$.
    Now we show that $P^* \implies P^* \ U \ G$. Assume $\vx(\cdot)$ leaves $P^*$ before reaching $G$. Let $T \geq 0$ be the first time that $\vx(T)$ reaches $\partial P^*$ (without reaching $G$) and $\vx^+(T) \not\in P^*$. Then $V(\vx(T)) = 0$ and $\dot{V}_{\vu(T)}(\vx(T)) > 0$. 
    On the other hand, $\dot{V}_{\vu(t)}(\vx(t)) < -\epsilon_s$ for all $t \leq T$, which is a contradiction and the trace would not leave $P^*$ before reaching $G$.
    
    If $(\forall t > 0) \ \vx(t) \in P^* \setminus G$, $\dot{V}(\vx(t)) \leq -\epsilon_s$, $V$ decreases to infinity as time progresses as long as $\vx(\cdot) \in P^* \setminus G$.
    However, the value of $V$ is bounded on bounded set $P^* \setminus G$. Therefore, $\vx(\cdot)$ cannot remain in $P^* \setminus G$ and cannot reach the boundary of $P^*$. The only possible outcome for the trace is to reach $G$. Therefore, there exists $T \geq 0$ s.t.
    \begin{compactenum}
    \item $(\forall t \in [0, T]) \ \vx(t) \in \inter{P^*} \subseteq S$
    \item $\vx(T) \in G$.     \QED
    \end{compactenum}
\end{proof}

\begin{figure}[t]
\begin{center}
\includegraphics[width=0.4\textwidth]%
    {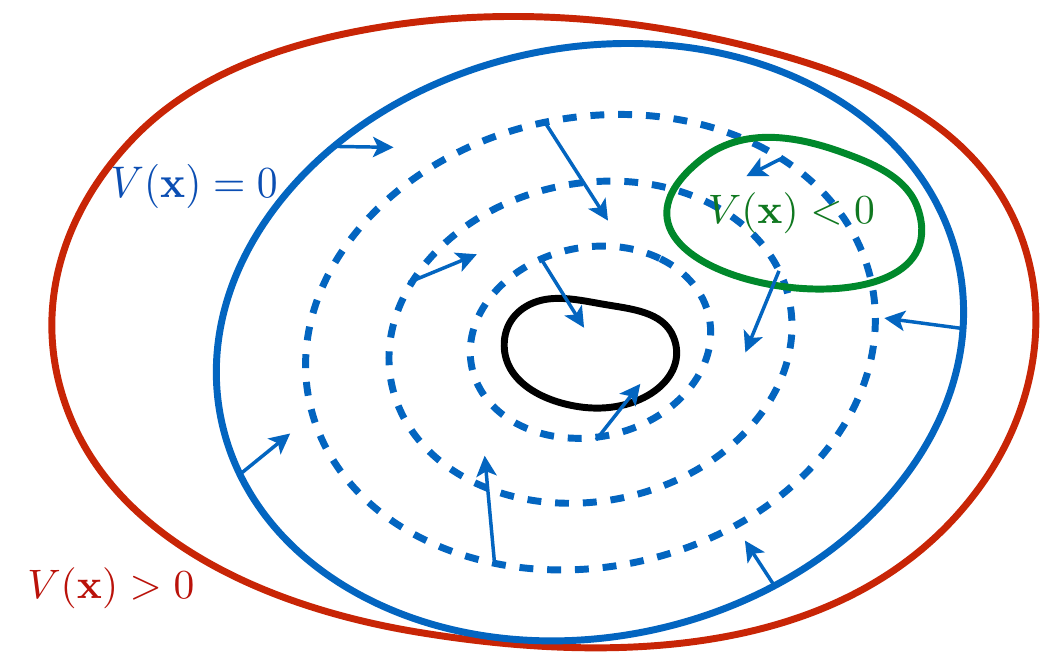}
\\ Boundary of safe, initial, and goal regions are shown in red, green, and black, respectively. The barrier is shown in blue along with the control vector field (blue arrows).
    \caption{A schematic view of CLBF.}
\label{fig:rws}
\end{center}
\end{figure}

\begin{theorem}\label{thm:unified-lyapunov-barrier} Given a plant $\P$, compact sets $ G, I, S $, and a smooth function $ V $ satisfying Eq.~\eqref{eq:cllf-rws}, there exists a smooth feedback law which satisfies $ I \implies S \ \U \ G$.
\end{theorem}
\begin{proof}
    We appeal directly to Sontag's result to obtain a smooth feedback function $\K(\vx)$ that guarantees that $\dot{V}(\vx(t)) < 0$ for all $\vx(t) \in S \setminus \inter{G}$~\cite{sontag1989universal}. Using Lemma~\ref{lem:rws}, $I \implies S \ \U \ G$ holds for the closed loop system.
  \QED
\end{proof}

The proof for switched feedback system is discussed in a more general form later in this chapter.

%% file: certificate/funnel.tex
\subsection{RWS with Reference Tracking}
As discussed in Chapter~\ref{ch:intro}, in RWS with reference tracking, in addition to initial set $\hat{I}$, goal set $\hat{G}$, and safe set $\hat{S}$, a (feasible) reference trajectory $\vx_r(\cdot)$ is also provided (as a hint). 
Notice that for convenience, we have defined $\hat{S}, \hat{I}, \hat{G}$ in the original state space $\vx$.
Consider a reference trajectory \emph{segment} defined by $\vx_r(\theta) $ for $\theta \in \Theta :\ [0, T]$, initial set $\hat{I} \ni \vx_r(0)$, a goal set $\hat{G} \ni \vx_r(T)$, and a safe sets $\hat{S}(\theta) \ni \vx_r(\theta)$ for $\theta \in \Theta$. $\hat{S}(\theta_1)$ defines the safe set when $\theta = \theta_1$. Moreover, $\hat{S}(0) \supseteq \hat{I}$ and $\hat{S}(T) \supseteq \hat{G}$.
Let $\hat{S}:\ \bigcup_{\theta \in \Theta} \hat{S}(\theta)$ denote the entirety of the safe set. We will now define them in terms of the deviation $\vx_d$ to define the following sets:
\begin{align*}
I &:\ \{\vx_d \ | \ \vx_d+\vx_r(0) \in \hat{I}\} \\
G &:\ \{\vx_d \ | \ \vx_d+\vx_r(T) \in \hat{G}\} \\
S(\theta) &:\ \{\vx_d \ | \ \vx_d + \vx_r(\theta) \in \hat{S}(\theta) \}.
\end{align*}
Finally, let us denote $S:\ \bigcup_{\theta \in \Theta} S(\theta)$.
Figure~\ref{fig:funnel}(a) shows a schematic view for these sets.

Recall the following dynamics for reference tracking:
\[
\dot{\vx}_d = f(\theta, \vx_d, \vu) = f(\vx_d + \vx_r(\theta), \vu) - \vr(\theta)\dot{\theta}\,.
\]

We consider $\dot{\theta} = u_0$, where $u_0$ is a virtual input and $u_0 \in [\underline{u_0}, \overline{u_0}]$ and $0 < \underline{u_0} \leq 1 \leq \overline{u_0} < \infty$.
As discussed, if we wish to consider RWS with timing constraints, we simply set $\underline{u_0} = \overline{u_0} = 1$.
To address the problem, we wish to use finite-time invariants a.k.a. \emph{funnels}~\cite{mason1985mechanics}.
\begin{definition}[Funnel] Given a closed loop system $\Psi(\P, \K)$, a funnel $F \subseteq X$ is a set with a head $F_h \subseteq F$ and a tail $F_t \subseteq F$. Moreover, if a trace $\vx(\cdot)$ reaches the head of the funnel, it remains inside the funnel until it reaches its tail:
\[
\vx(t) \in F_h \implies (\exists T \geq t) \begin{cases}
    \vx(T) \in F_t \\
    (\forall t' \in [t, T]) \ \vx(t') \in F\,.
\end{cases}
\]
\end{definition}
The control Lyapunov-barrier argument can get extended to control funnels~\cite{bouyer2017timed} for formally satisfying the RWS with reference tracking. We will define a \emph{control funnel} as a sub-level set of a smooth function $V(\theta, \vx_d)$.
For a smooth function $V$, and a relational operator $\bowtie \in \{ <, \leq, =, \geq , >\}$, let us define the following families of sets that are parameterized by $\theta$: $V^{\bowtie \beta}(\theta) :\ \{\vx \ | \ V(\theta, \vx) \bowtie \beta\}$. Furthermore, let $V^{\bowtie \beta} :\ \cup_{\theta \in \Theta} \ V_\theta^{\bowtie\beta} $.

\begin{definition}[Control Funnel Function]\label{def:control-funnel-function}
  A smooth function
  $V(\theta, \vx_d)$ is called a control funnel function (CFF) iff the
  following conditions hold:
\begin{equation}\label{eq:rules}
\begin{array}{lrl}
\noindent\textbf{(a)}:\ & (\forall \vx_d \in I) & \hspace{-0.2cm} V(0,\vx_d) < 0 \\
\noindent\textbf{(b)}:\ & (\forall \vx_d \in S(T) \setminus \inter{G}) & \hspace{-0.2cm} V(T,\vx_d) > 0 \\
\noindent\textbf{(c)}:\ & (\forall \theta \in \Theta, \vx_d \in \partial(S(\theta))) & \hspace{-0.2cm} V(\theta, \vx_d) > 0 \\
\noindent\textbf{(d)}:\ & (\forall \theta \in \Theta, \vx_d \in S(\theta)) & \hspace{-0.2cm} V(\theta, \vx_d) = 0 \implies (\exists \vv \in \V) \ \dot{V}(\theta,\vx_d, \vv) < 0.
\end{array}
\end{equation}
\end{definition}

The idea, depicted in Figure~\ref{fig:funnel}(b), is as follows. Initially (condition(a) in Eq.~\eqref{eq:rules}), $V(\vx(0)) < 0$ ($\vx \in V^{<0}$). Condition (d) guarantees that for all the states in a neighborhood of set $V^{=0}$, there exists a feedback which decreases the value of $V$. Therefore, by providing a proper feedback, the state never reaches $\partial V^{\leq 0}$ ($V^{=0}$) because the value of $V$ can be decreased just before reaching $V^{=0}$. As a result, $V$ remains $< 0$. This means that the state stays inside $V^{\leq 0}$ as long as $\theta \in \Theta$. Also the state remains in $S(\theta)$ as the value of $V$ on $\partial S(\theta)$ is $\geq 0$ (condition (c)). Since $\theta$ is increasing at minimum rate $\underline{u_0}$ at some point $\theta$ reaches $T$. Then, according to condition (b), the state must be in the interior of $G$ (top green ellipse), because otherwise value of $V$ would be $\geq 0$.

\begin{figure}[t]
\begin{center}
    \includegraphics[width=0.9\textwidth]{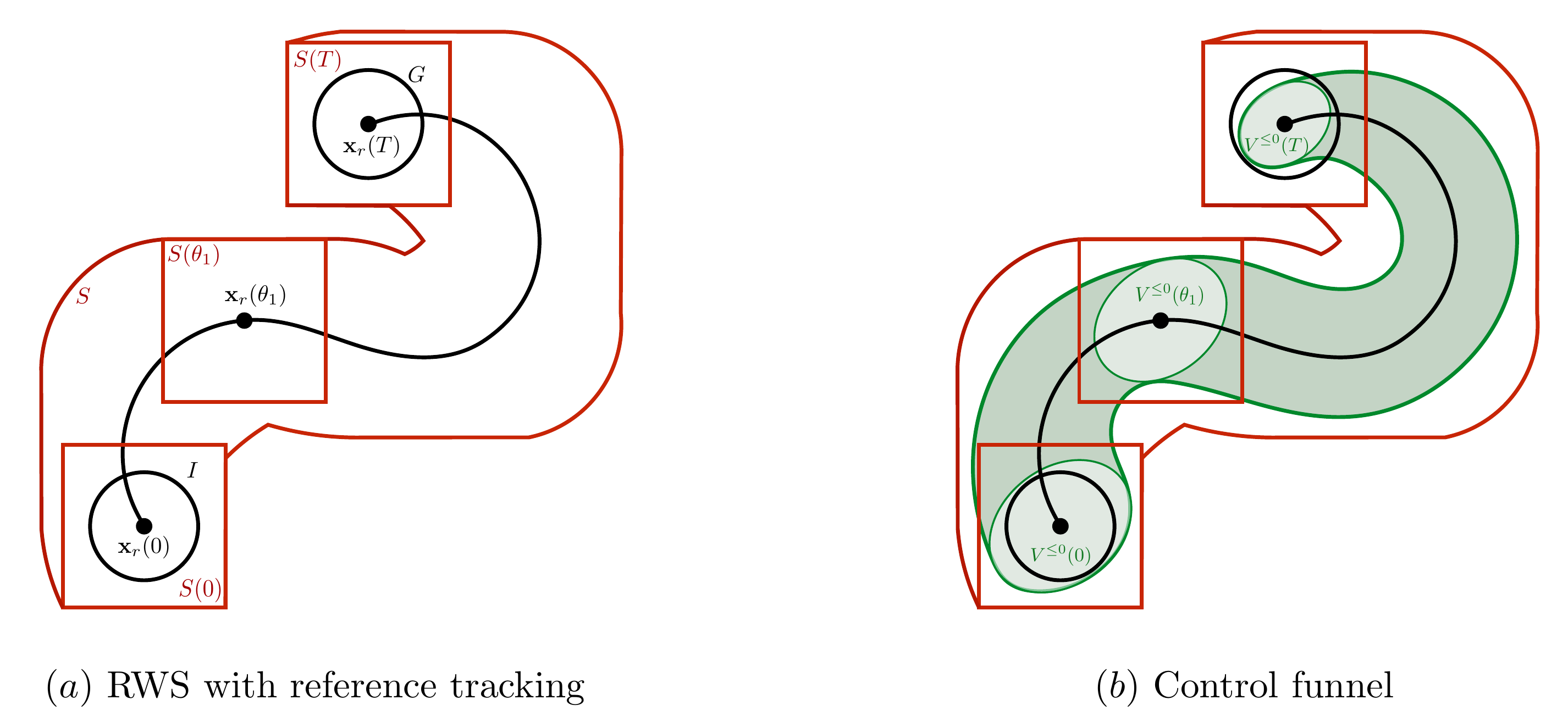} \\
\end{center}
\caption{A schematic view for RWS with reference tracking along with a control funnel. }\label{fig:funnel} 
\end{figure}

\begin{theorem}\label{thm:control-funnel}
  Given a plant $\P$ and a smooth control funnel function $V$, there exists a smooth feedback function for reaching $G$ such that for any initial state $\vx_d(0) \in I$, the goal state is eventually reached at some time $t^*$ satisfying $T/\underline{u_0} \geq t^* \geq T/\overline{u_0}$, while staying in set $S$ for $0 \leq t \leq t^*$.
\end{theorem}
\begin{proof}
    Initially $\vz(0) = [\theta(0), \vx_d(0)^t]^t$. According to condition (a), $V(\vx(0)) = \beta_0 < 0$ as $\theta(0) = 0$
    According to condition (d) in Eq.~\eqref{eq:rules} and by Sontag's result (cf.~\cite{sontag1989universal,WIELAND2007462}), there exists a smooth feedback function $\K$ which decreases the value of $V$ for all time instances that $\theta(t) \in \Theta \land \vx(t) \in S \land V(\vx(t)) = 0$. 
    Now, we assume $\vx(\cdot)$ reaches the boundary of $S$ before reaching $G$.
    Let $t_2$ be the first time instance that $\vx(\cdot)$ reaches the boundary of $S$. According to condition (c), $V(\vx(t_2)) > 0$. By smoothness of $V$ and the dynamics, there is a time $t$ ($0 \leq t < t_2$) for which $V(\vx(t)) = 0$. Let $t_1$ be the first time instance that $V(\vx(t_1)) = 0$ and $V^+(t_1) > 0$. However, the feedback law forces $V$ to decrease, which is a contradiction ($V^+(t_1) < 0$). Therefore, either $\vx_d(\cdot)$ remains inside $R = V^{\leq \beta} \cap S$ forever or remains inside $R$ until it reaches $G$. On the other hand, let $t_f$ be the time $\theta(t_f) = T$ and $\frac{T}{\overline{u_0}} \leq t_f \leq \frac{T}{\underline{u_0}}$. Since $\vx_d(\cdot)$ remains in $V^{\leq \beta}$, $V(\vx_d(t_f)) \leq \beta$. According to condition (b), $\vx_d(t_f)$ is in the interior of $G$. \QED
\end{proof}
Alternatively, similar to Eq.~\eqref{eq:cbf-safety}, we can replace condition (d) in Eq.~\eqref{eq:rules} with the following:
\begin{equation}\label{eq:cff-exp}
(\forall \theta \in \Theta, \vx_d \in S(\theta))\ (\exists \vv \in \V) \ \left( \begin{array}{l}
    \dot{V}(\theta,\vx_d, \vv) + \lambda^* V(\theta, \vx_d) < 0 \lor \\
    \dot{V}(\theta,\vx_d, \vv) - \lambda^* V(\theta, \vx_d) < 0
\end{array} \right).
\end{equation}
Extracting switched feedback law from control funnel function can be addressed using a similar protocol used in the proof of Theorem~\ref{thm:safety-switch}.

%% file: certificate/multiple-barrier.tex
\subsection{Uninitialized RWS}
Another interesting class of properties is control to facet problems, related to work of Habets et al.~\cite{habets2006reachability} and Kloetzer et al.~\cite{kloetzer2008fully}, wherein the control system is modeled with a finite automaton by solving local control-to-facet problems. 
These properties are RWS properties. However, the initial set $I$ is the same as safe set $ S $ (uninitialized RWS). Precisely, the specification is $ S \Rightarrow S \ \U \ G $, where $ S $ is a nondegenerate basic semi-algebraic compact set. Since $G$ is usually a facet, this property is also called control-to-facet.
Let $S$ be a nondegenerate basic semialgebraic sets, as in
Definition~\ref{Def:basic-semialgebraic-set}:
\[ 
S :\ \{\vx\ |\ p_{S,1}(\vx) \leq 0\ \land\ \cdots\ \land p_{S,i}(\vx) \leq 0\} \,,
\]  
where 
\[ H_{S,j} =
\{\vx\ |\ \vx \in S\ \land\ p_{S,j}(\vx) = 0\} \neq \emptyset \,.
\]
Let $\partial S$ be partitioned into nonempty facets $F_{1}, \ldots,$ $ F_{l}$. Each facet $F_{k}$ is, in turn, defined by two sets of polynomial inequalities $F_{k}^{<}$ of inactive constraints and $F_{k}^{=}$ of active constraints: $F_{k}=\{ \bigwedge_{p_{S,j} \in F_{k}^{<}} p_{S,j}(\vx) < 0\ \land\ \bigwedge_{p_{S,j} \in F_{k}^{=} } p_{S,j}(\vx) =  0 \}$. For example, Figure~\ref{fig:control-to-facet} shows facets for a polytope. For $i \in [0..4]$,  $F^{i=} :\ \{p_i\}$ and for $i \in [5..9]$, $F^{i=} :\ \{p_{i-5}, p_{(i-6)\%5}\}$.

For each state on a facet and not in $G$, we require the existence of an input $\vu$, whose vector field points inside $S$. Additionally, we require a certificate $V$ to decrease everywhere in $S \setminus G$. For any polynomial $p_{S,j}$, let $\dot{p}_{S,j,\vu}:\ (\nabla p_{S,j})\cdot f_\vu(\vx)$. Conventional methods~\cite{habets2004control} combine conditions for safety and reachability to define a \emph{Control Lyapunov Fixed-Barriers Function} (CLFBFs) $V$ (which is smooth) as the following:

\begin{equation}\label{eq:rws-basic}
\begin{array}{rl}
\noindent\textbf{(0)}:\ & \vx \in S\setminus \inter{G} \implies (\exists\ \vu \in U)\ \dot{V}_\vu(\vx) < 0  \\[4pt]
\noindent\textbf{(1)}:\ & \vx \in F_{1} \setminus \inter{G} \implies (\exists\ \vu \in U)
\left( \begin{array}{c}
     \dot{V}_\vu(\vx) < 0\ \land\ \bigwedge  \limits_{p \in F_{1}^{=}} \begin{array}{c} \dot{p}_\vu(\vx) <0\end{array} \end{array} \right)  \\
\vdots\\
\noindent\textbf{(l)}:\ & \vx \in F_{l} \setminus \inter{G} \implies (\exists\ \vu \in U)
\left(\begin{array}{c} 
\dot{V}_\vu(\vx) < 0\ \land\ \bigwedge  \limits_{p \in F_{l}^{=}} \begin{array}{c} \dot{p}_\vu(\vx) <0 \end{array}  \end{array} \right)
\,.
\end{array}
\end{equation}

Condition (0) in Eq.~\eqref{eq:rws-basic} states that $V$ must decrease everywhere in the set $S \setminus \inter{G}$. The subsequent conditions treat each facet $F_j$ of the set $S$ and posit the existence of a feedback $\vu$ for each state that causes the active constraints and the function $V$ to decrease. 

However, we note that as the number of state variables increases, the number of facets can be exponential in the number of inequalities that define $S$~\cite{helwa2013monotonic}. For example, a 4D ($n = 4$) box has $ 80 $ facets. This poses a serious limitation to the applicability of Eq.~\eqref{eq:rws-basic}. 

\begin{figure}[t!]
\begin{center}
\includegraphics[width=0.4\textwidth]%
    {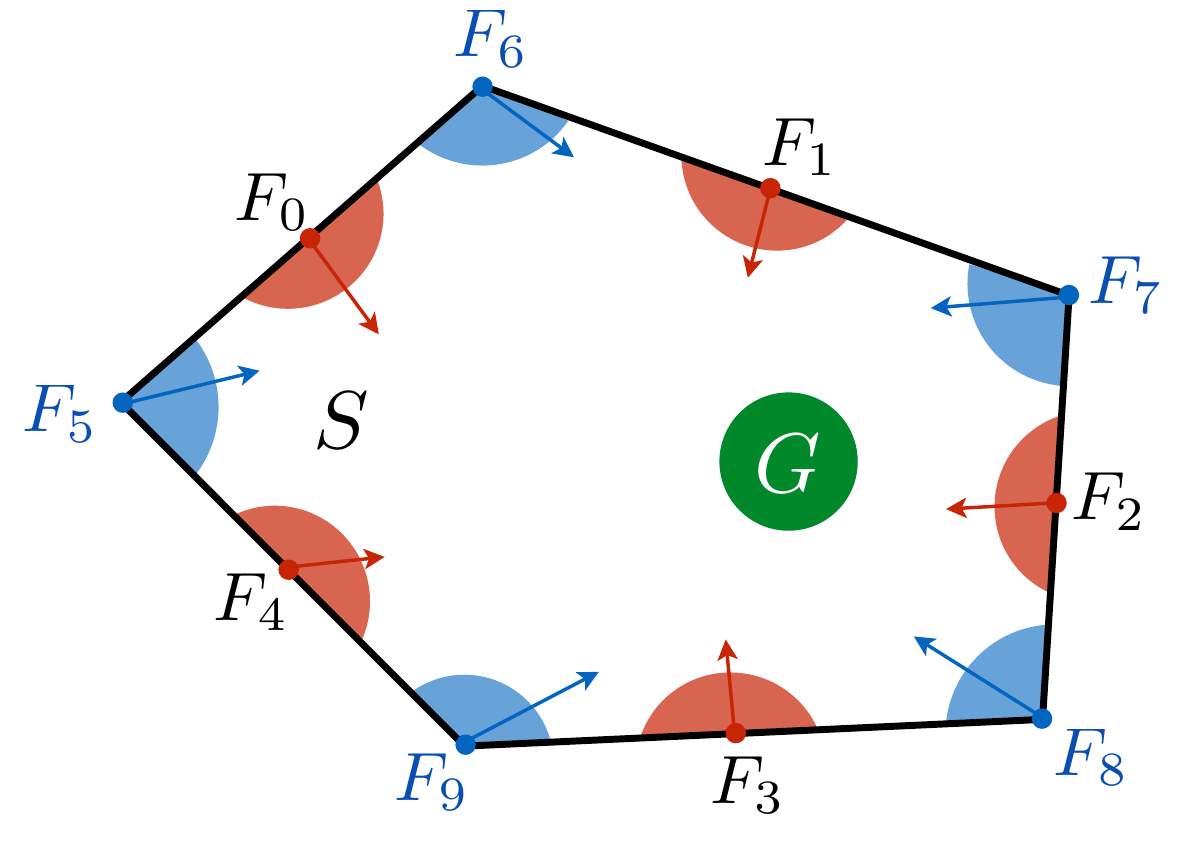}
\\ Black lines and blue dots are the facets. For blue dots and red dots, feasible directions for the vector field is shown.
    \caption{Facets for a 2D polytope.}
\label{fig:control-to-facet}
\end{center}
\end{figure}

Our solution to this problem, is based on the idea introduced for control barrier functions in Section~\ref{sec:certificate-safety}. Rather than force the vector field to point inwards at each facet, we simply ensure that each polynomial inequality $p_{S,j} \leq 0$ that defines $S$, satisfies a decrease condition outside the set $G$. Thus, Eq.~\eqref{eq:rws-basic} is replaced by a simpler (relaxed) condition:
\begin{equation}\label{eq:rws-basic-exp}
\begin{array}{l}
\vx \in S \setminus \inter{G} \implies (\exists\ \vu \in U)  \ \ \dot{V}_\vu(\vx) < 0 \ \ \land \bigwedge_j\ \left( \begin{array}{c}\dot{p}_{S,j,\vu}(\vx) + \lambda\ p_{S,j}(\vx)\ < 0
    \end{array} \right)\,.
\end{array}
\end{equation}
Again, $\lambda > 0$ is a user specified parameter and this rule is a relaxation of Eq.~\eqref{eq:rws-basic}.

This control certificate may not be \emph{constructive} for nonlinear smooth feedback systems because of the logical formulation. In other words, extracting a closed-form smooth feedback is not trivial and potentially a hard problem. As a result, here we discuss feedback law extraction only for switched systems.

Given a control certificate $V$ satisfying Eq.~\eqref{eq:rws-basic-exp}, the choice of a switching mode is dictated by a function $\cond_\vu(\vx)$ defined for each state $\vx \in X$ and mode $\vu \in U$ as follows:
\begin{equation}\label{eq:cond-for-control-to-facet}
\cond_\vu(\vx):\ \max\left(
\begin{array}{c}
\dot{V}_\vu(\vx),\ 
\cond_{S,1,\vu}(\vx),\ \cdots, \cond_{S,k,\vu}(\vx)
\end{array}
\right)\,,
\end{equation}
where $\cond_{S,j,\vu}$ is $\dot{p}_{S,j,\vu} +\lambda p_{S,j}$.
The goal of a controller is to switch to a mode $\vu$ that guarantees that $\cond_\vu(\vx) < - \epsilon$, which in turn guarantees decrease of $V$ as well as remaining in $S$.
Now, we can show there exists $\veps$ s.t. any feedback function $\K \in \SK(\cond, S \setminus \inter{G}, \veps)$ is a solution to the \RWS \ problem.

\begin{theorem}
\label{thm:control-to-facet}
    Given a switched plant $\P$, a nondegenerate \ basic \ semialgebraic set $S$, a compact set $G$, and a smooth function $V$ (satisfying Eq.~\eqref{eq:rws-basic-exp}), there exists $\veps$ s.t. the  $\SK(\cond, S \setminus \inter{G}, \veps)$ (wherein $\cond_\vu$ is defined by Eq.~\eqref{eq:cond-for-control-to-facet}) is non-empty, and any $\K \in \SK(\cond, S \setminus \inter{G}, \veps)$ guarantees the RWS property defined by $S,G$:\ $S \implies S \scr{U} G$.
\end{theorem}
\begin{proof}
     Note that $\cond_\vu(\vx)$ is defined as $\max(\alpha_0(\vx),\ldots,\alpha_{m'}(\vx))$ for some smooth functions $\alpha_0,\ldots,$ $ \alpha_{m'}$. Also, $\dot{x}$ is bounded as $D$ is bounded. Using Eq.~\eqref{eq:rws-basic-exp}, and Lemma~\ref{lem:min-dwell-exists} there exists $\veps$ s.t. $\K \in \SK(\cond, S, \veps)$ is non-empty. Moreover, for any $\K \in \SK(\cond, S \setminus \inter{G}, \epsilon)$, (i) the min-dwell time exists and (ii) $\dot{V}_{\vu(t)}(\vx(t)) \leq -\epsilon_s$ and $\bigwedge_j (\dot{p}_{S,j,\vu(t)}(\vx(t)) + \lambda p_{S,j}(\vx(t)) \leq -\epsilon_s$ as long as $\vx(t) \in S \setminus \inter{G}$.
    
    Assume $\vx(t)$ is on the boundary of $S$ (and not in $G$) at some time $t$. Because $S$ is assumed to be a nondegenerate basic semialgebraic set, there exists at least one $j$ s.t. $p_{S,j}(\vx(t)) = 0$. We obtain $\dot{p}_{S,j,\vu(t)}(\vx(t)) \leq -\epsilon_s < 0$.
    Therefore, there exits $\tau_j > 0$, s.t. $\forall s \in (t, t + \tau_j)$, $p_{S,j,\vu(s)}(\vx(s)) < 0$. As a result, the trajectory cannot leave the set $S$.

     Thus, the trace cannot leave $S$, unless it reaches $G$. Now, we show that the trajectory cannot stay inside $S \setminus \inter{G}$ forever. By the construction of the controller, we can conclude that time progresses as long as $\vx(\cdot) \in S \setminus \inter{G}$ (because the controller respects the min-dwell time property) and that $V$ decreases ($\dot{V}_{\vu(t)}(\vx(t)) \leq -\epsilon_s$). However, the value of $V$ is bounded on bounded set $S \setminus \inter{G}$. Therefore, $\vx$ cannot remain in $S \setminus \inter{G}$ and the only possible outcome for the trace is to reach $G$. \QED
\end{proof}

\begin{example}\label{ex:basic-control-to-facet}
This example is adopted from~\cite{nilssonincremental}. There are two variables and three control modes with the dynamics given below: \begin{align*} \begin{array}{c} \left[ \begin{array}{c} \dot{x_1} \\ \dot{x_2} \end{array} \right]
    = \left[ \begin{array}{c} -x_2-1.5x_1-0.5x_1^3 \\
    x_1 \end{array} \right] + B_\vu ,  B_{\vu_1}
    = \left[ \begin{array}{c} 0 \\ -x_2^2 +
    2 \end{array} \right],  B_{\vu_2} = \left[ \begin{array}{c}
    0 \\ -x_2 \end{array} \right],  B_{\vu_3}
    = \left[ \begin{array}{c} 2 \\
    10 \end{array} \right].  \end{array} \end{align*}
    The goal is to reach the target set $G:\ (x_1+0.75)^2 + (x_2-1.75)^2 \leq 0.25^2$, a circle centered at $(-0.75,1.75)$, as shown in Figure~\ref{fig:basic}, while staying in the safe region given by the rectangle $S_0:\ [-2,2] \times [-2,3]$:
    \[
    S_0 :\ \{ \vx | (x_1 + 2)(x_1 - 2) \leq 0 \land (x_2 + 2)(x_2 - 3) \leq 0 \} \,.
    \]

    First, we find the following control certificate:
        \begin{align*}  
        V(x_1,x_2):\ & 37.782349 x_1^2 -
    2.009762 x_1 x_2 + 60.190607 x_1 + 4.415093 x_2^2 -
    16.960145 x_2 + 37.411604 \,.
        \end{align*}
Using Eq.~\eqref{eq:suitable-feedbacks} we design a controller.  Figure~\ref{Fig:closed-loop-basic-example} shows some of the simulation traces of this closed-loop system, demonstrating the RWS property.
\begin{figure}[t!]
\begin{center}
\begin{subfigure}{.4\textwidth}
\centering
\includegraphics[width=1\textwidth]{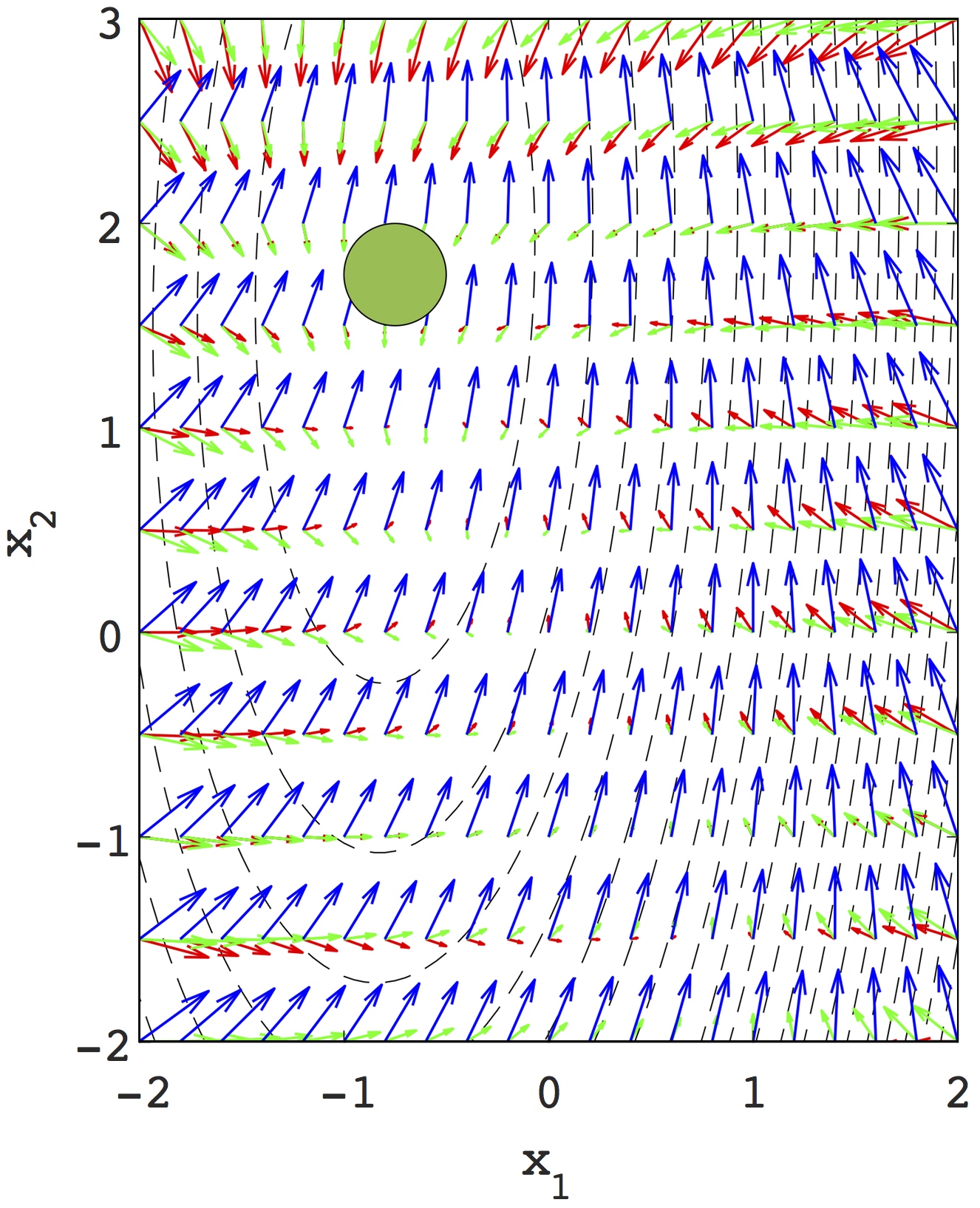}
\caption{Region $G$ is shown with a green circle and the vector fields for modes $\vu_1,\vu_2$ and $\vu_3$ are shown in red, green and blue, respectively. Level sets of $V$ are shown with black dashed lines}
\label{fig:basic}
\end{subfigure}
\qquad
\begin{subfigure}{.4\textwidth}
\centering
\includegraphics[width=1\textwidth]{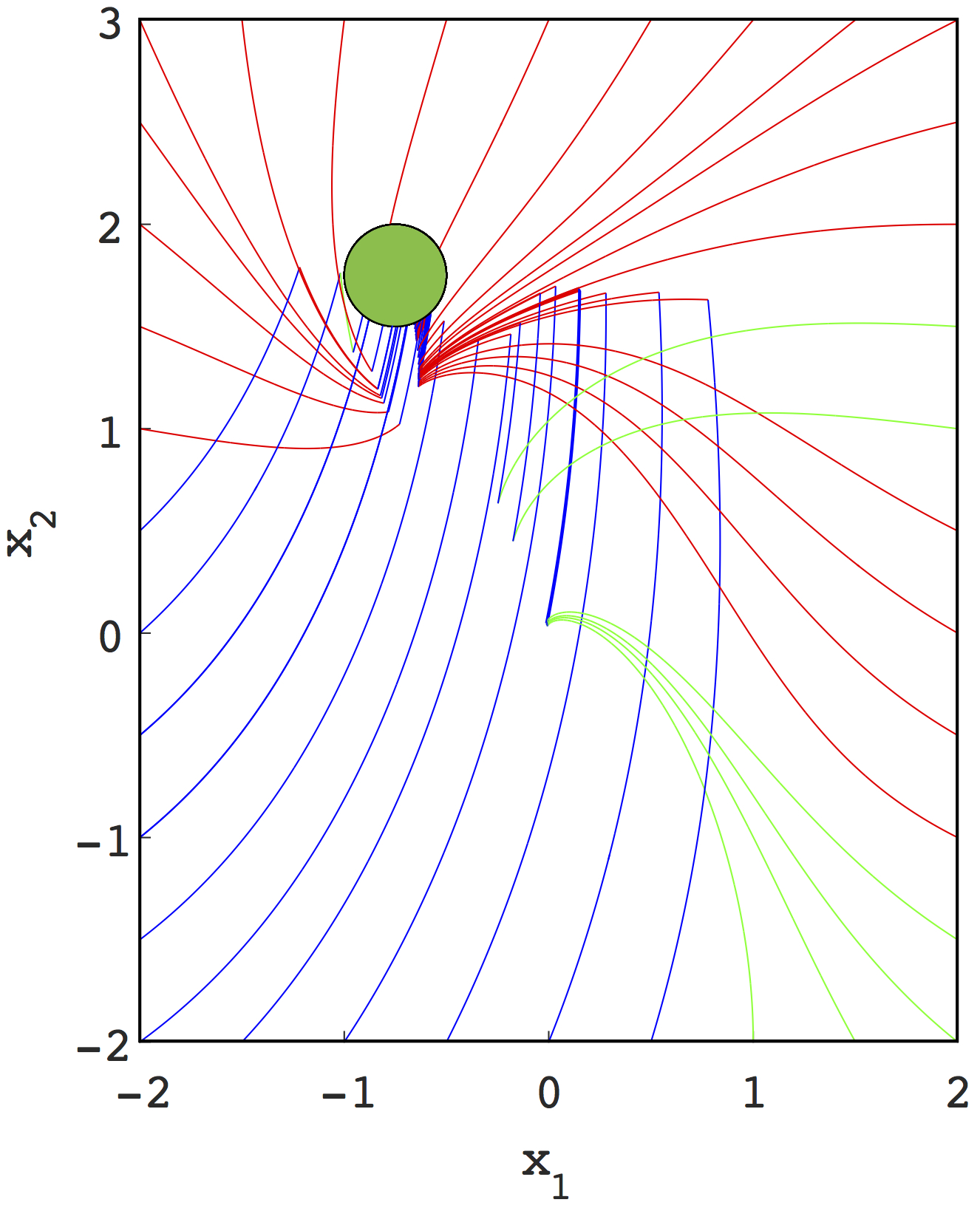}
\caption{Closed-loop trajectories using a CLFBF-based controller. The segments shown in colors red, green and blue correspond to the modes $\vu_1, \vu_2$ and $\vu_3$, respectively}
\label{Fig:closed-loop-basic-example}
\end{subfigure}
\caption{Plots for Example~\ref{ex:basic-control-to-facet}.}
\end{center}
\end{figure}
\end{example}

%% file: certificate/disturbance.tex
\section{Disturbances}\label{sec:dist-1}
Going one step further, one can consider the disturbances as well, and design a more robust solution. In this setting, $f$ is not only is a function of $\vx$ and $\vu$, but also function of disturbances $\vd$, which belongs to a compact set $D\subseteq\reals^p$: $\dot{\vx} = f(\vx, \vu, \vd)$ (Figure~\ref{fig:dist-model}). To avoid technical difficulties, we simply assume $D$ is a basic semi-algebraic set, and $\vd(\cdot)$ describes the uncontrollable input signal, and $\vd(\cdot)$ is assumed to be (locally) Lipschitz.

\begin{figure}[t]
\begin{center}
    \includegraphics[width=0.8\textwidth]{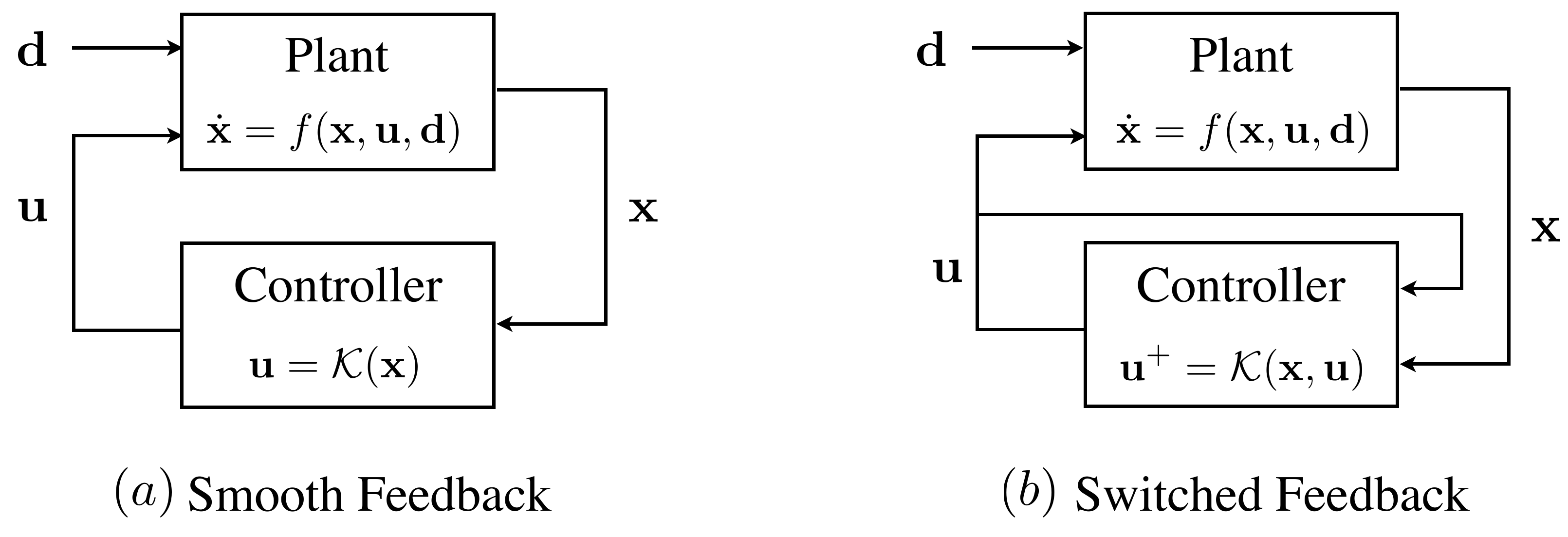}
\end{center}
\caption{Models of state feedback systems in presence of disturbances.}\label{fig:dist-model} 
\end{figure}

The controller has access to the state $\vx$ and needs to choose a proper input, without knowing the value of the disturbance $\vd$.
The idea is to find feedback $\vu$ such that the value of $V$ decreases under all possible disturbances.
To handle disturbances, the notion of CLF is extended to \emph{robust} CLFs (RCLF)~\cite{rotea1989stabilization,freeman1996inverse}.
We modify Definition~\ref{def:clf} to incorporate disturbances.
\begin{definition} [\emph{Robust} CLF] $\hspace{-0.2cm}$ A robust CLF (RCLF) is a smooth radially unbounded function $V$ with the following properties:
    \begin{equation}\label{eq:cllf-das}
      V(\vx_r) = 0 \,, \ \ \ \ 
          (\forall \vx \neq \vx_r) \ V(\vx) > 0 \,, \ \ \ \ 
          (\forall \vx \neq \vx_r) \ (\exists \vu \in U)
                \ {\color{red}(\forall \vd \in D)} \ (\nabla V) . f(\vx, \vu, {\color{red}\vd}) <
                0 \,.
    \end{equation}
\end{definition}

When compared to Eq.~\eqref{eq:clf}, the third condition for the RCLF is more complicated because of the extra $(\forall \vd)$ quantifier.

The solution of Freeman et al.~\cite{freeman1996inverse} for extracting a \emph{smooth} feedback function $\K$ from RCLF is not constructive. Battilotti~\cite{battilotti1999robust} provides a method for automatic design of $\K$, which depends on finding unknown functions with complicated constraints. We are not aware of any efficient method for extracting $\K$ from RCLFs. In fact, it can be quite complicated, requiring expensive quantifier elimination procedures. 

Here we discuss robust control certificates for \emph{switched} systems. For brevity, we only consider \emph{robust} CLBFs, and we mention that other robust control certificates will follow a similar protocol. Our solution is constructive and the design of the feedback function $\K$ is quite straightforward if we postpone quantifier elimination to runtime.

\begin{definition}[Robust CLBF] A smooth function $V$ is a robust control Lyapunov-barrier function iff
    \begin{equation} \label{eq:rcllf-rws}
    \begin{array}{rl} (\forall \vx \in I) & \ V(\vx) < 0 \\ (\forall \vx \in \partial S) & \ V(\vx) > 0 \\ (\forall \vx \in S \setminus \inter{G}) & \ (\exists \vu \in U) \ (\forall \vd \in D) \ \nabla V.f(\vx, \vu, \vd) < 0 \,.\end{array}
\end{equation}
\end{definition}

To design a switched feedback, let us define a function $\cond_\vu(\vx)$ over a state $\vx$ and mode $\vu$ as
\begin{equation}\label{eq:condq-definition}
\cond_\vu(\vx) = \max_{\vd \in D} \ \nabla V \cdot f_\vu(\vx, \vd) \,.
\end{equation}
The goal of a controller is to switch to a mode $\vu$ that guarantees that $\cond_\vu(\vx) < - \epsilon$, which in turn guarantees decrease of $V$.
Now, we show that having an RCLBF $V$, there exists $\veps$ s.t. any feedback function $\K \in \SK(\cond, S \setminus \inter{G}, \veps)$ is a solution to the \RWS \ problem. Notice that as $S$ is compact, and there is no need to check whether there is a finite escape time. Also, Zenoness must be avoided only for some time $T$ where $\vx(T) \in G$. Once the state reaches $G$, we assume the control is handed over to another controller.

\begin{theorem} \label{thm:robust-rws} Given a plant $\P$, sets $I$, $G$, $S$, and a RCLBF $V(\vx)$, there exists $\veps$ s.t. the  $\SK(\cond, S \setminus \inter{G}, \veps)$ (wherein $\cond_\vu :\ \max_{\vd \in D} (\nabla V) \cdot f_\vu(\vx, \vd)$) in non-empty, and any $\K \in \SK(\cond, S \setminus \inter{G}, \veps)$  guarantees RWS property defined by $I$, $S$, and $G$: $I \implies S \ \U \ G$.
\end{theorem}
\begin{proof}
        Note that as $D$ is a compact basic semi-algebraic set and thus, $f$ is bounded on $S \setminus G$. Also, $\cond_\vu$ is piecewise continuous.
By Eq.~\eqref{eq:cllf-das} and Lemma~\ref{lem:min-dwell-exists}, there exists $\veps$ s.t. $\K \in \SK(\cond, S \setminus \inter{G}, \veps)$ is non-empty. Moreover, for any $\K \in \SK(\cond, S \setminus \inter{G}, \veps)$ (i) a min dwell-time exists and (ii) $\cond_{\vu(t)}(\vx(t)) :\ \nabla V.f_{\vu(t)}(\vx(t), \vd(t)) \leq -\epsilon_s$ as long as $\vx(t) \in S \setminus G$. Using Lemma~\ref{lem:rws} we conclude $I \implies S \ \U \ G$. \QED
\end{proof}
Recall that the controller needs to track the value of $\cond_\vu(\vx)$, which involves an optimization. As these optimizations can be expensive, we use the following assumption to design an efficient controller: (a) $D$ is a polytope, and (b) $f$ is affine in $\vd$. Now, the optimization is equivalent to solving a linear programming problem. Furthermore, if $D$ is a hyper-box, the optimization problem is effectively solvable with complexity $O(p)$, which is practically appealing. Therefore, the whole mode selection is $O(pm')$, where $m'$ is the number of modes.

\paragraph{Summary:} In this chapter, we discussed already established control certificates for \emph{smooth} feedback systems to address stability, reference tracking, and safety. Moreover, for these systems, we introduced control Lyapunov-barrier functions and control funnel functions (using path-following) for RWS properties. We introduced non-Zeno CLFs for \emph{switched} feedback systems. In addition, we showed that other control certificates for smooth feedback systems are applicable to switched feedback systems as well. Furthermore, we introduced control Lyapunov fixed-barriers functions for switched feedback systems. Finally, we demonstrated that \emph{robust} control certificates can be used for switched feedback systems with disturbances. In the next chapter, we provide a framework for discovering these control certificates.

%% file: synt/synt.tex
\chapter{Inductive Synthesis}~\label{ch:search}
In the previous chapter, we described several classes of control certificates and how a control synthesis problem is reduced to that of finding a control certificate. We now propose a framework for finding such control certificates. In the proposed framework, the hypothesis space is defined using parameterization. More specifically, a template $T:\C \mapsto \Hy$ is defined over a set of parameters $\vc \in \C$. For each $\vc \in C$, $T(\vc)$ is a member of the hypothesis space $\Hy$, and the goal is to find $ \vc $ such that $ T(\vc) $ is a control certificate. All the control certificates we discussed in the previous chapter are functions that map $X$ to $\reals$.
We define a template $T:\C \mapsto (X \mapsto \reals)$ as $T(\vc) :\ V_\vc(\vx) :\ \sum_{k=1}^r c_k g_k(\vx)$, wherein $ c_k $ is the coefficient of basis function $ g_k(\vx) $. 
Suppose that we wish to find a CLF. If $g_k$'s are smooth, then $V$ is smooth. Moreover, if $g_k$'s are monomials, $V$ would be radially unbounded. Now, it is sufficient to solve for 
\begin{equation}
(\exists \vc \in \C) \ \left( V_\vc(\vzero) = 0 \, \land \, (\forall \vx \neq \vzero) \left( V_\vc(\vx) > 0 \land (\exists \vu) \ \nabla V \cdot f(\vx, \vu) < 0 \right)  \right)     \,.
\end{equation}
Either we prove no such $ \vc $ exists or we show that there exists $ \vc^* $ s.t. $ V_{\vc^*} $ is a CLF. The former case translates to ``non-existence of CLFs in the hypothesis space." Nevertheless, one can try to simply change the hypothesis space by adding more basis functions to the template. 

Finding control certificates is challenging as solving constraints which arise in a control certificate search is expensive~\cite{primbs1999}. Experts use their knowledge about the domain of interest to design a ``control certificate"~\cite{Nguyen2015RobustCLF,Nguyen2016CBF}. Alternatively, there are solutions for specific systems such as feedback linearizable systems~\cite{krstic1995nonlinear} and strict feedback systems~\cite{freeman2008robust}. 
In this chapter, we introduce an inductive framework for finding a control certificate $T(\vc):\ V_\vc$. 
First, we discuss the related work. 

\input{synt/background}

\input{synt/cegis} 
\input{synt/demonstration}

\input{synt/sdp}

%% file: synt/background.tex
\section{Background}
We go over the related work by first discussing the certificate synthesis problem using constraint solvers. Afterward, we investigate the control synthesis problem. For simplification, we mostly focus on the stability property. More specifically, we discuss constraint solving techniques used to analyze or synthesize control systems using Lyapunov functions.

\subsection{Lyapunov Function Synthesis}
Suppose we wish to find a Lyapunov function $V$. The search space is simplified through parameterization $V_\vc :\ \sum\limits_{k=1}^{r} c_k g_k(\vx)$.
Now, the problem is to find $\vc$ s.t. $V_\vc$ is a Lyapunov function. In this chapter, we assume $V$ is a polynomial with unknown coefficients ($g_k$'s are monomials). Other templates are left for future work. We note that using parameterization, the completeness is lost. However, any smooth function can be approximated with a polynomial.

\paragraph{Linear Systems:}
For a linear system where $\dot{\vx} = A \vx$, the problem of finding a Lyapunov function can be solved quite efficiently. More specifically, it is known that for a stable linear system, quadratic Lyapunov functions exist. Therefore, the parameterization of $V$ preserves the completeness as long as $V_\vc$ contains all quadratic terms. Now, we wish to find $\vc$ s.t.
\begin{equation}
    V_\vc(\vzero) = 0 \, \land \, 
    (\forall \vx \neq \vzero ) \ V_\vc(\vx) > 0 \, \land \, 
    (\forall \vx \neq \vzero ) \ \nabla V_\vc \cdot f(\vx) < 0 \,.
\end{equation}
I.e., $V_\vc$ should be positive definite and $\dot{V}_\vc$ is negative definite.
Notice that $V_\vc :\ \sum\limits_{k=1}^{r} c_k g_k(\vx)$ is linear in $\vc$ and therefore, $\nabla V_\vc \cdot f(\vx) :\ \sum\limits_{k=1}^{r} c_k (\nabla g_k \cdot f(\vx))$ is linear in $\vc$ as well. Also, if $V_\vc$ is purely quadratic in $\vx$ ($V_\vc:\ \sum\limits_{i=1}^n\sum\limits_{j=1}^{i} c_{ij} x_i x_j$), then $V_\vc(\vzero) = 0$, $V_\vc :\ \vx^t C \vx$, and $\nabla V_\vc \cdot f(\vx) :\ \vx^t C A \vx + \vx^t A^t C \vx$, where $C$ is a symmetric matrix and its entities have linear relations with $\vc$.
Alternatively, CLF conditions can be written in the following form:
\[
(\forall \vx \neq \vzero) \ \left( \tupleof{C, \vx \vx^t} > 0\, \land \, \tupleof{CA + A^tC, \vx \vx^t} < 0 \right) \,,
\]
wherein $\tupleof{A, B}$ is $\mbox{trace}(AB)$.
Using the fact that
\[
C \succ 0 \iff (\forall X \succeq 0, \neq \vzero) \ \tupleof{C, X} > 0 \iff (\forall \vx \neq \vzero) \ \tupleof{C, \vx \vx^t} > 0\,,
\]
any solution to the following semi-definite programming (SDP) problem yields a Lyapunov function:
\begin{align*}
    (\exists \vc) \ s.t. & \,  C \succ 0 \, \land \, C A + A^t C \prec 0 \, \wedge \ \tupleof{C, \vx\vx^t} = V_\vc(\vx),
\end{align*}
wherein the last condition defines the (linear) relation between $\vc$ and $C$.

\paragraph{Nonlinear Systems:}
The verification problem for a nonlinear system is harder. First, $V$ is not necessarily quadratic, and second, even if $V$ is quadratic in $\vx$, $\nabla V$ is not quadratic anymore.  Nevertheless, if $V$ and $f$ (and thus $\nabla V$) are polynomials over $\vx$, the same trick is extended to polynomials using SOS programming~\cite{papachristodoulou2002,prajna2002sos}. To represent a polynomial $g:X \mapsto \reals$, a vector collecting all monomials of degree up to $\degr$ is defined:
\[ \vz:\ \left[\begin{array}{c}
               1 \  x_1 \ x_2\ \ldots \ x_n^{\degr}\ \end{array}\right]^t \,,\]
wherein $\degr$ is chosen to be at least half 
of the maximum degree in $\vx$ among all monomials in $g(\vx):\ \degr \geq \frac{1}{2} \mbox{deg}(g)$.
Each polynomial $p$ of degree up to $2\degr$ (including $g$) may now be written as a trace inner product $p(\vx):\ \tupleof{ P, \vz \vz^t}$, wherein the matrix $P$ is symmetric and has real-valued entries that define the coefficients in $p$ corresponding to the various monomials. Then, it is easy to show
\[
    P \succ 0  \Rightarrow (\forall Z \succeq 0,\neq \vzero)\  \tupleof{P, Z} > 0 \Rightarrow (\forall \vz \neq \vzero)\  \tupleof{P, \vz \vz^t} > 0 \Rightarrow (\forall \vx) \ p(\vx \neq \vzero) > 0\,.
    \]
Notice that the reverse does not hold anymore. For example, Motzkin polynomial $m(\vx) :\ x_1^4x_2^2 + x_1^2x_2^4 - 3x_1^2x_2^2 + 1$ is a non-negative function, but for all $M$ s.t. $m(\vx) = \tupleof{M, \vz \vz^t}$, $M$ is not positive-semi definite~\cite{reznick2000some}. Similarly, one can show $m'(\vx) :\ m(\vx) + \epsilon$ ($\epsilon > 0$) is a positive polynomial, but for all $M'$ s.t. $m'(\vx) = \tupleof{M', \vz \vz^t}$, $M'$ is not positive definite (we leave the proof of this to the reader as it is slightly different from the one discussed in~\cite{reznick2000some}).

The method can be extended to cases where $\vx \in I$ for some basic semi-algebraic set $I :\ \{\vx \ | \ g_1(\vx) \geq 0 \land \ldots\land g_l(\vx) \geq 0\}$. Let $g_i(\vx) = \tupleof{G_i, \vz\vz^t}$ and $\I :\ \{\vz \ | \tupleof{G_1, \vz\vz^t}\geq 0\land\ldots\land \tupleof{G_l, \vz\vz^t}\geq 0 \}$. It is straightforward to show
\begin{align*}
    \left( P_0 \succ 0, \bigwedge_{i=1}^l P_i \succ 0 \right) & \Rightarrow (\forall \vz \in \I) \ \tupleof{P_0, \vz \vz^t} + \sum_{i=1}^l \tupleof{P_i, \tupleof{G_i, \vz \vz^t} \vz \vz^t} > 0 \\
    & \Rightarrow (\forall \vx \in I) \  p(\vx) = p_0(\vx) + \sum_{i=1}^l p_i(\vx)g_i(\vx) > 0\,,
\end{align*}
where $\degr \geq \frac{1}{2} \max(\mbox{deg} (p_0),\cup_i\mbox{deg} (p_ig_i))$.
This trick is known as S-procedure.

Let cone of $I$ be
\[
\mbox{Cone}(I) : \left\{ p \left| \left( (\exists P_0, \ldots, P_l \succ 0) \ \tupleof{P, \vz\vz^t} = \tupleof{P_0, \vz \vz^t} + \sum_{i=1}^l \tupleof{P_i, \tupleof{G_i, \vz \vz^t} \vz \vz^t}   \right) \right.\right\}\,.
\]
The following theorem shows the inverse is true under some assumptions.
\begin{theorem}[Putinar's Positivestellensatz~\cite{putinar1993positive}]
Given a set $I :\ \{\vx \ | g_1(\vx) \geq 0\land\ldots\land g_l(\vx) \geq 0\}$, assume there exists $h \in \mbox{Cone}(I)$ s.t. $\{\vx \ | \ h(\vx) \geq 0\}$ is a compact set, and $p$ is a positive polynomial on $K$, then, there exists $\degr$ and $P_0,P_1,\ldots,P_l \succ 0$ s.t.
\[
    p(\vx) = \tupleof{P, \vz \vz^t} = \tupleof{P_0, \vz \vz^t} + \sum_{i=1}^l \tupleof{P_i, \tupleof{G_i, \vz\vz^t}\vz \vz^t}\,.
\]
\end{theorem}
Given a compact set $I$, without loss of generality, one could add an extra constraint $g_{l+1}(\vx) :\ N - ||\vx||^2$ for a large enough $N$. Then, the above theorem is applicable to $I' :\ I \cap \{\vx \ | \ N - ||\vx||^2 \geq 0\} = I$.

This idea leads to an SDP relaxation for approximate polynomial optimization~\cite{lasserre2001global}, as well as sum of squares (SOS) programming~\cite{papachristodoulou2002}, which is used to find Lyapunov functions over semi-algebraic sets. More specifically, by fixing a template $V_\vc$, and a large enough $\degr$ (to define the size of $P_0,P_1,\ldots,P_l$), any constraint of the form $(\forall \vx \in I) \ p_\vc(\vx) > 0$ is relaxed to an SDP constraint:
\[
\bigwedge_{i=0}^l P_i \succ 0 \land \left( \tupleof{P_0, \vz\vz^t} + \sum_{i=1}^l \tupleof{P_i, \tupleof{G_i, \vz\vz^t} \vz\vz^t} \right) = p_\vc(\vx)\,,
\]
wherein the last condition defines the (linear) relation between $P_0,P_1,\ldots P_l$, and $\vc$. Thus, one could use an SDP solver to find a Lyapunov function for polynomial dynamical systems~\cite{papachristodoulou2002}.

\paragraph{Sample-Based Lyapunov Function Synthesis:}
The problem of synthesizing Lyapunov functions for a control system by 
observing states of the system in simulation (for sampling) has been 
investigated in the past by Topcu et al. to learn Lyapunov functions 
along with the resulting basin of attraction~\cite{topcu2007stability}. 
Whereas the original problem is bilinear, the use of simulation data makes 
it simpler to postulate states that belong to the region of attraction, and
therefore find Lyapunov functions that belong to this region.
We use similar ideas to find \emph{control} Lyapunov functions.

\subsection{Lyapunov Function + Feedback law Synthesis}
As discussed in Chapter~\ref{ch:overview}, to address the control synthesis problem, most of constraint solving methods find a ``certificate+feedback law." In other words, $\K$ is also parameterized $\K_{\vc'}(\vx) :\ \sum_{i=1}^{r'} \vc_i' g'(\vx)$.
The problem of finding feedback function and Lyapunov function at the same time is harder as there are two sets of unknowns.

\paragraph{Linear Systems:}
For state feedback linear systems where $\dot{\vx} = A\vx + B\vu$, the problem remains simple. In fact, to stabilize to the origin, one would only need a linear feedback function $K$. Then, the goal is to find $\vc$ and $\vc'$ s.t.
\begin{align*}
    V_\vc(\vzero) = 0 \, \land \,
    (\forall \vx \neq \vzero) \ V_\vc(\vx) > 0 \, \land \, 
    (\forall \vx \neq \vzero) \ \nabla V_\vc \cdot f(\vx, \K_{\vc'}(\vx)) < 0\,.
\end{align*}
Written in matrix form, we wish to find $C$ and $C'$ s.t.
\[
C \succ 0 \, , \, C(A+BC') + (A+BC')C \prec 0\,.
\]
This problem seems harder as unknowns are multiplied which yields a bilinear problem. However, the following trick solves the bilinearity problem~\cite{lofberg2004yalmip}. Let $Q=C^{-1}$ and $Y = C'Q$. Then, find $Q$ and $Y$ s.t.
\[
Q \succ 0 \land QA + A^tQ + Y^tB^t + BY \prec 0 \,,
\]
which is linear (not biliear).
Then, $C$ is $Q^{-1}$ and $C' = CY$. Another approach is to use LQR~\cite{boyd1994linear} in which a Riccati differential equation is solved.

\paragraph{Nonlinear Systems:}
The idea mentioned above is not extendable to nonlinear systems. In fact, for a nonlinear system, such a problem is not convex (the feasible set is not even connected) anymore~\cite{prieur1999uniting}. Usually, this problem is formulated as a bilinear SOS programming with two sets of unknowns. There are expensive methods which can solve such nonconvex bilinear problem~\cite{lofberg2004yalmip}. Another approach is to use alternating optimizations.
To form an optimization problem, a scalar variable $\gamma$ is added and we wish to solve the following optimization problem:
\[
\min_{\vc, \vc', {\color{red}\gamma}} {\color{red}\gamma} \mbox{ s.t. } \ V_\vc(\vzero) = 0 \, \land \,
    (\forall \vx \neq \vzero) \ V_\vc(\vx) > 0 \, \land \,
    (\forall \vx \neq \vzero) \ \nabla V_\vc \cdot f(\vx, \K_{\vc'}(\vx)) < {\color{red}\gamma}\,.
\]

Then, alternatively (i) $\vc$ is fixed while $\gamma$ is minimized by changing $\vc'$, and (ii) $\vc'$ is fixed while $\gamma$ is minimized by changing $\vc$. This procedure repeats until it converges to a local minimum. If $\gamma$ is negative, the solution is feasible. This method (which is also called policy iteration) has poor guarantees in practice, it often \emph{gets stuck} on a saddle point that does not allow the technique to make progress in finding a feasible solution~\cite{el1994synthesis,BENSASSI2017} . To combat this, Majumdar et al. (ibid) use LQR controllers and their associated Lyapunov functions for the linearization of the dynamics as a good initial seed solution~\cite{majumdar2013control}. However, the linearization of the dynamics may not be controllable. In addition, the complexity of their method is exponential in the number of inputs when inputs are saturated.

\subsection{Control Lyapunov Function Synthesis}
The problem of finding control Lyapunov function synthesis is harder as there is one addition quantifier alternation.
The problem for linear systems is not considered as complete solutions exists for finding a ``Lyapunov function + feedback law" (as discussed). For nonlinear systems Tan et al. consider systems with control affine dynamics~\cite{tan2004searching}. Then, they reduce the problem of finding a CLF to the following problem:
\[
(\exists \vc \in \C) \ V_\vc(\vzero) = 0 \, , \, (\forall \vx \neq \vzero) \ V_\vc(\vx) > 0 \, , \, (\forall \vx \neq \vzero) \left(\bigwedge_{i=1}^m \nabla V.f_i(\vx) = 0 \right) \implies f_0(\vx) < 0,
\]
which gives a bilinear formulation through SOS programming. Then, the bilinearity is solved using alternative optimizations which may trap in a local infeasible solution.

In this chapter, we propose a framework to find CLFs as well as other control certificates. For other control certificates we use similar templates $V_\vc(\vx) : \sum_{i=1}^r c_i g_i(\vx)$, where $g_i(\vx)$ is a monomial in $\vx$. For \emph{smooth} feedback systems, we need to solve a formula of the following form:
\begin{equation}\label{eq:smooth-general}
    (\exists \vc \in \C)\ (\forall \vx \in X) \ 
    \ \begin{cases}
    \vx \in R_1 \implies \bigvee\limits_{q \in Q} (\exists \vu \in U) \ p_{\vc,1,q}(\vx, \vu) < 0 \\
    \vdots \\
    \vx \in R_l \implies \bigvee\limits_{q \in Q} (\exists \vu \in U)\ p_{\vc,l,q}(\vx,\vu) < 0\,,
\end{cases}
\end{equation}
wherein $R_i$ is a \emph{basic semi-algebraic set} and $p_{\vc,i,q}$ is linear in $\vc$. 
The disjunction over $Q$ is used for control barrier functions. To find a control barrier function, we wish to solve the following formula:
\begin{equation*}
    (\exists \vc \in \C) \ (\forall \vx \in X) \begin{cases}\begin{array}{rl}
\vx \in \partial S \implies & \ B_\vc(\vx) > 0 \\
\vx \in  I \implies & \  B_\vc(\vx) < 0 \\
\vx \in S \setminus \inter{I} \implies & \ \ \left( \begin{array}{c}
        (\exists \vu \in U) \ \nabla B_\vc \cdot f(\vx,\vu) - \lambda^* B_\vc(\vx) < 0 \\
        \lor \\
        (\exists \vu \in U) \ \nabla B_\vc \cdot f(\vx,\vu) + \lambda^* B_\vc(\vx) < 0 \\
    \end{array} \right) \,,
    \end{array}\end{cases}
\end{equation*}
which fits the general form in Eq.~\eqref{eq:smooth-general}. Similar argument applies for control funnel functions (see Eqs.~\eqref{eq:rules} and~\eqref{eq:cff-exp}).

For \emph{switched} feedback systems, we wish to solve the following formula:
\begin{equation}\label{eq:switched-general}
    (\exists \vc \in \C)\ (\forall \vx \in X) \ 
    \ \begin{cases}
    \vx \in R_1 \implies \bigvee\limits_{q \in Q} \bigwedge\limits_{s \in S} \ p_{\vc,1,q,s}(\vx) < 0 \\
    \vdots \\
    \vx \in R_l \implies \bigvee\limits_{q \in Q} \bigwedge\limits_{s \in S} \ p_{\vc,l,q,s}(\vx) < 0\,,
\end{cases}
\end{equation}
wherein $R_i$ is a basic semi-algebraic set and $p_{\vc,i,q,s}$ is linear in $\vc$. 
Notice that here $(\exists \vu \in U)$ is merged into the disjunction over $Q$ as $U$ is finite. Recall that we consider control Lyapunov fixed-barriers functions only for switched feedback systems. The disjunction over $S$ is needed when we search for a control Lyapunov fixed-barriers function:
\begin{equation*}
(\exists \vc \in \C)\ (\forall \vx \in X) \ 
    \ \begin{cases}
\begin{array}{rl}
\vx \in S \setminus \inter{G} \implies & \bigvee_{\vu \in U}  \left( \begin{array}{c}
    \nabla V_\vc \cdot f_\vu(\vx) < 0 \ \land \\
    \bigwedge_j\ \left( \begin{array}{c}\nabla p_{S,j} \cdot f_\vu(\vx) + \lambda\ p_{S,j}(\vx)\ < 0
\end{array}\right)
\end{array} \right)\,.
\end{array}
\end{cases}
\end{equation*}

Now, we propose our framework to solve Eqs.~\eqref{eq:smooth-general} and~\eqref{eq:switched-general}.

%% file: synt/cegis.tex
\section{Counterexample Guided Search}\label{sec:cegis}
The general idea for our proposed framework is to use a finite number of samples to learn a control certificate. The idea is simple: given a finite set of points (states), if a hypothesis $T(\vc)$ satisfies conditions for being a control certificate on these sample points (witness), $T(\vc)$ potentially satisfies those conditions for all states.

\begin{example} \label{ex:cegis-idea} Consider a smooth feedback system with two state variables $x_1$ and $x_2$, where $\dot{x_1} = x_2$, $\dot{x_2} = -x_1 + u$, and $u \in [-1, 1]$. The goal is to stabilize the system. It is easy to show (cf.~\cite{ravanbakhsh2015counter}) that $V$ is a CLF for the smooth feedback system iff $V$ is a CLF for a switched feedback system with two modes $ U :\ \{u_1, u_2\} $ and the following dynamics:
\[
\vu_1 \begin{cases}
    \dot{x_1} = x_2 \\
    \dot{x_2} = -x_1-1
\end{cases} \,,
\vu_2 \begin{cases}
    \dot{x_1} = x_2 \\
    \dot{x_2} = -x_1+1\,.
\end{cases} 
\]
We use a quadratic template $V_\vc(x_1, x_2):\ c_1 x_1^2 + c_2 x_1x_2 + c_3 x_2^2$. Then, we find a $\vc$ s.t. $V_\vc(\vx)$ satisfies Eq.~\eqref{eq:clf} only for states shown in red in Figure~\ref{fig:example-points}. In other words, for red states, (a) $V(\vx)$ is positive, and (b) there is an input ($\vu_1$ or $\vu_2$) which if selected, the value of $V$ decreases. We find one such $V_\vc$. $V_\vc(x_1, x_2) = 3 x_1^2 + 1.5 x_1x_2 + 1.5 x_2^2$. Interestingly, it is verified that this $V_\vc$ is in fact a control Lyapunov function.

\begin{figure}[t!]
\begin{center}
\includegraphics[width=0.8\textwidth]%
    {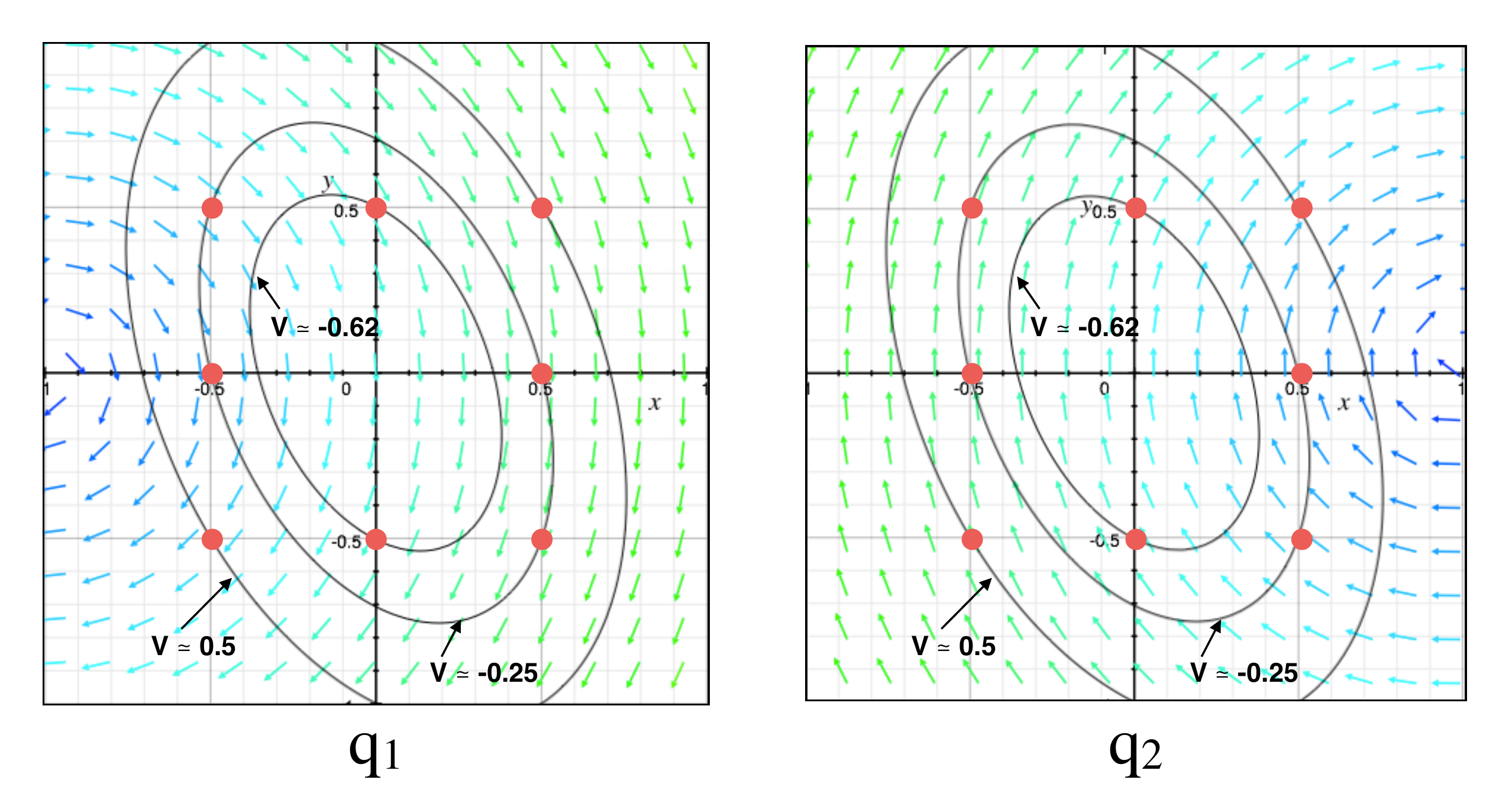}
\\
Left plot shows the vector
      field for mode $\vu_1$ and right plot shows the vector field for mode
      $\vu_2$ along with level sets of the CLF discovered.
    \caption{Learning a CLF from finite samples.}
\label{fig:example-points}
\end{center}
\end{figure}
\end{example}

Unfortunately, the set of samples may yield a candidate solution $T(\vc)$ which is not a control certificate. In these cases, we wish to generate a new sample and repeat the process iteratively. Ultimately, we want to carefully select samples and learn suitable parameters $\vc$ with few witnesses.
For learning Lyapunov functions, Kapinski et al.~\cite{kapinski2014simulation} propose to generate samples iteratively using counterexamples through counterexample guided inductive synthesis (CEGIS) framework, which originally is proposed in verification community by
Solar-Lezama et al.~\cite{solar2006combinatorial,solar2008program}. Here, we extend this method to learn \emph{control} certificates.

First, we briefly discuss how CEGIS works. The procedure is iterative and solves $ (\exists \vc \in \C) \ (\forall \vx \in X) \ \psi_\vc(\vx)$. By solving, we mean either prove no $ \vc^* \in C $ exists, or find a $ \vc^* \in \C $ for which $ (\forall \vx \in X) \ \psi_{\vc^*}(\vx) $ holds. The framework consists of two components: (i) a learner which generates a candidate solution using samples, and (ii) a verifier that tests whether the candidate solution is valid.

The learner uses the following concept to generate a candidate solution.
\begin{definition}[Sample Compatibility]
    A formula $\psi(\vx) :\ X \mapsto \bools$ is compatible with a set of samples $\X$
    iff $\bigwedge\limits_{\vx_i \in \X} \psi(\vx_i)$.
\end{definition}

For each iteration $ j $, we define a finite set of witnesses (counterexamples) $ \X_j :\ \{\vx_1,\ldots,\vx_j\} \subset X$. Then, an implicit set of candidate solutions $ \C_j $ defined as all $\vc \in \C$ for which $\psi_\vc(\vx)$ is compatible with the samples:
\[
\C_j :\ \{\vc \in \C | \bigwedge_{\vx_i \in \X_j} \psi_\vc(\vx_i)\} \,.
\]

Set $ \C_j $ over-approximates feasible solutions in $ \C $ ($\C_j$ is a set of potential solutions). $ \C_j $ is defined by $ \X_j $ and the procedure starts with $ \X_0 = \emptyset$ ($\C_0 = \C$). As depicted in Figure~\ref{fig:cegis}, in the $ j^{th} $ iteration starting from $j=1$, the following steps are executed:
\begin{compactenum}
    \item \textsc{findCandidate:} The learner checks wether $ \C_{j-1} $ is empty
    \begin{compactenum}
        \item If yes, CEGIS terminates, proving no solution exists,
        \item Otherwise, the learner returns a candidate $ \vc_j \in \C_{j-1} $.
    \end{compactenum}
    \item \textsc{verify:} The verifier checks whether candidate $\vc_j$ yields a solution: $ (\forall \vx \in X) \ \psi_{\vc_j}(\vx) $
    \begin{compactenum}
        \item If yes, CEGIS terminates with $ \vc^* = \vc_j $ as a solution,
        \item Otherwise, the verifier returns a counterexample (witness) $ \vx_j \in X $ s.t. $ \lnot \psi_{\vc_j}(\vx_j) $.
    \end{compactenum}
    \item \textsc{update} add the new witness to the witnesses set:
\begin{align}
    \X_{j} :& \ \X_{j-1} \cup \{\vx_{j}\} \\
    \C_{j} :& \ \C_{j-1} \cap \{\vc \in \C\ | \ \psi_{\vc}(\vx_{j}) \} \,.
\end{align} 
\end{compactenum}

\begin{figure}[t!]
\begin{center}
\includegraphics[width=0.7\textwidth]%
    {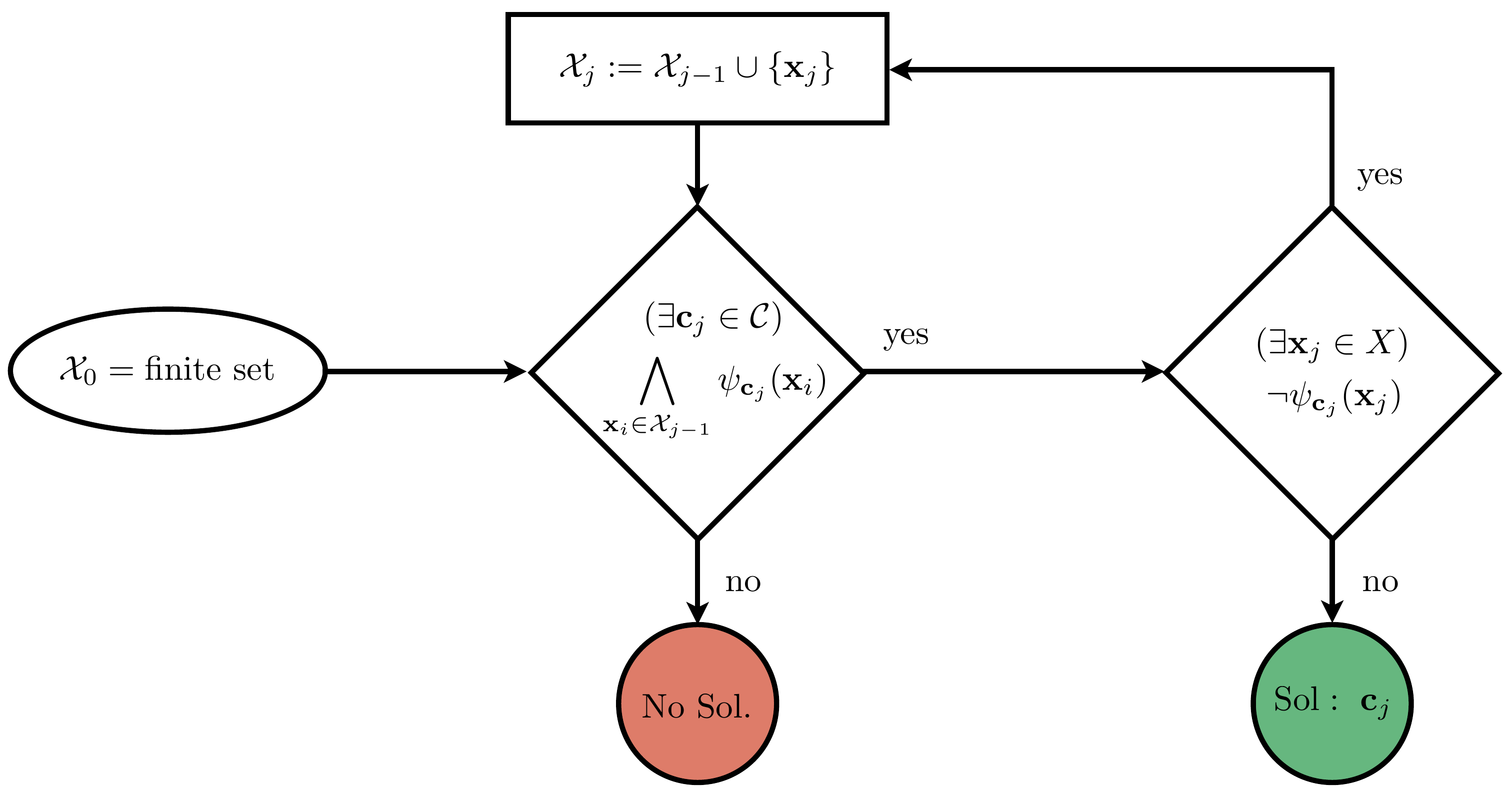}
\\
\caption{Visualization of the CEGIS framework.}
\label{fig:cegis}
\end{center}
\end{figure}

First, the learner checks whether a candidate solution exists. Since $\C_j$ over-approximates feasible solution, $\C_j = \emptyset$ implies that no solution exists. Otherwise, the verifier tests feasibility of a candidate solution. If no counterexample is found, the candidate is feasible. Otherwise, the counterexample is added to the set of witnesses. We note that $ \vx_{j} \not\in \X_{j-1} $ since it is guaranteed for all $\vc_j \in \C_{j-1}$, $ (\forall \vx_j \in \X_{j-1})\ \psi_{\vc_j}(\vx_j) $. Intuitively, set $ \X_{j} $ has one more member ($ \vx_{j} $) that $ \X_{j-1} $ and $ \C_{j} \subset \C_{j-1} $. If $ X $ is a finite set, CEGIS terminates with two possible outcomes: Either (a) $ \C_j = \emptyset $ for some iteration $ j $ (no solution exists), or $ \vc_j $ is a solution (no $ \vx_j $ exists).

\begin{remark}
    The initial set of witnesses $\scr{X}_0$ does not need to be empty, and it can be an arbitrary finite set.
\end{remark}

The CEGIS method is not directly applicable to our problem, and we need some adaptations. Recall Eq.~\eqref{eq:clf}:
\begin{equation}\label{eq:clf-Vc}
    V_\vc(\vzero) = 0 \, , \, (\forall \vx \neq \vzero) \ \left( V_\vc(\vx) > 0 \land (\exists \vu \in U) \ \nabla V_\vc \cdot f(\vx, \vu) < 0 \right) \,.
\end{equation}
Let $\C_0 :\ \{ \vc \in \C | \ V_\vc(\vzero) = 0 \}$. We wish to solve $(\exists \vc \in \C) \ (\forall \vx \neq \vzero) \psi_\vc(\vx)$, where $\psi_\vc(\vx)$ is

\[
    \psi_\vc(\vx) :\ V_\vc (\vx) > 0 \land (\exists \vu \in U) \ \nabla V_\vc \cdot f(\vx, \vu) < 0\,.
\]

To implement the CEGIS framework, two main steps should be addressed: (a) \emph{Discovering} a $\vc_j \in \C_{j-1} $, and (b) \emph{Verifying} whether candidate $ V_{\vc_j} $ is a solution. For the learning step, we need to solve the following formula for $\vc \in \C$:

\begin{equation} \label{eq:candidate-generation}
    \bigwedge_{\vx_i \in \X_{j-1}} \left( V_\vc(\vx_i) > 0 \land (\exists \vu \in U) \ \nabla V_\vc \cdot f(\vx_i, \vu) < 0 \right) \,.
\end{equation} 

Also, in the verification step, we need to solve two formulae for $\vx \in X$:
\begin{align}
    (a) & :\ \vx \neq \vzero \land V_{\vc_j}(\vx) \leq 0 \label{eq:ce-generation-pos} \\
    (b) & :\ \vx \neq \vzero \land (\exists \vu \in U) \ \nabla V_{\vc_j} \cdot f(\vx, \vu) \geq 0\,. \label{eq:ce-generation-neg}
\end{align}

However, these formulae are hard to solve mainly because of the quantifier over $\vu$. To make CEGIS applicable to our problem we consider switched systems (i.e., we assume that the set $ U $ is finite). Theoretically, each $U$ can be safely under-approximated with a finite set $ \hat{U} $ and a solution for the switched system is also a solution for the original problem. This reduction is complete for certain systems/properties~\cite{ravanbakhsh2015counter}. For now, we assume that $U$ is finite. We will remove this restriction, later in this chapter. Having a finite $U$, each formula mentioned above is a quantifier-free formula as $(\exists \vu \in U)$ is replaced with $\bigvee_{\vu \in U}$. More precisely, Eq.~\eqref{eq:candidate-generation} is a Quantifier-Free formula in Linear Real Arithmetic theory (QF-LRA), which can be solved with satisfiability modulo theories (SMT) solvers such as Z3~\cite{DeMoura-Bjorner-08-Z3} or mixed integer linear programming (MILP) solvers. Eqs.~\eqref{eq:ce-generation-pos} and~\eqref{eq:ce-generation-neg} are quantifier-free formula as well. Solving these formulae is \emph{in general} undecidable if the dynamics (or basis functions) include trigonometric and exponential functions. However, $\delta$-decision procedures can solve these problems approximately under certain assumptions~\cite{gao2012delta}.
Assuming that the dynamics and chosen bases are polynomials in $\vx$, Eqs.~\eqref{eq:ce-generation-pos} and~\eqref{eq:ce-generation-neg} are Quantifier Free formula in Nonlinear Real Arithmetic theory (QF-NRA). Such problems are decidable with high complexity (NP-hard)~\cite{Basu+Pollock+Roy/03/Algorithms}.
Exact approaches using semi-algebraic geometry~\cite{Brown:2007:CQE:1277548.1277557} or branch-and-bound solvers (e.g. dReal~\cite{DBLP:conf/cade/GaoKC13}) can tackle this problem precisely.

\emph{Numerical} SMT Solvers such as dReal~\cite{DBLP:conf/cade/GaoKC13} are shown to be more efficient. However, they use some numerical thresholds, which result in two issues. First, numerical SMT solvers need the region of interest (here $X \setminus \{\vzero\}$) to be bounded. Second, an output for a given formula is either \emph{UNSAT} (unsatisfiable) or \emph{$\delta$-SAT}, when the formula is satisfiable under some $\delta$-perturbation~\cite{gao2012delta} (for $\delta > 0$). These two issues prevent us from verifying a CLF because first, the CLF conditions are defined over an unbounded open set ($\vx \neq \vzero$), and second, an SMT solver would return $\delta$-SAT for any CLF as infinite precision is needed around the origin. For these reasons, we consider other control certificates (where the regions of interest are compact semi-algebraic sets) in the rest of this section. We will deal with CLFs in the next section.

Recall that for other control certificates we use similar templates: $V_\vc(\vx) = \sum_{i=1}^r c_i g_i(\vx)$, where $g_i(\vx)$ is a monomial in $\vx$. Then, assuming $\C_0 :\ \C$, using learning a proper $\vc$ is equivalent to solving $(\exists \vc \in \C) \ (\forall \vx \in X) \ \psi_\vc(\vx)$ where $\psi_\vc(\vx)$ has the following general form (see Eq.~\eqref{eq:switched-general}):
\begin{equation}\label{eq:cegis-psi}
\psi_\vc(\vx) :\ \begin{cases}
    \vx \in R_1 \implies \bigvee\limits_{q \in Q} \bigwedge\limits_{s \in S} \ p_{\vc,1,q,s}(\vx) < 0 \\
    \vdots \\
    \vx \in R_l \implies \bigvee\limits_{q \in Q} \bigwedge\limits_{s \in S} \ p_{\vc,l,q,s}(\vx) < 0\,,
\end{cases}
\end{equation}
wherein $R_i$ is a basic semi-algebraic set, and $p_{\vc,q,s,i}$ is linear in $\vc$. Moreover, for all control certificates except for CLFs, $R_i$ is \emph{compact} (for all $i$). At each iteration a candidate $\vc_j$ is generated s.t. $\psi_{\vc_j}(\vx)$ holds for all $\vx \in \X_{j-1}$:
\begin{equation}\label{eq:cegis-lra}
\bigwedge_{\vx \in \X_{j-1}} \begin{cases}
    \vx \in R_1 \implies \bigvee\limits_{q \in Q} \bigwedge\limits_{s \in S} \ p_{\vc_j,1,q,s}(\vx) < 0 \\
    \vdots \\
    \vx \in R_l \implies \bigvee\limits_{q \in Q} \bigwedge\limits_{s \in S} \ p_{\vc_j,l,q,s}(\vx) < 0 \,. \\
\end{cases}
\end{equation}
 This formula belongs to QF-LRA. Then, having a candidate $\vc_j$, we solve $l$ different verification problems:
\begin{align*}
    (1) & :\ \vx \in R_1 \land \bigwedge\limits_{q \in Q} \bigvee\limits_{s \in S} p_{\vc_j,1,q,s}(\vx) \geq 0 \\
    \vdots \\
    (l) & :\ \vx \in R_l \land \bigwedge\limits_{q \in Q} \bigvee\limits_{s \in S} p_{\vc_j,l,q,s}(\vx) \geq 0
\,.
\end{align*}
Now, we could use numerical SMT solvers to solve these problems because $R_i$ is a compact set. 
They can correctly conclude that the current candidate yields a valid control certificate.  In case one of the problems is $\delta$-SAT, the witness may not be a witness to the original problem. Using the spurious witness may cause the CEGIS procedure to potentially continue (needlessly) even when a solution $\vc_i$ has been found. Even worse, a $\delta$-SAT procedure may cause $\C_{j-1} = \C_{j}$ (because of a spurious witness), i.e., CEGIS does not progress. We address this problem next.

\subsection{Termination}
The termination of CEGIS is not guaranteed in a continuous domain, even if we use symbolic SMT solvers to make sure the method progresses. We noted that termination is possible if a solution of the desired form exists, or the hypothesis space is exhausted. However, neither situation may happen, and the algorithm may run forever. 
We provide a strengthening of Eq.~\eqref{eq:cegis-lra} that guarantees termination:
 \begin{equation}\label{eq:cegis-lra-relaxed}
\bigwedge_{\vx \in \X_{j-1}} \psi_\vc(\vx) : \begin{cases}
    \vx \in R_1 \implies \bigvee\limits_{q \in Q} \bigwedge\limits_{s \in S} \ p_{\vc_j,1,q,s}(\vx) < {\color{red}-\epsilon_T} \\
    \vdots \\
    \vx \in R_l \implies \bigvee\limits_{q \in Q} \bigwedge\limits_{s \in S} \ p_{\vc_j,l,q,s}(\vx) < {\color{red}-\epsilon_T} \,, \\
\end{cases}
 \end{equation}
wherein $\epsilon_T > 0$ is fixed and larger than the threshold used for the numerical SMT solver.

Let $\vc_j$ be a candidate examined at the $j^{th}$ iteration of the CEGIS procedure modified to use Eq.~\eqref{eq:cegis-lra-relaxed}. Suppose $V_{\vc_j}$ fails to be a control certificate and we compute $\C_{j}$. It is easily shown the $\vc_j \not \in \C_{j}$.
Furthermore, by using~\eqref{eq:cegis-lra-relaxed}, we obtain the following result that any candidate in a $\eta$-ball around $\vc_j$ is also eliminated if $\C$ is compact.
\begin{theorem}
   \label{thm:eta-exists}
Given a compact $\C$, if the CEGIS procedure were modified using Eq.~\eqref{eq:cegis-lra-relaxed} with a given $\epsilon_T > 0$, then there exists a constant $\eta > 0$ such that at each iteration $j$, $ \B_{\eta}(\vc_j) \cap \C_{j} = \emptyset$.
\end{theorem}
\begin{proof}
 Given a counterexample ($\vx_j$) for $V_{\vc_j}$, there is a $k \in [1,\ldots,l]$ s.t.
 \begin{align*}
    (k) :\ \vx_j \in R_k \land \bigwedge\limits_{q \in Q} \bigvee\limits_{s \in S} \ p_{\vc_j,k,q,s}(\vx_j) \geq 0\,,
 \end{align*}
 and for next iteration
 \begin{align}
  \label{eq:next-iteration-restrictions}
  \C_{j} \subseteq \C_{j-1} \cap \left\{\vc \ | \bigvee\limits_{q \in Q} \bigwedge\limits_{s \in S} \ p_{\vc,k,q,s}(\vx_j) < -\epsilon_T \right\} \,.
 \end{align}
Let $p'(\vc) :\ \min_{q \in Q} \max_{s \in S} \ p_{\vc,k,q,s}(\vx_j, \vu)$.
As $p'$ is continuous and piecewise differentiable. Therefore, there is a Lipschitz constant $A_{p'}$ s.t.
 \[
 (\forall \vc_1, \vc_2 \in \C) \ ||p'(\vc_1) - p'(\vc_2)|| \leq A_{p'} ||\vc_1 - \vc_2||\,.
 \]
Let $\eta :\ \frac{\epsilon_T}{A_{p'}}$. Then
\begin{align*}
(\forall \vc \in \B_{\eta}(\vc_j)) & \ ||p'(\vc) - p'(\vc_j) || \leq  A_{p'} ||\vc - \vc_j|| \leq A_{p'} \eta = \epsilon_T \\
\implies & p'(\vc) \geq p'(\vc_j) - \epsilon_T \geq -\epsilon_T 
\implies \vc \not\in \C_{j}\,.
\end{align*}\QED
\end{proof}

As a result, starting from a compact initial set $\C_0:\C$, we note that employing the stronger rule (Eq.~\eqref{eq:cegis-lra-relaxed}) guarantees that at each step, an $\eta$-ball around the current solution is also removed. Thus, either a control certificate is found, or the hypothesis space is empty within finitely many iterations.
If we exhaust the hypothesis space for a given value of $\epsilon_T$, it is possible to repeat the search by halving $\epsilon_T$ to alleviate the loss of possible solutions due to the strengthening of Eq.~\eqref{eq:cegis-lra} by Eq.~\eqref{eq:cegis-lra-relaxed}.

\subsection{Implementation Heuristics}
A first cut application of the CEGIS approach, presented thus far, resulted in a prohibitively large number of witnesses, failing on most of our benchmarks. Such a failure happens because candidate control certificates returned by the SMT solvers are similar (parameters are close in term of Euclidean distance). We discuss a heuristic for witnesses selection at each step, that leads to successful implementation of the overall procedure.

Given a current candidate $\vc_j$,  we may split the search for a witness into $l$ parts, for each of which we find a witness that violates the corresponding condition.
For the $k^{th}$ condition, we search for a counterexample that produces the ``most-egregious'' violation of the constraint possible. Therefore, we wish to maximize 
\[
\min\limits_{q \in Q} \max\limits_{s \in S} p_{\vc_j,k,q,s}(\vx)\,.
\]
However, many SMT solvers currently lack the ability to optimize. Therefore, we simply fix a constant $\gamma$ and search for $\vx_j$ satisfying
\[
\min\limits_{q \in Q} \max\limits_{s \in S} p_{\vc_j,k,q,s}(\vx_j) > \gamma\,.
\]
A larger $\gamma$ leads to a more ``egregious'' violation and a larger set of candidates ruled out in the hypothesis space and it is less likely to find a  candidate that is similar to the previously selected candidate. The parameter $\gamma$ itself is iteratively reduced to find a witness or conclude that no witness exists when $\gamma = 0 $.
Another heuristic is to seed the process with an initial set of points $\X_0$, which in our experiments improves the performance.

\paragraph{Equality Constraints:}
Recall the conditions for a parameterized control barrier function $B_\vc$
\begin{align*}
    (\forall \vx \in I) & \ B_\vc(\vx) < 0 \\
    (\forall \vx \in \partial S) & \ B_\vc(\vx) > 0 \\
    (\forall \vx \in S \setminus \inter{I}) & \ B_\vc(\vx) = 0 \implies \bigvee\limits_{\vu \in U} \ \nabla B_\vc \cdot f(\vx, \vu) < 0 \,.
\end{align*}

The condition in Eq.~\eqref{eq:cbf-original} can be encoded into the CEGIS framework. However, the presence of the equality $ B(\vx) = 0$ poses practical problems. In particular, it requires to find a candidate $B_\vc$, s.t. for all $\vx \in \X$:
\[
\ B_\vc(\vx) \neq 0 \lor \bigwedge\limits_{\vu \in U} \ \nabla B_\vc \cdot f(\vx, \vu) \geq 0 \,.
\]
Unfortunately, such an assertion is easy to satisfy ($B_\vc(\vx) \neq 0$ can easily be satisfied), resulting in the procedure always exceeding the maximum number of iterations permitted. However, our experiments suggest that the following relaxation (as discussed in Chapter~\ref{ch:certificates}) is particularly effective:
\begin{equation*}
    \begin{array}{rl}
    (\forall \vx \in I) & \ B_\vc(\vx) < 0 \\
    (\forall \vx \in \partial S) & \ B_\vc(\vx) > 0 \\
    (\forall \vx \in S \setminus \inter{I}) & \ \bigvee\limits_{\vu \in U}\left( \begin{array}{c}
        \nabla B \cdot f(\vx, \vu) - \lambda B(\vx) < 0
        \lor
        \nabla B \cdot f(\vx, \vu) + \lambda B(\vx) < 0
    \end{array} \right) \,.
    \end{array}
\end{equation*}
The same argument applies to control funnel functions (see Eq.~\eqref{eq:cff-exp}).

\subsection{An Illustrative Example}
We consider a simple example and follow the CEGIS procedure step by step. The system of interest has two continuous variables $x$ and $y$ and two modes $\vu_1$ and $\vu_2$ with the following vector field:
\begin{align*}
\vu_1 \begin{cases}
    \dot{x} = -x + y^2 \\
    \dot{y} = 1
\end{cases} \,,
\vu_2 \begin{cases}
    \dot{x} = -x - y^2 \\
    \dot{y} = -1 \,.
\end{cases} 
\end{align*}
We are interested in a reach-while-stay (RWS) property where $S = [-0.5, 0.5]^2$, $I :\ \B_{0.2}(\vzero)$, and $G:\B_{0.05}(\vzero)$. 
We use the following template for control Lyapunov-barrier function $V:\ c_1x^2 + c_2xy + c_1 y^2$.  Notice that we assumed coefficients of $x^2$ and $y^2$ are equal so that the $\C$ becomes a 2D space suitable for illustration.
We use Eq.~\eqref{eq:cegis-lra-relaxed} by setting $\epsilon_T = 1$. Initial sets are $\C_0 = [-10, 10]^2$, $\X_0 =\ \{(z_1, z_2) |  z_1, z_2 \in \{ -0.5, 0, 0.5\} \} \setminus \{(0, 0)\}$.

\begin{figure}[t!]
\begin{center}
\includegraphics[width=1.0\textwidth]%
    {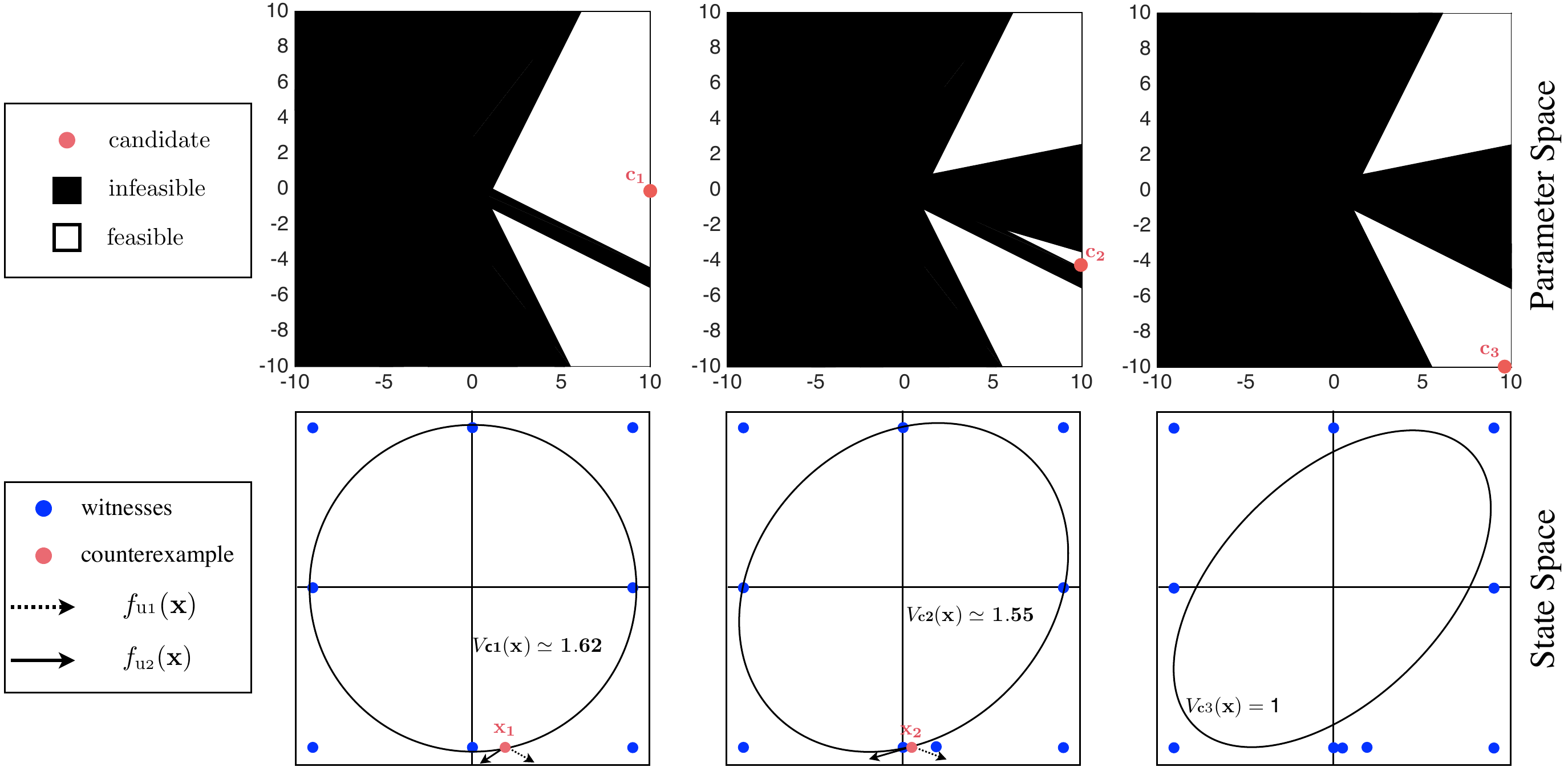}
\\
(\textbf{top}) Parameter space for the various iterations. The feasible set is shown in white and the current candidate is shown in red. (\textbf{bottom}) The phase plots and the current witness points (witness violates the candidate CLBF is shown in red as well as the vector fields for the witness).
\caption{Iterations of the CEGIS algorithm.}
\label{fig:example-complete}
\end{center}
\end{figure}

The whole procedure is summarized in Figure~\ref{fig:example-complete}.  The candidate $\vc_1=\ [10 \ \ 0]^t$ is learned in the first iteration. $V_{\vc_1}$ fails to be a CLBF as a witness $\vx_1 = [0.115, -0.4989]^t$ is found with $\gamma = 8$. The figure shows that the flow at $\vx_1$ locally violates the candidate $V_{\vc_1}$ in each mode.  In the next iteration, by adding this counterexample to $\X$, feasible parameter set gets smaller and then $\vc_2=\ [10 \ \ -4]^t$ is learned. This fails once again with witness $\vx_2 = [0.024 \ \ -0.5]^t$ for $\gamma = 2$. For the third iteration, the feasible region is further refined as shown in Figure~\ref{fig:example-complete} (top right). In this iteration $\vc_3 = [9.92288424748 \ \ -10]^t$ is learned as a candidate and verified by the verifier. The final CLBF is $V(\vx) = 9.92288424748x^2 - 10xy + 9.92288424748y^2$.

\subsection{Completeness and Complexity}
There are many sources of incompleteness: (i) The polynomial template on the control certificate; (ii) The use of  $\epsilon_T$ in Eq.~\eqref{eq:cegis-lra-relaxed}; and finally (iii) the use of a $\delta$-satisfiability solver for nonlinear constraints. However, it is possible to reduce this incompleteness by making $\delta$ and $\epsilon_T$ smaller and using larger polynomial templates.

Regarding the complexity, solving linear arithmetic constraints and quantifier-free nonlinear constraints are well-known to be NP-hard. Furthermore, while it is guaranteed that there will be a finite number of iterations, this number can be prohibitively large. Though we provided some heuristics to decrease the number of iterations, the worst case can be in the order of $O(d^r)$, where $r$ is the number of unknown coefficients in the template and $d$ is a function of $\C_0$ and $\eta$ in Theorem~\ref{thm:eta-exists}.

\subsection{Handling Disturbances}
Recall that modeling disturbances would result in additional quantifier alternations. Moreover, we address disturbances only for switched feedback systems.
To find a robust control certificate, we need to extend Eq.~\eqref{eq:switched-general} and solve a problem with the following
structure:
\begin{equation*}
(\exists\ \vc \in C)\ (\forall\ \vx \in X) \begin{cases}
    \vx \in R_1 \implies \bigvee_{q\in Q} \bigwedge_{s\in S} (\forall\ \vd \in D)\ p_{\vc,1,q,s}(\vx,\vd) < 0 \\
    \vx \in R_2 \implies \bigvee_{q\in Q} \bigwedge_{s\in S} (\forall\ \vd \in D)\ p_{\vc,2,q,s}(\vx,\vd) < 0 \\
    \vdots \\
    \vx \in R_l \implies \bigvee_{q\in Q} \bigwedge_{s\in S} (\forall\ \vd \in D)\ p_{\vc,l,q,s}(\vx,\vd) < 0\,.
\end{cases}
\end{equation*}

First, we replace conjunction over $s \in S$ ($|S| = k'$) by a quantifier over $\vs' \in S :\ \{\vs' \ | \ \vs' \geq 0 \land \vec{1} \cdot \vs' = 1\}$ s.t.
\[
\bigwedge_{s \in S} h_s(\vx) < 0 \iff (\forall \vs' \in S') \  \vs' \cdot \vh(\vx) < 0 \,,
\]
where $\vh$ is vector of function $h_s$'s. 
Then, the formula is rewritten in a shorter form as
\begin{equation*}
(\exists\ \vc \in C)\ (\forall\ \vx \in X) \begin{cases}
    \vx \in R_1 \implies \bigvee_{q\in Q} (\forall\ \vd' \in D')\ p_{\vc,1,q}(\vx,\vd') < 0 \\
    \vx \in R_2 \implies \bigvee_{q\in Q} (\forall\ \vd' \in D')\ p_{\vc,2,q}(\vx,\vd') < 0 \\
    \vdots \\
    \vx \in R_l \implies \bigvee_{q\in Q} (\forall\ \vd' \in D')\ p_{\vc,l,q}(\vx,\vd') < 0\,,
\end{cases}
\end{equation*}
for a quantifier over $\vd'^t :\ [\vs'^t \ \vd^t] \in D' :\ S' \times D$.

A first solution consists of applying CEGIS for $\exists\forall$ described previously. However, doing so yields quantified constraints for the candidate and witness generation steps. Since our objective was to avoid these quantified constraints in the first place, we modify the witness structure. 

Our solution is conceptually simple: we will extend the witnesses structure.  Rather than witnesses which are simply states $\vx_i \in X$, we will now allow witnesses that are of the form $(\vx_i, (Q \mapsto D') )$, i.e, a combination of a state $\vx_i \in X$ and a map from each disjunction to a disturbance vector. Since $Q$ is finite, this map is explicitly stored as $(\vx_i, (q_1, \vd'_{q_1}), \ldots, (q_k,\vd'_{q_k}))$.

\paragraph{Witness Structure:} A witness to the violation of a given robust control certificate candidate $T(\vc)$ includes a state $\vx \in X$ at which the violation happens along with for each disjunction $q$, a disturbance witness
$\vd'_q \in D'$ that will violate the formula.
With disturbances, each witness then has the following structure:
\begin{equation}\label{eq:witness-structure}
 \vy_i:\ \left( \vx_i, (q_1, \vd'^{(i)}_{q_1}), \ldots, (q_m, \vd'^{(i)}_{q_k}) \right) \,.
\end{equation}
With this witness structure, the overall CEGIS procedure now extends naturally.

\paragraph{Learning:} Let $\mathcal{Y}_j:\ \{ \vy_1,\ldots,\vy_j \}$ be the set of witnesses at the $j^{th}$ iteration, starting from $\mathcal{Y}_0:\ \emptyset$. The learner solves the formula:
\begin{equation}\label{eq:cegis-instantiation-dist}
(\exists \vc \in \C)\ \land\ \bigwedge_{(\vx_i,(q_1,\vd'^{(i)}_{q_1}),\ldots,(q_k,\vd'^{(i)}_{q_k})) \in \mathcal{Y}_{j-1}}\ \begin{cases}
 \vx_i \in R_1 \implies \bigvee\limits_{q \in Q} p_{\vc,1,q}(\vx_i,\vd'^{(i)}_q) < 0 \\
 \vx_i \in R_2 \implies \bigvee\limits_{q \in Q} p_{\vc,2,q}(\vx_i,\vd'^{(i)}_q) < 0 \\
    \vdots \\
\vx_i \in R_l \implies \bigvee\limits_{q \in Q} p_{\vc,l,q}(\vx_i,\vd'^{(i)}_q) < 0 \,. \\
\end{cases}
\end{equation}
We now use an SMT solver to find if the unquantified formula holds.

\paragraph{Verification:} Once a candidate $\vc_j$ is generated we now evaluate if it yields a robust control certificate. Since $\psi_\vc$ itself is the conjunction of $l> 0$ conditions, its negation is a disjunction and we can check each disjunct separately for satisfiability.  Each disjunct has the following form:
\[ \vx \in R_i\ \land\ \bigwedge_{q \in Q} \ (\exists\ \vd' \in D')\
p_{\vc_j,i,q}(\vx, \vd') \geq 0 \,.\]
We can remove the existential quantifier over $\vd'$ equivalently through $k$ fresh set of variables $\vd'_{q_1}, \ldots, \vd'_{q_k}$. The new disjunct is written:
\begin{equation}\label{eq:witness-generation-rcc-form}
 \vx \in R_i\ \land\ \bigwedge_{q \in Q} p_{\vc_j,i,q}(\vx,\vd'_{q}) \geq 0 \,. 
\end{equation}
If satisfiable, we obtain a witness $(\vx_{j}, (q_1, \vd'^{(j)}_{q_1}),\ldots, (q_k,$ $\vd'^{(j)}_{q_k}) )$. Otherwise, we conclude that $\vc_j$ yields a valid robust control certificate.

%% file: synt/demonstration.tex
 \section{Demonstration Guided Search}
The CEGIS framework can be viewed as a framework in which a learner interacts with a verifier \emph{oracle}.
A more general framework is Oracle-Guided Inductive Synthesis (OGIS) framework wherein, a learner interacts with different input/output oracles~\cite{jha2010oracle}. As depicted in Figure~\ref{fig:ogis}, if the only oracle is the verifier, then OGIS is equivalent to CEGIS. Here, in addition to a  verifier oracle, we wish to use a demonstrator oracle.
\begin{figure}[t]
\begin{center}
\includegraphics[width=0.8\textwidth]{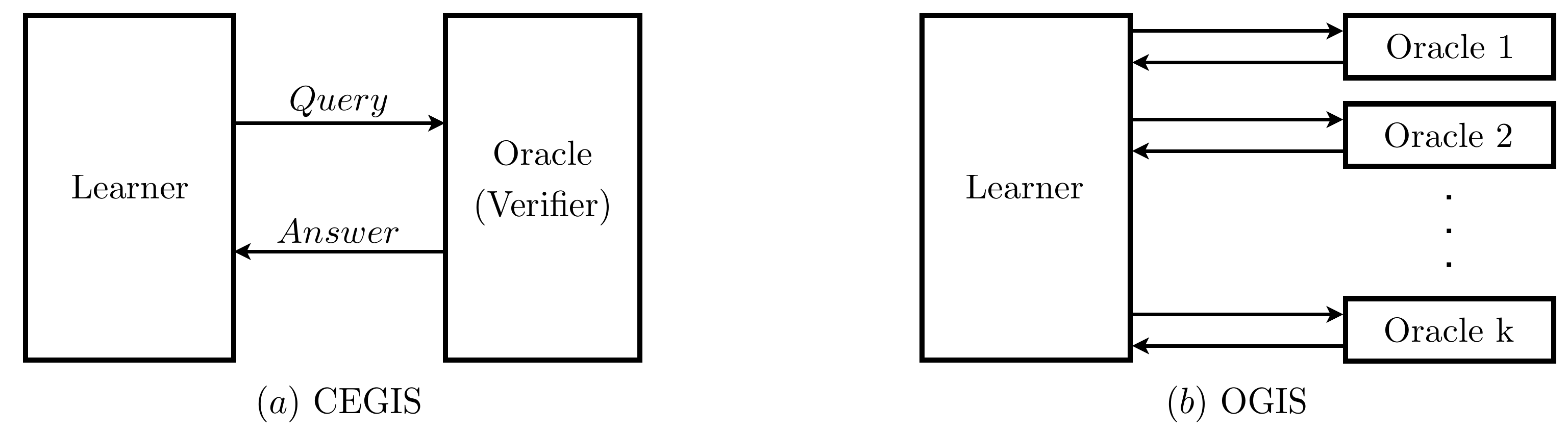}\\
\end{center}
\caption{OGIS vs. CEGIS.}
\label{fig:ogis}
\end{figure}
 The idea of learning from demonstrations has a long history~\cite{ARGALL2009469}. A demonstrator can, in fact, be a human operator~\cite{KHANSARIZADEH2014,Khansari-Zadeh2017} or a complex model predictive control~\cite{stolle2006policies,atkeson2013trajectory,ross2011reduction,zhong2013value,Mordatch-RSS-14,zhang2016learning}. However, the goal in these articles is to learn a policy similar to that of the demonstrator through statistical optimizations over large models such as deep neural networks.
In a more related method, Khansari-Zadeh et al. use human demonstrations to generate data and enforce CLF conditions for the data points, to learn a CLF candidate~\cite{KHANSARIZADEH2014}. However, their method does not include a verifier, and therefore, the CLF candidate may not, in fact, be a CLF. In this section, we extend the CEGIS framework to use a demonstrator oracle.

We investigate the problem of learning a control certificate using a black-box \emph{demonstrator} that can be queried with a given system state, and responds by demonstrating control inputs to address the specification starting from that state. Such a demonstrator can be realized using an expensive nonlinear model predictive controller (MPC) that uses a local optimization scheme or even a human operator.
The framework includes three components: (i) a \textsc{Learner} which selects a candidate control certificate, (ii) a \textsc{Verifier} that tests whether this control certificate is valid, and (iii) a \textsc{Demonstrator}. 
When verification fails, the \textsc{Verifier} returns a state (counterexample) at which the current candidate fails and the \textsc{Learner} queries the \textsc{Demonstrator} to obtain a control input corresponding to this state.

\subsection{Illustrative Example: TORA System}
\begin{figure*}[t]
\begin{center}
\includegraphics[width=0.95\textwidth]{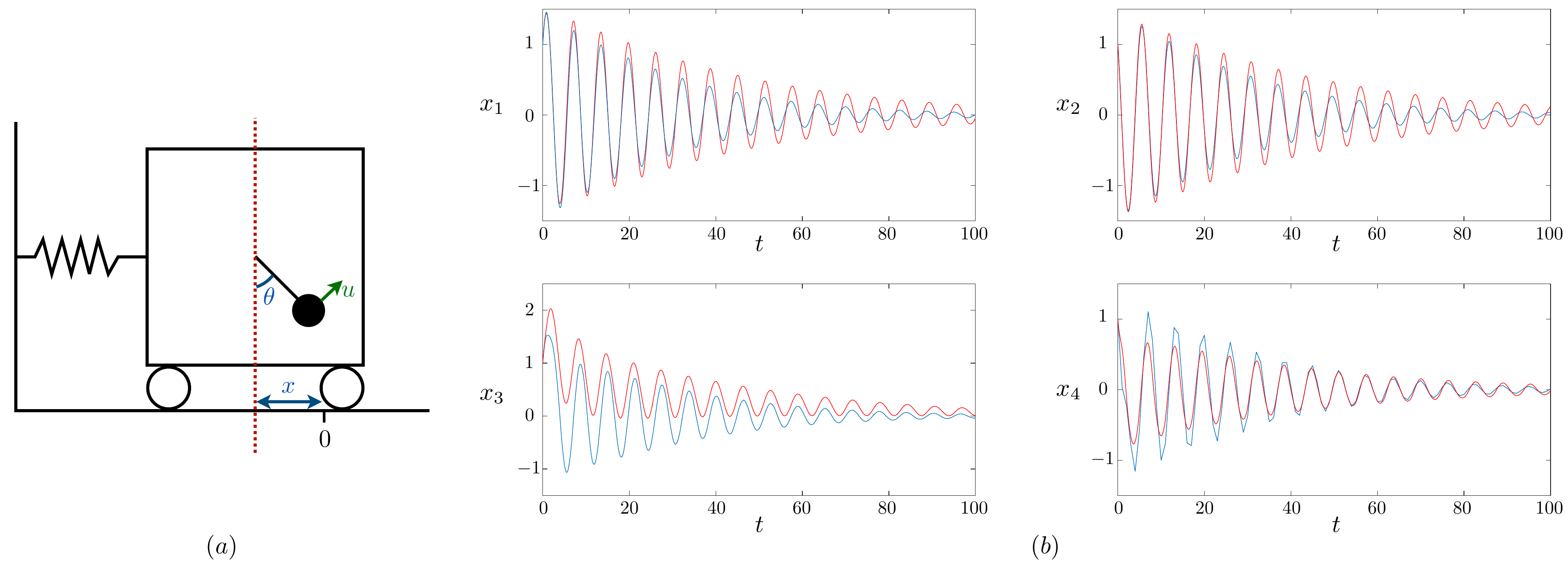}\\
(a) A schematic diagram of the TORA 
system. (b) Execution traces of the system using MPC control 
(blue traces) and Lyapunov based control (red traces) starting 
from same initial state $[1 \ 1 \ 1 \ 1]$.
\end{center}
\caption{TORA System.}
\label{Fig:tora-example}
\end{figure*}

Figure~\ref{Fig:tora-example}(a) shows a mechanical system, called translational oscillations with a rotational actuator (TORA).
The system consists of a cart attached to a wall using a spring. Inside the cart, there is an arm with a weight which can rotate. The cart itself can oscillate freely, and there are no friction forces. The system has two degrees of freedom, including the position of the cart $x$, and the rotational position of the arm $\theta$. The controller can rotate the arm through input $u$.
The goal is to stabilize the cart to $x=0$, with its velocity, angle, and angular velocity $\dot{x} = \theta=\dot{\theta} = 0$. We refer the reader to Jankovic et al.~\cite{jankovic1996tora} for a derivation of the dynamics, shown below in terms of state variables $(x_1,\ldots,x_4)$ and control input $u_1$, after a suitable basis transformation:
\begin{equation}\label{eq:tora-dyn}
        \dot{x_1} = x_2,\, \dot{x_2} = -x_1 + \epsilon \sin(x_3),\, \dot{x_3} = x_4,\, \dot{x_4} = u_1\,.
\end{equation}
$\sin(x_3)$ is approximated using a degree three polynomial approximation which is quite accurate over the range $x_3 \in [-2,2]$ (region of interest).
The equilibrium $x = \dot{x} = \theta = \dot{\theta} = 0$ now corresponds to $x_1 = x_2 = x_3 = x_4 = 0$.  The system has a single control input $u_1$ that is bounded $u_1 \in [-1.5, 1.5]$.  

\paragraph{MPC Scheme:} A first approach for solving the problem uses a nonlinear model-predictive control (MPC) scheme using a discretization of the system dynamics with a time step $\tau = 1$.
The time $t$ belongs to set $\{0, \tau, 2\tau,\ldots,N\tau = \T\}$ and:
\begin{equation}\label{eq:discretized-dynamics}
  \vx(t + \tau) =  \vx(t) + \tau f(\vx(t), \vu(t)) \,,
\end{equation}
Fixing the time horizon $\T = 30$, we use a simple cost function $J(\vx(0), \vu(0), \vu(\tau), \ldots, \vu(\T-\tau))$:
\[
 \left( \sum_{t \in \{0,\tau,...,\T-\tau\}} \left(||\vx(t)||_2^2 +
    ||\vu(t)||_2^2\right) \right) + N \ ||\vx(\T)||_2^2 \,.\] Here, we constrain $\vu(t) \in [-1.5, 1.5]$ for all $t$ and define $\vx(t+\tau)$ in terms of $\vx(t)$ using the discretization in Eq.~\eqref{eq:discretized-dynamics}.
Such a controller is implemented using a first/second order numerical gradient descent method to minimize the cost function~\cite{Nocedal+Wright/2006/Numerical}. The stabilization of the system was informally confirmed through hundreds of simulations from different initial states. However, the MPC scheme is expensive, requiring repeated solutions to (constrained) nonlinear optimization problems in real-time. Furthermore, in general, the closed loop lacks formal guarantees despite the ``high confidence" gained from numerous simulations.

\paragraph{Learning a Control Lyapunov Function:} Now, we introduce an approach which uses the MPC scheme as a \textsc{Demonstrator} and attempts to learn a control Lyapunov function. Then, a simpler control law is obtained from the CLF. The overall idea, depicted in Figure~\ref{fig:learning-framework}, is to pose queries to the \emph{offline} MPC at finitely many \emph{witness} states $\X_j : \{ \vx_1, \ldots, \vx_j \}$. Then, for each witness state $\vx_i$, the MPC yields the corresponding instantaneous control inputs $\vu_i$. The \textsc{learner} attempts to find a candidate function $V(\vx)$ that is positive for all $\vx_i \in \X_j$, which also decreases at each witness state $\vx_i$ through the control input $\vu_i$. This function $V$ is potentially a CLF function for the system.
This function is fed to the \textsc{verifier}, which checks whether $V(\vx)$ is indeed a CLF, or discovers a state $\vx_{j+1}$ which refutes $V$. This new state is added to the witness set and the process is iterated. The procedure described in this section synthesizes the control Lyapunov function $V(\vx)$ below:
 \begin{align*}
 V =& 1.22 x_2^2 + 0.31 x_2x_3 + 0.44 x_3^2 - 0.28 x_4x_2
    + 0.80 x_4x_3 + \\
    & 1.69 x_4^2 + 0.07 x_1x_2 - 0.66 x_1x_3
    - 1.85 x_4x_1 + 1.6 x_1^2\,.
 \end{align*}

Next, this function is used to design a simple associated control law that guarantees the stabilization of the model (Eq.~\eqref{eq:tora-dyn}). Figure~\ref{Fig:tora-example}(b) shows a closed loop trajectory for this control law vs. control law extracted by the MPC. The advantage of this law is that its calculation is \emph{much simpler}, appealing for control implementation. 

 \begin{figure}[t]
\begin{center}
\begin{tikzpicture}
\matrix[every node/.style={rectangle, draw=black, line width=1.5pt}, row sep=20pt, column sep=15pt]{ & \node[fill=myBlue!60](n0){\begin{tabular}{c} 
\textsc{Learner} \end{tabular}}; & \\
\node[fill=myGreen!60](n1){\begin{tabular}{c}
\textsc{Verifier} \end{tabular}}; & & \node[fill=myRed!60](n2){\begin{tabular}{c}
\textsc{Demonstrator} \end{tabular} }; \\
};
\path[->, line width=2pt] (n0) edge[bend right] node[left]{$V(\vx)?$} (n1)
(n1) edge[bend right] node[below]{\;\;\;\begin{tabular}{c}
Yes or \\
No($\vx_{j}$)\end{tabular}} (n0)
(n0) edge [bend left] node[above]{$\vx_j$} (n2)
(n2) edge [bend left] node[above]{$\vu_j$} (n0);
\draw (n0.north)+(0,0.3cm) node {$(\vx_1, \vu_1), \ldots, (\vx_{j-1}, \vu_{j-1})$};
\end{tikzpicture}
\end{center}
\caption{Overview of the learning framework for learning a control Lyapunov function.}\label{fig:learning-framework}
\end{figure}
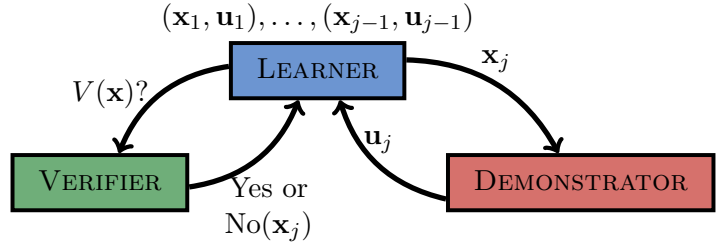
\subsection{Formal Learning Framework}
As mentioned earlier, the learning framework has three components: a demonstrator, a learner, and a verifier (see Figure~\ref{fig:learning-framework}).
The demonstrator inputs a state $\vx$ and returns a control input $\vu \in U$, that is an appropriate ``instantaneous'' feedback for $\vx$.  Formally, demonstrator is a function $\D :\ X \mapsto U$.

\begin{remark}[Demonstrator]
  The demonstrator is treated as a black box. This allows using a variety of approaches ranging from trajectory  optimizations~\cite{zhang2016learning}, human expert demonstrations~\cite{KHANSARIZADEH2014}, and sample-based methods~\cite{lavalle2000rapidly,kocsis2006bandit}, which can be probabilistically complete. 
  While the demonstrator is \emph{presumed} to address the specification, our method can work even if the demonstrator is faulty. Specifically, a faulty demonstrator in the worst case scenario may cause our method to terminate without finding a control certificate. However, if our approach finds a control certificate, it is guaranteed to be correct.
\end{remark}

Formally, in addition to a (switched or smooth) plant $\P$, the specification $\varphi$, the hypothesis space defined by $T$ over $\C$, we also assume that a \emph{black-box} demonstrator function $\D:\ X \mapsto U$ is provided as input.

For simplicity, we focus on finding CLFs for smooth feedback systems and we define the following terminologies regarding CLFs.

Instead of set of witness points $\X$, we define a set of observations $O$ as
\[
O :\ \{(\vx_1, \vu_1),\ldots,(\vx_j, \vu_j)\} \subset X \times U\,,
\]
where $\vu_i$ is the demonstrated feedback for state $\vx_i \neq \vzero$, i.e, $\vu_i:\ \D(\vx_i)$.

\begin{definition}[Observation Compatibility] \label{def:compatible-data}
    A function $V$ is said to be compatible with a set of 
    observations $O$ iff $V$ respects the CLF conditions 
    (Eq.~\eqref{eq:clf}) for every observation in $O$:
    \[
    V(\vzero) = 0 \ \wedge \ \bigwedge\limits_{(\vx_i, \vu_i) \in O_j}
\left(\begin{array}{c} V(\vx_i) > 0\ \land\ \nabla V \cdot f(\vx_i, \vu_i) < 0 \end{array}\right)\,.
    \]
\end{definition}

We note that not every CLF (satisfying the conditions in
  Eq.~\eqref{eq:clf}) will necessarily be compatible with a given
  observation set $O$.

\begin{definition}[ Demonstrator Compatibility ] \label{def:compatible-dem}
    A function $V$ is said to be compatible with a demonstrator $\D$ iff $V$ respects the CLF conditions (Eq.~\eqref{eq:clf}) for every observation that can be generated by the demonstrator:
    \[
    V(\vzero) = 0 \ \wedge \ \forall{\vx \neq \vzero}
\left(\begin{array}{c} V(\vx) > 0\ \land\ \nabla V \cdot f(\vx, \D(\vx)) < 0 \end{array}\right)\,.
    \]
    In other words, $V$ is a Lyapunov function for the closed-loop system $\Psi(\P, \D)$.
\end{definition}

The framework works iteratively and at each iteration $j$, the learner maintains a set of observations
\[
O_j :\ \{(\vx_1, \vu_1),\ldots,(\vx_{j}, \vu_{j})\} \subset X \times U\,.
\]
Given a template $V_\vc(\vx) :\ \vc^t \cdot \vg(\vx)$, corresponding to $O_j$, $\C_j \subseteq \C$ is defined as a set of candidate unknowns for function $V_\vc(\vx)$. Formally, $\C_j$ is a set of all $\vc$ s.t. $V_\vc$ is compatible with $O_j$:
\begin{equation} \label{eq:C_j_0}
\C_j :\ \left\{ \vc \in \C\left|
\begin{array}{c} V_\vc(\vzero) = 0 \ \land
\bigwedge\limits_{(\vx_i, \vu_i) \in O_j}
\left(\begin{array}{c} V_\vc(\vx_i) > 0\ \land\ \nabla V_\vc \cdot f(\vx_i, \vu_i) < 0 \end{array}\right)
\end{array}
\ \right.\right\}.
\end{equation}

\begin{figure}[!t]
\begin{center}
    \includegraphics[width=0.7\textwidth]{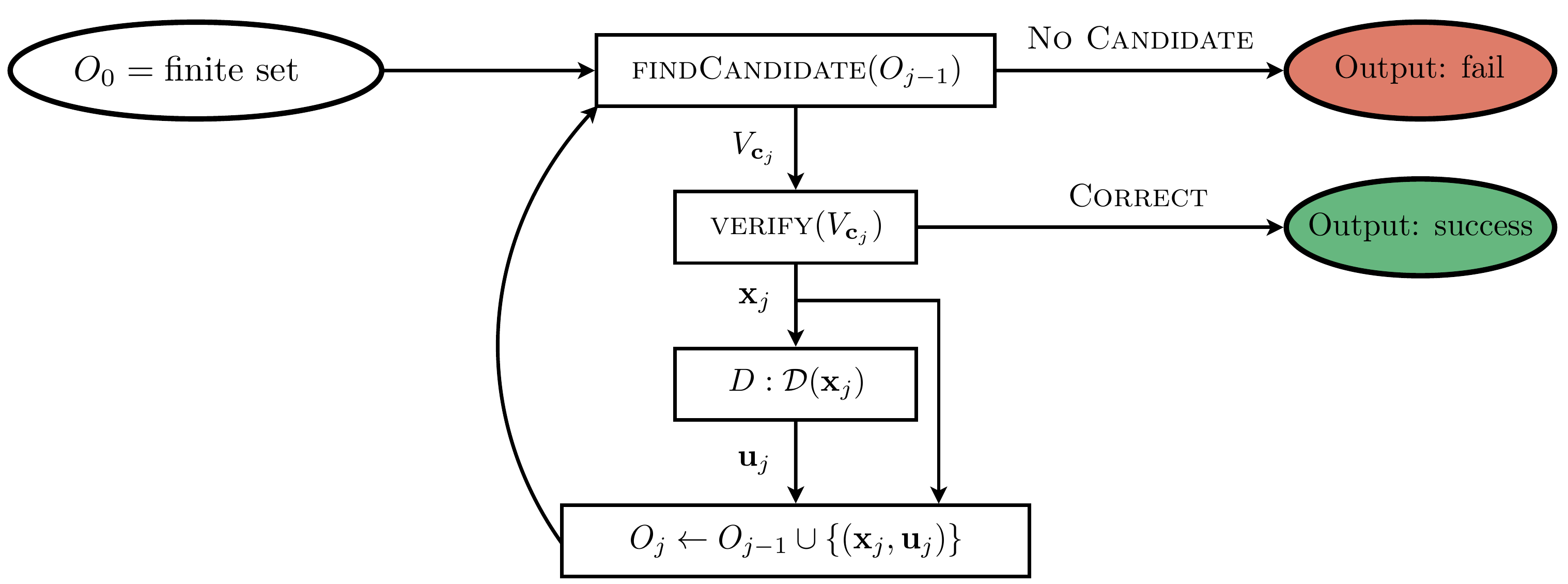}
\end{center}
\caption{Visualization of the learning framework.}\label{fig:framework} 
\end{figure}

The overall procedure is shown in Figure~\ref{fig:framework}. The procedure starts with an empty set $O_0=\emptyset$ and the corresponding set of compatible parameters $\C_0 :\ \{\vc \in \C \ | \ V_\vc(\vzero) = 0 \}$.
Each iteration involves the following steps:
\begin{compactenum}
    \item \textsc{findCandidate}: The learner checks if there exists a $V_\vc$ compatible with $O_{j-1}$
    \begin{compactenum}[(a)]
    \item If no such $\vc$ exists, the learner declares failure ($\C_{j-1} = \emptyset$),
    \item Otherwise, a candidate $\vc_j \in \C_{j-1}$ is chosen and the corresponding function $V_{\vc_j}(\vx):\ \vc_j^t \cdot \vg(\vx)$ is considered for verification.
    \end{compactenum}
    \item \textsc{verify}: The verifier oracle tests whether $V_{\vc_j}$ is a CLF (Eq.~\eqref{eq:clf})
    \begin{compactenum}[(a)]
    \item If yes, the process terminates successfully ($V_{\vc_j}$ is a CLF),
    \item Otherwise, the oracle provides a witness $\vx_j \neq \vzero$ for the negation
    of Eq.~\eqref{eq:clf}.
    \end{compactenum}
    \item \textsc{update}: Using the demonstrator $\vu_j :\ \D(\vx_j)$, a new observation $(\vx_j, \vu_j)$ is added to the training set:
    \begin{align}
    O_{j} &:\ O_{j-1} \cup \{(\vx_j, \vu_j)\} \\
        \C_{j} &:\ \C_{j-1} \cap \left\{ \vc \in \C\ |\ \begin{array}{c} V_{\vc}(\vx_{j}) > 0\ \land\ \nabla V_{\vc}\cdot f(\vx_{j}, \vu_{j}) < 0 \end{array} \right\} \label{eq:C_j+1}\,.
    \end{align}
\end{compactenum}

\begin{theorem}\label{thm:formal-learning-thm}
  The learning framework as described above has the following property:
\begin{compactenum}
\item $\vc_j \not\in \C_{j}$. I.e., the candidate found at the $j^{th}$ step is eliminated from further consideration
\item If the algorithm succeeds at iteration $j$, then the output function $V_{\vc_j}$ is a valid CLF for stabilization
\item The algorithm declares failure at iteration $j$ if and only if no linear combination of the basis functions is a CLF compatible with the demonstrator. 
\end{compactenum}
\end{theorem}
\begin{proof}
    1) Suppose that $\vc_j \in \C_{j}$. Then, $\vc_j$ satisfies the following conditions (Eq.~\eqref{eq:C_j+1}):
    \[
    V_{\vc_j}(\vx_{j}) > 0\ \land\ \nabla V_{\vc_j}\cdot f(\vx_{j}, \vu_{j}) < 0 \,.
    \]
    However, the verifier guarantees that $\vc_j$ is a counterexample for Eq.~\eqref{eq:clf}). I.e.,
    \[
        V_{\vc_j}(\vx_{j}) \leq 0\ \lor\ \nabla V_{\vc_j}\cdot f(\vx_{j}, \vu_{j}) \geq 0 \,,
    \]
    which is a contradiction. Therefore, $\vc_j \not\in \C_{j}$.
    
    2) The algorithm declares success if the verifier could not find a counterexample. In other words, $V_{\vc_j}$ satisfies conditions of Eq.~\eqref{eq:clf} and therefore a CLF.
    
    3) The algorithm declares failure if $\C_j = \emptyset$. On the other hand, by definition, $\C_j$ yields the set of all $\vc$ s.t. $V_\vc$ (which is a linear combination of basis functions) is compatible with the observations $O_j$. Therefore, $\C_j = \emptyset$ implies that no linear combination of the basis functions is compatible with the $O_j$ and therefore compatible with the demonstrator. \QED
\end{proof}

When compared to CEGIS, the verifier remains unchanged. However, the learning process differs.
Recall that $V_\vc :\ \vc^t \cdot \vg(\vx)$ and the learner needs to check if there exists a $\vc$ s.t. $V_\vc$ is compatible with the observation set $O$(Definition~\ref{def:compatible-data}). In other words, we wish to check
\[
(\exists \vc \in \C) \ V_\vc(\vzero) = 0 \wedge \bigwedge_{(\vx_i, \vu_i) \in O} 
\left(\begin{array}{c} V_\vc(\vx_i) > 0\ \land\ \nabla V_\vc \cdot f(\vx_i, \vu_i) < 0 \end{array}\right)\,.
\]
The (initial) space of all candidates $\C$ is assumed to be a hyper-rectangular box, and therefore a polytope. Let $\overline{\C_j}$ represent the topological closure of the set $\C_j$ obtained at the $j^{th}$ iteration (see Eq.~\eqref{eq:C_j_0}).
\begin{lemma}\label{lemma:cj-convex}
For each $j \geq 0$, $\overline{\C_j}$ is a polytope.
\end{lemma}
\begin{proof}
    We prove by induction. Initially, $\C$ is an hyper-rectangular box. 
    Also, $\C_0 :\ \C \cap H_0$, where
    \[
    H_0 = \{\vc \ | \ V_\vc(\vzero) = \sum_{i=1}^r c_i g_i(\vzero) = 0 \} \,.
    \] 
As $V_\vc$ is linear in $\vc$, $H_0 :\ \{\vc \ | \ \va_0^t . \vc = b_0\}$ is a hyperplane, where $\va_0$ and $b_0$ depend on the values of, $g_k(\vzero)$ ($k = 1,\ldots,r$). $\C_0$ would be the intersection of a polytope and a hyperplane, which is a polytope.
Now, assume $\overline{\C_{j-1}}$ is a polytope. Recall that $\C_{j}$ is defined as $\C_{j} :\ \C_{j-1} \cap P_j$ (Eq.~\eqref{eq:C_j+1}), where
    \[
    P_{j}:\  \left\{ \vc \ \left|\ \begin{array}{c} \sum_{i=1}^{r} (c_i \ g_i(\vx_{j})) > 0\ \land\ \\ \sum_{i=1}^{r} (c_i \ \nabla g_i(\vx_j) \cdot f(\vx_{j}, \vu_{j})) < 0 \end{array} \right.\right\} \,.
    \]
    Notice that $f(\vx_{j}, \vu_{j})$ and $g_i(\vx_i)$ are constants and
    \begin{align*}
    P_j :\ &H_{j1} \cap H_{j2} \ \  \begin{cases}
    H_{j1} :\ &\{\vc \ | \ \va_{j1}^t . \vc > b_{j1}\}
    = \{\vc | \sum_{i=1}^{r} (c_i \ g_i(\vx_{j})) > 0\} \\
    H_{j2} :\ &\{\vc \ | \ \va_{j2}^t . \vc > b_{j2}\} =
    \{\vc | 
    \sum_{i=1}^{r} (c_i \ \nabla g_i(\vx_j) \cdot f(\vx_{j}, \vu_{j})) < 0 \}\,.
    \end{cases}
    \end{align*}
    Therefore, $\overline{\C_{j}}$ is the intersection of a polytope ($\overline{\C_{j-1}}$) and two half-spaces ($P_j$) which yields another polytope. \QED
\end{proof}

The learner should sample a point $\vc_j \in \C_{j-1}$ at $j^{th}$ iteration, which is equivalent to checking the emptiness of a polytope with some strict inequalities. This is solved using a slight modification of the simplex method, using infinitesimals for strict inequalities or using interior point methods~\cite{Vanderbei/2004/Linear}. We will now demonstrate that by choosing $\vc_j$ carefully, we can guarantee the polynomial time termination of our learning framework.

\subsection{Termination}
Jha et al.~\cite{Jha2017} prove bounds on the number of queries (iterations) for discrete hypothesis space using results on exact concept learning in discrete spaces~\cite{GOLDMAN199520}. Here, we prove bounds on the number of queries needed to learn in a continuous linear hypothesis space using results from convex optimization.

Recall that in the framework, the learner provides a candidate, and the verifier refutes the candidate by a counterexample, and the demonstrator generates a new observation.  The following lemma relates the sample $\vc_j \in \C_{j-1}$ at the $j^{th}$ iteration and the set $\C_{j}$ in the subsequent iteration.

\begin{lemma}\label{lemma:cj-half-space}
  There exists a half-space  $H_{j}:\ \va^t \vc \geq b$ such that (a) $\vc_j$ lies on the boundary of hyperplane $H_{j}$, and (b) $\C_{j} \subseteq \C_{j-1} \cap H_{j}$.
\end{lemma}
\begin{proof}
    Recall that we have $\vc_j \in \C_{j-1}$ but $\vc_j \not\in \C_{j}$ by Theorem~\ref{thm:formal-learning-thm}. Let $\hat{H}_j:\ \va^t \vc = \hat{b}$ be a separating hyperplane between the (convex) set $\C_{j}$ and the point $\vc_j$, such that $\C_{j} \subseteq \{ \vc\ |\ \va^t \vc \geq \hat{b}\}$. By setting the offset $b:\ \va^t \vc_j$, we note that $b \leq \hat{b}$. Therefore, by defining $H_{j}$ as $\va^t \vc \geq b$, we obtain the required half-space that satisfies conditions (a) and (b).
\QED
\end{proof}

While sampling a point from $\C_j$ is solved efficiently by solving a linear programming problem, Lemma.~\ref{lemma:cj-half-space} suggests that the choice of $\vc_j$ governs the convergence of the algorithm. Figure~\ref{fig:learning-iteration} demonstrates the importance of this choice by showing candidate $\vc_j$, $\C_{j}$, and $\C_{j+1}$.

\begin{figure}[t]
\begin{center}
\begin{tikzpicture}[scale=0.8]
\draw[draw=black, line width = 1.5pt, fill=myGreen!40](0,0) -- (2,0.5) -- (2,2.5) -- (0,4) -- (-1,1) -- cycle;
\draw[draw=black, line width=1.5pt, fill=myBlue!40] (0,0) -- (2,0.5) -- (2,1.1) -- (1.2,1.5) -- ( -0.35,0.35) -- cycle;
\draw[draw=black, line width=1.5pt, dashed](-0.9,-0.05) -- (3.0, 2.83);
\draw[draw=black, line width=1.5pt, dashed](2.5,0.86) -- (-1, 2.56);
\draw[draw=black, line width=1.5pt, dashed](-1.3,0.9) -- (3.0, 2.2);
\draw[fill=myRed] (0.5,1.45) circle (0.1);
\node at (0.22,1.7) {$\vc_j$};
\node at (0.6,2.8) {$\C_{j-1}$};
\node at (0.8,0.6) {$\C_{j}$};
\node at (-1.2,2.8){$H_{j1}$};
\node at (3.5,3.2){$H_{j2}$};
\node at (-1.7,0.8){$H_{j}$};
\end{tikzpicture}
\\ Original candidate region $\C_{j-1}$ (green) at the start of the
  $j^{th}$ iteration, the candidate $\vc_j$, and the new region
  $\C_{j}$ (blue region).
\end{center}
\caption{Effect of candidate selection on feasible parameter sets.}\label{fig:learning-iteration}
\end{figure}

For a faster termination, we wish to remove a ``large portion" of $\C_{j-1}$ to obtain a ``smaller" $\C_{j}$.  There are two important factors affecting this: (i) counterexample $\vx_j$ selection and (ii) candidate $\vc_j$ selection.  Counterexample $\vx_j$ would affect $\vu_j:\ \D(\vx_j)$, $g(\vx_j)$, and $f(\vx_j, \vu_j)$ and therefore affects the separating hyperplane $H_{j}$. On the other hand, candidate  $\vc_j \not\in \C_{j}$. We have already discussed the counterexample selection in Section~\ref{sec:cegis}. In the following, we focus on different techniques to generate a candidate $\vc_j \in \C_{j-1}$. 

The goal is to find a $\vc_j$ s.t. 
\begin{equation} \label{eq:volume-reduction}
    \Vol(\C_{j}) \leq \alpha \Vol(\C_{j-1}) \,,
\end{equation}
for each iteration $j$ and a fixed constant $0 \leq \alpha < 1$, independent of the hyperplane $H_{j}$. Here, $\Vol(\C_j)$ represents the volume of the (closure) of the set $\C_j$. Since the closure of $\C_j$ is contained in $\C$ which is compact, this volume will always be finite. Note that if we can guarantee Eq.~\eqref{eq:volume-reduction}, it immediately follows that $\Vol(\C_j) \leq \alpha^j \Vol(\C_0)$. This implies that the volume of the remaining candidates ``vanishes" rapidly.

\begin{remark} By referring to $\Vol(\C_j)$, we are implicitly assuming that $\C_j$ is not embedded inside a subspace of $\reals^r$, i.e, it is full-dimensional. However, this assumption is not strictly true. Specifically, $\C_0 :\ \C \cap H_0$, where $H_0$ is a hyperplane. Thus, strictly speaking, the volume of $\C_0$ in $\reals^r$ is $0$. This issue is easily addressed by first factoring out the linearity space of $\C_j$, i.e., the affine hull of $\C_j$. This is performed by using the equality constraints that describe the affine hull to eliminate variables from $\C_j$. Subsequently, $\C_j$ can be treated as a full dimensional polytope in $\reals^{r-d_j}$, wherein $d_j$ is the dimension of its linearity space.
Furthermore, since $\C_{j} \subseteq \C_{j-1}$, we can continue to express $\C_{j} $ inside $\reals^{r-d_j}$ using the same basis vectors as $\C_{j-1}$.  A further complication arises if $\C_{j}$ is embedded inside a smaller subspace. We do not treat this case in our analysis. However, note that this can happen for at most $r$ iterations and thus, does not pose a problem for the termination analysis.
\end{remark}

Intuitively, it is clear from Figure~\ref{fig:learning-iteration} that a candidate at the \emph{center} of $\C_j$ would be a good one. We now relate the choice of $\vc_j$ to an appropriate definition of center, so that Eq.~\eqref{eq:volume-reduction} is satisfied.

\paragraph{Center of Maximum Volume Ellipsoid:}
Maximum volume ellipsoid (MVE) inscribed inside a polytope is unique with many useful characteristics. Let $\E_j$ be the MVE inscribed inside $\C_j$ (Figure~\ref{fig:ellipsoid}).

\begin{theorem}[\cite{tarasov1988method,khachiyan1990inequality}] \label{thm:mve-2}
Let $\vc_j$ be chosen as the center of $\E_{j-1}$. Then, $\Vol(\E_{j}) \leq \left(\frac{8}{9}\right) \Vol\left(\E_{j-1}\right)$. 
\end{theorem}

Recall, here that $r$ is the number of basis functions such that $\C_j \subseteq \reals^r$. This leads us to a scheme that guarantees termination of the overall procedure within finitely many steps under some assumptions. The idea is simple. Select the center of the MVE inscribed in $\C_{j-1}$ at $j^{th}$ iteration (Figure~\ref{fig:ellipsoid}).
We consider some robustness for the candidate.
\begin{definition}[Robust Compatibility]
A candidate $\vc$ is $\delta$-robust for $\delta > 0$ w.r.t. observations (the demonstrator), iff for each $\hat{\vc} \in \B_{\delta}(\vc)$, $V_{\hat{\vc}}(\vx):\ \hat{\vc}^t \cdot \vg(\vx)$ is compatible with observations (the demonstrator) as well. 
\end{definition}

Let $\C \subseteq (-\Delta, \Delta)^r$ for $\Delta > 0$. Following the robustness assumption, it is sufficient to terminate the procedure whenever
\begin{equation}\label{eq:termin-cond-2}
\Vol(\E_j) < \gamma \delta^r\,,
\end{equation}
where $\gamma$ is the volume of $r$-ball with radius $1$. This additional termination condition is easily justified when one considers the precision limits of floating point numbers and sets of small volumes. Clearly, as the $\Vol(\C_j)$ decrease exponentially (as $j$ increases), each point inside the set will be quite close to one that is outside, requiring high precision arithmetic to represent and sample from the sets $\C_j$. 
  
  \begin{theorem} \label{thm:termination2}
If at each step $\vc_j$ is chosen as the center of $\E_{j-1}$, the learning loop condition defined by Eq.~\eqref{eq:termin-cond-2} is violated in at most $\frac{r (\log(\Delta) - \log(\delta))}{- \log\left(\frac{8}{9}\right) } = O(r)$ iterations.
\end{theorem}
\begin{proof}
    Initially, $\B_{\Delta}(\vzero)$ is the MVE strictly inside box $[-\Delta, \Delta]^r$ and therefore, $\Vol(\E_0) < \gamma \Delta^r$. Then by Theorem~\ref{thm:mve-2}
    \begin{align*}
    &\Vol(\E_j) \leq (\frac{8}{9})^j \ \Vol(\E_0) < (\frac{8}{9})^j \gamma \Delta^r \\
    \implies & \log(\Vol(\E_j)) - \log(\gamma \Delta^r) < j \  \log(\frac{8}{9}) \,.
    \end{align*}
    After $k = \frac{r(\log(\Delta)-\log(\delta))}{-\log(\frac{8}{9})}$ iterations:
    \begin{equation*}
        \log(\Vol(\E_k)) - \log(\gamma \Delta^r) < \frac{r(\log(\Delta)-\log(\delta))}{-\log(\frac{8}{9})} \  \log(\frac{8}{9}) \,,
    \end{equation*}
    and
    \begin{align*}
        \implies & \log(\Vol(\E_k)) - \log(\gamma \Delta^r) < r(\log(\delta) - \log(\Delta)) \\
        \implies & \log(\Vol(\E_k)) - \log(\gamma \Delta^r) < \log(\gamma \delta^r) - \log(\gamma \Delta^r) \\
        \implies & \log(\Vol(\E_k)) < \log(\gamma \delta^r)\,.
    \end{align*}    
    It is concluded that $\Vol(\E_k) < \gamma \delta^r$, which is the termination condition. And asymptotically, the maximum number of iterations would be $O(r)$. \QED
\end{proof}

  The volume of an ellipsoid is effectively computable, and thus, such termination condition can be checked quickly. The MVE itself can be computed by solving a convex optimization problem\cite{tarasov1988method,vandenberghe1998determinant}.

\begin{figure}[t]
\begin{center}
    \includegraphics[width=0.4\textwidth]{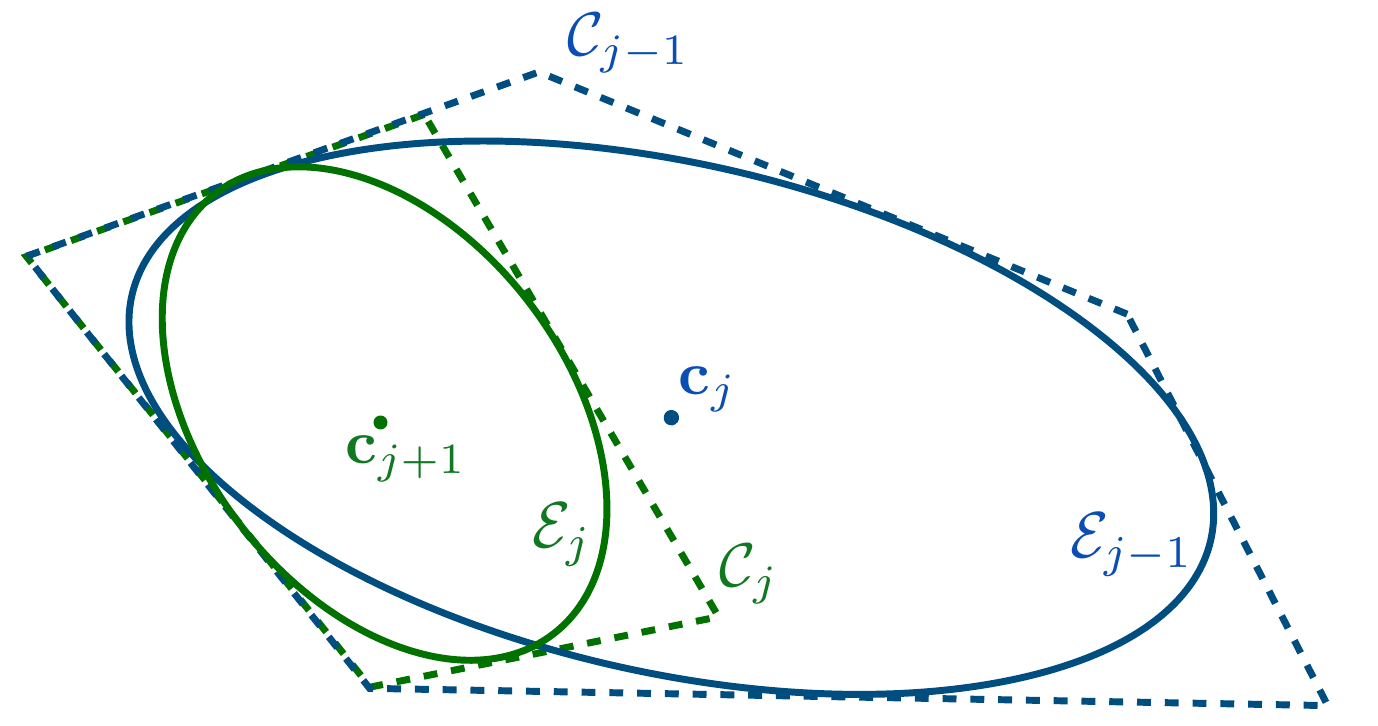}
\\
Candidate region $\C_{j-1}$ ($\C_{j}$) is shown in blue (green) polygon. The maximum volume ellipsoid $E_{j-1}$ ($E_{j}$) is inscribed in $\C_{j-1}$ ($\C_{j}$) and its center is the candidate $\vc_j$ ($\vc_{j+1}$).
\end{center}
\caption{Effect of selecting center of MVE.}
\label{fig:ellipsoid}
\end{figure}

\begin{theorem}\label{thm:clf-or-no-robust-solution}
    The learning framework either finds a control Lyapunov function or proves that no linear combination of the basis functions would yield a function with robust compatibility with the demonstrator.
\end{theorem}
\begin{proof}
    By Theorem~\ref{thm:formal-learning-thm}, if verifier certifies correctness of a solution $V$, then $V$ is a CLF. Assume that the framework terminates after $k$ iterations and no solution is found. Then, by Theorem~\ref{thm:termination2}, $\Vol(\E_k) < \gamma \delta^r$. This means that a ball with radius $\delta$ would not fit in $\C_k$ as $\E_k$ is the MVE inscribed inside $\C_k$.
    In other words
    \[
    (\forall \vc \in \C_k) \ (\exists \hat{\vc} \in \B_{\delta}(\vc)) \ \hat{\vc} \not\in \C_k\,.
    \]
    On the other hand, for all $\vc \not\in \C_k$, $V_\vc$ is not compatible with the observations $O_j$. Therefore, even if there is a CLF $V_\vc$ s.t. $\vc \in \C_k$, the CLF is not robust in its compatibility with the demonstrator. \QED
\end{proof}

\paragraph{Other Definitions for Center of Polytope:}
Besides the center of MVE inscribed inside a polytope, there are other notions for defining the center of a polytope. These include the center of gravity and  Chebyshev center. Center of gravity provides the following inequality~~\cite{bland1981ellipsoid}
\[ \Vol\left(\C_{j}\right) \leq \left(1-\frac{1}{e}\right) \Vol\left(\C_{j-1}\right)
< 0.64 \ \Vol(\C_{j-1}) \,,\]
meaning that the volume of candidate set is reduced by at least 36\% at each iteration. Unfortunately, computing center of gravity is costly. Chebyshev center~\cite{elzinga1975central} of a polytope is the center of the largest Euclidean ball that lies inside the polytope. Finding a Chebyshev center for a polytope is equivalent to solving a linear program, and while it yields a good heuristic, it would not provide an inequality in the form of Eq.~\eqref{eq:volume-reduction}.

There are also notions for defining the center for a set of constraints, including analytic center, and volumetric center. Assuming $I :\ \{ \vc \ | \bigwedge_i \va_i^t . \vc < b_i \}$, then analytic center for $\bigwedge_i \va_i^t . \vc < b_i$ is defined as
\[
ac(\bigwedge_i \va_i^t . \vc < b_i) = argmin_{\vc} - \sum_i \log (b_i - \va_i^t . \vc) \,.
\]
Notice that infinitely many inequalities can represent $I$ and any point inside $I$ can be an analytic center depending on the inequalities. Atkinson et al.~\cite{atkinson1995cutting} and Vaidya~\cite{vaidya1996new} provide candidate generation techniques based on these centers, along with appropriate termination conditions and convergence analysis.

\subsection{Other Control Problems}
While we discussed the problem for smooth feedback systems, its extension for switched feedback systems is straightforward. Moreover, the framework is directly applicable for finding control exponential-barrier function~\cite{kong2013exponential}, control Lyapunov-barrier functions, and control Lyapunov fixed-barriers functions (for uninitialized RWS) as well. 

To extend the method for finding control barrier functions (or control funnel functions), we need a slightly different demonstrator. Recall that for a control barrier certificate $B_\vc(\vx)$ the following conditions must hold:
\begin{align*}
(\forall \vx \in I) & \ B_\vc(\vx) < 0\\
(\forall \vx \in \partial S) & \ B_\vc(\vx) > 0 \\
(\forall \vx \in S \setminus \inter{I}) & \left( \begin{array}{c}
(\exists \vu \in U) \ \nabla B_\vc \cdot f(\vx, \vu) + \lambda^* B_\vc(\vx) < 0 \\ \lor \\
(\exists \vu \in U) \ \nabla B_\vc \cdot f(\vx, \vu) - \lambda^* B_\vc(\vx) < 0\,.
\end{array} \right) \,.
\end{align*}
If $\lambda^* = 0$, the previously discussed method can address the problem. Otherwise, because of the disjunction on the third condition, $\overline{\C_j}$ would not be a polytope if we use the demonstrator discussed previously.

Considering the third condition, if $B_\vc(\vx) < 0$, then it is easier to satisfy 
\[(\exists \vu \in U) \ \nabla B_\vc \cdot f(\vx, \vu) + \lambda B_\vc(\vx) < 0\,,\]
and in case $B_\vc(\vx) > 0$, the other condition is weaker. Also recall that $B_{\vc}^{< 0}$ contains states for which safety is guaranteed. We wish the demonstrator to inform the learner which condition is weaker (better to use) for a specific $\vx_j$. A simple solution is to assume a demonstrator that for a given state $\vx_j$ outputs: (i) a proper instantaneous control input $\vu_j$, and (ii) $b_j \in \bools$ indicating whether safety can be guaranteed by the demonstrator policy, starting from $\vx_j$. If safety is guaranteed using the demonstrator, we prefer to search a CBF $B_\vc$ s.t. $\vx_j \in B_\vc^{< 0}$ (safety would be guaranteed using the CBF-based controller) which hints the learner to use the first condition. Such a demonstrator would yield an observation set of the form:
\[
O_{j} :\ \{(\vx_1, \vu_1, b_1),\ldots(\vx_{j}, \vu_{j}, b_{j})\}\,.
\]
Then, $C_j$ is defined as
\[
C_j :\ \left\{ \vc \in \C \ | \ \bigwedge_{(\vx_i, \vu_i, b_i) \in O_j} 
\left( \begin{array}{l}
    b_i \implies \nabla B_\vc \cdot f(\vx_i, \vu_i) + B_\vc(\vx_i) < 0 \land \\
    \lnot b_i \implies \nabla B_\vc \cdot f(\vx_i, \vu_i) - B_\vc(\vx_i) < 0
\end{array} \right)
\right\}\,.
\]
Now, one could show $\overline{C_j}$ is a polytope and other parts of the framework remain unchanged.

%% file: synt/sdp.tex
\section{SDP Relaxation}
Recall that if the regions of interest are semi-algebraic sets, the dynamics and chosen bases are polynomials in $\vx$, then the verification problem for switched systems reduces to checking if a given semi-algebraic set defined by polynomial inequalities is empty. Therefore, the verification problem decidable with high complexity (NP-hard)~\cite{Basu+Pollock+Roy/03/Algorithms}.
However, for scalability, we consent to a relaxation using SDP solvers. We now present a relaxation using semidefinite programming (SDP) solvers.

Abandoning numerical SMT solvers, we no longer require the regions of interest to be compact. Moreover, we do not assume $U$ is finite anymore. However, we require nonlinear functions to be polynomials. For simplicity, we only consider CLF verification for smooth feedback systems.

The verifier checks the CLF conditions in Eq.~\eqref{eq:clf} for a candidate $V_{\vc_j}(\vx):\ \vc_j^t \cdot \vg(\vx)$. Since the CLF is generated by the learner, it is guaranteed that $V_{\vc_j}(\vzero) = 0$ (Eq.~\eqref{eq:C_j_0}). 
Accordingly, verification is split into two separate checks:

\noindent\textbf{(A)} Check if $V_{\vc_j}(\vx)$ is a positive polynomial for
  $\vx \neq \vzero$, or equivalently
  \begin{equation}\label{eq:positivity-cond}
 (\exists\ \vx \neq \vzero)\ V_{\vc_j}(\vx) \leq 0 \,.
\end{equation}
\noindent\textbf{(B)} Check if the Lie derivative of $V_{\vc_j}$ can be made negative for each $\vx \neq \vzero$ by a choice $\vu \in U$: 
\begin{equation}\label{eq:decrease-cond-init}
 (\exists \vx \neq \vzero ) \ (\forall \vu \in U)\ (\nabla V_{\vc_j}) \cdot f(\vx, \vu) \geq 0 \,.
\end{equation}
This problem \emph{seems} harder  due to the presence of a \emph{quantifier alternation}.
\begin{lemma}[Farkas Lemma for LP]\label{lemma:farkas}
Exactly one of the following holds:
\begin{compactitem}
    \item $(\exists \vx) \ A \vx \leq \vb$
    \item $(\exists \vlam) \ \vlam \geq \vzero \,, \vb^t\vlam < 0 \,, A^t \vlam = \vzero$.
\end{compactitem}
\end{lemma}
\begin{lemma} \label{lem:control-dual}
Eq.~\eqref{eq:decrease-cond-init} holds for some $\vx \neq \vzero$ iff
\begin{equation}\label{eq:decr-condition}
\begin{array}{ll}
(\exists\ \vx \neq \vzero,\vlam) \ \vlam\geq \vzero \, , \, \vlam^t \vb \geq -\nabla V_{\vc_j}.f_0(\vx) \,,\, A_i^t \vlam=\nabla V_{\vc_j}.f_i(\vx) (i \in \{1 \ldots m\}).
\end{array}
\end{equation}
\end{lemma}
\begin{proof}
    Suppose Eq.~\eqref{eq:decrease-cond-init} holds. Then, for the given $V$, there
    exists a $\vx \neq \vzero$ s.t.
    \begin{equation} \label{eq:u-cond-before-farkas-lemma}
    (\forall \vu \in U) \ \nabla V \cdot f(\vx,\vu) = \left( \begin{array}{c} \nabla V \cdot f_0(\vx) + \sum\limits_{i=1}^{m} \nabla V \cdot f_i(\vx) u_i \end{array} \right)\hspace{-0.1cm} \geq0,
    \end{equation}
    which is equivalent to:
    \[
        (\not\exists \vu) A \vu \geq \vb \land \nabla V \cdot f_0(\vx) + \sum_{i=1}^{m} \nabla V \cdot f_i(\vx) u_i < 0 \,.
    \]
     This yields a set of linear inequalities (w.r.t. $\vu$). Using Lemma~\ref{lemma:farkas}, this is equivalent to
    \begin{align*}
    (\exists \vlam \geq 0) \ A_i^t \vlam = \nabla V \cdot f_i(\vx) (i \in \{1...m\}) \, , \,
    \vlam^t \vb \geq -\nabla V \cdot f_0(\vx).
    \end{align*}
    Thus, for a given $V$, Eq.~\eqref{eq:decrease-cond-init} is equivalent
    to Eq.~\eqref{eq:decr-condition}. \QED
\end{proof}

Let $\vw^t:\ [\vx^t, \vlam^t] \in \reals^{n'}$ collect the state variables $\vx$ and the dual variables $\vlam$ involved in the conditions stated in Eq.~\eqref{eq:decr-condition}. The core idea behind the SDP relaxation is to consider a vector collecting all monomials of degree up to $\degr$: $\vz:\ [ 1 \  w_1 \  w_2\ \ldots \ w_{n'}^{\degr} ]^t$, wherein $\degr$ is chosen to be at least half of the maximum degree in $\vx$ among all monomials in $g_j(\vx)$ and $\nabla g_j \cdot f_i(\vx)$:$ \degr \geq \frac{1}{2} \max\left( \bigcup_{j} \left( \{ \mbox{deg}(g_j) \} \cup \{ \bigcup_{i}\mbox{deg}(\nabla g_j \cdot f_i ) \} \right) \right)$.
Let us define $Z(\vw):\ \vz \vz^t$, which is a symmetric matrix
  of monomial terms of degree at most $2\degr$. Each polynomial of degree up to $2\degr$ may now be written as a trace inner product $p(\vx, \vlam):\ \tupleof{ P, Z(\vw)} = \mathsf{trace}( P Z(\vw) )$, wherein the symmetric matrix $P$ has real-valued entries that define the coefficients in $p$ corresponding to the various monomials. Although, $Z$ is a function of $\vx$ and $\vlam$, we will write $Z(\vx)$ as a function of just $\vx$ to denote the matrix $Z([\vx^t, \vzero^t]^t)$ (i.e., set $\vlam = \vzero$).

The constraint in Eq.~\eqref{eq:positivity-cond} is equivalent to solving the following optimization problem over $\vx$
\begin{equation}\label{eq:positivity-cond-relax}
\begin{array}{ll}
 \mathsf{max}_{\vx} \tupleof{I,Z(\vx)} 
 \mathsf{ s.t. } \tupleof{\mathcal{V}_{\vc_j}, Z(\vx)} \leq 0\,, \\
\end{array}
\end{equation}
wherein $I$ is the identity matrix, and $V_{\vc_j}(\vx)$ is written in the inner product form as $\tupleof{\mathcal{V}_{\vc_j}, Z(\vx)}$.  Let $\tupleof{\Lambda_k, Z(\vw)}$ represent the variable $\lambda_k$. $\vlam$ is represented as vector $\Lambda(Z(\vw))$, wherein the $k^{th}$ element is $\tupleof{\Lambda_k, Z(\vw)}$.  Then, the conditions in Eq.~\eqref{eq:decr-condition} are written as
\begin{equation}\label{eq:decr-cond-relax}
 \begin{array}{l}
 \mathsf{max}_{\vw} \tupleof{I,Z(\vw)} 
 \mathsf{s.t.}  \\
  \tupleof{F_{{\vc_j},i}, Z(\vw)} = A_i^t 
 \Lambda(Z(\vw))\ i \in \{1,\ldots, m\} \, , \, 
\tupleof{-F_{{\vc_j},0}, Z(\vw)} \leq \vb^t \Lambda(Z(\vw)) \, , \, 
\Lambda(Z(\vw)) \geq 0 \,,
\end{array}
\end{equation}
wherein the components $\nabla V_{\vc_j} \cdot f_i(\vx)$ defining the Lie derivatives of $V_{\vc_j}$ are now written in terms of $Z(\vw)$ as $\tupleof{F_{{\vc_j},i},Z(\vw)}$. Notice that $Z(\vzero)$ is a square matrix where the first element ($Z(\vzero)_{1,1}$) is $1$ and the rest of the entries are zero. Let $Z_0 = Z(\vzero)$ . Then $\tupleof{I, Z_0} = 1$, and $(\forall \vw) \ Z(\vw) \succeq Z_0$.

The SDP relaxation is used to solve these problems approximately and $D$ defines the degree of relaxation~\cite{henrion2009gloptipoly}. The relaxation treats $Z(\vw)$ as a fresh matrix variable $Z$ that is no longer a function of $\vw$. The constraint $Z \succeq Z_0$ is added.
However, $Z(\vw):\ \vz \vz^t$ is a rank one matrix, and ideally, $Z$ should be constrained to be rank one as well. However, such a constraint is non-convex, and therefore, will be dropped.
Also, constraints in Eqs.~\eqref{eq:positivity-cond-relax} and~\eqref{eq:decr-cond-relax} are added as support constraints.
In other words, for a constraint $h(\vw) \geq 0$, the relaxation treats $Z(\vw)h(\vw)$ as a fresh matrix variable $Z_h \succeq Z_0$ (cf.~\cite{lasserre2009moments,lasserre2010positivity,henrion2009gloptipoly} for details).
Both optimization problems (Eqs.~\eqref{eq:positivity-cond-relax} and~\eqref{eq:decr-cond-relax}) are feasible by setting $Z$ to be $Z_0$. Furthermore, if the optimal solution for each problem is $1$ in the SDP relaxation, then we will conclude that the given candidate is a CLF. Unfortunately, the converse is not necessarily true: the relaxation may fail to recognize that a given candidate is, in fact, a CLF.

\begin{lemma}\label{lem:non-zero-sol}
Whenever the relaxed optimization problems in Eqs.~\eqref{eq:positivity-cond-relax} and~\eqref{eq:decr-cond-relax} yield $1$ as a solution, then the given candidate $V_{\vc_j}(\vx)$ is in fact a CLF. 
\end{lemma}
\begin{proof}
    Suppose that $V_{\vc_j}$ is not a CLF but both optimization problems yield an optimal value of $1$. Then, one of Eq.~\eqref{eq:positivity-cond} or Eq.~\eqref{eq:decrease-cond-init} is satisfied. 
    I.e., $(\exists \vx^* \neq \vzero, \vlam^* \geq \vzero)$ s.t. $V_{\vc_j}(\vx^*) \leq 0$ or $A_i^t \vlam^*=\nabla V_{\vc_j}.f_i(\vx^*) \ (i \in \{1 \ldots m\}) , \ \vlam^{*t} \vb \geq - \nabla V_{\vc_j}.f_0(\vx^*)$.
        Let $\vw^{*t} = [\vx^{*t}, \vlam^{*t}]$ and therefore, $Z(\vw^*) \succeq Z_0$ is a solution for Eq.~\eqref{eq:positivity-cond-relax} or Eq.~\eqref{eq:decr-cond-relax}.  Let $Z' = Z(\vw^*) - Z_0$. 
        As $\vw^* \neq \vzero$, $Z'$ has a non-zero diagonal element, and since $Z' \succeq 0$, we may also conclude that at least one of the eigenvalues of $Z'$ must be positive. Therefore, $\tupleof{I, Z'} > 0$ as the trace of $Z'$ is the sum of eigenvalues of $Z'$. Thus, $\tupleof{I, Z(\vw)} > \tupleof{I, Z_0} = 1$. Therefore, the optimal solution of at least one of the two problems has to be greater than one. This contradicts our original assumption. \QED
\end{proof}

However, the converse is not true. It is possible for $Z \succeq Z_0$ to be optimal for one of the relaxed conditions, but $Z \not= Z(\vw)$ for any $\vw$.
This happens because (as mentioned earlier) the relaxation drops two key constraints to convexify the conditions: (i) $Z$ has to be a rank one matrix written as $Z:\ \vz \vz^t$, and (ii) there is a $\vw$ such that $\vz$ is the vector of monomials corresponding to $\vw$.

\subsection{Lifting the Counterexamples}
Thus far, we have observed that the relaxed optimization problems (Eqs.~\eqref{eq:positivity-cond-relax} and~\eqref{eq:decr-cond-relax}) yield matrices $Z$ as counterexamples, rather than vectors $\vx$. Furthermore, given a solution $Z$, there is no way for us to extract a corresponding $\vx$ for reasons mentioned above. We solve this issue by ``lifting'' our entire learning loop to work with  observations of the form
\[ O_j:\ \{ (Z_1, \vu_1),\ldots,(Z_{j}, \vu_{j})\} \,,\]
effectively replacing states $\vx_i$ by matrices $Z_i$.

Also, each basis function $g_k(\vx)$ in $\vg$ is now written instead as $\tupleof{G_k, Z}$.
The candidates are therefore, $\sum\limits_{k=1}^r  c_k \tupleof{ G_k, Z}$. Likewise, we write the components of its Lie derivative $\nabla g_k \cdot f_i$ in terms of $Z$ ($\tupleof{G_{ki}, Z}$).
Therefore
\begin{align}\label{eq:relaxed-template}
    \mathcal{V}_\vc = \sum\limits_{k=1}^r c_{k} G_k \ , \ F_{\vc,i} = \sum\limits_{k=1}^r c_{k} G_{ki}\,.
\end{align}
 
\begin{definition}[Relaxed CLF]\label{def:relaxed-CLF}
    A polynomial function $V_\vc(\vx) = \sum\limits_{k=1}^r c_k g_k(\vx)$ defined by $\vc$ is a $D$-relaxed CLF iff $\tupleof{\mathcal{V}_\vc, Z_0} = 0$, and for all
    $Z \not= Z_0$
    \begin{equation}\label{eq:relaxed-clf}
    \begin{array}{l}  \tupleof{\mathcal{V}_\vc, Z} > 0 \ \land
        (\exists \vu \in U) \ \tupleof{F_{\vc,0}, Z} + \sum\limits_{i=1}^m \tupleof{F_{\vc,i}, Z} u_i < 0\,.
    \end{array}
        \end{equation}
\end{definition}

\begin{theorem}\label{thm:relaxed-CLF-vs-CLF}
    A relaxed CLF is a CLF.
\end{theorem}
\begin{proof}
    Suppose that $V_\vc$ is not a CLF. The proof is complete by showing that $V_\vc$ is not a relaxed CLF. If $V_\vc(\vzero) \neq 0$, then $\tupleof{\mathcal{V}_\vc, Z_0} \neq 0$ and $V_\vc$ is not a relaxed CLF. Otherwise, according to Eq.~\eqref{eq:clf} there exists a $\vx \neq \vzero$ s.t.
    \[
    V_\vc(\vx) \leq 0  \ \lor \ (\forall \vu \in U) \ \nabla V_\vc.f(\vx, \vu) \geq 0 \,.
    \]
    Therefore, there exists $\vx \neq \vzero$ s.t.
    \begin{align*}
    \tupleof{\mathcal{V}_\vc, Z(\vx)} \leq 0 \ \lor \ 
    (\forall \vu \in U) \ \tupleof{F_{\vc,0}, Z(\vx)} + \sum_{i=1}^m \tupleof{F_{\vc,i}, Z(\vx)} u_i \geq 0    \,.
    \end{align*}
    Setting $Z:\ Z(\vx)$ shows that $V_\vc$ is not a relaxed CLF, since the negation of Eq.~\eqref{eq:relaxed-clf} holds. \QED
\end{proof}

We lift the overall formal learning framework to work with matrices $Z$ as counterexamples using the following modifications to various parts of the framework:
\begin{enumerate}
\item First, for each $(Z_j, \vu_j)$ in the observation set, $Z_j$ is the feasible solution returned by the SDP solver while solving Eqs.~\eqref{eq:decr-cond-relax} and~\eqref{eq:positivity-cond-relax}.

\item However, the demonstrator $\D$ requires its input to be a state $\vx \in X$.  We define a projection operator $\pi:\ \zeta \mapsto X$ mapping each $Z$ to a state $\vx:\ \pi(Z)$, such that the demonstrator operates over $\pi(Z_j)$ (instead of $\vx_j$). Note that the vector of monomials $\vz$ used to define $Z$ from $\vx$ includes the degree one terms $x_1, \ldots, x_n$. The projection operator simply selects the entries from $Z$ corresponding to these variables. Other more sophisticated projections are also possible, but not considered in this work.

\item The space of all candidates $\C$ remains unaltered except that each basis polynomial is now interpreted as $g_k:\ \tupleof{G_k, Z}$ and similarly for the Lie derivative $(\nabla g_k)\cdot f(\vx, \vu) :\ \tupleof{G_{k0}, Z} + \sum\limits_{i=1}^m\tupleof{G_{ki}, Z}u_i$. Thus, the learner is effectively unaltered.
\end{enumerate}

\begin{definition}[Relaxed Observation Compatibility] \label{def:compatible-data-relaxed}
    A polynomial function $V_\vc(\vx)$ is said to be compatible with a set of $D$-relaxed-observations $O$ iff $V_\vc$ respects the $D$-relaxed CLF conditions (Eq.~\eqref{eq:clf}) for every point in $O$:
    \begin{align*}
    \tupleof{\mathcal{V}_\vc, Z_0} = 0 \ \wedge 
    \bigwedge\limits_{(Z, \vu) \in O}
\left(\begin{array}{c} \tupleof{\mathcal{V}_\vc, Z} > 0\ \land\ \tupleof{F_{\vc,0}, Z} + \sum\limits_{i=1}^m \tupleof{F_{\vc,i}, Z}u_{i} < 0 \end{array}\right)\,.
    \end{align*}
\end{definition}

\begin{definition}[Relaxed Demonstrator Compatibility] \label{def:compatible-dem-relaxed}
    A polynomial function $V_\vc(\vx)$ is said to be compatible with a relaxed-demonstrator $\D \circ \pi$ iff $V_\vc$ respects the $D$-relaxed CLF conditions (Eq.~\eqref{eq:clf}) for every observation that can be generated by the relaxed-demonstrator:
    \begin{align*}
    \tupleof{\mathcal{V}_\vc, Z_0} = 0 \ \wedge &(\forall Z \succeq Z_0, \ Z \neq Z_0)
    \left(\begin{array}{c} \tupleof{\mathcal{V}_\vc, Z} > 0\ \land\ \tupleof{F_{\vc,0}, Z} + \sum\limits_{i=1}^m \tupleof{F_{\vc,i}, Z} \D(\pi(Z))_i < 0 \end{array}\right)\,.    
    \end{align*}
    In other words, $V_\vc$ is a relaxed Lyapunov function for the closed-loop system $\Psi(\P, \D \circ \pi)$.
\end{definition}

\begin{lemma}~\label{lem:no-lam-relaxation}
    Suppose Eq.~\eqref{eq:decr-cond-relax} has a solution $Z \not= Z_0$. Then,
    \begin{align*}
        & (\forall \vu \in U) \ \tupleof{F_{{\vc_j},0}, Z} + \sum_{i=1}^m \tupleof{F_{{\vc_j},i}, Z} u_i \geq 0\,.
    \end{align*}
\end{lemma}
    
\begin{proof}
    While in the relaxed problem, the relation between monomials are lost, each inequality in Eq.~\eqref{eq:decr-cond-relax} holds. Let $\hat{\vlam} = \Lambda(Z)$. Then, we have
    \begin{align*}
        \tupleof{F_{{\vc_j},i}, Z} = A_i^t \hat{\vlam},\ i \in \{1,\ldots, m\}  \, , \,
        \tupleof{-F_{{\vc_j},0}, Z} \leq \vb^t \hat{\vlam} , \ \hat{\vlam} \geq 0\,.
    \end{align*}    
    Similar to Lemma.~\ref{lem:control-dual} (using Farkas Lemma) this is equivalent to
    \begin{equation*}
(\forall \vu \in U) \ \tupleof{F_{{\vc_j},0}, Z} + \sum_{i=1}^m \tupleof{F_{{\vc_j},i}, Z} u_i \geq 0 \,.
    \end{equation*} \QED
\end{proof}

\begin{theorem}\label{thm:robust-relaxed-termination}
    The adapted formal learning framework terminates within a finite number of iterations. The procedure either finds a CLF $V$ or proves that no linear combination of basis functions would yield a CLF, with robust compatibility w.r.t. the (relaxed) demonstrator.
\end{theorem}
\begin{proof}
    $\C_j$ represents all $\vc$ s.t. $V_\vc$ is compatible with relaxed-observation $O_j$. Still $\mathcal{V}_\vc$ and $F_{\vc,i}$ are linear in $\vc$ (Eq.~\eqref{eq:relaxed-template}), and therefore $\C_{j-1}$, which is the set of all $\vc \in \C$ s.t.
    \begin{equation*}
    \begin{array}{l}
        \tupleof{\mathcal{V}_\vc, Z_0} = 0 \ \wedge
        \bigwedge\limits_{(Z, \vu) \in O_{j-1}}
            \left( \begin{array}{c} \tupleof{\mathcal{V}_\vc, Z} > 0\ \land\ \\ 
                \sum_{i=1}^m \tupleof{F_{\vc,i}, Z}u_{i}
                + \tupleof{F_{\vc,0}, Z} < 0 
            \end{array} \right)
    \end{array} \,,
    \end{equation*}
     is a polytope (similar to Lemma~\ref{lemma:cj-convex}).
    Suppose that at the $j^{th}$ iteration, $V_{\vc_j}(\vx) :\ \vc_j^t \cdot \vg(\vx)$ is generated by the learner. 
    The relaxed verifier solves Eqs.~\eqref{eq:positivity-cond-relax} and~\eqref{eq:decr-cond-relax}. If the optimal solution for these problems are $1$, by Lemma~\ref{lem:non-zero-sol}, $V_{\vc_j}$ is a CLF. Otherwise, it returns a counterexample $Z_j \succeq Z_0$ and $Z_j \neq Z_0$. Furthermore, according to Eqs.~\eqref{eq:positivity-cond-relax} and~\eqref{eq:decr-cond-relax} and Lemma~\ref{lem:no-lam-relaxation}
  \[
  \tupleof{\mathcal{V}_{\vc_j}, Z_j} \leq 0 \ \lor \ (\forall \vu \in U) \ \tupleof{F_{{\vc_j},0}, Z_j} + \sum_{i=1}^m \tupleof{F_{{\vc_j},i}, Z_j} u_i \geq 0\,.
  \]
      In other words, $V_{\vc_j}$ is not a $D$-relaxed CLF. Next, the demonstrator generates a proper feedback for $\pi(Z_j)$ and observation $(Z_j, \D(\pi(Z_j)))$ is added to the set of observations. Notice that $V_{\vc_j}$ does not respect the $D$-relaxed CLF conditions for $(Z_j, \D(\pi(Z_j)))$. I.e.,
    \[
    \tupleof{\mathcal{V}_{\vc_j}, Z_j} \leq 0 \ \lor \ \tupleof{F_{{\vc_j},0}, Z_j} + \sum_{i=1}^m \tupleof{F_{{\vc_j},i}, Z_j} \D(\pi(Z_j))_i \geq 0 \,.
    \]
    Therefore, the new set $\C_{j}$ does not contain $\vc_j$. Now, the learner uses the center of maximum volume ellipsoid to generate the next candidate. This process repeats and the learning procedure terminates in finite iterations (Theorem~\ref{thm:termination2}). 
    When the algorithm returns with no solution, it means that $\Vol(\C_j)$ $\leq \gamma \delta^r$. Similar to Theorem~\ref{thm:clf-or-no-robust-solution}, this guarantees that no ball of radius $\delta$ fits inside $\C_j$, which represents the set of all linear combinations of the basis functions, which are compatible with the relaxed observations. Therefore, no linear combination of basis functions would yield a CLF with robust compatibility with the relaxed observation, and thus, with the relaxed-demonstrator.
    \QED
\end{proof}

\subsection{Other Control Certificates}
While we discussed SDP relaxation to verify CLFs for smooth feedback systems, it is applicable for verification of other control certificates. Recall the formula we wish to solve for smooth feedback systems (Eq.~\eqref{eq:smooth-general}):
\begin{equation*}
    (\exists \vc \in \C)\ (\forall \vx \in X) \ 
    \ \begin{cases}
    \vx \in R_1 \implies \bigvee\limits_{q \in Q} (\exists \vu \in U) \ p_{\vc,1,q}(\vx, \vu) < 0 \\
    \vdots \\
    \vx \in R_l \implies \bigvee\limits_{q \in Q} (\exists \vu \in U)\ p_{\vc,l,q}(\vx,\vu) < 0\,.
    \end{cases}
\end{equation*}
To verify a control certificate ($\vc$ is fixed), one needs to solve $l$ different formulae:
\begin{equation*}
    \begin{array}{rl}
    (1)\: &\vx \in R_1 \land \bigwedge\limits_{q \in Q} (\forall \vu \in U) \ p_{\vc,1,q}(\vx, \vu) \geq 0 \\
    \vdots \\
    (l)\: &\vx \in R_l \land \bigwedge\limits_{q \in Q} (\forall \vu \in U)\ p_{\vc,l,q}(\vx,\vu) \geq 0\,.
    \end{array}
\end{equation*}
Similar to Lemma~\ref{lem:control-dual}, one could use Farkas Lemma to eliminate $\forall \vu$ and get the following formula:
\begin{equation*}
    \begin{array}{rl}
    (1)\: &\vx \in R_1 \land \bigwedge\limits_{q \in Q} \ p'_{\vc,1,q}(\vx, \vlam_q) \geq 0 \\
    \vdots \\
    (l)\: &\vx \in R_l \land \bigwedge\limits_{q \in Q} \ p'_{\vc,l,q}(\vx,\vlam_q) \geq 0\,,
    \end{array}
\end{equation*}
where $p'_{\vc,i,q}$ is polynomial in $\vx$ and $\vlam_q$.

Regarding the switched systems, SDP relaxation is applicable to Eq.~\eqref{eq:witness-generation-rcc-form} as Eq.~\eqref{eq:witness-generation-rcc-form} is essentially conjunction of polynomial inequalities.

\paragraph{Summary:} In this chapter we first discussed related work on constraint-based techniques for finding certificates. We also introduced an inductive learning framework, which uses counterexamples and demonstrations. Then, we provided a candidate selection technique, which guarantees quick termination. Finally, we integrated the SDP relaxation into the framework to reduce the complexity of the verification process.

%% file: eval/eval.tex
\chapter{Evaluation}~\label{ch:eval}
In this chapter, we investigate the applicability of the proposed framework. First, we demonstrate that even a simple inductive synthesis framework implemented using SMT solvers is comparable with the state-of-the-art controller synthesis toolboxes. Next, we show how the performance is affected when SDP relaxation is integrated into the inductive synthesis framework. In fact, the SDP relaxation makes robust controller synthesis (control problems involving disturbances), feasible for systems with more state variables. Afterward, by adding a demonstrator, we show the framework becomes powerful enough to address interesting control problems which arise in robotics and autonomous vehicles. Also, we evaluate the behavior and performance of controllers which are extracted from control certificates. Finally, we perform a physical platform experience to show the effectiveness of the automatically designed controllers.

\input{eval/cegis}

\input{eval/multiple}
\input{eval/sdp}
\input{eval/demonstration}
\input{eval/tracking}

%% file: eval/cegis.tex
\section{Basic Inductive Synthesis Framework}
In this section, we investigate the performance of CEGIS framework using SMT solvers.
Our approach is implemented as a Python script that wraps around constraint solvers. As discussed in Chapter~\ref{ch:search}, we only consider switched systems to avoid quantifier alternation. For the learner (solving a QFLRA for candidate generation), Z3~\cite{DeMoura-Bjorner-08-Z3} is used. For the verifier (solving a QFNRA for counterexample generation), dReal~\cite{DBLP:conf/cade/GaoKC13} is the chosen solver. 
Recall that \emph{numerical} SMT solvers like dReal~\cite{DBLP:conf/cade/GaoKC13} solve the problem using a numerical threshold $\delta$. dReal either proves a formula is \emph{UNSAT} or returns \emph{$\delta$-SAT}. As such, we consider reach-while-stay (RWS) properties to avoid numerical problems when dealing with stability properties.

We collected $15$ benchmark instances that are used in our evaluation. These benchmarks are taken from many sources and adapted to produce problem instances for our evaluation. We manually formulated a reach-while-stay (RWS) specification where one was not available. More specifically, we define the safe set $S$, initial set $I$, and the goal set $G$. Finally, for now, we do not consider disturbances--- benchmarks with disturbances were modified by setting to nominal values. \emph{A detailed description of each problem instance can be found in the appendix~\ref{ch:benchmark}.} 

The inputs to our procedure consists of (i) a description of the plant model, which includes the dynamics for each mode of the system, (ii) the specification (sets $S$, $I$ and $G$), (iii) a template for the control Lyapunov-barrier function (CLBF), and (iv) a couple of parameters described below.

The safe set $S$ is taken to be a box, while $G$ and $I$ are provided as balls of radius $r_G$ and $r_I$, respectively. These balls are centered at the origin.
For the template, we assume a quadratic form for the CLBF for all benchmarks:
\[
V (\vc, \vx) = \sum_{i=1}^n \sum_{j=1}^i c_{ij} x_i x_j - 1 \,.
\]
We assume that template coefficients belong to a compact set. More specifically  $c_{ij} \in (-\Delta, \Delta)$ and $\Delta$ can be specified by the user. We fixed $\Delta = 100$ for all the experiments. Moreover, we put $\epsilon$, which is used for control design (see Eq.~\eqref{eq:suitable-feedbacks}), directly into the condition of CLBF ($\dot{V} < -\epsilon$).
Finally, to guarantee termination, we choose an $\epsilon_T > 0$ (see Eq.~\eqref{eq:cegis-lra-relaxed}) larger than the numerical threshold used for dReal. More precisely, for each of the three inequalities in Eq.~\eqref{eq:cllf-rws} we choose a different constants:
\begin{align*}
    (\forall \vx \in I) & \ V_\vc(\vx) < {\color{red}-\epsilon_{T_1}} \\
    (\forall \vx \in \partial S) & \ V_\vc(\vx) > {\color{red}\epsilon_{T_2}} \\
    (\forall \vx \in S \setminus \inter{G}) & \ (\exists \vu \in U) \ \nabla V \cdot f_\vu(\vx) < {\color{red}-\epsilon_{T_3}}\,.
\end{align*}
For the experiments we choose $\epsilon_{T_1} = \epsilon_{T_2}$.

\begin{table*}[t]
\caption{Results of running basic CEGIS framework on the
  benchmark suite.}
\label{tab:cegis}

\textbf{Legend}: $n$: \# state variables,
  , $\delta$: dReal precision,
  $itr$ : \# iterations, Tot. T: total computation time, Z3 T: time taken by Z3, dReal T: time taken by dReal,
   \tick: control certificate found, \crossMark: failed, TO: time out ($>$ 3 hours). All timings are in seconds and rounded.
 
\begin{center}
{
\begin{tabular}{||l|l||l|l|l||r|r|r|r|r|c||}
\hline
\multicolumn{2}{||c||}{Problem} & \multicolumn{3}{c||}{Parameters}
& \multicolumn{6}{c||}{dReal Results}
\tabularnewline \hline
System ID & $n$ & $\epsilon$ & $\epsilon_{T_1}$ & $\epsilon_{T_3}$ & $\delta$ & $itr$ & z3 T & dReal T & Tot. T & Status
\tabularnewline \hline
~\ref{sys:harmonic} & 2  & 0.01 & 0.1 & 0.01 & 
10$^{-3}$ & 3 & 0.0 & 1.0 & 1.0 & \tick
\tabularnewline \hline
~\ref{sys:linear-ss-1} & 2  & 0.001 & 0.1 & 0.001 & 
10$^{-3}$ & 1 & 0.0 & 1.2 & 1.3 & \tick
\tabularnewline \hline
~\ref{sys:dc-motor} & 2  & 0.01 & 0.1 & 0.01 & 
10$^{-4}$ & 3 & 0.0 & 0.8 & 0.8 & \tick
\tabularnewline  \hline
~\ref{sys:dc-dc} & 2  & 0.0001 & 0.2 & 0.0001 & 
10$^{-4}$ & 21 & 0.1 & 25.1 & 26.0 & \tick
\tabularnewline  \hline
~\ref{sys:tulip-2d} & 2 &  0.1 & 0.1 & 0.1 & 
10$^{-3}$ & 3 & 0.0 & 1.4 & 1.5 & \tick
\tabularnewline  \hline
~\ref{sys:sliding-motion-2} & 2  & 0.0001 & 0.1 & 0.0001 & 
10$^{-4}$ & 5 & 0.0 & 1.6 & 1.7 & \crossMark
\tabularnewline
 & & & & & 
10$^{-5}$ & 3 & 0.0 & 0.7 & 0.8 & \tick
\tabularnewline  \hline
~\ref{sys:inverted-pendulum} (a) & 2 & 0.01 & 0.1 & 0.01 & 
10$^{-4}$ & 5 & 0.0 & 8.6 & 8.9 & \crossMark
\tabularnewline \hline
~\ref{sys:inverted-pendulum} (b) & 2 & 0.05 & 0.1 & 0.01 & 
10$^{-4}$ & 7 & 0.0 & 23.6 & 23.9 & \tick
\tabularnewline
 & & 0.01 & & & 
10$^{-4}$ & 3 & 0.0 & 7.4 & 7.5 & \tick
\tabularnewline  \hline
~\ref{sys:linear-ss-2} & 3  & 0.0001 & 0.1 & 0.01 & 
10$^{-5}$ & 26 & 6.1 & 83.2 & 90.8 & \tick
\tabularnewline
 &  &  &  &  & 
10$^{-4}$ & 24 & 2.8 & 75.6 & 79.6 & \tick
\tabularnewline \hline
~\ref{sys:linear-ss-3} & 3  & 0.0001 & 0.05 & 0.0001 & 
10$^{-5}$ & 2 & 0.0 & 8.6 & 8.6 & \tick
\tabularnewline  \hline
~\ref{sys:non-equilibrium-stabilization} & 3  & 0.05 & 0.1 & 0.05 & 
10$^{-3}$ & 2 & 0.0 & 5.2 & 5.2 & \tick
\tabularnewline  \hline
~\ref{sys:tulip-pipe-3d} & 3  & 0.01 & 0.1 & 0.1 & 
10$^{-4}$ & 31 & 1.2 & 69.5 & 71.9 & \tick
\tabularnewline  \hline
~\ref{sys:lorenz} & 3 & 0.0001 & 0.1 & 0.01 & 
10$^{-4}$ & 17 & 3.3 & 290.3 & 294.8 & \tick
\tabularnewline  \hline
~\ref{sys:nonholonomic} & 3 & 0.1 & 0.1 & 1.0 & 
10$^{-4}$ & 55 & 208.8 & 44.2 & 255.5& \crossMark
\tabularnewline  \hline
~\ref{sys:LQR} & 4 & 0.01 & 0.1 & 0.01 & 
10$^{-3}$ & 1 & 0.0 & 12.1 & 12.1 & \tick
\tabularnewline  \hline
~\ref{sys:heater} (a) & 3 & 0.001 & 0.1 & 0.001 & 
10$^{-3}$ & 1 & 0.0 & 3.6 & 3.6 & \tick
\tabularnewline \hline
~\ref{sys:heater} (b) & 4 & 0.001 & 0.1 & 0.001 & 
10$^{-3}$ & 1 & 0.0 & 16.1 & 16.2 & \tick
\tabularnewline \hline
~\ref{sys:heater} (c) & 5 & 0.001 & 0.1 & 0.001 & 
10$^{-3}$ & 1 & 0.0 & 403.1 & 403.1 & \tick
\tabularnewline \hline
~\ref{sys:heater} (d) & 6 & 0.001 & 0.1 & 0.001 & 
10$^{-3}$ & 1 & 0.0 & 703.2 & 703.2 & \tick
\tabularnewline \hline
~\ref{sys:heater} (e) & 9 & 0.001 & 0.1 & 0.001 & 
10$^{-3}$ & - & - & - & TO & \tick
\tabularnewline \hline
\end{tabular}\\
}
\end{center}
\end{table*}
\clearpage

All the computations are
performed on a Mac Book Pro with 2.9 GHz Intel Core i7 processor and
16GB of RAM. The results  are summarized in Table~\ref{tab:cegis}. Our approach finds a control certificate for all but two problem instances. For example, our approach fails to find a control certificate for System $13$, a nonholonomic unicycle system for which it is known that no polynomial CLF exists~\cite{brockett1983}. Our technique was successful on some benchmarks with up to six state variables.

\paragraph{A Comparison:}
We also perform an \emph{extensive comparison} with the SCOTS tool~\cite{rungger2016scots}, a recently developed state-of-the-art control synthesizer.
SCOTS uses fixed-point computation to find a (maximum) controllable region. Therefore, one would find out whether the initial set is in the controllable region or not, only after the computation. On the other hand, our method forces the initial condition to be inside the controllable region by definition. 
Therefore, for a fair comparison, we declare success for a fixed-point computation method on a problem instance, if the initial region is a subset of the controllable region.
SCOTS has three main parameters: (i) state discretization parameter $\eta$, (ii) time discretization parameter $\tau$, and (iii) input discretization parameter $\mu$. $\mu$ is defined according to the problem instance. The user should set the other two parameters. To come up with these parameters, roughly speaking, in the first step, we try to use small values for $\tau$ and $\eta$ to get a solution. Afterward, in the quest for faster termination, we try larger values for parameters and check whether the computation remains successful or not.
Tables~\ref{tab:scots-1} and~\ref{tab:scots-2} show the results for SCOTS using different parameters.

\begin{table*}[t!]
\caption{Results of running SCOTS on the
  benchmark suite -- Part I.}
\label{tab:scots-1}

\textbf{Legend}: $n$: \# state variables,
  $itr$ : \# iterations, abs T: abstraction computation time, FP T: fixed-point computation time, Tot. T: total computation time, \tick: control certificate found, \crossMark: failed. TO: time out ($>$ 3 hours). All timings are in seconds and rounded.
 
\begin{center}
{
\begin{tabular}{||l|l||l|l||r|r|r|r|c||r|c||}
\hline
\multicolumn{2}{||c||}{Problem} & \multicolumn{2}{c||}{Parameters}
& \multicolumn{5}{c||}{SCOTS Results} & \multicolumn{2}{c||}{CEGIS Results} 
\tabularnewline \hline
System ID & $n$ & $\eta$ & $\tau$ & $itr$ & abs T & FP T & Tot. T & Status & Tot. T & Status
\tabularnewline \hline
~\ref{sys:harmonic} & 2 & 
0.02 & 0.2 &
26 & 0.3 & 0.4 & 0.7 & \tick & 1.0 & \tick
\tabularnewline
 & & 0.03 & 0.2 &
5 & 0.1 & 0.0 & 0.0 & \crossMark &  & 
\tabularnewline
 & & & 0.3 &
6 & 0.1 & 0.0 & 0.0 & \crossMark &  & 
\tabularnewline \hline
~\ref{sys:linear-ss-1} & 2 &
0.3 & 0.5 &
13 & 0.0 & 0.0 & 0.0 & \tick & 1.3 & \tick
\tabularnewline
 & & 0.4 & 0.5 & 
3 & 0.0 & 0.0 & 0.0 & \crossMark & & 
\tabularnewline
 & & & 1.0 & 
3 & 0.0 & 0.0 & 0.0 & \crossMark & & 
\tabularnewline \hline
~\ref{sys:dc-motor} & 2 & 
0.008 & 0.001 &
153 & 0.9 & 7.7 & 8.6 & \tick & 0.8 & \tick
\tabularnewline
 & & 0.009 & 0.001 & 
13 & 0.7 & 0.3 & 1.0 & \crossMark & & 
\tabularnewline
 & & & 0.002 & 
13 & 0.7 & 0.3 & 1.0 & \crossMark & & 
\tabularnewline \hline
~\ref{sys:dc-dc} & 2 & 
0.02 & 0.75 &
145 & 0.0 & 0.3 & 0.3 & \tick & 26.0 & \tick
\tabularnewline
 & & 0.025 & 0.75 & 
62 & 0.0 & 0.1 & 0.1 & \crossMark & & 
\tabularnewline
 & & & 1.5 & 
28 & 0.0 & 0.0 & 0.0 & \crossMark & & 
\tabularnewline \hline
~\ref{sys:tulip-2d} & 2 &
0.16 & 0.12 &
18 & 0.0 & 0.0 & 0.0 & \tick & 1.5 & \tick
\tabularnewline
 & & 0.17 & 0.12 & 
9 & 0.0 & 0.0 & 0.0 & \crossMark & & 
\tabularnewline
 & &  & 0.15 & 
5 & 0.0 & 0.0 & 0.0 & \crossMark & & 
\tabularnewline \hline
~\ref{sys:sliding-motion-2} & 2 &
0.001 & 0.001 & 
213 & 681.5 & 2501.3 & 3182.8 & \crossMark & 0.8 & \tick
\tabularnewline
 & & & 0.002 &
429 & 692.7 & 7310.4 & 8003.1 & \crossMark &  &
\tabularnewline
 & & & 0.004 &
269 & 718.4 & 4842.7 & 5561.1 & \crossMark &  &
\tabularnewline
 & & 0.002 & 0.001 &
3 & 89.0 & 1.9 & 90.9 & \crossMark &  &
\tabularnewline
 & & & 0.004 &
226 & 91.6 & 643.3 & 734.9 & \crossMark &  &
\tabularnewline
 & & & 0.005 &
205 & 98.2 & 634.4 & 732.6 & \crossMark &  &
\tabularnewline \hline
~\ref{sys:inverted-pendulum} (a) & 2 & 
0.04 & 0.07 & 
55 & 0.1 & 0.3 & 0.4 & \tick & 8.9 & \crossMark
\tabularnewline \hline
~\ref{sys:inverted-pendulum} (b) & 2 & 
0.03 & 0.05 & 
53 & 1.3 & 3.2 & 4.5 & \tick & 7.5 & \tick
\tabularnewline
  & & 0.04 & 0.05 & 
20 & 0.6 & 0.4 & 1.0 & \crossMark &  & 
\tabularnewline
  & & & 0.06 & 
11 & 0.6 & 0.2 & 0.8 & \crossMark &  & 
\tabularnewline \hline
\end{tabular}\\
}
\end{center}
\end{table*}

\begin{table*}[t!]
\caption{Results of running SCOTS on the
  benchmark suite -- Part II.}
\label{tab:scots-2}

\textbf{Legend}: $n$: \# state variables,
  $itr$ : \# iterations, abs T: abstraction computation time, FP T: fixed-point computation time, Tot. T: total computation time, \tick: control certificate found, \crossMark: failed, TO: time out ($>$ 3 hours). All timings are in seconds and rounded.
 
\begin{center}
{
\begin{tabular}{||l|l||l|l||r|r|r|r|c||r|c||}
\hline
\multicolumn{2}{||c||}{Problem} & \multicolumn{2}{c||}{Parameters}
& \multicolumn{5}{c||}{SCOTS Results} & \multicolumn{2}{c||}{CEGIS Results} 
\tabularnewline \hline
System ID & $n$ & $\eta$ & $\tau$ & $itr$ & abs T & FP T & Tot. T & Status & Tot. T & Status
\tabularnewline \hline
~\ref{sys:linear-ss-2} & 3 & 
0.01 & 0.1 & 
53 & 2111.5 & 2566.2 & 4677.7 & \tick & 79.6 & \tick 
\tabularnewline
 & & 0.015 & 0.1 & 
19 & 364.1 & 114.3 & 478.4 & \crossMark & & 
\tabularnewline
 & & & 0.15 & 
21 & 352.6 & 147.7 & 500.2 & \crossMark & & 
\tabularnewline \hline
 ~\ref{sys:linear-ss-3} & 3 &
0.03 & 0.2 & 
46 & 42.7 & 92.7 & 135.4 & \tick & 8.6 & \tick 
\tabularnewline
 & & 0.04 & 0.2 & 
9 & 13.7 & 2.1 & 15.8 & \crossMark & & 
\tabularnewline
 & &  & 0.3 & 
7 & 12.5 & 1.7 & 14.2 & \crossMark & & 
\tabularnewline \hline
 ~\ref{sys:non-equilibrium-stabilization} & 3 &
  0.025 & 0.02 & 
62 & 57.5 & 246.4 & 303.9 & \tick & 5.2 & \tick 
\tabularnewline
 & & 0.03 & 0.02 & 
53 & 62.7 & 78.7 & 141.4 & \crossMark & & 
\tabularnewline
 & & & 0.03 & 
54 & 27.0 & 79.7 & 106.7 & \crossMark & & 
\pagebreak
\tabularnewline \hline
 ~\ref{sys:tulip-pipe-3d} & 3 &
 0.3 & 0.1 & 
18 & 0.5 & 0.2 & 0.7 & \tick & 71.9 & \tick 
\tabularnewline
 & & 0.4 & 0.1 & 
5 & 0.2 & 0.0 & 0.0 & \crossMark & & 
\tabularnewline
 & & 0.4 & 0.2 & 
4 & 0.2 & 0.0 & 0.2 & \crossMark & & 
\tabularnewline \hline
 ~\ref{sys:lorenz} & 3 & 
0.09 & 0.02 & 
88 & 71.1 & 266.8 & 337.9 & \tick & 294.8 & \tick 
\tabularnewline
& & 0.1 & 0.02 & 
23 & 42.3 & 31.2 & 73.5 & \crossMark & & 
\tabularnewline
& & 0.1 & 0.03 & 
68 & 40.0 & 136.8 & 176.8 & \crossMark & & 
\tabularnewline \hline
 ~\ref{sys:nonholonomic} & 3 & 
 0.07 & 0.1 & 
44 & 4.4 & 4.4 & 8.8 & \tick & 255.5 & \crossMark 
\tabularnewline
 & & 0.08 & 0.1 & 
2 & 2.4 & 0.1 & 2.5 & \crossMark & & 
\tabularnewline
 & & & 0.2 & 
4 & 2.4 & 0.1 & 2.5 & \crossMark & & 
\tabularnewline \hline
 ~\ref{sys:LQR} & 4 & 
 0.05 & 0.01 & 
5 & 2931.7 & 53.3 & 2985 & \crossMark & 12.1 & \tick 
\tabularnewline
 & &  & 0.02 & 
6 & 2800.7 & 90.5 & 2891.2 & \crossMark & & 
\tabularnewline
 & &  & 0.03 & 
5 & 2748.6 & 75.4 & 2824.0 & \crossMark & & 
\tabularnewline
 & & 0.04 & 0.03 & 
- & TO & - & TO & \crossMark & & 
\tabularnewline \hline
 ~\ref{sys:heater} (a) & 3 & 
 0.3 & 6.0 & 
23 & 1.5 & 0.6 & 2.1 & \tick & 3.6 & \tick 
\tabularnewline
 & & 0.4 & 6.0 & 
3 & 0.6 & 0.0 & 0.6 & \crossMark & & 
\tabularnewline
 & &  & 8.0 & 
4 & 0.5 & 0.0 & 0.5 & \crossMark & & 
\tabularnewline \hline
 ~\ref{sys:heater} (b) & 4 &  
 0.2 & 4.0 & 
64 & 1305.3 & 1016.0 & 2321.3 & \tick & 16.1 & \tick
\tabularnewline
 & & 0.3 & 4.0 & 
9 & 98.9 & 8.4 & 107.3 & \crossMark & & 
\tabularnewline
 & & & 6.0 & 
9 & 99.9 & 8.4 & 108.3 & \crossMark & & 
\tabularnewline \hline
 ~\ref{sys:heater} (c) & 4 & 0.2 & 4.0 & 
- & TO & -& TO & \crossMark & 403.1 & \tick
\tabularnewline \hline
 ~\ref{sys:heater} (d) & 6 & 0.2 & 4.0 & 
- & TO & -& TO & \crossMark & 703.2 & \tick
\tabularnewline \hline
 ~\ref{sys:heater} (e) & 9 & 0.2 & 4.0 & 
- & TO & -& TO & \crossMark & TO & \crossMark
\tabularnewline \hline
\end{tabular}\\
}
\end{center}
\end{table*}

Tables~\ref{tab:scots-1} and~\ref{tab:scots-2} also compare results for our method and SCOTS for the provided benchmark. We found that SCOTS can solve two problem instances where our method fails. Unlike SCOTS, our method is restricted to systems for which control certificates exist with a simple structure (polynomials), and in that sense, SCOTS is more precise. Our method fails for System~\ref{sys:inverted-pendulum} (a) mainly because the safe set is very small and quadratic CLBFs do not exist. The problem is solved when a bigger safe set is used (System~\ref{sys:inverted-pendulum} (b)). Also, as mentioned before, our method fails for System~\ref{sys:nonholonomic} simply because a smooth feedback law and thus a smooth control certificate does not exist~\cite{brockett1983}.

On the other hand, SCOTS uses abstraction for the system. For an abstraction, all the states in one abstract state (a cell) are treated the same way which yields a less precise transition relation compared to our method. To compensate, usually smaller values for $\eta$ is considered, which in turn, increases the computation cost, especially for higher dimensional problems. As demonstrated in Table~\ref{tab:scots-1} and~\ref{tab:scots-2}, SCOTS performs really well on $2D$ problems. However, our method is relatively faster with $3D$ problems. Also, given a computation time limit, SCOTS fails on five problem instances four of which are problems with four or higher variables. While SCOTS performs well on $2D$ problems, it still fails on System~\ref{sys:sliding-motion-2}, for which our method succeeds. We attribute this to the use of a more precise transition relation in our method.

\clearpage

%% file: eval/multiple.tex
\subsection{Uninitialized RWS}
In this section, we consider uninitialized RWS problems. Recall that we wish to find a control Lyapunov fixed-barriers function (CLFBF) $V$ with the following conditions to address the uninitialized RWS problem:
\begin{align*}
    (\forall \vx \in S \setminus \inter{G}) \ (\exists \vu \in U) \ \left(
    \nabla V_\vc \cdot f_\vu(\vx) < 0 \land \bigwedge_j \left(
    \dot{p}_{S,j,\vu}(\vx) + \lambda p_{S,j}(\vx) < 0\,,
    \right) \right)\,,
\end{align*}
where $S :\ \{\vx \ | \ p_{S,1}(\vx) \leq 0 \land \ldots \land p_{S,l}(\vx) \leq 0 \}$.
\begin{example}
\label{ex:basic-4D}
    This example is adopted from~\cite{habets2004control}. There are four variables and two control inputs. The dynamics are as follows:
    \begin{align*}
        \dot{x_1} = x_1 + x_2 + 8 + u_1\, , \,
        \dot{x_2} =     -x_2 + x_3 + 1 - u_2\, , \,
        \dot{x_3} =     -2x_3 + 2x_4 + 1 -2u_1\, , \,
        \dot{x_4} =     -3x_4 + 1 + u_2
    \end{align*} 
    The region of interest $S$ is hyber-box $[-1, 1]^4$ and the input belongs to set $[0, 1] \times [0, 2]$. The goal is to reach facet $x_1 = 1$, while staying in $S$ as the safe region. 
    
    First, we discretize the control input to model the system as a switched system. For this purpose, we assume $u_1 \in \{0, 1\}$ and $u_2 \in \{0, 0.5, 1, 1.5, 2\}$. Then, we use a linear template for the CLF ($c_1 x_1 + c_2 x_2 + c_3 x_3 + c_4 x_4$) , $\epsilon = 0.1$, $\lambda = 5$. CEGIS framework finds control certificate $V(\vx):\ -0.13333344 (x_1+x_2+x_3+x_4)$.
\end{example}

\begin{example} \label{ex:obstacle-1} Consider again the system from Example~\ref{ex:basic-control-to-facet}, with the addition of some obstacles~\cite{nilssonincremental}. More precisely, as shown in Fig~\ref{fig:obstacle1-basic}(a), the safe set is $S = S_0 \setminus (O_1 \cup O_2)$. First, the safe set is decomposed into four basic semialgebraic sets, which are shown with $R_0$ to $R_3$ in Figure~\ref{fig:obstacle1-basic}(a).
\begin{figure}[t]
\centering
\begin{subfigure}{0.3\textwidth}
  \includegraphics[width=1\linewidth]{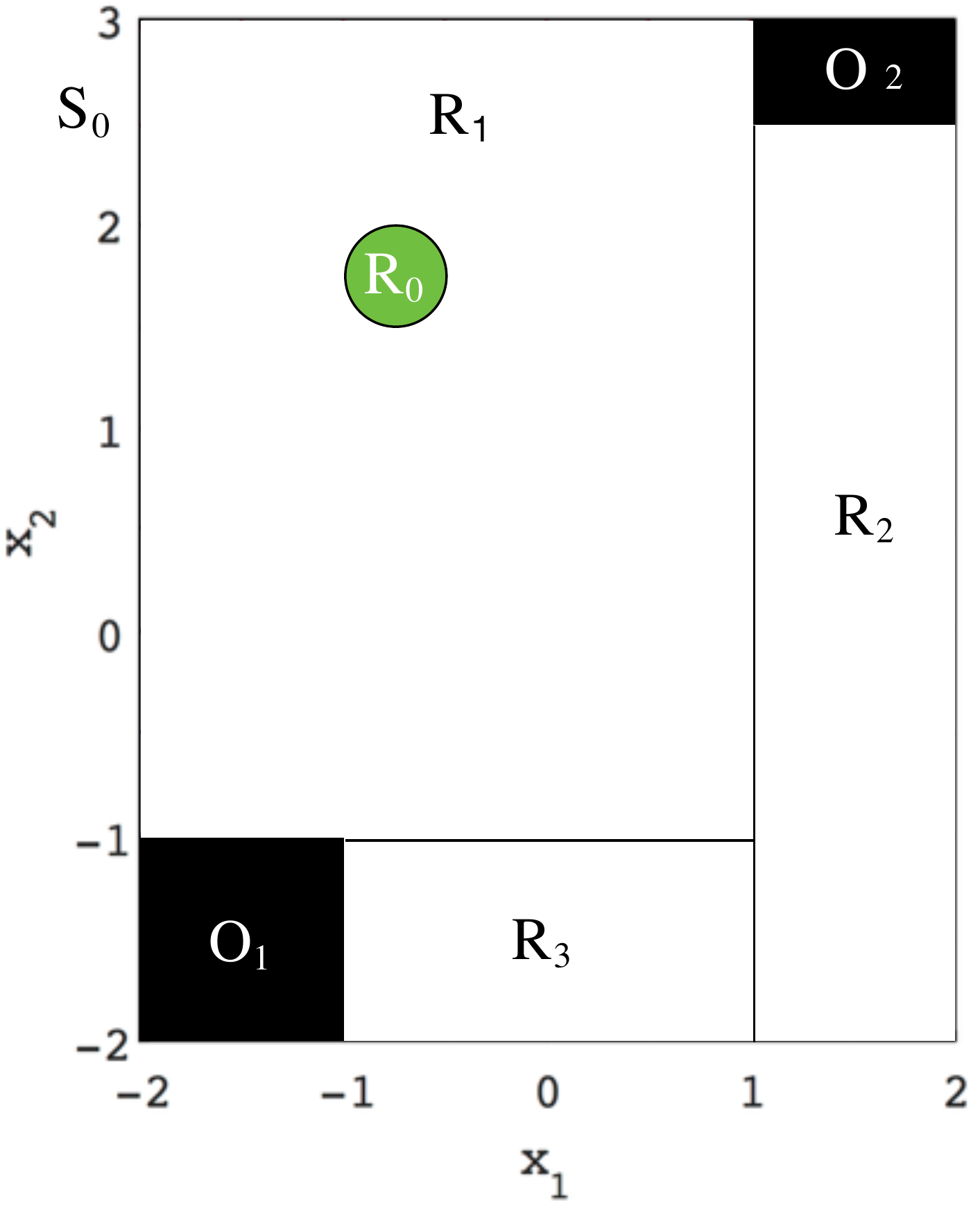}
  \caption{Schematic view of state decomposition}
  \label{fig:obstacle1-basic-a}
  \end{subfigure} \qquad
  \begin{subfigure}{0.3\textwidth}
  \includegraphics[width=1\linewidth]{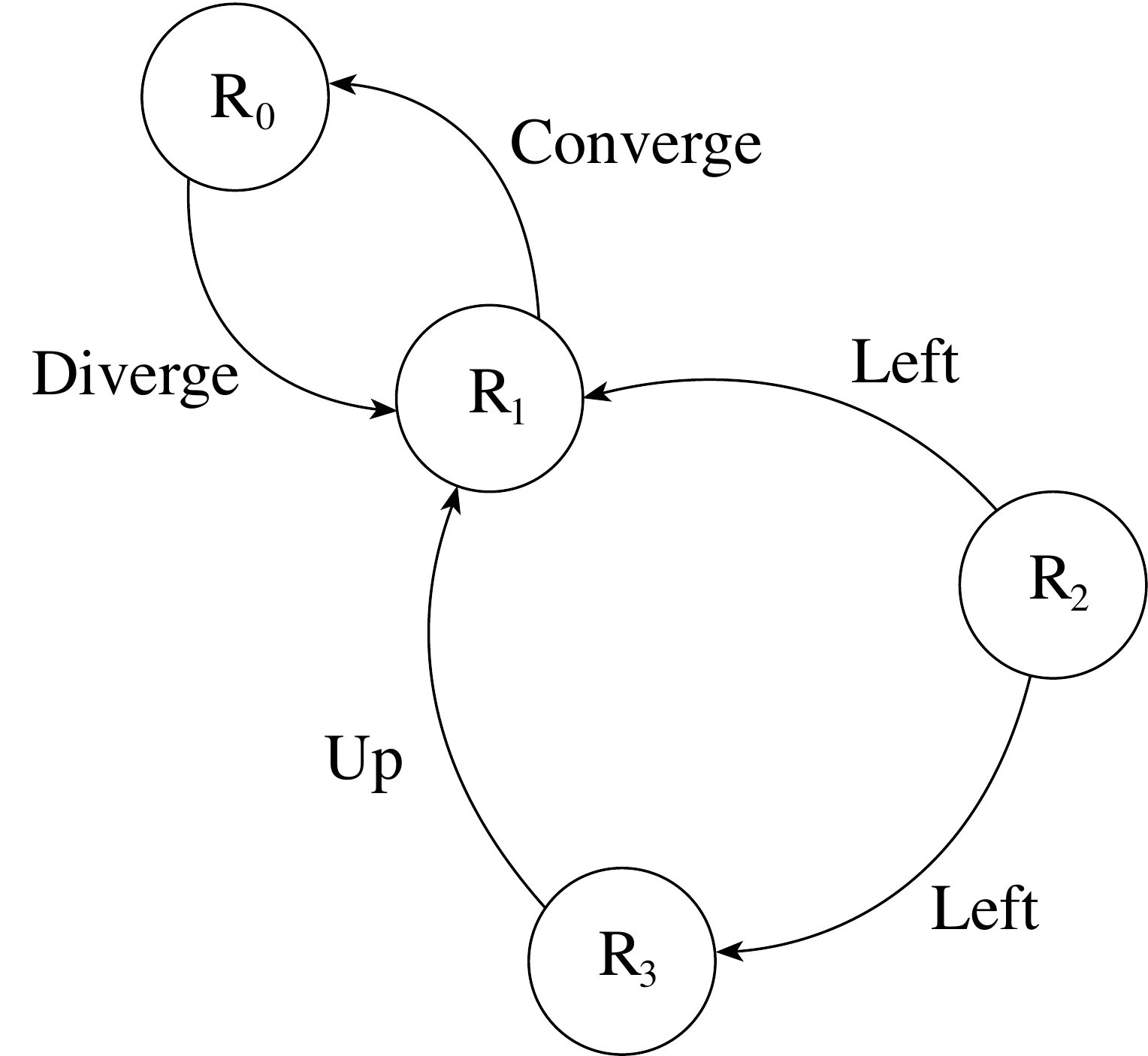} \\ \\ \\
  \caption{Finite abstraction for the original problem}
  \label{fig:obstacle1-basic-b}
  \end{subfigure}
  
\caption{Visualizations for Example~\ref{ex:obstacle-1}.} \label{fig:obstacle1-basic}
\end{figure}

$R_0$ is the target set. Next, we build a transition relation between four abstract states, representing four basic semialgebraic sets. This is done by solving seven RWS problems for basic semialgebraic sets. For the $R_1 \rightarrow R_0$ transition, we use a quadratic template for $V$, and for other problems, we use linear templates. The abstract system is shown in Figure~\ref{fig:obstacle1-basic}(b). Next, the problem is solved for the abstract system. The solution to the abstract system is simple: if the state is in $R_2$, the controller uses the left facet to reach $R_1$ or $R_3$. Otherwise, if the state is in $R_3$, the controller uses the upper facet to reach $R_1$, and finally, if the state is in $R_1$, the controller makes sure the state reaches $R_0$.
\end{example}

\begin{example}\label{ex:unicycle}
    A unicycle~\cite{zamani2012symbolic} has three variables. $x$ and $y$ define the position of the car and $\theta$ is its angle. The dynamics of the system is
    $\dot{x} = u_1 cos(\theta) \,, \ 
        \dot{y} = u_1 sin(\theta) \,,$ and $
        \dot{\theta} = u_2$,
    where $u_1$ and $u_2$ are inputs. Assuming a switched system, we consider $u_1 \in \{-1, 0, 1\}$ and $u_2 \in \{-1, 0, 1\}$. The safe set is $[-1, 1]\times[-1, 1]\times[-\pi, \pi]$ and the target facet is $x = 1$.
    We use a template that is linear in $(x,y)$ and quadratic in $\theta$. Using $\epsilon = 0.1$ and $\lambda = 0.5$, the following CLF is found after $22$ iterations:
    \[ V(\vx) :\ - x - y - 0.5881 \theta +\theta^2 - 0.1956 \theta x + \theta y \,.\]

\paragraph{Problem (a):}
Now we consider a path planning problem for the unicycle~\cite{rungger2016scots}. Projection of the safe set on the $x$-$y$ plane yields a maze. The target set is placed at the right bottom corner of the maze (Figure~\ref{fig:unicycle}).
 Using specification-guided technique, we modeled the system with 53 polyhedra. Each polyhedron is treated as a single state, and a transition relation is built by solving 113 control-to-facet problems. Then, the problem is solved over the finite graph. The total computation took 1484 seconds. The figure also shows a single trajectory of the closed-loop system.
 
 \paragraph{Problem (b):}
     This problem is similar to the previous one, except for the fact that there is no direct control over the angular velocity. More precisely, only the angular acceleration is controllable and the system would have the following dynamics
    $\dot{x} = u_{1} \cos (\theta),\,
    \dot{y} = u_{1} \sin (\theta),\,
    \dot{\theta} = \omega,\,
    \dot{\omega} = u_{2}\,.
    $
    Also, we assume $\omega \in [-1, 1]$. By changing the coordinates one can use $r = \sqrt{x^2 + y^2}$, $z_1 = x\cos (\theta)+y\sin (\theta)$ and $z_2 = y\cos (\theta)-x\sin (\theta)$ to define the position and angle of the car(cf.~\cite{liberzon2012switching} for details). Then, we use the following template $V(x, y, \theta, \omega) = c_1 r^2 + c_2 z_1 + c_3 z_2\omega + c_4 \omega^2$,
    where the origin is located just outside of the target facet. Using this template, we find control certificates for all 113 control-to-facet problems in 5296 seconds.
\end{example}

\begin{figure}[t]
\centering
\includegraphics[width=0.4\textwidth]
    {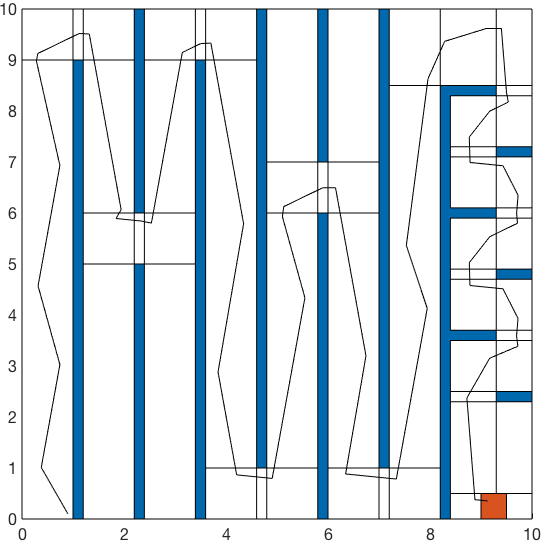}
    \\
    Region $G$ is shown shaded in Orange, and unsafe regions are shown in blue. An execution trace of the car is shown for $x$ and $y$ variables.
\caption{A $x$-$y$ view for Example~\ref{ex:unicycle}.}
\label{fig:unicycle}
\end{figure}

\paragraph{A Comparison:} While SCOTS provides a relatively complete solution, it is computationally expensive. On the other hand, our method is a Lyapunov-based method which uses polynomial templates.  As such, the existence of certificates is not guaranteed, and thus, our approach lacks the general applicability of a fixed-point based synthesis. 
However, our method is relatively more scalable thanks to recent development in SMT solvers. Here, for the sake of completeness, we provide a brief comparison with SCOTS toolbox~\cite{rungger2016scots} for the examples provided in this section.
To compare Example~\ref{ex:basic-4D}, we use ``fat" facet and assume that the target set has a volume (otherwise, because of time discretization, SCOTS cannot find a solution). More precisely, we use target set $[1, 1.2]\times[-1, 1]^3$ instead of $[1,1]\times[-1,1]^3$. 

The results are reported in Table~\ref{tab:scots}. We also note that if we use larger values for SCOTS parameters, SCOTS fails to solve these problems (the initial set is not a subset of the controllable region).  Table~\ref{tab:scots} shows that SCOTS performs better for Example~\ref{ex:obstacle-1}, for which there are only two state variables. For Example~\ref{ex:unicycle}(a), both methods have similar performances, while for Example~\ref{ex:basic-4D} and Example\ref{ex:unicycle}(b) with four state variables, our method is faster.

\begin{table}[t!]
\caption{Results for uninitialized RWS problems.}
\label{tab:scots}

\textbf{Legend}: $n$: \# state variables,
  $itr$ : \# iterations, Time: total computation time, $\eta$: state discretization step, $\tau$: time step. All timings are in seconds and rounded, TO: timeout ($> 10$ hours).
\begin{center}
{
\begin{tabular}{||l|l||l|l|r|r||l|r||}
\hline
\multicolumn{2}{||c||}{Problem} & \multicolumn{4}{c||}{SCOTS}
& \multicolumn{2}{c||}{CEGIS} 
\tabularnewline \hline
ID & $n$ & $\eta$ & $\tau$ & $itr$ & Time  & $\delta$ & Time
\tabularnewline \hline
Example~\ref{ex:obstacle-1} & 2 &
0.16$^2$ & 0.12 &
18 & 0 & 10$^{-4}$ & 3 
\tabularnewline \hline
Example~\ref{ex:unicycle}(a) & 3 &
0.2$^2$$\times$0.1 & 0.3 &
404 & 989 & 10$^{-4}$ & 1484
\tabularnewline \hline
Example~\ref{ex:basic-4D} & 4 &
0.03$\times$0.1$^3$ & 0.005 &
48 & 304 & 10$^{-5}$ & 3
\tabularnewline \hline
Example~\ref{ex:unicycle}(b) & 4 &
0.1$^2\times$0.05$^2$ & 0.3 &
\multicolumn{2}{|c||} {TO} & 10$^{-4}$ & 5296
\tabularnewline \hline
\end{tabular}\\
}
\end{center}
\end{table}

%% file: eval/sdp.tex
\section{Integrating SDP Relaxation}

In this section, we demonstrate the effectiveness of the SDP relaxation in generating counterexamples. We extend the CEGIS framework to work with SDP relaxation by lifting the counterexamples.
We use the Gloptipoly tool~\cite{henrion2009gloptipoly} for counterexample generation. Gloptipoly, in turn, uses YALMIP~\cite{lofberg2004yalmip}, which is configured to use Mosek~\cite{mosek2010mosek} as the SDP solver.

We use the same set of benchmarks that we used in the previous section and the results are shown (and compared with the SMT solver) in Table~\ref{tab:cegis-sdp}. Compared to SDP relaxation, the SMT solver is much faster for $2D$ problems. For $3D$ problems, these two methods are comparable, and for higher dimensional problems, SDP relaxation clearly wins. 

\begin{table*}[t!]
\caption{Comparing SDP relaxation with SMT solvers on the
  benchmark suite.}
\label{tab:cegis-sdp}

\textbf{Legend}: $n$: \# state variables,
  , $\delta$: dReal precision, $\degr$: SDP relaxation degree,
  $itr$ : \# iterations, T: total computation time, Z3 T: time taken by Z3, dReal T: time taken by dReal, Glp T: time taken by Gloptipoly, St: Status, \tick: control certificate found, \crossMark: failed, TO: time out ($>$ 3 hours). All timings are in seconds and rounded.
 
\begin{center}
{
\begin{tabular}{||l|l||r|r|r|r|r|c||r|r|r|r|r|c||}
\hline
\multicolumn{2}{||c||}{Problem}
& \multicolumn{6}{c||}{Numerical SMT Solver Results} & \multicolumn{6}{c||}{SDP Relaxation Results}
\tabularnewline \hline
ID & $n$ & $\delta$ & $itr$ & z3 T & dReal T & T & St & $\degr$ & $itr$ & z3 T & Glp T & T & St
\tabularnewline \hline
~\ref{sys:harmonic} & 2  & 
10$^{-3}$ & 3 & 0.0 & 1.0 & 1.0 & \tick &
3 & 17 & 0.4 & 41.5 & 42.0 & \tick
\tabularnewline \hline
~\ref{sys:linear-ss-1} & 2  &
10$^{-3}$ & 1 & 0.0 & 1.2 & 1.3 & \tick &
2 & 1 & 0.0 & 8.4 & 8.4 & \tick
\tabularnewline \hline
~\ref{sys:dc-motor} & 2  & 
10$^{-4}$ & 3 & 0.0 & 0.8 & 0.8 & \tick &
3 & 8 & 0.0 & 60.0 & 60.0 & \tick
\tabularnewline
 &  & 
 & & &  &  &  &
4 & 5 & 0.0 & 40.1 & 40.2 & \tick
\tabularnewline  \hline
~\ref{sys:dc-dc} & 2  & 
10$^{-4}$ & 21 & 0.1 & 25.1 & 26.0 & \tick &
2 & 13 & 0.0 & 54.6 & 54.7 & \tick
\tabularnewline  \hline
~\ref{sys:tulip-2d} & 2 &  
10$^{-3}$ & 3 & 0.0 & 1.4 & 1.5 & \tick &
3 & 3 & 0.0 & 12.4 & 12.5 & \tick
\tabularnewline  \hline
~\ref{sys:sliding-motion-2} & 2  &
10$^{-4}$ & 5 & 0.0 & 1.6 & 1.7 & \crossMark &
3 & 7 & 0.0 & 26.2 & 26.3 & \crossMark
\tabularnewline
 & & 
10$^{-5}$ & 3 & 0.0 & 0.7 & 0.8 & \tick &
4 & 3 & 0.0 & 11.8 & 11.9 & \tick
\tabularnewline  \hline
~\ref{sys:inverted-pendulum} (a) & 2 & 
10$^{-4}$ & 5 & 0.0 & 8.6 & 8.9 & \crossMark &
5 & 5 & 0.0 & 27.0 & 27.1 & \crossMark
\tabularnewline \hline
~\ref{sys:inverted-pendulum} (b) & 2 & 
10$^{-4}$ & 7 & 0.0 & 23.6 & 23.9 & \tick &
5 & 22 & 0.4 & 136.2 & 136.8 & \tick
\tabularnewline
 & &
10$^{-4}$ & 3 & 0.0 & 7.4 & 7.5 & \tick &
5 & 3 & 0.0 & 12.5 & 12.5 & \tick
\tabularnewline  \hline
~\ref{sys:linear-ss-2} & 3  & 
10$^{-5}$ & 26 & 6.1 & 83.2 & 90.8 & \tick &
3 & 13 & 0.1 & 67.4 & 67.7 & \tick
\tabularnewline
 &  & 
10$^{-4}$ & 24 & 2.8 & 75.6 & 79.6 & \tick &
 &  &  &  &  & 
\tabularnewline \hline
~\ref{sys:linear-ss-3} & 3  &
10$^{-5}$ & 2 & 0.0 & 8.6 & 8.6 & \tick &
2 & 2 & 0.0 & 13.0 & 13.0 & \tick
\tabularnewline  \hline
~\ref{sys:non-equilibrium-stabilization} & 3  & 
10$^{-3}$ & 2 & 0.0 & 5.2 & 5.2 & \tick &
2 & 2 & 0.0 & 12.2 & 12.2 & \tick
\tabularnewline  \hline
~\ref{sys:tulip-pipe-3d} & 3  & 
10$^{-4}$ & 31 & 1.2 & 69.5 & 71.9 & \tick &
3 & 14 & 0.2 & 50.2 & 50.6 & \tick
\tabularnewline  \hline
~\ref{sys:lorenz} & 3 & 
10$^{-4}$ & 17 & 3.3 & 290.3 & 294.8 & \tick &
4 & 24 & 12.0 & 223.9 & 236.3 & \tick
\tabularnewline  \hline
~\ref{sys:nonholonomic} & 3 &
10$^{-4}$ & 55 & 208.8 & 44.2 & 255.5& \crossMark &
3 & - & - & - & TO & \crossMark
\tabularnewline  \hline
~\ref{sys:LQR} & 4 & 
10$^{-3}$ & 1 & 0.0 & 12.1 & 12.1 & \tick &
3 & 1 & 0.0 & 29.4 & 29.5 & \tick
\tabularnewline  \hline
~\ref{sys:heater} (a) & 3 & 
10$^{-3}$ & 1 & 0.0 & 3.6 & 3.6 & \tick &
2 & 1 & 0.0 & 12.0 & 12.0 & \tick
\tabularnewline \hline
~\ref{sys:heater} (b) & 4 & 
10$^{-3}$ & 1 & 0.0 & 16.1 & 16.2 & \tick &
2 & 1 & 0.0 & 15.3 & 15.4 & \tick
\tabularnewline \hline
~\ref{sys:heater} (c) & 5 & 
10$^{-3}$ & 1 & 0.0 & 403.1 & 403.1 & \tick &
2 & 1 & 0.0 & 24.4 & 24.4 & \tick
\tabularnewline \hline
~\ref{sys:heater} (d) & 6 & 
10$^{-3}$ & 1 & 0.0 & 703.2 & 703.2 & \tick &
2 & 1 & 0.0 & 25.4 & 25.4 & \tick
\tabularnewline \hline
~\ref{sys:heater} (e) & 9 & 
10$^{-3}$ & - & - & - & TO & \crossMark &
2 & 1 & 0.0 & 54.6 & 54.6 & \tick
\tabularnewline \hline
\end{tabular}\\
}
\end{center}
\end{table*}

\section{Handling Disturbances}
We also extended the CEGIS framework implementation to handle disturbances and find \emph{robust} CLBFs (see Section~\ref{sec:dist-1}).
We compare the robust controller synthesis (RS) approach with a simple \emph{Synthesize and Verify Robustness} (SVR) approach that uses a nominal disturbance value (e.g., $\vd = 0$) and checks whether the resulting controller is robust, as the last step.  Specifically, the disturbance-free case uses CLBF for synthesis. In doing so, we also check whether adding a ``margin'' by increasing the value of $\epsilon$ during controller synthesis necessarily makes the resulting design more robust to disturbances.
For comparison, several RWS examples are considered where $\dot{\vx} = f_\vu(\vx) + \vd$ and disturbances $\vd \in D:\ [-\sigma_D, \sigma_D]^n$, with varying values of $\sigma_D$. All times are reported in seconds.

\begin{example}
\label{ex:obstacle-disturbance}
This example considers a DC-motor model with two variables $\omega$ and $i$, described in Appendix~\ref{ch:benchmark} (System~\ref{sys:dc-motor}).

The results for the RS method are shown in Table~\ref{tab:obstacle-disturbance}.  To evaluate the effect of disturbances on the CEGIS procedure, we use different disturbance sizes. The results suggest that bigger disturbances impose harder restrictions on the RCLBF and many more iterations are needed, as the size of disturbance gets bigger.

 \begin{table}[t!]
    \caption{Results for Example~\ref{ex:obstacle-disturbance} using RS method.}
    \label{tab:obstacle-disturbance}
    \begin{center}
    \begin{tabular}{|c|c|c|c|}
    \hline
    $\sigma_D$ & Itr & Time & Status    \tabularnewline \hline
    0.0 & 3 & 36.1 & \tick \tabularnewline \hline
    0.4 & 4 & 54.4 & \tick \tabularnewline \hline
    0.8 & 5 & 115.3 & \tick  \tabularnewline \hline
    1.2 & 8 & 177.5 & \tick  \tabularnewline \hline
    1.4 & 34 & 800.9 & \tick  \tabularnewline \hline
    1.6 & 83 & 1500.0 & \tick  \tabularnewline \hline
    1.7 & 184 & 4367.8 & \tick  \tabularnewline \hline
    1.8 & 494 & 10565.4 & \crossMark  \tabularnewline \hline
    \end{tabular}     
    \end{center}
\end{table}

\begin{table}[t!]
    \caption{Results for Example~\ref{ex:obstacle-disturbance} using SVR method.}
    \label{tab:obstacle-disturbance-svr}
    \begin{center}
    \begin{tabular}{|c|c|c|c|c|}
    \hline
    $\epsilon$ & $\sigma_D$ & Itr & Time & Status    \tabularnewline \hline
    0.2 & 0.1 & 3 & 32.4 & \tick \tabularnewline \hline
    0.3 & 0.2 & 5 & 40.7 & \tick \tabularnewline \hline
    0.3 & 0.3 & 5 & 46.6 & \crossMark \tabularnewline \hline
    0.4 & 0.0 & 536 & 5074.3 & \crossMark \tabularnewline \hline
    \end{tabular}
    \end{center}
\end{table}

On the other hand, using the SVR technique, first, a CLBF is found with preferably higher values for $\epsilon$. The most robust controller is obtained using $\epsilon = 0.03$ and it is verified that this controller can handle disturbances for $r_D = 0.2$ . The results are shown in Table~\ref{tab:obstacle-disturbance-svr}. These results suggest that RS method can provide provably more robust controllers where the SVR approaches fails to synthesize a controller for larger values of $\epsilon$ and fails to verify for larger disturbance values.
\end{example}

\begin{example}\label{ex:ozay-example}
Consider Example~\ref{ex:basic-control-to-facet}, wherein the goal is to reach a region around $(-0.75, 1.75)$ ($G :\ \{(x_1, $ $x_2) | (x_1+0.75)^2 + (x_2-1.75)^2 \leq 0.25^2\}$).
For each method, we check for the biggest disturbance for which the problem can be solved. Using RS method, we were able to solve the problem when $\sigma_D = 0.5$ using $\epsilon = 0.01$. For the SVR method, the most robust controller (obtained by setting $\epsilon = 0.2$) is verified to decrease $V$ when $\sigma_D = 0.03$. Detailed results are shown in Table~\ref{tab:ozay-example}. Again, these results suggest RS method yields more robust controllers.

\begin{table}[t!]
    \caption{Results for Example~\ref{ex:ozay-example} using SVR
    method.} \label{tab:ozay-example} \begin{center} \begin{tabular}{|c|c|c|c|c|} \hline
    $\epsilon$ & $\sigma_D$ & Itr & Time &
    Status \tabularnewline \hline 0.1 & 0.01 & 3 & 30.4
    & \tick \tabularnewline \hline 0.2 & 0.03 & 15 & 87.6
    & \tick \tabularnewline \hline 0.2 & 0.04 & 15 & 87.8
    & \crossMark \tabularnewline \hline 0.3 & 0.0 & 35 & 290.9
    & \crossMark \tabularnewline \hline \end{tabular}

    \end{center}
\end{table}

 \end{example}

\begin{example}
The following example is taken from~\cite{bolzern2004quadratic}. The system has three continuous variables with four different modes. The dynamics are provided in Appendix~\ref{ch:benchmark} (System~\ref{sys:non-equilibrium-stabilization}). Again, the goal is to find the most robust controller. 
The RS method can find a RCLBF with disturbance $\sigma_D = 0.1$ when $\epsilon = 0.01$ is used.
The SVR method failed to synthesize a controller for $\epsilon = 0.3$ and using $\epsilon = 0.2$, it failed to verify the controller for $\sigma_D = 0.009$. The robustness is guaranteed for $\sigma_D = 0.008$.
\end{example}

\begin{table*}[t!]
    \caption{Results for Example~\ref{ex:heater}.}
    \textbf{Legend:} $n$ : \# state variables, $m$: \# modes, OM: Out of Memory.
    \label{tab:heater}
    \begin{center}
    \begin{tabular}{||c|c||c|c|c|c|c||c|c|c|c|c||}
    \hline
    \multicolumn{2}{||c||}{Problem} & \multicolumn{5}{|c||}{RS} & \multicolumn{5}{|c||}{SVR}
    \tabularnewline \hline
    $n$ & $m$ & 
    $\sigma_D$ & $\epsilon$ & Itr & Time & Status & 
    $\sigma_D$ & $\epsilon$ & Itr & Time & Status
    \tabularnewline \hline
    \multirow{3}{*}{3} & \multirow{3}{*}{4} & 
    0.04 & 0.0001 & 1 & 25.1 & \tick &
    0.005 & 0.02 & 4 & 65.6 & \tick 
    \tabularnewline 
     & & & & & & &
     0.006 & 0.02 & 4 & 54.8 & \crossMark 
    \tabularnewline 
     & &  &  & &  &  &
     0.0 & 0.03 & 4 & 49.0 & \crossMark 
    \tabularnewline \hline
    
    \multirow{3}{*}{4} & \multirow{3}{*}{5} & 
    0.02 & 0.0001 & 1 & 155.6 & \tick &
    0.001 & 0.01 & 4 & 237.0 & \tick 
    \tabularnewline 
     & & & & & & &
     0.002 & 0.01 & 4 & 159.5 & \crossMark 
    \tabularnewline 
     & &  &  & &  &  &
     0.0 & 0.02 & 6 & 71.9 & \crossMark 
    \tabularnewline \hline
    
    \multirow{3}{*}{5} & \multirow{3}{*}{6} & 
    0.001 & 0.0001 & 1 & 1500.3 & \tick &
    0.0002 & 0.002 & 4 & 2243.5 & \tick 
    \tabularnewline 
     & & & & & & &
     0.0003 & 0.002 & 4 & 872.5 & \crossMark 
    \tabularnewline 
     & &  &  & &  &  &
     0.0 & 0.003 & 3 & 89.7 & \crossMark 
    \tabularnewline \hline
    
    \multirow{3}{*}{6} & \multirow{3}{*}{4} & 
    0.01 & 0.0001 & 1 & 9224.3 & \tick &
     0.001 & 0.01 & 4 & 11559.4 & \tick  
    \tabularnewline 
     & & & & & & &
     0.002 & 0.01 & 4 & 4237.4 & \crossMark 
    \tabularnewline 
     & &  &  & &  &  &
     0.0 & 0.02 & 4 & 333.5 & \crossMark 
    \tabularnewline \hline
    
    9 & 4 & 
    0.01 & 0.0001 & - & - & OM &
    0.001 & 0.01 & - & - & OM 
    \tabularnewline \hline
    \end{tabular}
    
    \end{center}    
\end{table*}

\begin{example}
\label{ex:heater}
This benchmark includes five problem instances, the details of which are available in Appendix~\ref{ch:benchmark} (System~\ref{sys:heater}). The goal is to keep different rooms of an apartment warm, using few numbers of active heaters. While these examples do not have disturbances, we incorporate disturbances of the form $\dot{\vx} = f_\vu(\vx) + \vd$. We use these problem instances to demonstrate the scalability.
The results for both methods are shown in Table~\ref{tab:heater}. These results demonstrate that our method is scalable to larger problems while dealing with robustness. Notice that both methods fail for the last problem instance as the verification of such a big problem, even when we use the SDP relaxation scheme, is expensive.
\end{example}

%% file: eval/demonstration.tex
\section{Integrating the Demonstrator Oracle}\label{sec:eval-demonstration}
In this section, we investigate the demonstration-based inductive synthesis. For the demonstrator, a nonlinear MPC scheme is used, which is solved using a gradient descent algorithm.  For each benchmark, the following parameters are tuned to obtain the cost function:
\begin{compactenum}
    \item time step $\tau$
    \item number of horizon steps $N$
    \item $Q$, $R$, and $H$ for the cost function:
    \[
    \begin{array}{c}
    \left( \sum_{0 = 1}^{N-1} \vx(i\tau)^t \ Q \ \vx(i\tau) + \vu(i\tau)^t \ R \ \vu(i\tau) \right) + \vx(N\tau)^t \ H \ \vx(N\tau)\,.
    \end{array}
    \]
      \end{compactenum}
As such, an MPC cost function is designed to enforce RWS. 
However, since the approach provides no guarantees, we run hundreds of simulations of the closed-loop system starting from randomly selected initial states to check whether the specifications are met. Failing this, the cost function is adjusted, repeating the testing process.
For the verifier, we use the SDP relaxation. As discussed in the previous chapter, the learner at each iteration returns the center of the maximum volume ellipsoid (MVE) that fits inside the set of candidate solutions (see Theorem~\ref{thm:termination2}). To implement the learner, we use Mosek~\cite{mosek2010mosek} as the SDP solver to find the center of the MVE.

Now, we address the problems discussed in Chapter~\ref{ch:intro}. The reported control certificates are rounded to two decimal points. 

\paragraph{Bicycle Problem:} We wish to solve the RWS for the bicycle model described in Example~\ref{ex:bicycle}. The $\sin$ function appeared in the dynamics is approximated with polynomial of degree one. Next, our framework searches and finds the following CLBF:
\begin{align*}
V(\vx):\ 
&0.37 y^2
+ 0.52 y\theta
+ 3.11 \theta^2
+ 0.98 y\sigma
+ 2.23 \sigma\theta +
4.46 \sigma^2
- 0.36 vy
- 0.29 v\theta
+ 0.95 v\sigma
+ 3.86 v^2\,.
\end{align*}

This CLBF is used to design a controller. Figure~\ref{fig:bicycle-sim} shows the projection of traces onto the $x$-$y$ plane for such controller in red.  The blue traces are generated using the MPC controller. The behaviors of the system for both controllers are similar but not identical. Notice that the initial state in
  Figure~\ref{fig:bicycle-sim}(c) is not in the region of attraction
  (guaranteed region). Nevertheless, the certificate-based controller can still stabilize the system while keeping the system in the safe region. On the other hand, the MPC violates the safety constraints even when the safety constraints are formulated in the MPC scheme. We note that the safety is violated because in the beginning, $\theta$ gets larger than $1$ and it gets close to $\pi/2$ (the vehicle moves almost vertically).

\begin{figure}[t!]
\begin{center}
    \includegraphics[width=0.6\textwidth]{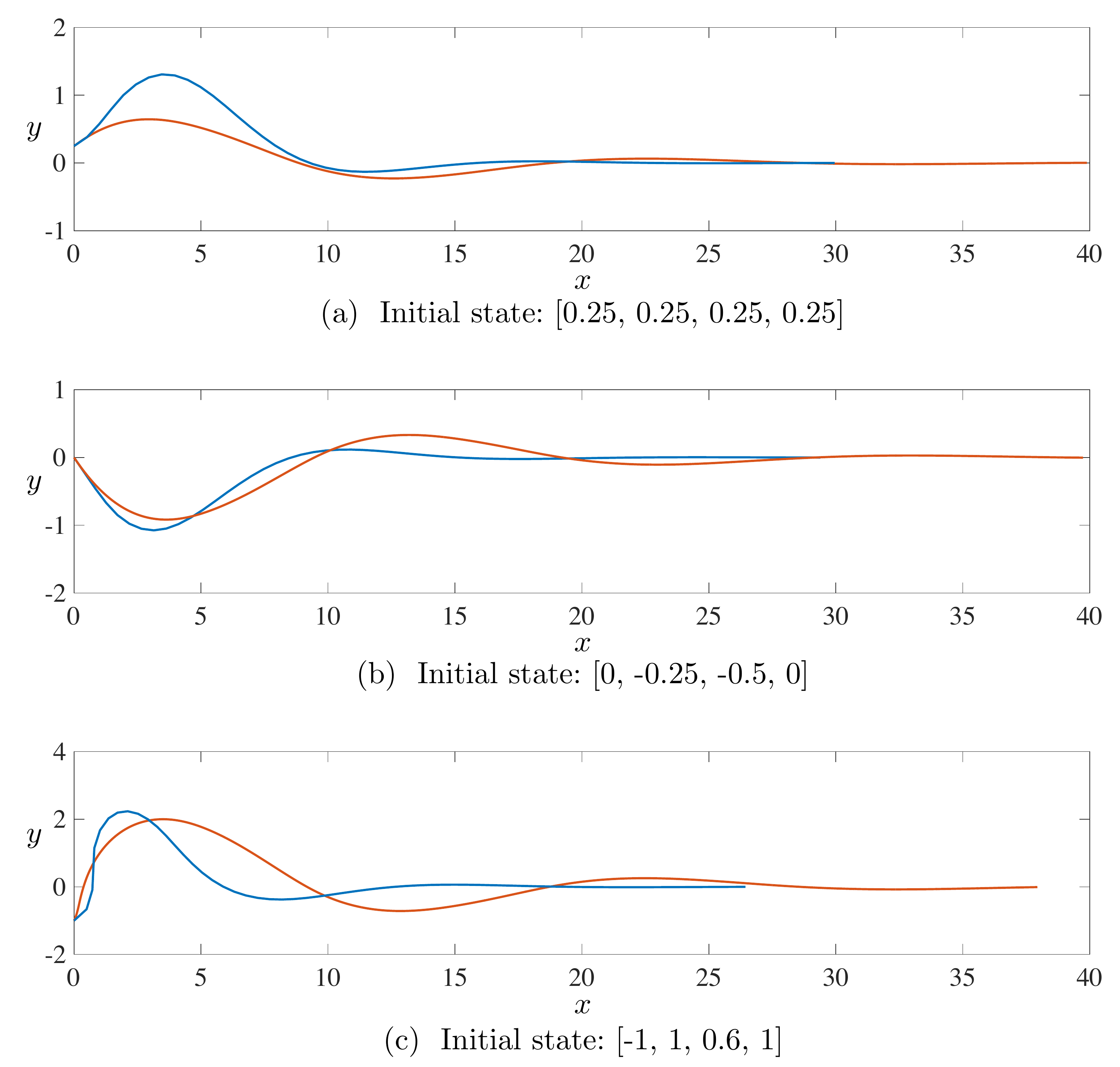}
    \\
     Simulation traces are plotted for three different initial states. Blue (red) traces corresponds to traces of the system for MPC controller (certificate-based controller).
      \end{center}
\caption{Simulation for the bicycle model - Projected on x-y
 plane.}\label{fig:bicycle-sim} 
\end{figure}

\paragraph{Inverted Pendulum Problem:} To solve the safety problem of the inverted pendulum on a cart (Example~\ref{ex:inverted-pendulum}), first, a partial
linearization is performed over the dynamics which results in the following dynamics:
\begin{align*}\label{eq:inverted-pendulum-dyn}
        \ddot{x} = 4u + \frac{4(M+m)g \tan(\theta) - 3mg\sin(\theta)\cos(\theta)}{4(M+m)-3m\cos^2(\theta)} \,, \ \ \ \ \ddot{\theta} = \frac{- 3 u \cos(\theta)}{l} \,.
\end{align*}
Then, the trigonometric and rational functions are approximated with polynomials of degree three.

Figure~\ref{fig:inverted-sim} shows some of the traces of the closed-loop system for the certificate-based controller as well as the MPC controller. Notice that the certificate-based controller can behave differently, Especially in regions where a demonstration is not provided. For example, for Figure~\ref{fig:inverted-sim}(b), the behaviors of these controllers are similar outside the initial set $I$. However, inside $I$ (near the equilibrium) the behavior is different as the demonstrations are only generated for states outside $I$.
The certificate-based controller is designed using the following CLBF generated by the learning framework:
\begin{align*}
V(\vx):\ & 16.37 \dot{\theta}^2 + 50.37 \dot{\theta}\theta
+ 75.16 \theta^2 + 13.51 x\dot{\theta} 
+ 43.26 x\theta + \\
& 10.44 x^2 + 23.30 \dot{\theta}\dot{x} + 38.09 \dot{x}\theta
+ 11.13 \dot{x}x + 9.55 \dot{x}^2 \,.
\end{align*}

\begin{figure*}[t]
\begin{center}
    \includegraphics[width=1\textwidth]{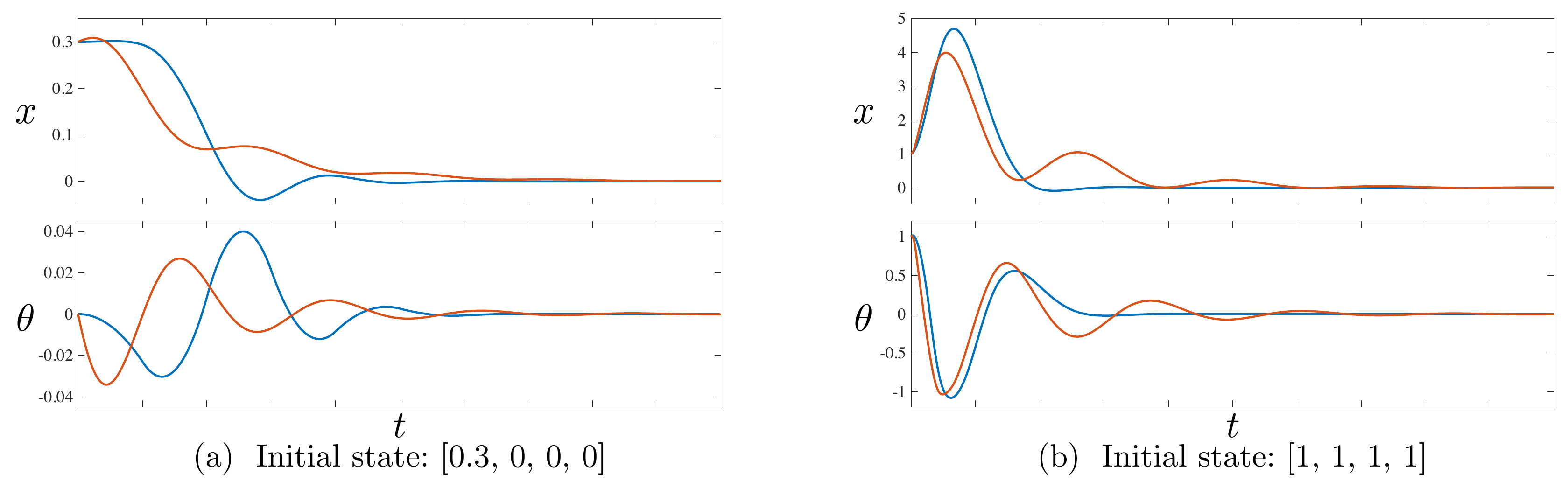}
\\ Simulation traces are plotted 
for two initial states. Red (blue) traces show the simulation traces for the certificate-based (MPC) controller.
\end{center}
\caption{Simulation for the inverted pendulum system.}\label{fig:inverted-sim} 
\end{figure*}

\paragraph{Forward Flight Problem:}
Here we wish to stabilize the Caltech ducted-fan in a forward flight (Example~\ref{ex:ducted-fan-forward}). Our framework is not directly applicable to the problem, because the system is not affine in control. To address this problem, we replace inputs $u$ and $\delta_u$ with $u_s = u \sin(\delta_u)$ and $u_c = u \cos(\delta_u)$:
\begin{align*}
        \dot{v} &= \frac{-D(v, \alpha) - W \sin(\gamma) + u_c \cos(\alpha) -  u_s \sin(\alpha)}{m} & 
        \dot{\gamma} &= \frac{L(v, \alpha) - W \cos(\gamma) + u_c \sin(\alpha) + u_s \cos(\alpha)}{mv} \\
        \dot{\theta} &= q & \dot{q} &= \frac{M(v, \alpha) - l_T u_s}{J}\,.
\end{align*}
Projection of $U$ into the new coordinate will yield a sector of a circle.
Then, $U$ is safely under-approximated by a polytope $\hat{U}$ as shown in Figure~\ref{fig:uhat}.
\begin{figure}
\begin{center}
    \includegraphics[width=0.3\textwidth]{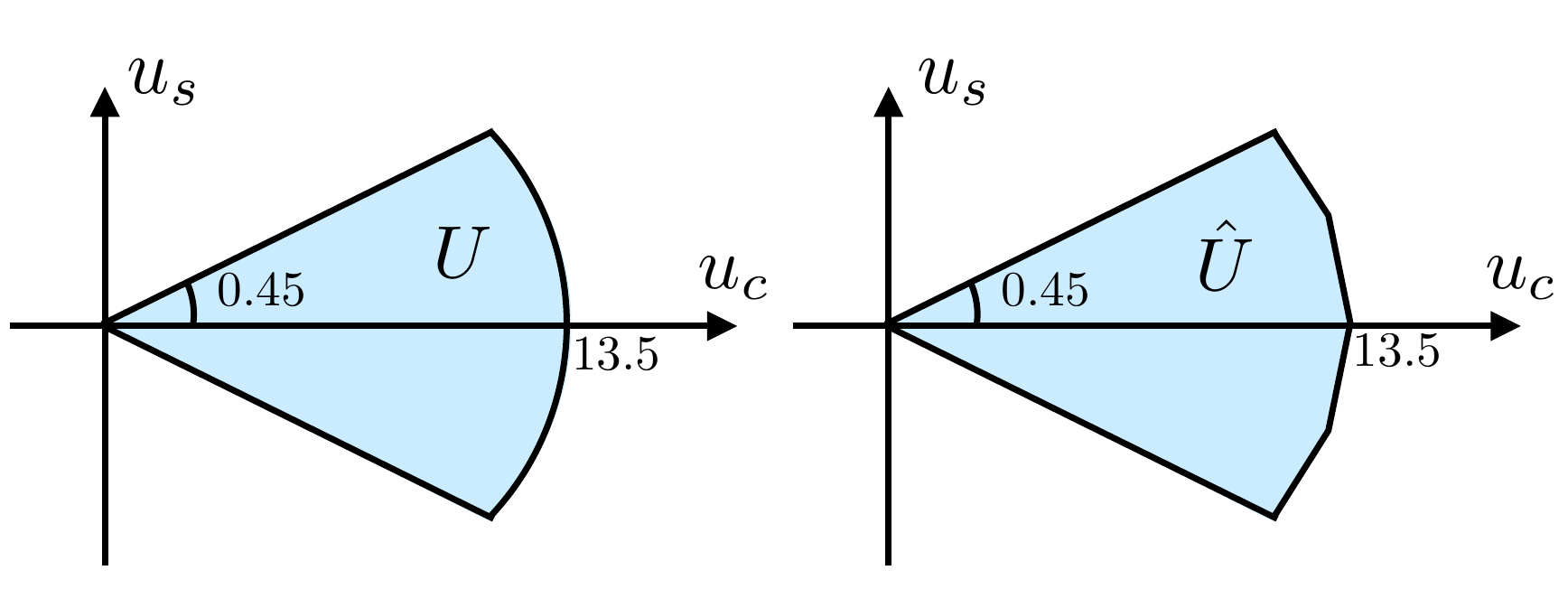}
\end{center}
\caption{Set of inputs $U$ and its under approximation $\hat{U}$ for Example~\ref{ex:ducted-fan-forward}.}\label{fig:uhat} 
\end{figure}
Next, we perform a translation so that the $\vx^*$ ($\vu^*$) is the origin of the  state (input) space in the new coordinate system.
In order to obtain a polynomial dynamics, we approximate $v^{-1}$, $\sin$ and $\cos$ with polynomials of degree one, three, and three, respectively.
These changes yield a polynomial control affine dynamics, which fits the description of the supported model. 

The projection of some of the traces of the system in $x$-$y$ plane is shown in Figure~\ref{fig:forward-sim}. We set $x_0 = y_0 = 0$ and $\dot{x} = v \cos(\gamma) , \ \dot{y} = v \sin(\gamma)$.
The certificate-based controller is designed using the following generated CLBF:
\begin{align*}
V(\vx):\ &3.23 q^2 + 2.17 q\theta
+ 3.90 \theta^2 - 0.2 qv
- 0.45 v\theta 
& + 0.53 v^2 + 1.66 q\gamma - 1.33 \gamma\theta
+ 0.48 v\gamma + 3.90 \gamma^2 \,.
\end{align*}
The traces show that the certificate-based controller stabilizes faster, however, the MPC controller uses the aerodynamics to achieve the same goal with a lower cost.

\begin{figure}[t!]
\begin{center}
    \includegraphics[width=1\textwidth]{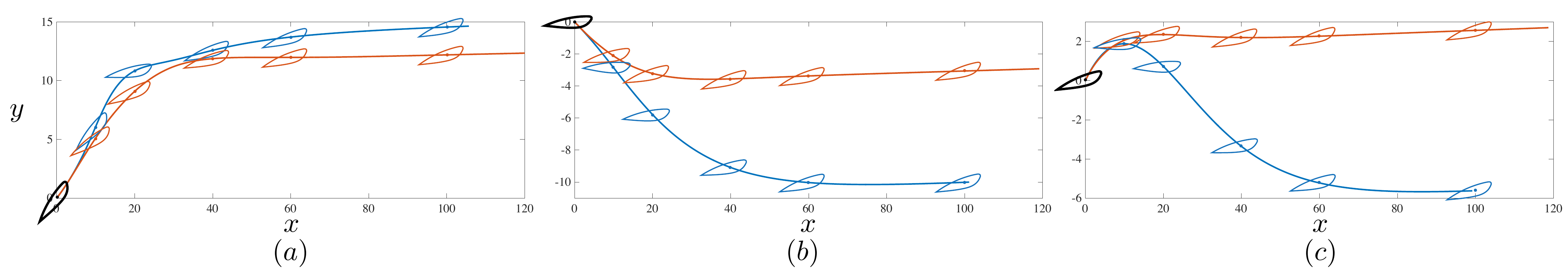}
\\ The initial rotational position is shown with the black ducted-fan. Blue (red) traces are trajectories of the closed loop system with the MPC (certificate-based) controller. The rotational position is shown for some of the states for each trajectory. Initial states are $[2, 0.4, 0.717, 0]$, $[-1, -0.25, -0.133, 0]$, and $[-1, 0.4, 0.177, 0]$ for (a), (b), and (c), respectively.
\end{center}
\caption{Simulation for forward flight of 
Caltech ducted-fan - Projected on x-y
 plane.}\label{fig:forward-sim} 
\end{figure}

\paragraph{Hover Mode Problem:}
To stabilize the Caltech ducted-fan in a hover mode (Example~\ref{ex:ducted-fan-hover}), the trigonometric functions in the dynamics are approximated with degree two polynomials and the procedure finds a quadratic CLBF:
\begin{align*}
V(\vx):\ & 1.64 \dot{\theta}^2 - 0.56 \dot{\theta}\dot{y}
+ 13.53 \dot{y}^2 + 0.07 \dot{\theta}y + 1.15 y\dot{y} +
1.16 y^2 + 1.74 \theta\dot{\theta} + 0.03 \dot{y}\theta - 0.77 y\theta + \\
&4.80 \theta^2 - 4.57 \dot{\theta}\dot{x} + 0.85 \dot{x}\dot{y} + 0.34 y\dot{x} - 8.59 \dot{x}\theta + 12.77 \dot{x}^2 -
0.45 \dot{\theta}x + 0.06 \dot{y}x +  0.51 yx -  \\
&3.71 x\theta + 4.12 x\dot{x} +1.88 x^2 \,.
\end{align*}
Some of the traces are shown in Figure~\ref{fig:hover-sim}. As the simulations suggest, the MPC controller behaves very differently and the certificate-based controller yields solutions with more oscillations. 
Also, once the trace is inside the target region, the certificate-based controller does not guarantee a decrease in $V$ as this fact is intuitively visible in Figure~\ref{fig:hover-sim}(c).

\begin{figure}[t!]
\begin{center}
    \includegraphics[width=1\textwidth]{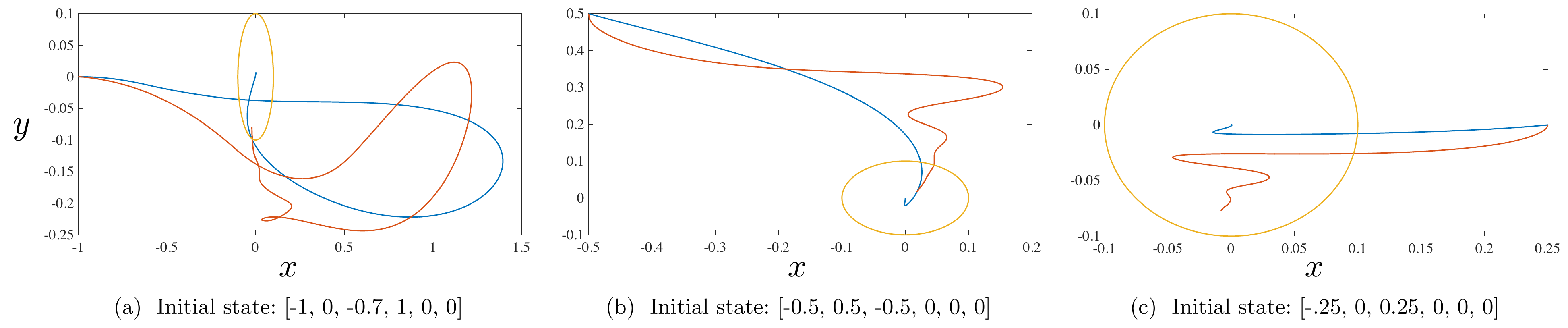}
\\ The trajectories corresponding to the certificate-based (MPC) controller are shown in red (blue) lines. The boundary of the target set $G$ is shown in yellow.
\end{center}
\caption{Simulation for the Caltech ducted-fan in hover mode  - Projected on x-y
 plane.}\label{fig:hover-sim} 
\end{figure}

\subsection{Performance}
As mentioned earlier, the inputs to the learning framework are a plant $\P$, a specification $\varphi$, monomial basis functions $\vg$, a demonstrator $\D$, and the degree of relaxation $\degr$.
At each iteration, first an MVE inscribed inside a polytope is calculated. This task is performed rather efficiently. The MPC scheme used inside the demonstrator is part of the input and we do not consider its computation performance here. Nevertheless, MPC is known to be very effective if it is carefully tuned.
We mention that the MPC parameters used here are selected by a non-expert and usually the time step is very small, while the horizon is very long. Nevertheless, the MPC is used offline and suitable for our framework. Also, costs matrices $Q$, $R$, and $H$ are diagonal:
\[
Q = diag(Q') \ , \ R = diag(R') \ , \ H = N diag(Q') \,,
\]
where $Q' \in \reals^n$ and $R' \in \reals^m$.  Two other important factors determine the performance of the whole learning framework: (i) the time taken by the verifier and (ii) the number of iterations. Table~\ref{tab:result} shows the results of the learning framework for the set of case studies described thus far. For each problem instance, the parameters of the MPC, as well as the degree of the SDP relaxation is provided.
Also, the performance of the learning framework is tabulated.
First, the procedure starts from $\C :\ [-\Delta, \Delta]^r$ and terminates  whenever $\Vol(\E_j) < \gamma \delta^r$. We set $\Delta = 100$ and $\delta = 10^{-3}$. The results demonstrate that the method terminates in few iterations, even for the cases where a compatible control certificate does not exist.

Notice that the number of demonstrations is different from the number of iterations. The trick is as follows. If we can find a counterexample by considering only conditions involving $V$ (and not $\nabla V$), there is no need for a demonstration and the procedure can move to the next iteration.
This optimization is added to speed up the procedure by avoiding expensive calls to the offline MPC.  
As Table~\ref{tab:result} shows, using this trick, the number of demonstrations can be much smaller than the total number of iterations.

At each iteration, several verification problems are solved which involve solving large SDP problems. While the complexity of solving SDP is polynomial in the number of variables, they are still hard to solve. The verification problem is quite expensive when the number of variables as well as the degree of relaxation is large. Nevertheless, as the SDP solvers mature, we believe our method can solve larger problems since the verification procedure is currently the computational bottleneck for the learning framework.  

\begin{table*}[t!]
\caption{Results of running demonstration-based CEGIS.}\label{tab:result} 
\textbf{Legend:} $n$: \# variables, $m$: \# control inputs, $\tau$: MPC time step, $N$: number of horizon steps, $Q'$: defines MPC state cost, $R'$: defines MPC input cost, $D$: SDP relaxation degree bound, \#D : number of demonstrations, \#Itr: number of iterations, VT: total computation time for verification (minutes), T: total computation time (minutes), St: Status, \tick: control certificate found, \crossMark: fail.
\begin{center}
\begin{tabular}{ ||l||c|c|c|c||c||c|c|c|c|c|| } 
 \hline
 \multicolumn{1}{||c||}{Problem} & \multicolumn{4}{c||}{Demonstrator} & Ver. & 
 \multicolumn{5}{c||}{Performance} \\
 \hline
 System Name & $\tau$ & $N$ & $Q'$ & $R'$ & $D$ & \#D & \# Itr & VT & T & St \\
 \hline
 \multirow{2}{*}{TORA} & \multirow{2}{*}{1} & \multirow{2}{*}{30} & \multirow{2}{*}{[1 1 1 1]} & \multirow{2}{*}{[1]} & 3 & 52 & 118 & 7 & 14 & \crossMark \\ 
  &  &  & & & 4 & 19 & 76 & 5 & 8 & \tick \\ 
  \hline
 \multirow{3}{*}{Inv. Pendulum} & \multirow{3}{*}{0.04} & \multirow{3}{*}{50} & \multirow{3}{*}{[10 1 1 1]} & \multirow{3}{*}{[10]} & 3 & 56 & 85 & 7 & 27 & \crossMark \\ 
  &  &  & & & 4 & 53 & 69 & 9 & 25 & \tick \\ 
  &  &  & & & 5 & 34 & 50 & 7 & 19 & \tick \\ 
  \hline
 \multirow{2}{*}{Bicycle} & \multirow{2}{*}{0.4} & \multirow{2}{*}{20} & \multirow{2}{*}{[1 1 1 1]} & \multirow{2}{*}{[1 1]} & 2 & 14 & 32 & 2 & 2 & \crossMark \\ 
  &  &  & & & 3 & 7 & 25 & 1 & 1 & \tick \\ 
  \hline
 \multirow{2}{*}{Bicycle $\times$ 2} & \multirow{2}{*}{0.4} & \multirow{2}{*}{20} & \multirow{2}{*}{[1 1 1 1 1 1 1 1]} & \multirow{2}{*}{[1 1 1 1]} & 2 & 119 & 225 & 77 & 90 & \crossMark \\ 
  &  &  & & & 3 & 30 & 81 & 43 & 46 & \tick\\ 
  \hline
 \multirow{2}{*}{Forward Flight} & \multirow{2}{*}{0.4} & \multirow{2}{*}{40} & \multirow{2}{*}{[1 1 1 1]} & \multirow{2}{*}{[1 1]} & 4 & 14 & 77 & 16 & 18 & \crossMark \\ 
  &  &  & & & 5 & 4 & 64 & 10 & 10 & \tick \\ 
  \hline
 \multirow{3}{*}{Hover Flight} & \multirow{3}{*}{0.4} & \multirow{3}{*}{40} & \multirow{3}{*}{[1 1 1 1 1 1]} & \multirow{3}{*}{[1 1]} & 2 & 57 & 147 & 12 & 40 & \crossMark \\ 
  &  &  & & & 3 & 57 & 124 & 21 & 47 & \tick \\ 
  &  &  & & & 4 & 51 & 116 & 30 & 54 & \tick \\ 
 \hline
\end{tabular}
\end{center}
\end{table*}

In the previous chapter, we discussed that two important factors govern the convergence of the search process: (i) candidate selection, and (ii) counterexample selection. To study the effect of these processes, we investigate different techniques to evaluate their performances. For candidate selection, we consider three different methods.
In the first method, a Chebyshev center of $\C_j$ is used as a candidate. In the second method, the analytic center of constraints defining $\C_j$ is the selected candidate, and redundant constraints are not dropped. Finally, in the last method, the center of
MVE inscribed in $\C_j$ yields the candidate. Also, for each of these methods, we compare the performance for two different cases: (i) a counterexample is generated without any specific property, (ii) the generated counterexample maximizes constraint violations (see Sec.~\ref{sec:cegis}). Table~\ref{tab:selection} shows the performance for each of these six cases, applied to the same set of problems. For each case, the number of iterations, the verification time and the total computation time is reported.
The results demonstrate that selecting good counterexamples would increase the convergence rate (fewer iterations) \emph{for all cases}. Nevertheless, the time it takes to generate these counterexamples (verification time) increases, and therefore, the overall performance degrades. In conclusion, while generating good counterexamples provides a better reduction in the space of candidates, it is computationally expensive, and thus, it seems to be beneficial to only rely on candidate selection for fast termination.
Table~\ref{tab:selection} also suggests that the method based on the Chebyshev center has the worst performance. Also, the MVE-based method performs better (fewer iterations) compared to the method which is based on the analytic center.
\begin{table*}[t!]
\caption{Results on different variations.}\label{tab:selection}
\textbf{Legend:} I: number of iterations, VT: computation time for verification (minutes), T: total computation time (minutes), Simple CE: any counterexample, Max CE: counterexample with maximum violation.
\begin{center}{\tiny
\begin{tabular}{ ||l||rrr|rrr||rrr|rrr||rrr|rrr||} 
 \hline
 \multirow{3}{*}{Problem} & \multicolumn{6}{c||}{Chebyshev Center} & \multicolumn{6}{c||}{Analytic Center} & \multicolumn{6}{c||}{MVE Center} \\
 \cline{2-19}
 & \multicolumn{3}{c|}{Simple CE} & \multicolumn{3}{c||}{Max CE}  & \multicolumn{3}{c|}{Simple CE} & \multicolumn{3}{c||}{Max CE}  & \multicolumn{3}{c|}{Simple CE} & \multicolumn{3}{c||}{Max CE} \\
 \cline{2-19}
  & I & VT & T & I & VT & T & I & VT & T & I & VT & T & I & VT & T & I & VT & T  \\
 \hline
 TORA & 185 & 7 & 10 & 52 & 12 & 15 & 95 & 5 & 9 & 36 & 9 & 11 & 76 & 5 & 8 & 36 & 12 & 14  \\
 Inverted Pend. & 163 & 10 & 23 & 85 & 22 & 30 & 57 & 8 & 20 & 51 & 22 & 32 & 50 & 7 & 19 & 35 & 18 & 25  \\
 Bicycle & 99 & 3 & 3 & 40 & 5 & 5 & 31 & 2 & 2 & 20 & 3 & 3 & 25 & 1 & 2 & 15 & 3 & 3  \\
 Bicycle $\times$ 2 & 759 & 121 & 127 & 438 & 244 & 246 & 96 & 47 & 50 & 77 & 141 & 143 & 81 & 43 & 46 & 66 & 132 & 133  \\
 Forward Flight & 676 & 20 & 21 & 34 & 30 & 31 & 113 & 15 & 16 & 21 & 18 & 19 & 64 & 10 & 10 & 16 & 16 & 16  \\
 Hover Flight & 499 & 65 & 90 & 196 & 113 & 127 & 146 & 36 & 67 & 90 & 92 & 109 & 116 & 30 & 54 & 75 & 69 & 82  \\
 \hline
 \end{tabular}
} \end{center}
 \end{table*}
 

%% file: eval/tracking.tex
\section{Physical Experiments}
In this section, we address ``RWS with reference tracking" problems for a bicycle model discussed in Example~\ref{ex:obstacle-car}. In particular, in addition to sets $\hat{I}$, $\hat{G}$, and $\hat{S}$, a reference trajectory over a finite interval $t \in [0, T]$ is provided as input: $\sigma_r :\ (\vx_r(t),\vu_r(t))$, wherein $\sigma_r$ is a valid trace. We also assume $\hat{S}$ depends on $t$ : $\hat{S}(t)$ represents the safe set at time $t$.
Next, we replace the time variable $t$ with $\theta$ to parameterize the trajectory with $\theta$, which is controllable with input $u_0$ ($\dot{\theta} = u_0$). Also, for $u_0 \in U_0$, we set $U_0 :\ [-0.9, 9]$. We solve all tracking problems using control funnel functions. However, we use a ``body fixed frame'', wherein the state of the vehicle is given by  $\vz^t :\ [\theta, \vx_R^t]$, and $\vx_R :\ [\alpha_R, x_R, y_R, v_R]^t$. The state variables in the inertial frame $\vx(t) :\ [\alpha(t), x(t), y(t), v(t)]^t$ are written in terms of $\vx_R$ as follows:
\[
\left[ \begin{array}{c}\alpha_R(t) + \alpha_r(\theta(t)) \\ \cos(\alpha_r(\theta(t))) x_R(t) - \sin(\alpha_r(\theta(t))) y_R(t) + x_r(\theta(t)) \\ \sin(\alpha_r(\theta(t))) x_R(t) + \cos(\alpha_r(\theta(t))) y_R(t) + y_r(\theta(t)) \\ v_R(t) + v_r(\theta(t)) \end{array} \right]\,.
\]
In this frame, $y_R$ axis is always aligned to axis of the vehicle in the reference trajectory. We now use $I$, $G$, and $S(\theta)$ to represent sets $\hat{I}$, $\hat{G}$, and $\hat{S}(\theta)$ in terms of $\vx_R$.
We observe that the change of coordinates allows for accurate low order polynomial approximations. Also, our experimental results suggest that the learning framework succeeds in finding a control funnel function of lower degree over the new coordinates when compared to the inertial frame.  For demonstration, we use a nonlinear MPC with $\tau = 0.05$, $N = 24$, $Q' = [1 \ 1 \ 1 \ 0.01]$, and $R' = [0.1 \ 1 \ 1]$ (see Section~\ref{sec:eval-demonstration}).
As an alternative to \emph{Path-Following based Control Funnel} (PF-CF) we compare with \emph{Trajectory Tracking based Control Funnel} (TT-CF) obtained by setting $\dot{\theta} = 1$, and eliminating the control input $u_0$.

\paragraph{Parkour Car:}
\begin{figure}
\vspace{0.2cm}
\centering
    \includegraphics[width=0.4\columnwidth]{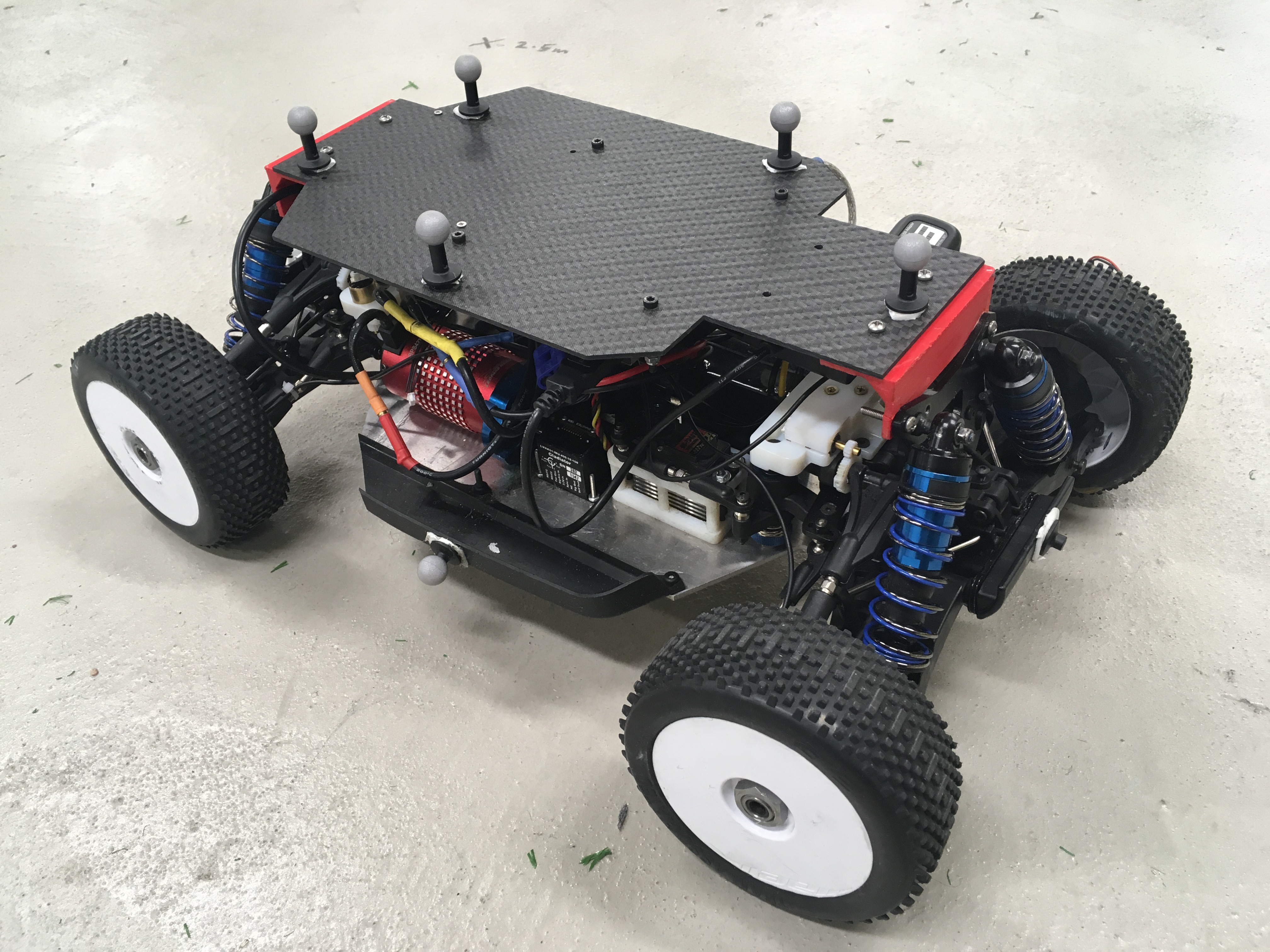}
\caption{Parkour car platform used for experiments.}\label{fig:parkourcar} 
\end{figure}
To verify the functionality of the proposed method we perform experiments on a $\frac{1}{8}^{th}$ scale, four-wheel drive vehicle platform known as Parkour car (Figure\ref{fig:parkourcar}) in a lab environment equipped with an OptiTrack motion capture system~\cite{OptiTrack}. The Parkour car has a wheelbase of $l = 34\text{cm}$ and includes an onboard computer to perform all computation on the vehicle. While in action, the main computer receives a pose update from motion capture system through a WiFi connection, after which new control action is calculated based on the synthesized control law which then gets transmitted to an ECU (Electronic Control Unit). The ECU handles signal conditioning for acceleration and steering motors on the Parkour car. One iteration of this control action calculation can be performed in less than $300\mu\text{s}$ on a single CPU core running at 3.5GHz.

We now investigate the certificate based controller for different paths.

\paragraph{Straight Path:}
\begin{figure*}[t]
\begin{center}
    \includegraphics[width=1\textwidth]{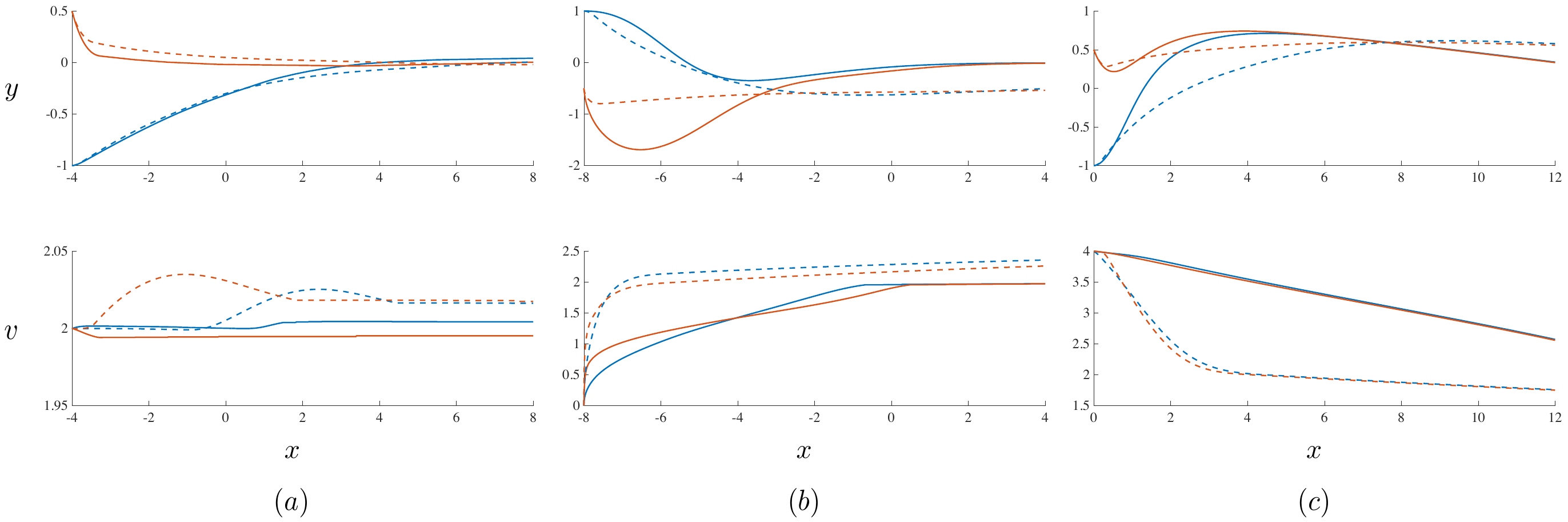}
\end{center}
Solid (dashed) lines are simulation trajectories corresponding to the PF-CF (TT-CF). Blue (red) trajectories start from the same initial condition.
\caption{Simulation results for a straight path.}\label{fig:straight} 
\end{figure*}
In the first experiment, we consider a straight path from $x = -2$ to $x = 2$ and the reference trajectory is $\vx_r(t) :\ [-\pi/2, -2+2t, 0, 2]^t$. The sets are:
\[ S(\theta) :\ [-1, 1]^3 \times [-3, 3] \, , \, I :\ \B_{0.5}(\vzero) \, , \, G :\ \B_{0.5}(\vzero) \,.\]
Then, the learning framework successfully finds a path-following based control funnel (PF-CF). However, the learning framework fails to find a trajectory tracking based control funnel (TT-CF). We note that the verification procedure is not complete and we do not claim that no TT-CF compatible with the demonstrator exists. Nevertheless, even if a TT-CF exists, the solution is fragile compared to the founded PF-CF, and proving its correctness is harder.
Next, we increase the length of the path to $8m$ (from $x=-4$ to $x=4$) to allow discovery of less robust solutions. In this case, the learning framework can find a TT-CF. Figure~\ref{fig:straight} shows simulation trajectories corresponding to the PF-CF and the TT-CF. For comparison, starting from same initial conditions, the simulation is performed until $x$ reaches $x(0) + 12$. Figure~\ref{fig:straight}(a) shows the results for initial states where the initial state is near $I$. The simulations suggest that both methods perform similarly and all trajectories converge to the path ($y$ converges to zero). The simulation time for all cases are similar and around $6s$. 
Also, the velocity of the vehicle is almost constant for both methods. Figure~\ref{fig:straight}(b) shows the results for cases when the initial states are further away from $G$ (it needs more forces/time to reach $G$). In this case, the path-following method takes a longer time to reach $x = 4$ as the speed increases smoothly.
Figure~\ref{fig:straight}(c) considers initial states that are closer to $G$. For these cases, the path-following method takes a shorter time to reach $x = 12$ as the speed decreases smoothly.
The results demonstrate that the path-following method yields a faster convergence to the reference path. Moreover, the velocity changes smoothly while the trajectory tracking method settles the target velocity immediately.

We also investigate the same problem (straight path from $x = -4$ to $x = 4$) where the velocity is more restricted:
\begin{align*}
S(\theta):\ [-1,1]^3 \times [-0.5, 0.5] \,, 
    I :\ G :\ \{[\alpha \ x \ y \ v]^t | 4\alpha^2 + 4x^2 + 4y^2 + 16v^2 \leq 1\} \,.
\end{align*}
Again, under these circumstances, learning TT-CF fails while finding PF-CF is feasible. In other words, in trajectory tracking the change of velocity is crucial for reducing the tracking error.

\paragraph{Circular Path:}
To carry out experiments on the parkour car platform and examine the behavior over long trajectories, we consider a circular path with radius $1.5m$. The vehicle moves with a constant velocity $\frac{\pi}{2} m/s$ and the reference trajectory would be $\vx_r(t) :\ [\frac{\pi}{3}t, 1.5 \cos(t), 1.5 \sin(t), \frac{\pi}{2}]^t$. For the learning process we consider a finite trajectory (moving on the path for one round) where $t \in [0, 6]$. The sets are
\begin{align*}
    S(\theta) :\ [-1, 1] \times [-3, 3], \ I :\ \B_{0.5}(\vzero), \ G :\ \B_{0.5}(\vzero)\,.
\end{align*}

We can find a path-following based control funnel (PF-CF) and a trajectory tracking based control funnel (TT-CF) when we use the ``body fixed frame." However, the learning fails in the differential frame.
Figure~\ref{fig:circular-5rounds} shows trajectories when the controller runs on the parkour car platform. Despite the uncertainties in the measurements and simple modeling, both controllers do a good job of following the reference. Figure~\ref{fig:circular-inits} shows trajectories for different initial states. Figure~\ref{fig:circular-inits}(b) suggest that the trajectory tracking method may take shortcuts to satisfy time constraints.

\begin{figure}[t]
\begin{center}
    \includegraphics[width=0.45\textwidth]{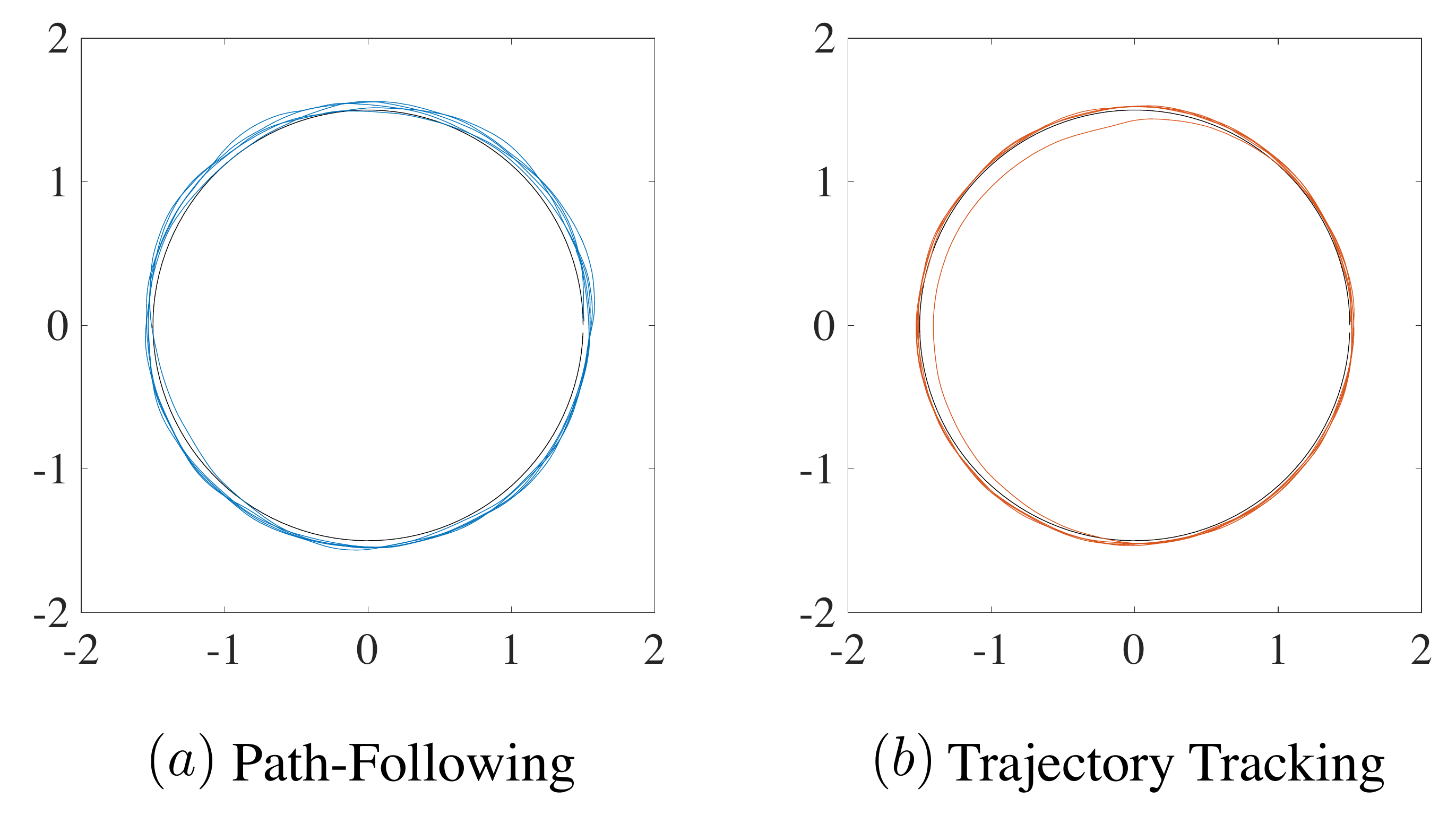}
\\ The parkour car finishes five rounds around the circle. The reference trajectory is shown in black.
\end{center}
\caption{Trajectories of the parkour car platform for the circular path.}\label{fig:circular-5rounds} 
\end{figure}

\begin{figure}[t]
\begin{center}
    \includegraphics[width=0.45\textwidth]{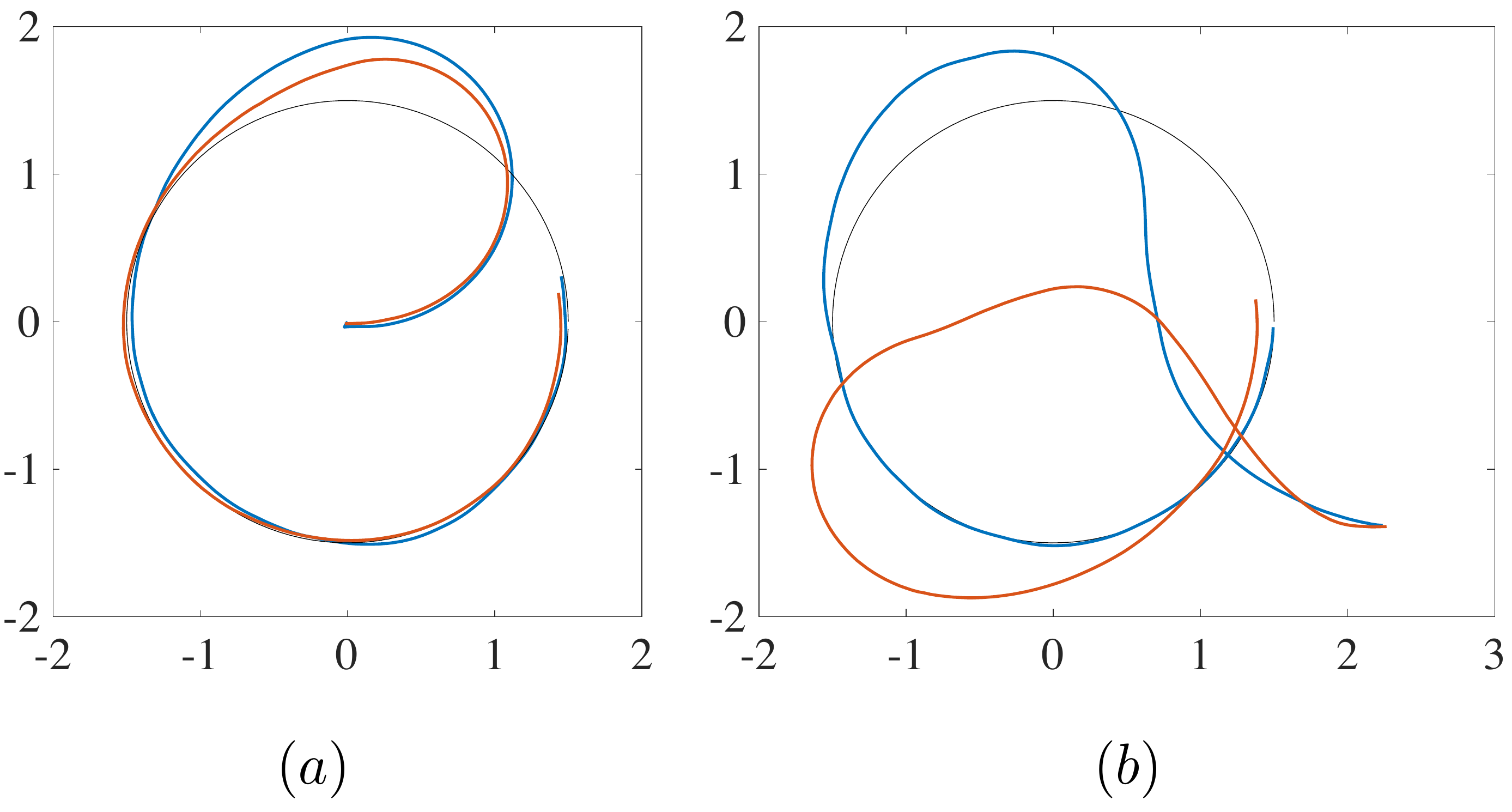}
\\ Blue (red) lines corresponds to the path-following (trajectory tracking) method. The reference trajectory is shown in black. Initial state: (a) $[-\pi/2, 0, 0, 0]$, and (b) $[\pi, 2.25, -1.4, 0]$.
\end{center}
\caption{Trajectories of the parkour car platform, for different initial states.}\label{fig:circular-inits} 
\end{figure}

We also investigated the same problem with higher reference velocity. When the reference velocity is increased to $\pi$ (from $\pi/2$), we could not find a TT-CF. Nevertheless, increasing reference velocity does not seem to affect the process of finding PF-CF, and we can discover solutions even if the reference velocity is $10\pi$.

\paragraph{Oval Path:} Following a circular path is easy as the curvature remains fixed. However, the problem is more challenging when the path is an oval. The goal is to follow an oval path $P :\ \{[x \ y]^t \ | \ \frac{y^2}{1^2} + \frac{x^2}{2^2} = 1\}$. First, a reference trajectory is generated to follow this path closely. As polynomial approximations for the reference path become more challenging, we divide the reference path into two similar parts. Then, we find a funnel for each part and make sure we can concatenate these two funnels. For the first part, the goal is to reach from $\B_{0.5}([2, 0])$ to $\B_{0.5}([-2, 0])$ going in a CCW direction and then reach from $\B_{0.5}([-2,0])$ to $\B_{0.5}([2, 0])$ again in a CCW direction. For both segments we use the following sets:
\begin{align*}
    S(\theta) :\ [-1, 1] \times [-3, 3], \ I :\ \B_{0.5}(\vzero), \ G :\ \B_{0.5}(\vzero)\,.
\end{align*}
Notice that since $G$ for the first segments fits in $I$ for the second segment, we can safely concatenate the funnels. If a trajectory tracking method is being used, the learning procedure fails to find solutions. However, the path following method yields proper control funnels. Figure~\ref{fig:oval-inputs} shows trajectories generated from our experiments using the control-funnel-based controller. The tracking is not precise when the curvature is at its maximum. We believe the main reason is input saturation for the steering, which occurs because of the imprecise model we use (Figure~\ref{fig:oval-inputs}).

\begin{figure}[t!]
\begin{center}
    \includegraphics[width=1\textwidth]{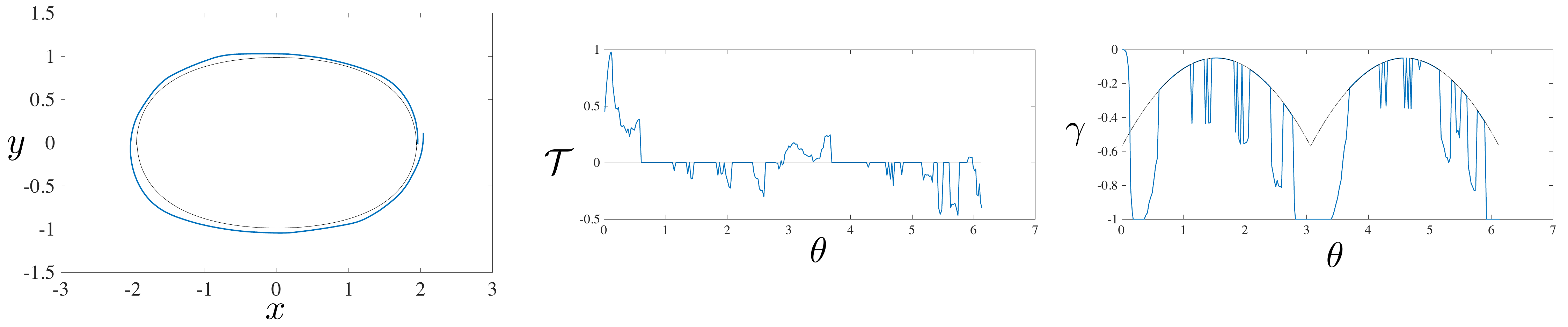}
\\ The reference trajectory is shown in black.
\end{center}
\caption{Trajectories of the parkour car platform for the oval path.}\label{fig:oval-inputs} 
\end{figure}

\paragraph{Obstacle Avoidance:} Going back to the scenario of Example~\ref{ex:obstacle-car}, we wish to find a control funnel to guarantee safety (avoiding the obstacle). Recall that having a reference trajectory, instead of defining $S(\theta)$, we simply define $\hat{S}:\ \{ \vx \ | \ ([x, y] \oplus \B_{0.25}) \cap O = \emptyset \}$. 
We were able to find a solution (only if the path-following method is being used). For the experiment, the parkour car moves toward the obstacle with different initial states and the control-funnel-based controller engages when $x$ difference between the car and the obstacle reaches $1.5m$. Figure~\ref{fig:obs-trace} shows the projection of the funnel on $x$-$y$ plain. We note that if a trajectory starts from the head of the funnel, not only its initial $x$ and $y$, but also its initial $v$ and $\alpha$ should also be inside the funnel. Figure~\ref{fig:obs-trace} (a) shows trajectories where the initial state is inside the head of the funnel. As shown, trajectories remain inside the funnel and reach the tail. However, as demonstrated in Figure~\ref{fig:obs-trace} (b), even if the trajectory starts outside of the funnel head, the whole body of the car may remain in the guaranteed region (blue region). Nevertheless, the safety is not guaranteed any longer as Figure~\ref{fig:obs-trace} (c) shows trajectories where the parkour car leaves the guaranteed region.

\begin{figure}[t!]
\begin{center}
    \includegraphics[width=1\textwidth]{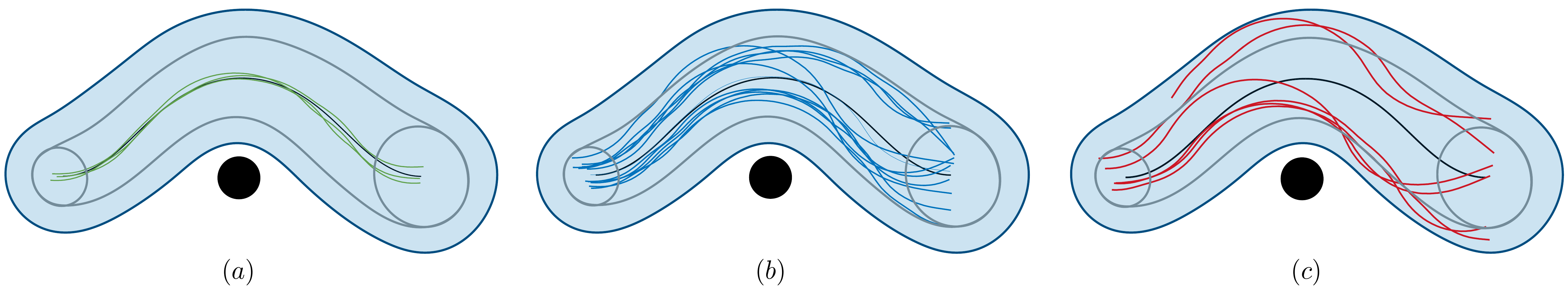}
\\ Funnel boundary is shown in gray and as long as the center of car is in the funnel, the whole body of the car remains in the blue region. (a) guaranteed traces, (b) not guaranteed but safe traces, (c) not guaranteed and unsafe traces.
\end{center}
\caption{Trajectories of the parkour car platform for the obstacle avoidance problem.}\label{fig:obs-trace}
\end{figure}

%% file: conclusion.tex
\chapter{Conclusions and Future Work}
In this thesis, we have introduced different classes of control certificates for smooth and switched feedback systems. We showed how control synthesis problems are reduced to finding control certificates. We have proposed an algorithmic framework for synthesizing these control certificates. Our framework uses different constraint solvers along with a demonstrator oracle to efficiently find a control certificate or prove certain types of control certificates do not exists.
In the rest of this section, we discuss some current limitations as well as possible extensions.

\paragraph{Extensions to Discrete-Time Systems:} Control problems on discrete-time systems have been widely studied. MPC schemes are naturally implemented over such systems, and furthermore, certificate conditions extend quite naturally. As such, our approach can be extended to discrete-time nonlinear systems defined by maps as opposed to ODEs. However, polynomial discrete systems are known to pose computational challenges: when the Lie derivative is replaced by a difference operator, the degree of the resulting polynomial can be larger.
  
 \paragraph{Extensions to Stochastic Systems:} While we have addressed disturbances using worst-case analysis, many control problems consider stochastic disturbances. Our framework is extendable to stochastic systems using stochastic certificates, namely super-martingales~\cite{prajna2007martingale}.
For such extensions, one needs proper templates (cf.~\cite{steinhardt2012}), a formal definition of a counterexample for stochastic systems, and a stochastic demonstrator (e.g.~\cite{todorov2005}).

\paragraph{Optimizing Performance Criteria:} Our framework searches for a feasible solution and stops as soon as a control certificate is discovered. An important extension to our work is finding control certificates so that the resulting controllers optimize some performance metric. 

\paragraph{Other Verifiers:} While in theory, the SDP relaxation addresses
verification problems for polynomial systems, the scalability is still an issue. There are alternative solutions to the SDP relaxation, which promise better scalability. In particular linear programming is attractive for our framework~\cite{ahmadi2014dsos,bensassi2015linear}.

For a highly nonlinear system, the degree of polynomials for the dynamics as well as basis functions get larger. For these systems, the scalability is even more challenging. In future, we wish to explore the use of falsifiers (instead of verifiers) and move towards more scalable solutions~\cite{Abbas+Others/2013/Probabilistic,AnnapureddyLFS11tacas,Donze+Maler/2010/Robust}. While falsifiers would not guarantee correctness, they can be used to find concrete counterexamples.  Furthermore, by dropping formal correctness, a falsifier can replace the verifier in the learning framework.

\paragraph{Beyond Polynomial Control Certificates:} We assumed that the template is a linear combination of some given basis functions. While this model is precise enough to for specific systems~\cite{peet2008polynomial}, there are systems for which a smooth $V$ does not exist. Moreover, for some specifications such as STL~\cite{maler2004monitoring}, the structure of a certificate gets more complicated~\cite{dimitrova2014deductive}. Nevertheless, our framework can also handle nonlinear templates such as Gaussian mixtures or feed-forward neural network models, especially if the verifier is replaced by a falsifier that can be implemented through simulations.
However, there are some serious drawbacks, including more expensive candidate generation, and weaker convergence guarantees.

\paragraph{Beyond MPC-based Demonstrations:}
As mentioned earlier, we use a black-box demonstrator. We have investigated to use MPC as they are easy to design, and can provide smooth feedbacks which in our experiments is the key to find smooth control certificates. 
However, if we employed human demonstrators (for example, an expert who operates the system), the demonstrator may include errors, and we may need to consider approaches that can reject a subset of the given demonstrations~\cite{KHANSARIZADEH2014}. Also, the demonstrations can lead to inconsistent data, wherein nearby queries are handled using different strategies by the demonstrator, leading to no single control certificate that is compatible with the given demonstrations~\cite{chernova2008learning,BREAZEAL2006385}.

In the end, we note that correct-by-construction controllers rely on models and their correctness is guaranteed only w.r.t. the models. From a practical point of view, modeling can be challenging and models may be considerably different from real systems. In this thesis, we demonstrated the applicability of correct-by-construction controllers to real systems \emph{only} through experiments. However, can we make any claim about the correctness of the real closed-loop systems? How does the model compare versus the reality? We need to investigate these issues to achieve \emph{perfectly reliable systems}.

%% file: benchmarks/benchmarks.tex
\chapter{Benchmark}\label{ch:benchmark}
The benchmark used in the experiments are examples adopted from the literature. We consider each of these systems as a switched system with RWS as the specification, where the safe set $S$ is a box, and the initial (goal) set is a ball with radius $r_I$ ($r_G$) centered at the origin.

\begin{system} 
\label{sys:harmonic}
This system is adopted from ~\cite{liberzon1999basic}.
There are two continuous variables $x$ and $y$ and the dynamics are
$\dot{x} = y \,, \dot{y} = - x + u$. 
We assume $u \in \{-1, 0, 1\}$ and instead of stability, we consider RWS with region $S:\ [-1 \ \ 1]^2$, $r_G = 0.2$ and $r_I = 0.8$.
\end{system}

\begin{system}
\label{sys:linear-ss-1}
This system is a switched system adopted from ~\cite{greco2005stability} (Example 3.1). There are two continuous variables $x$ and $y$ and five modes ($\vu_1,..., \vu_5$) the dynamics of each mode is described below
\begin{align*}
\vu_1 & \begin{cases} \dot{x} = 0.0403x+0.5689y  \\
                             \dot{y} = 0.6771x-0.2556y
       \end{cases} & 
\vu_2 & \begin{cases} \dot{x} = 0.2617x-0.2747y  \\
                             \dot{y} = 1.2134x-0.1331y
       \end{cases} &
\vu_3 & \begin{cases} \dot{x} = 1.4725x-1.2173y  \\
                             \dot{y} = 0.0557x-0.0412y
       \end{cases}\\
\vu_4 & \begin{cases} \dot{x} = -0.5217x+0.8701y  \\
                             \dot{y} = -1.4320x+0.8075y
       \end{cases}& 
\vu_5 & \begin{cases} \dot{x} = -2.1707x-1.0106y  \\
                             \dot{y} = -0.0592x+0.6145y \,.
       \end{cases} &
\end{align*}
The original specification is stability. We consider RWS with $S :\ [-1 \ \ 1]^2$, $r_I = 0.5$ and $r_G = 0.1$.
\end{system}

\begin{system}
\label{sys:dc-motor}This system (adopted from ~\cite{Mazo+Others/2010/PESSOA}--Section 7.1) is a DC motor system. There are two continuous variables $\omega$ and $i$, and input $u$ is the source voltage:
\begin{align*}
    \dot{\omega} & = - \frac{B}{J}\omega + \frac{k}{J} i \,, &
    \dot{i} & = - \frac{k}{L} \omega - \frac{R}{L} i + \frac{1}{L}u \,,
\end{align*}
where $B = 10^{-4}$, $J = 25\times 10^{-5}$, $k = 0.05$, $R = 0.5$, $L = 15\times 10^{-4}$, and $u \in \{-10, 0, 10\}$. In this example, the goal is to bring $\omega$ close to $20.0$ ($1 9.5$ to $20.5$) and $i$ close to $0$ ($-0.7$ to $0.7$). The safe set $S$ is $[-1, 30]\times[-3, 3]$. For a more challenging problem, we restrict the control input to be in range $u \in \{-3, 3\}$, target set to a ball with radius $0.5$. Since the desired point is $[\omega \ i] = [20 \ 0]$, by a change of basis, the following system is obtained:
\begin{align*}
    \dot{\omega'} & = - \frac{B}{J}(\omega' + 20) + \frac{k}{J} i \,, &
    \dot{i} & = - \frac{k}{L} (\omega' + 20) - \frac{R}{L} i + \frac{1}{L}u \,.
\end{align*}
Also, the specification is originally RS. 
Here, we just consider the RWS with $r_I = 2$ and $r_G = 0.5$.
\end{system}

\begin{system}
\label{sys:dc-dc}
This system is a DCDC converter adopted from ~\cite{mouelhi2013cosyma} with two discrete modes ($\vu_1$, $\vu_2$) and two continuous variables $i$ and $v$ ($\vx :\ [i \ v]^t$. The safe set is $[0.65,1.65]\times[4.95,5.95]$ and the goal set is $[1.1, 1.6]\times[5.4,5.9]$. The dynamics are
\begin{align*}
 \vu_1 & \begin{cases} \dot{i} = 0.0167i + 0.3333  \\ 
                             \dot{v} = - 0.0142v
       \end{cases} &
 \vu_2 & \begin{cases} \dot{i} = - 0.0183i - 0.0663v + 0.3333  \\ 
                             \dot{v} = 0.0711i - 0.0142v \,.
       \end{cases}\\
\end{align*}

The specification is RWS and we choose a new origin ($i=1.25$, $v=5.55$). Then, we set $r_G = 0.15$ to under-approximate the original goal set. For the initial region, we consider $r_I = 0.35$.
\end{system}

\begin{system}
\label{sys:tulip-2d}
This system is adopted from~\cite{nilssonincremental}. There are two continuous variables $x_1$ and $x_2$, and the controller can choose between three different modes ($\vu_1$, $\vu_2$, and $\vu_3$).
Dynamics for these modes are
\begin{align*}
 \vu_1 & \begin{cases} \dot{x_1} = - x_2 -1.5 x_1 - 0.5 x_1^3  \\ 
                             \dot{x_2} = x_1 - x_2^2 + 2
       \end{cases} &
 \vu_2 & \begin{cases} \dot{x_1} = - x_2 -1.5 x_1 - 0.5 x_1^3  \\ 
                             \dot{x_2} = x_1 - x_2
       \end{cases} \\
  \vu_3 & \begin{cases} \dot{x_1} = - x_2 -1.5 x_1 - 0.5 x_1^3 + 2 \\ 
                             \dot{x_2} = x_1 + 10\,.
       \end{cases}
\end{align*}

The safe set is $[-2, 2]\times[-1.5, 3]$ with some obstacles at corners of the safe set. For simplicity, we assume that there are no obstacles. Also, the goal set is $[-1, -0.5]\times[1.5, 2]$. By setting $x_1 = -0.75$ and $x_2 = 1.75$ as the origin, $r_G$ is defined as $0.25$. Furthermore, we consider $r_I = 1.0$.  
\end{system}

\begin{system}
\label{sys:sliding-motion-2}
This system is adopted from ~\cite{perruquetti1996lyapunov}(Example 8). There are two continuous variables $x$ and $y$, and the dynamics are
$\dot{x} = u \,, \dot{y} = y^2x$,
where $u \leq |k|$ for some constant $k$. We assume $k = 4$ and discretize the input ($u \in \{-4, 0, 4\}$). Also, in the original problem, $y$ is the output. However, we consider state feedback problem here. The specification is RWS (instead of stability) with $S :\ [-1 \ \ 1]^2$, $r_G = 0.1$, and $r_I = 0.5$.
\end{system}

\begin{system}
\label{sys:inverted-pendulum} (a) This system (adopted from ~\cite{PESSOA:Website}) is a model of inverted pendulum on a cart. There are two continuous variables $\theta$ (angular position)and $\omega$ (angular velocity), and input $u$ is the applied force to the cart.
\begin{align*}
\dot{\theta} & =\omega\,, &
\dot{\omega} & =\frac{g}{l}sin(\theta)-\frac{h}{ml^2}\omega+\frac{1}{ml}cos(\theta)u \,.
\end{align*}
, where $g = 9.8$, $h = 2$, $l = 2$, $m = 0.125$, and $u \in [-3, 3]$. The specification is \RS \ with region $S :\ \{[\theta \ \ \omega]^t | \theta \in [-1.5 \ \ 1.5], \omega \in [-1 \ \ 1]\}$ and $G = [-0.25, 0.25]^2$. We consider RWS and use $r_G = 0.25$ to under-approximate the target region and $r_I = 0.5$.

(b) We consider the same problem, except for the fact that we enlarge the safe set $S :\ \{[\theta \ \ \omega]^t | \theta \in [-1.5 \ \ 1.5], \omega \in [-4 \ \ 4]\}$ (and also increase the input range $u \in \{-15, 0, 15\}$) to make another instance for the inverted pendulum example.
\end{system}

\begin{system}
\label{sys:linear-ss-2}
The system is a linear switched system, adopted from~\cite{pettersson2001stabilization}. There are three continuous variables $x$, $y$, $z$ in this system and the dynamics for three modes ($\vu_1$, $\vu_2$, and $\vu_3$) are
\begin{align*}
\vu_1 & \begin{cases} \dot{x} = 1.8631x - 0.0053y + 0.9129z  \\ 
                             \dot{y} = 0.2681x - 6.4962y + 0.0370z \\
                             \dot{z} = 2.2497x - 6.7180y + 1.6428z
       \end{cases} &
\vu_2 & \begin{cases} \dot{x} = - 2.4311x - 5.1032y + 0.4565z  \\ 
                             \dot{y} = - 0.0869x + 0.0869y + 0.0185z \\
                             \dot{z} = 0.0369x - 5.9869y + 0.8214z
       \end{cases}\\
\vu_3 & \begin{cases} \dot{x} = 0.0372x - 0.0821y - 2.7388z  \\ 
                             \dot{y} = 0.1941x + 0.2904y - 0.1110z \\
                             \dot{z} =  - 1.0360x + 3.0486y - 4.9284z \,.
       \end{cases}
\end{align*}
The original specification is stability. Here we consider RWS with $S :\ [-1 \ \ 1]^3$, $r_I = 0.7$, and $r_G = 0.1$.
\end{system}

\begin{system}
\label{sys:linear-ss-3}
This system is a switched system adopted from ~\cite{greco2005stability} (Example 3.2). There are three continuous variables $x$, $y$, $z$, and five modes ($\vu_1,..., \vu_5$) the dynamics of each mode is described below
\begin{align*}
\vu_1 & \begin{cases} \dot{x} = 0.1764x + 0.8192y - 0.3179z  \\ 
                             \dot{y} = -1.8379x-0.2346y-0.7963z \\
                             \dot{z} = -1.5023x-1.6316y+0.6908z
       \end{cases}&
\vu_2 & \begin{cases} \dot{x} = -0.0420x-1.0286y+0.6892z  \\ 
                             \dot{y} = 0.3240x+0.0994y+1.8833z \\
                             \dot{z} = 0.5065x-0.1164y+0.3254z
       \end{cases}\\
\vu_3 & \begin{cases} \dot{x} = -0.0952x-1.7313y+0.3868z  \\ 
                             \dot{y} = 0.0312x+0.4788y+0.0540z \\
                             \dot{z} = -0.6138x-0.4478y-0.4861z
       \end{cases}&
\vu_4 & \begin{cases} \dot{x} = 0.2445x+0.1338y+1.1991z  \\ 
                             \dot{y} = 0.7183x-1.0062y-2.5773z \\
                             \dot{z} = 0.1535x+1.3065y-2.0863z
       \end{cases}\\
\vu_5 & \begin{cases} \dot{x} = -1.4132x-1.4928y-0.3459z  \\ 
                             \dot{y} = -0.5918x-0.0867y+0.9863z \\
                             \dot{z} = 0.5189x-0.0126y+0.6433z \,.
       \end{cases}
\end{align*}
The original specification is stability. However, here we consider RWS with the safe region $S :\ [-1 \ \ 1]^3$, $r_G = 0.2$, and $r_I = 0.8$.
\end{system}

\begin{system}
\label{sys:non-equilibrium-stabilization}
This system with three continuous variables and four modes is adopted from~\cite{bolzern2004quadratic} (Example 2). The dynamics are
\begin{align*}
\vu_1 & \begin{cases} \dot{x} = 4.15x - 1.06y - 6.7z + 1  \\ 
                             \dot{y} = 5.74x+4.78y-4.68z -4\\
                             \dot{z} = 26.38x-6.38y-8.29z+1
       \end{cases} &
\vu_2 & \begin{cases} \dot{x} = -3.2x -7.6y -2z +4  \\ 
                             \dot{y} = 0.9x + 1.2y -z -2 \\
                             \dot{z} = x + 6y +5z -1
       \end{cases}\\
\vu_3 & \begin{cases} \dot{x} = 5.75x -16.48y -2.41z -2  \\ 
                             \dot{y} = 9.51x -9.49y +19.55z +1 \\
                             \dot{z} = 16.19x + 4.64y +14.05z -1
       \end{cases} &
\vu_4 & \begin{cases} \dot{x} = -12.38x +18.42y +0.54z -1  \\ 
                             \dot{y} = -11.9x +3.24y -16.32z +2 \\
                             \dot{z} = -26.5x -8.64y -16.6z +1\,.
       \end{cases}
\end{align*}
The original specification is stability, while here we consider RWS as the specification with $S :\ [-1 \ \ 1]^3$, $r_G = 0.2$, and $r_I = 0.8$.
\end{system}

\begin{system}
\label{sys:tulip-pipe-3d}
This system is a radiant system in building adopted from ~\cite{nilssonincremental}, which is a switched linear system with three continuous variables ($T_c$, $T_1$, and $T_2$) and two modes ($\vu_1$, $\vu_2$). 
The dynamics for mode $\vu_1$ is as follows:
 \begin{align*}
C_r \dot{T_c} &= K_{r,1} (T_1-T_c) + K_{r,2}(T_2-T_c) - K_w(T_w-T_c) \\ 
C_1 \dot{T_1} &= K_{r,1}(T_c-T_1)+K_1(7-T_1)+K_{1,2}(T_2-T_1) + p_1 \\
C_2 \dot{T_2} &= K_{r,2}(T_c-T_2)+K_2(7-T_2)+K_{2,1}(T_1-T_2) + p_2 \,,
\end{align*}
where $C_1 = C_2 = 2000$, $C_r = 3500$, $K_{r,1} = K_{r,2} = 7.8740$, $K_1 = K_2 = 0.4651$, $K_w - 16.6667$, $K_{1,2} = K_{2,1} = 5.5556$, $T_w = 18$, $T_a = 31$, and $p_1 = p_2 = 12.8$. For mode , $\vu_2$, $K_{w}$ is zero. These parameters are taken from~\cite{nghiem2013event}. The original specification is \RS, with region $[20, 28]^3$ and target set $[21, 27]\times[22, 25]$. Here we just consider RWS with a smaller target region $[23, 25]\times[22, 24]^2$. More precisely, we consider state $T_c = 24$, $T_1 = T_2 = 23$ as the new origin and $r_G = 1$. Also $r_I = 3$. Additionally, since the changes in dynamics are slow, we enlarge the derivatives $1000$ times. Furthermore, the problem scaled down ten times to avoid numerical issues.
\end{system}

\begin{system}
\label{sys:lorenz}
This system is adopted from ~\cite{saat2011nonlinear}. There are three continuous variables $x$, $y$, $z$, and the dynamics are
\begin{align*}
    \dot{x} &= -10x + 10y + u \,, &
    \dot{y} &=  28x - y -xz \,, &
    \dot{z} &=  xy - 2.6667z \,.
\end{align*}
The original specification is stability with output feedback. Here we consider RWS with state feedback and region $S :\ [-4, 4]^3$, and $r_G = 1$. Also, we limit $u \in \{-100, 0, 100\}$. The initial set is defined by $r_I = 3$. Also, the problem scaled down ten times times to avoid numerical issues.
\end{system}

\begin{system}
\label{sys:nonholonomic}
This system is a uni-cycle example, adopted from~\cite{liberzon2012switching}. There are three variables $x$, $y$, and $z$, with two control input $u$ and $v$. The dynamics are
\begin{align*}
    \dot{x} & = u \,, &
    \dot{y} & = v \,, &
    \dot{z} & = xv - yu \,.
\end{align*}
The original goal is stability. Here, RWS is considered instead with $S :\ [-1, 1]^3$, $r_I = 0.5$, and $r_G = 0.2$.
\end{system}

\begin{system}
\label{sys:LQR}
The original system is a switched control system with continuous input $u$ from~\cite{zhang2009exponential} (Example 7.2). There are four variables ($w$, $x$ ,$y$, and $z$) and four original modes. The dynamics are 
\begin{small}
\begin{align*}
 \vu_1 & \begin{cases} \dot{w} &= -0.693w   -1.099x    +2.197y    +3.296z -7.820u  \\ 
                     \dot{x} &= -1.792x    +2.197y    +4.394z   -8.735u \\
                     \dot{y} &= -1.097x    +1.504y    +2.197z   -2.746u\\
                     \dot{z} &= 0.406z    +3.244u
       \end{cases}\\
 \vu_2 & \begin{cases} \dot{w} &= -1.792w   -1.099x    +2.197y    +1.099z    +6.696u  \\ 
                     \dot{x} &= 0.406x   -2.197y     +4.734u \\
                     \dot{y} &= -0.693y    +2.773u\\
                     \dot{z} &= -2.197w   -1.099x    +2.197y    +1.504z    
                     +4.263u
       \end{cases}
\end{align*}
\begin{align*}
  \vu_3 & \begin{cases} \dot{w} &= 0.406w    +0.811u  \\ 
                     \dot{x} &= 1.099w   -0.144x    +0.549y   -0.549z    +1.910u \\
                     \dot{y} &= 0.549x   -0.144y   -0.549z    +3.871u\\
                     \dot{z} &= 1.099w   -0.693z    +4.970u
       \end{cases}\\
  \vu_4 & \begin{cases} \dot{w} &= -0.693w    +2.000x    +1.863u  \\ 
                     \dot{x} &= -0.693x    +4.159u \\
                     \dot{y} &= -0.693y    +2.773u\\
                     \dot{z} &= 4.000x   -4.000y   -0.693z   -1.069u \,.
       \end{cases}
\end{align*}
\end{small}
, where $u \in \{-1, 0, 1\}$. Instead of stability, we consider RWS with $S :\ [-1, 1]^4$, $r_G = 0.2$, and $r_I = 0.5$.
\end{system}

\begin{system}
\label{sys:heater}
The goal of this benchmark is to keep $R$ rooms warm, given some limited number of heaters.

The first three instances are adopted from~\cite{mouelhi2013cosyma}. Temperature of each room $i$ is shown with $t_i$. For mode, $\vu_0$ the heater is off, and the dynamics are
\begin{align*}
    \dot{t_i} = 0.01 (- 10.5t_{i} + 5t_{(i+1)\%R} + 5t_{(i-1)\%R} + 5)\,.
\end{align*}

If the heater is on in room $i$, the dynamics for room $i$ changes as follows
\begin{equation*}
    \dot{t_i} = 0.01 (- 11.5t_{i} + 5t_{(i+1)\%R} + 5t_{(i-1)\%R} + 55)\,.
\end{equation*}

(a) In this instance, $R = 3$ and there is one heater which can be off or in one of the rooms.

(b) In this instance, $R = 4$ and there is one heater which can be off or in one of the rooms.

(c) In this instance, $R = 5$ and there is one heater which can be off or in one of the rooms. 

(d) In this instance, $R = 6$ and there are two heaters which can be off, or they are both on at the same time. Also, if the first heater is in room $i$, the other heater must be in room $(i+3)\%R$.

(e) In this instance, $R = 9$ and there are three heaters which can be off, or they are all on at the same time. Also, if the first heater is in room $i$, the other two heaters must be in room $(i+3)\%R$ and $(i+6)\%R$.

The original specification is safety with safe set $[20, 22]^R$. Here we consider RWS with target set $[20, 22]^R$, and safe set $[17, 25]^R$. To under-approximate the target region, we consider $t_i = 21$ (for all $i$) to be the origin and $r_G = 1$. Also, $r_I = 3$.
\end{system}